\documentclass[acmsmall,10pt,nonacm,screen]{acmart}
\acmJournal{PACMPL}
\acmVolume{1}
\acmNumber{CONF} 
\acmYear{2026}
\acmMonth{1}
\acmDOI{} 

\setcopyright{none}

\usepackage[shortlabels]{enumitem}
\usepackage{trimclip}
\usepackage{graphicx}
\usepackage{float}
\usepackage{pgfplots}
\usepackage{tikz}
\usepackage{tikz-cd}
\usetikzlibrary{
  positioning,
  automata,
  shapes.geometric,
  arrows,
  arrows.meta,
  trees,
  backgrounds
}
\usepackage{subcaption}
\usepackage{mathpartir}
\usepackage{xcolor}
\usepackage{listings}
\usepackage{tabularray}
\usepackage{array}   
\usepackage{mathtools}
\usepackage[scaled=1.0]{beramono}
\usepackage{bm}
\usepackage[utf8x]{inputenc}
\usepackage[page,header]{appendix}
\usepackage{cleveref}
\usepackage{soul}
\usepackage{titletoc}
\newcommand\DoToC{%
  \startcontents
  \printcontents{}{1}{\textbf{Contents}\vskip3pt\hrule\vskip5pt}
  \vskip3pt\hrule\vskip5pt
}

\newcolumntype{L}{>{$}l<{$}} 
\newcolumntype{C}{>{$}c<{$}} 
\newcolumntype{R}{>{$}r<{$}} 

\tikzstyle{startstop} = [rectangle, minimum width=0.6cm, minimum height=0.8cm, text centered, draw=black]
\tikzstyle{process}   = [rectangle, minimum width=0.6cm, minimum height=0.8cm, text centered, text width=2.2cm, draw=black]
\tikzstyle{arrow}     = [thick,->,>=stealth]

\definecolor{codegreen}{rgb}{0,0.6,0}
\definecolor{eminence}{RGB}{108,48,130}

\lstdefinelanguage{tll}{
  emph={program, logical, theorem, def, inductive, where},
  morekeywords={let, let*, in, match, as, with, if, then, else, fork, fun, ln, fn},
  extendedchars=true, 
  alsoletter=*,
  morecomment=[l]{--},
  emphstyle=\color{blue},
  keywordstyle=\color{eminence},
  commentstyle=\color{codegreen},
  basicstyle=\footnotesize\ttfamily,
}

\lstnewenvironment{tllisting}{
\lstset{
  language=tll,
  basicstyle=\scriptsize\ttfamily,
  extendedchars=true,
  mathescape=true,
}}{}

\newcommand{\bsf}[1]{\textbf{\textsf{#1}}}

\newcommand{\mcP}{\mathcal{P}}
\newcommand{\mcQ}{\mathcal{Q}}
\newcommand{\mcE}{\mathcal{E}}
\newcommand{\mcI}{\mathcal{I}}
\newcommand{\mcJ}{\mathcal{J}}
\newcommand{\mcM}{\mathcal{M}}
\newcommand{\mcN}{\mathcal{N}}
\newcommand{\mcK}{\mathcal{K}}
\newcommand{\mcC}{\mathcal{C}}

\newcommand{\TLLC}{TLL$_{\mathcal{C}}$}

\newcommand{\flq}{\texttt{\guilsinglleft}}
\newcommand{\frq}{\texttt{\guilsinglright}}

\newcommand{\Un}{\textsf{U}}
\newcommand{\Ln}{\textsf{L}}
\newcommand{\ty}[1]{:_{#1}}
\newcommand{\tL}{:_{\Ln}}
\newcommand{\tU}{:_{\Un}}
\newcommand{\PiR}[3]{\Pi_{#1}({#2}).{#3}}
\newcommand{\PiI}[3]{\Pi_{#1}\{{#2}\}.{#3}}
\newcommand{\lamR}[3]{\lambda_{#1}({#2}).{#3}}
\newcommand{\lamI}[3]{\lambda_{#1}\{{#2}\}.{#3}}
\newcommand{\appR}[2]{{#1}\;{#2}}
\newcommand{\appI}[2]{{#1}\;\{{#2}\}}
\newcommand{\SigR}[3]{\Sigma_{#1}({#2}).{#3}}
\newcommand{\SigI}[3]{\Sigma_{#1}\{{#2}\}.{#3}}
\newcommand{\pairR}[3]{\langle{{#1},{#2}}\rangle_{#3}}
\newcommand{\pairI}[3]{\langle{\{{#1}\},{#2}}\rangle_{#3}}
\newcommand{\SigElim}[3]{\text{R}_{#1}^{\Sigma}({#2},{#3})}

\newcommand{\dotcup}{\ensuremath{\mathbin{\mathaccent\cdot\cup}}}

\newcommand{\FV}[1]{\text{FV}({#1})}
\newcommand{\FC}[1]{\text{FC}({#1})}

\newcommand{\fix}[2]{\mu({#1}).{#2}}
\newcommand{\ActR}[3]{{{#1}({#2}).\ {#3}}}
\newcommand{\ActI}[3]{{{#1}\{{#2}\}.\ {#3}}}

\newcommand{\End}{\bsf{1}}
\newcommand{\CH}[1]{{\bsf{ch}\langle{#1}\rangle}}
\newcommand{\HC}[1]{{\bsf{hc}\langle{#1}\rangle}}
\newcommand{\Proto}{\bsf{proto}}

\newcommand{\fork}[2]{\Fork\,({#1})\,\With\ {#2}}
\newcommand{\recvI}[1]{\underline{\bsf{recv}}\,{#1}}
\newcommand{\recvR}[1]{\bsf{recv}\,{#1}}
\newcommand{\sendI}[1]{\underline{\bsf{send}}\,{#1}}
\newcommand{\sendR}[1]{\bsf{send}\,{#1}}
\newcommand{\close}[1]{\bsf{close}\,{#1}}
\newcommand{\wait}[1]{\bsf{wait}\,{#1}}
\newcommand{\C}{\mathcal{C}}
\newcommand{\CM}[1]{\mathcal{C}({#1})}
\newcommand{\return}[1]{\bsf{return}\,{#1}}
\newcommand{\letin}[3]{\bsf{let}\;{#1}\Leftarrow{#2}\bsf{ in }{#3}}
\newcommand{\unit}{\textsf{unit}}
\newcommand{\ii}{{()}}
\newcommand{\Bool}{\textsf{bool}}
\newcommand{\bTrue}{\textsf{true}}
\newcommand{\bFalse}{\textsf{false}}
\newcommand{\boolElim}[4]{\text{R}_{#1}^{\Bool}({#2},{#3},{#4})}
\newcommand{\proc}[1]{{\langle{#1}\rangle}}
\newcommand{\scope}[2]{{\nu{#1}.{#2}}}
\newcommand{\Match}{\bsf{match}}
\newcommand{\With}{\bsf{with}}
\newcommand{\Inductive}{\textsf{\color{blue}inductive}}
\newcommand{\Def}{\textsf{\color{blue}def}}
\newcommand{\Theorem}{\textsf{\color{blue}theorem}}
\newcommand{\Type}{\textsf{\color{blue}type}}
\newcommand{\Let}{\textsf{\color{eminence}let}}
\newcommand{\In}{\textsf{\color{eminence}in}}
\newcommand{\Send}{\bsf{send}}
\newcommand{\SendI}{\underline{\bsf{send}}}
\newcommand{\Recv}{\bsf{recv}}
\newcommand{\RecvI}{\underline{\bsf{recv}}}
\newcommand{\Close}{\bsf{close}}
\newcommand{\Wait}{\bsf{wait}}
\newcommand{\Fork}{\bsf{fork}}
\newcommand{\Return}{\bsf{return}}
\newcommand{\Rewrite}{\bsf{rewrite}}
\newcommand{\Refl}{\textsf{refl}}
\newcommand{\ch}[2]{\textsf{ch}^{#1}\langle{#2}\rangle}

\newcommand{\arity}[1]{\text{arity}(#1)}
\newcommand{\guard}[1]{\text{guard}(#1)}

\newcommand{\HeadSim}[2]{\text{headSim}({#1},{#2})}
\newcommand{\Sim}[2]{\text{sim}({#1},{#2})}

\newcommand{\Root}[2]{\textsf{Root}({#1}, {#2})}
\newcommand{\Node}[3]{\textsf{Node}({#1}, {#2}, {#3})}
\newcommand{\Poised}[1]{\text{Poised}({#1})}
\newcommand{\Terminal}{\text{Terminal}}

\makeatletter
\newcommand*{\Leadsto}{\leadsto\joinrel\mathrel{\mathpalette\@Leadsto\relax}}
\newcommand*{\@Leadsto}[2]{%
   \clipbox{{.68\width} 0pt 0pt {-.2\height}}{$\m@th#1\leadsto$}%
}

\begin{document}

\title{Dependent Session Types for Verified Concurrent Programming}

\author{Qiancheng Fu}
\affiliation{
  \institution{Boston University}
  \city{Boston}
  \state{MA}
  \country{USA}
}
\email{qcfu@bu.edu}

\author{Hongwei Xi}
\affiliation{
  \institution{Boston University}
  \city{Boston}
  \state{MA}
  \country{USA}
}
\email{hwxi@bu.edu}

\author{Ankush Das}
\affiliation{
  \institution{Boston University}
  \city{Boston}
  \state{MA}
  \country{USA}
}
\email{ankushd@bu.edu}

\begin{abstract}
We present \TLLC{} which extends the Two-Level Linear dependent type theory
(TLL) with session-based concurrency. Equipped with Martin-L\"{o}f style
dependency, the session types of \TLLC{} allow protocols to specify
properties of communicated messages. When used in conjunction with the dependent
type machinery already present in TLL, dependent session types facilitate a
form of relational verification by relating concurrent programs with their
idealized sequential counterparts. Correctness properties proven for sequential
programs can be easily lifted to their corresponding concurrent implementations.
\TLLC{} makes session types a powerful tool for intrinsically verifying the
correctness of data structures such as queues and concurrent algorithms such as
map-reduce. To extend TLL with session types, we develop a novel formulation of
intuitionistic session type which we believe to be widely applicable for
integrating session types into other type systems beyond the context of \TLLC{}.
We study the meta-theory of our language, proving its soundness as both a term
calculus and a process calculus.
To demonstrate the practicality of \TLLC{}, we have implemented a prototype compiler
that translates \TLLC{} programs into concurrent C code, which has been extensively
evaluated.
\end{abstract}

\maketitle

\section{Introduction}\label{sec:intro}
Session types~\cite{honda93,caires10,pfenning11} are an effective typing discipline for coordinating
concurrent computation. Through type checking, processes are forced to adhere to
communication protocols prescribed by interaction devices like channels.
This allows session type systems to statically rule out communication errors similar to
how standard type systems rule out bugs for sequential programs. While (simple)
session types guarantee concurrent programs do not crash catastrophically
and remain deadlock-free,
they do not provide any support for writing programs that are semantically correct.

Consider the session-typed concurrent queue~\cite{SilvaFP17} which is a commonly encountered
data structure in the session type literature. A queue is described by the following type:
\begin{align*}
  \textsf{queue}_A := \&\{
  \textsf{ins}: A \multimap \textsf{queue}_A,
  \textsf{del}: \oplus\{\textsf{none}: \End, \textsf{some}: A \otimes \textsf{queue}_A\}
  \}
\end{align*}
The following diagram illustrates the channel topology of a client interacting with
a queue server.
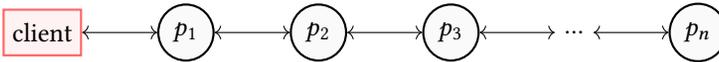
\begin{center}
\vspace{0.4em}
\begin{tikzpicture}[
squarednode/.style={rectangle, draw=red!60, fill=red!5, thick, minimum size=6mm},
roundnode/.style={circle, draw, thick, fill=black!2, minimum size=6mm},
ghostnode/.style={minimum size=5mm},
]
\node[squarednode]      (c0)       {client};
\node[roundnode]        (p1)       [right=of c0] {$p_1$};
\node[roundnode]        (p2)       [right=of p1] {$p_2$};
\node[roundnode]        (p3)       [right=of p2] {$p_3$};
\node[ghostnode]        (px)       [right=of p3] {$...$};
\node[roundnode]        (pn)       [right=of px] {$p_n$};

\draw[<->] (c0.east) -- (p1.west);
\draw[<->] (p1.east) -- (p2.west);
\draw[<->] (p2.east) -- (p3.west);
\draw[<->] (p3.east) -- (px.west);
\draw[<->] (px.east) -- (pn.west);
\end{tikzpicture}
\end{center}
Each of the $p_i$ nodes here represents a queue cell which holds a value and
these nodes are linked together by channels of type $\text{queue}_A$. As
indicated by the external choice type constructor $\&$, the first queue node $p_1$ first
receives either an $\textsf{ins}$ or a $\textsf{del}$ label from the client. In
the case of an $\textsf{ins}$ label, $p_1$ receives an element $x$ of type $A$
(indicated by $\multimap$) from the client.  The $p_1$ node then sends an
$\textsf{ins}$ label to $p_2$ and forwards $x$ to it.
This forwarding procedure repeats until the element reaches the end of the queue where a new queue cell
$p_{n+1}$ is allocated to store $x$. On the other hand, if $p_1$ receives a
$\textsf{del}$ label, the type constructor $\oplus$ requires that $p_1$ send
either $\textsf{none}$ or $\textsf{some}$.  The $\textsf{none}$ label is sent to
signify that the queue is empty and ready to terminate (indicated by $\End$).
The $\textsf{some}$ label is sent along with a value of type $A$ (indicated by $\otimes$)
which is the dequeued element. Finally, $p_1$ terminates by \emph{identifying} the two channels
connected to it, one to the client and the other to $p_2$, thus making the client communicate
with $p_2$ for future interactions.

It is clear from the example above that the session type $\textsf{queue}_A$ only lists
what operations a queue should support, but does not specify the expected behavior of
these operations. For instance, it does not specify that inserts must be performed at the tail
of the queue while deletes at the head.
In fact, any data structure that supports inserts and deletes (e.g., stack, queue, priority queue)
will satisfy the type above.
Thus, a correct implementation needs to maintain additional invariants and satisfy a more
sophisticated specification that goes beyong the session type. In fact, due to the under
specification of the $\textsf{queue}_A$ type, it is possible to implement a ``queue''
which always returns $\textsf{none}$ on $\textsf{del}$
without performing any deletion.


To address this issue, we develop \TLLC{}, a dependent session type system which
extends the Two-Level Linear dependent type theory (TLL)~\cite{fu25} with
session-based concurrency. In \TLLC{}, one could define queues through the following
dependent session type:
\begin{alignat*}{2}
  &\textsf{queue} (\textit{xs} : \textsf{list}\ A) :=\ ?(\ell : \textsf{opr}). \Match\ \ell\ \With \\
  &\qquad\mid \textsf{ins}(v) \Rightarrow \textsf{queue}(\textsf{snoc}(xs, v)) \\
  &\qquad\mid \textsf{del} \Rightarrow
    \Match\ \textit{xs}\ \With
    \ (x :: xs') \Rightarrow\ !(\textsf{sing}\ x). !(\HC{\textsf{queue}(xs')}). \End
    \mid [] \Rightarrow \End
\end{alignat*}
Here, the type $\textsf{queue}(\textit{xs})$ is parameterized by a list $\textit{xs}$
which represents the current contents of the queue. Notice that the type no longer needs
the $\oplus$ and $\&$ type constructors to describe branching behavior. Instead, it uses
type-level pattern matching to inspect the label $\ell$ received from the client.
The \textsf{opr} type which $\ell$ inhabits is defined as a simple inductive type with
two constructors:
\begin{align*}
  \Inductive\ \textsf{opr} := \textsf{ins}: A \rightarrow \textsf{opr} \mid \textsf{del}: \textsf{opr}
\end{align*}
When a queue server receives an $\textsf{ins}(v)$ value, the type of the server becomes
$\textsf{queue}(\textsf{snoc}(xs, v))$ where $\textsf{snoc}$ appends $v$ to the end of $xs$.
Conversely, when a $\textsf{del}$ label is received, the type-level pattern matching on $xs$
enforces that if the queue is non-empty (i.e. $x :: xs'$ case), then the server must send
the front element $x$ of the queue to the client (indicated by the \emph{singleton type}
$\textsf{sing}\ x$) along with the channel $\HC{\textsf{queue}(xs')}$ connecting to the remainder
of the queue. If the queue is empty (i.e. $[]$ case), then the server simply terminates.

This type eventually leads to the implementation of a convenient queue interface:
\begin{alignat*}{2}
  &\textsf{insert} &&: \forall \{xs : \textsf{list}\ A\}\;(x: A) \rightarrow \textsf{Queue}(xs) \rightarrow \textsf{Queue}(\textsf{snoc}(xs, x)) \\
  &\textsf{delete} &&: \forall \{x: A\}\;\{xs : \textsf{list}\ A\} \rightarrow
    \textsf{Queue}(x :: xs) \rightarrow \mcC (\textsf{sing}\ x \otimes \textsf{Queue}(xs)) \\
  &\textsf{free}   &&: \textsf{Queue}([]) \rightarrow \mcC(\textsf{unit})
\end{alignat*}
The \textsf{Queue} type here is an alias for the \emph{channel type} of queues
(explained later in detail) and the $\mcC$ type constructor here is the \emph{concurrency monad}
which encapsulates concurrent computations. Notice in the signature of \textsf{insert} and
\textsf{delete} that there are dependent quantifiers surrounded by braces.
These are the \emph{implicit} quantifiers of TLL which indicate that the corresponding arguments
are ``ghost'' values used for type checking and erased prior to runtime. For our purposes here,
ghost values are especially useful for \emph{relationally} specifying the expected
behaviors of queue interactions in terms of sequential list operations. For instance, the
signature of \textsf{insert} states that the queue obtained after inserting $x$ is related to
the original queue by the list operation $\textsf{snoc}$. Similarly, the signature of
\textsf{delete} states that deleting from a non-empty queue returns the front element $x$.
Even though neither of these $xs$ ghost values exists at runtime, they \emph{statically} ensure
that concurrent processes implementing these interfaces behave like actual sequential queues, i.e.,
are first-in-first-out data structures.
Therein lies \TLLC{}'s main novelty: dependent session types enable the
usage of \emph{sequential programs} as \emph{specifications} for concurrent programs,
thus allowing well-typed concurrent programs to naturally inherit the correctness properties of
their sequential counterparts.

Integrating session-based concurrency into a dependently typed functional language
poses several technical challenges that the metatheory of \TLLC{} overcomes. While prior
works~\cite{gay10,wadler12} have successfully combined \emph{classical} session
types with functional languages, it is well known that classical session types
do not easily support recursive session types~\cite{gay20}.
This is because classical session types are
defined in terms of a \emph{dual} operator which does not easily commute with
recursive type definitions. The addition of arbitrary type-level computations
through dependent types further complicates this matter.  On the other hand,
\emph{intuitionistic} session types~\cite{caires10} eschew the dual operator and
define dual \emph{interpretations} of session types based on their \emph{left} or
\emph{right} sequent rules.  Because intuitionistic session types do not rely on
a dual operator, they are able to support recursive session types without
commutativity issues. However, intuitionistic session types are often formulated
in the context of process calculi without a functional layer. To enjoy the
benefits of intuitionistic session types in a functional setting, we develop a
novel form of intuitionistic session types where we separate the notion of
\emph{protocols} from \emph{channel types}. The $\textsf{queue}(\textit{xs})$
type from before is, in actuality, a protocol whereas
$\HC{\textsf{queue}(\textit{xs})}$ is a channel type. In general, a channel type
is formed by applying the $\CH{\cdot}$ and $\HC{\cdot}$ type constructors to
protocols. These constructors provide dual interpretations to protocols,
allowing dual channels of the same protocol to be connected together. For
example, $!A. P$ would be interpreted dually as follows:
\begin{alignat*}{2}
  &\CH{!A. P}\quad &&(\textsf{send message of type } A) \\
  &\HC{!A. P}\quad &&(\textsf{receive message of type } A)
\end{alignat*}
Such channel types can be naturally included into the contexts of functional
type systems without needing to instrument the underlying language into a
sequent calculus formulation.  We believe our treatment of intuitionistic
session types is not specific to \TLLC{} and is widely applicable for
integrating intuitionistic session types with other functional languages.

In order to show that \TLLC{} ensures communication safety, we develop a process
calculus based concurrency semantics. Process configurations in the calculus are
collections of \TLLC{} programs interconnected by channels. At runtime,
individual processes are evaluated using the program semantics of base TLL. When
two processes at opposing ends (i.e. dually typed) of a channel are synchronized
and ready to communicate, the process level semantics transmits their messages
across the channel. We study the meta-theory of \TLLC{} and prove that it is
indeed sound at both the level of terms and at the level of process
configurations.

We implement a prototype compiler for compiling \TLLC{} programs into safe C
code. The compiler implements advanced language features such as dependent
pattern matching and type inference. The unique ownership property of linear types
also facilitates optimizations such as in-place programming~\cite{lorenzen23}.
All examples presented in this paper can be compiled using our
prototype compiler. The compiler source code and example programs are available
in our git repository\footnote{\url{https://anonymous.4open.science/r/ESOP26-anonymous-B66B}}.

In summary, we make the following contributions:
\begin{itemize}
  \item We extend the Two-Level Linear dependent type theory (TLL) with session
        type based concurrency, forming the language of \TLLC{}. \TLLC{} inherits the
        strengths of TLL such as Martin-L\"{o}f style linear dependent types and the
        ability to control program erasure.
  \item We develop a novel formulation of intuitionistic session types
        through a clear separation of protocols and channel types. We believe
        this formulation to be widely applicable for integrating session types into
        other functional languages.
  \item We study the meta-theoretical properties of \TLLC{}. We show that
        \TLLC{}, as a term calculus, possesses desirable properties such as confluence and
        subject reduction and, as a process calculus, guarantees communication safety.
  \item We implement a prototype compiler which compiles \TLLC{} into safe and
        efficient C code. The compiler implements additional features such as
        dependent pattern matching, type inference and in-place programming for linear types.
\end{itemize}


\section{Overview of Dependent Session Types}\label{sec:overview}
Session types in \TLLC{} are \emph{minimalistic} by design and yet surprisingly expressive
due to the presence of dependent types. Through examples, we provide an overview of
how dependent session types facilitate verified concurrent programming in \TLLC{}.

\subsection{Message Specification}\label{sec:message-specification}
An obvious, but important, use of dependent session types is the precise specification
of message properties communicated between parties. This is useful in practical network
systems where the content of messages may depend on the value of a prior request.
Consider the following protocol:
\begin{align*}
  !(\textit{sz}: \textsf{nat}).\
  ?(\textit{msg}: \textsf{bytes}).\ ?\{\textsf{sizeOf}(\textit{msg}) = \textit{sz}\}.\ \End
\end{align*}
Informally speaking, this protocol first expects a natural number \textit{sz} to be sent (the
$!$ operator) followed by receiving a byte string \textit{msg} (the $?$ operator).
With simple session types, there would be no way of specifying the relationship between \textit{sz} and
\textit{msg}. However, dependent session types allow us to express relations between messages.
Notice in the third interaction expected by the protocol, the party sending \textit{msg} must
provide a \emph{proof} that the size of \textit{msg} is indeed \textit{sz} according to
an agreed upon \textsf{sizeOf} function. Finally, the protocol terminates with $\End$ and
communication ends. Notice that the proof here, as indicated by the curly braces, is a
\emph{ghost message}: it is used for type checking and erased prior to runtime. Even though
the proof does not participate in actual communication, the necessity for the sender of
\textit{msg} to provide such a proof ensures that the protocol is followed.

This example showcases the main primitives for constructing dependent protocols in
\TLLC{}: the $!(x : A).B$ and $?(x : A).B$ \emph{protocol actions}. The syntax of these
constructs takes inspiration from binary session types~\cite{gay10,wadler12} and label-dependent
session types~\cite{ldst}, however the meaning of these constructs in \TLLC{} is
subtly different. In prior works, the $!$ marker indicates that the channel is to send
and the $?$ marker indicates that the channel is to receive. In \TLLC{}, neither marker
expresses sending or receiving per se, but rather an abstract action that needs to be
interpreted through a \emph{channel type}. Hence, the description of the messaging protocol
above is stated to be informal. To assign a precise meaning to the protocol, we need to
view it through the lens of channel types:
\begin{align*}
  &\CH{!(\textit{sz}: \textsf{nat}).\ ?(\textit{msg}: \textsf{bytes}).\ ?\{\textsf{sizeOf}(\textit{msg}) = \textit{sz}\}.\ \End} \\
  &\HC{!(\textit{sz}: \textsf{nat}).\ ?(\textit{msg}: \textsf{bytes}).\ ?\{\textsf{sizeOf}(\textit{msg}) = \textit{sz}\}.\ \End}
\end{align*}
Here, these two channel types are constructed using \emph{dual} channel type
constructors: $\CH{\cdot}$ and $\HC{\cdot}$.  The $\CH{\cdot}$ constructor
interprets $!$ as sending and $?$ as receiving while the $\HC{\cdot}$
constructor interprets $!$ as receiving and $?$ as sending. In general, dual
channel types interpret protocols in opposite ways. These constructors act just like
the duality of left and right rules for intuitionistic session types~\cite{caires10}.
Unlike intuitionistic session types which require the base type system to be
based on sequent calculus, our channel types can be integrated into the type
systems of functional languages so long as linear types are supported.

\subsection{Dependent Ghost Secrets}
Dependent ghost messages have interesting applications when it comes to message specification.
Consider the following encoding of an idealized Shannon cipher protocol:
\begin{align*}
  H(E, D) &:= \forall \{k : \mcK\}\;\{m : \mcM\} \rightarrow D(k, E(k, m)) =_{\mcM} m
\qquad\text{(correctness property)}
  \\
  \mcE(E, D) &:=\ !\{k : \mcK\}.\ !\{m : \mcM\}.\ !(c : \mcC).\ !\{H(E,D) \times (c =_\mcC E(k, m))\}.\ \End
\end{align*}
Given public encryption and decryption functions
$E: \mcK \times \mcM \rightarrow \mcC$ and
$D: \mcK \times \mcC \rightarrow \mcM$ respectively, the protocol $\mcE(E,D)$
begins by sending ghost messages: key $k$ of type $\mcK$ and message $m$ of type
$\mcM$.  Next, the ciphertext $c$ of type $\mcC$, indicated by round
parenthesis, is actually sent to the client. Finally, the last ghost message
sent is a proof object witnessing the correctness property of the
protocol: $c$ is obtained by encrypting $m$ with key $k$.  Observe that for the
overall protocol, \emph{only} ciphertext $c$ will be sent at runtime while the
other messages (secrets) are erased. The Shannon cipher protocol basically
forces communicated messages to always be encrypted and prevents the accidental
leakage of plaintext.

It is important to note that ghost messages and proof specifications, by
themselves, are \emph{not} sufficient to guaranteeing semantic security.
An adversary can simply use a different programming language and circumvent the
proof obligations imposed by \TLLC{}. However, these obligations are useful in
ensuring that honest parties correctly follow \emph{trusted} protocols to defend
against attackers. For example, in the Shannon cipher protocol above, an honest
party is required by the type system to send a ciphertext that is indeed encrypted
using the (trusted) algorithm $E$.

Another, more concrete, example of using ghost messages to specify secrets is the
Diffie-Hellman key exchange~\cite{DH76} protocol defined as follows:
\begin{align*}
  \textsf{DH}(p\ g: \textsf{int})
  :=\ & !\{a: \textsf{int}\}.\ !(A: \textsf{int}).\ !\{A = \textsf{powm}(g, a, p)\}.\\
      & ?\{b: \textsf{int}\}.\ ?(B: \textsf{int}).\ ?\{B = \textsf{powm}(g, b, p)\}.\ \End
\end{align*}
The \textsf{DH} protocol is parameterized by publicly known integers $p$ and $g$.
Without loss of generality, we refer to the message sender for the first row of the
protocol as Alice and the message sender for the second row as Bob. From Alice's
perspective, she first sends her secret value $a$ as a dependent ghost message to
initialize her half of the protocol. Next, her public value $A$ is sent as a real
message to Bob along with a proof that $A$ is correctly computed from values $p, g$ and $a$
(using modular exponentiation \textsf{powm}). At this point, Alice has finished sending
messages and waits for message from Bob to complete the key exchange. She first
``receives'' Bob's secret $b$ as a ghost message which initializes Bob's half of the
protocol. Later, Bob' public value $B$ is received as a real message along with a proof
that $B$ is correctly computed from $p, g$ and $b$. Notice that between Alice and Bob,
only the real messages $A$ and $B$ will be exchanged at runtime. The secret values
$a$ and $b$ and the correctness proofs are all ghost messages that are erased prior to
runtime. Basically, the \textsf{DH} protocol forces communication between Alice and Bob
to be encrypted and maintain secrecy.

\vspace{-0.4em}
\begin{center}
\begin{minipage}{0.45\textwidth}
\begingroup
\small
\addtolength{\jot}{-0.25em}
\begin{alignat*}{4}
  &\Def\ \textsf{Alice}\ (a\ p\ g: \textsf{int})\ (c : \CH{\textsf{DH}(p,g)}) \\
  &: \CM{\textsf{unit}} := \\
  &\quad\Let\ c \Leftarrow \Send\ c\ \{ a \}\ \In \\
  &\quad\Let\ c \Leftarrow \Send\ c\ (\textsf{powm}(g, a, p))\ \In \\
  &\quad\Let\ c \Leftarrow \Send\ c\ \{\textsf{refl}\}\ \In \\
  &\quad\Let\ \langle{\{b\}, c}\rangle \Leftarrow \Recv\ c\ \In \\
  &\quad\Let\ \langle{B, c}\rangle \Leftarrow \Recv\ c\ \In \\
  &\quad\Let\ \langle{\{\textit{pf}\}, c}\rangle \Leftarrow \Recv\ c\ \In \\
  &\quad\Close(c)
\end{alignat*}
\endgroup
\end{minipage}
\begin{minipage}{0.5\textwidth}
\begingroup
\small
\addtolength{\jot}{-0.25em}
\begin{alignat*}{4}
  &\Def\ \textsf{Bob}\ (b\ p\ g: \textsf{int})\ (c : \HC{\textsf{DH}(p,g)}) \\
  &: \CM{\textsf{unit}} := \\
  &\quad\Let\ \langle{\{a\}, c}\rangle \Leftarrow \Recv\ c\ \In \\
  &\quad\Let\ \langle{A, c}\rangle \Leftarrow \Recv\ c\ \In \\
  &\quad\Let\ \langle{\{\textit{pf}\}, c}\rangle \Leftarrow \Recv\ c\ \In \\
  &\quad\Let\ c \Leftarrow \Send\ c\ \{ b \}\ \In \\
  &\quad\Let\ c \Leftarrow \Send\ c\ (\textsf{powm}(g, b, p))\ \In \\
  &\quad\Let\ c \Leftarrow \Send\ c\ \{\textsf{refl}\}\ \In \\
  &\quad\Wait(c)
\end{alignat*}
\endgroup
\end{minipage}
\end{center}
\vspace{0.5em}

The \textsf{DH} key exchange protocol can be implemented through two simple monadic
programs \textsf{Alice} and \textsf{Bob} as shown above. The $\mcC$ type constructor
here is the concurrency monad for integrating the \emph{effect} of concurrent
communication with the \emph{pure} functional core of \TLLC{}. There are two kinds of
\textsf{send} (and respectively \textsf{recv}) operations at play here.
The first kind, indicated by $\Send\ c\ \{v\}$ is for sending a ghost message
$v$ on channel $c$. After type checking, these ghost sends are compiled to no-ops
so that they do not participate in runtime communication. The second kind, indicated by
$\Send\ c\ (v)$, is for sending a real message $v$ on channel $c$. These
real sends are compiled to actual messages in the generated code. Finally, the
\textsf{close} and \textsf{wait} operations synchronize the termination of the protocol.
Notice that the duality of channel types $\CH{\textsf{DH}(p,g)}$ and $\HC{\textsf{DH}(p,g)}$
ensures that every send in \textsf{Alice} is matched by a corresponding receive in
\textsf{Bob} and vice versa. Moreover, \textsf{Alice} and \textsf{Bob} are enforced by the
type checker to correctly carry out the key exchange.


\section{Relational Verification via Dependent Session Types}\label{sec:relational}
Earlier in the introduction section, we showed a sketch of how dependent session types
can be used for verified concurrent programming through the example of a concurrent queue.
In this section, we provide a detailed account of how we can use dependent session types
to construct a generic map-reduce system. Similarly to the queue example, we will verify
the correctness of the map-reduce system by relating it to sequential operations on trees.

\subsection{Construction of Map-Reduce}
Map-reduce is a commonly used programming model for processing large data sets in parallel.
Initially, map-reduce creates a tree of concurrently executing workers as illustrated in
\Cref{fig:map-reduce}. The client partitions the data into smaller chunks and sends them to
the leaf workers of the tree. Next, each leaf worker applies a user-specified
function $f$ to each of its received data chunks and sends the results to its parent worker.
When an internal worker receives results from its children, it combines the results using another
user-specified binary function $g$. This procedure continues until the root worker computes
the final result and sends it back to the client. Due to the fact that workers without
data dependencies can operate concurrently, the overall system can achieve significantly better
performance than sequential implementations of the same operations.
\vspace{-0.3em}
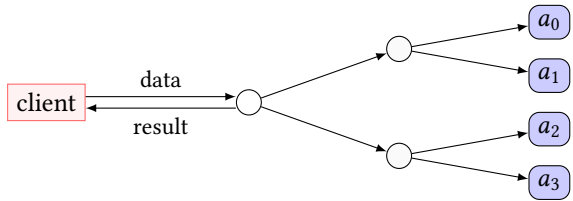
\begin{figure}[H]
\begin{tikzpicture}[
  treenode/.style = {shape=rectangle, rounded corners,
                     draw, align=center,
                     fill=blue!20},
  client/.style   = {shape=rectangle, draw=red!60, fill=red!5},
  root/.style     = {circle,draw},
  env/.style      = {treenode},
  dummy/.style    = {circle,draw,fill=black!2},
  grow = right,
  edge from parent/.style = {draw, -latex},
  sloped,
  level distance=2cm,
  level 1/.style={sibling distance=4em},
  level 2/.style={sibling distance=2em}
]
  \node [client]            (c0) {client};
  \node [root, right= 2cm of c0] (t0) {}
    child { node [dummy] {}
      child { node [env] {$a_3$} }
      child { node [env] {$a_2$} }
    }
    child { node [dummy] {}
      child { node [env] {$a_1$} }
      child { node [env] {$a_0$} }
    };
\draw[-latex,transform canvas={yshift=+0.5ex}] (c0.east) -- (t0.west) node [midway, above] {\small{data}};
\draw[latex-,transform canvas={yshift=-0.5ex}] (c0.east) -- (t0.west) node [midway, below] {\small{result}};
\end{tikzpicture}
\vspace{-0.8em}
\caption{Tree Diagram of Map-Reduce}
\label{fig:map-reduce}
\end{figure}
\vspace{-0.8em}

The first step in constructing the map-reduce system is to build a model of our desired
computation in a sequential setting. For this purpose, we define a simple binary tree
inductive type:

\vspace{-1em}
\begingroup
\small
\addtolength{\jot}{-0.25em}
\begin{alignat*}{4}
  &\Inductive\ \textsf{tree}\ (A : \Un) :=\
    \textsf{Leaf} : A \rightarrow \textsf{tree}(A) \mid\
    \textsf{Node} : \textsf{tree}(A) \rightarrow \textsf{tree}(A) \rightarrow \textsf{tree}(A)
  \\[0.4em]
  &\begin{alignedat}{4}
    &\Def\ \textsf{map} : \forall \{A\ B : \Un\}\ (f : A \rightarrow B) \rightarrow \textsf{tree}(A) \rightarrow \textsf{tree}(B) \\
    &\quad\mid\ \textsf{Leaf}\ x \Rightarrow \textsf{Leaf}\ (f\ x) \\
    &\quad\mid\ \textsf{Node}\ l\ r \Rightarrow \textsf{Node}\ (\textsf{map}\ f\ l) (\textsf{map}\ f\ r)
  \end{alignedat}
  \\[0.4em]
  &\begin{alignedat}{4}
    &\Def\ \textsf{reduce} : \forall \{A\ B : \Un\}\ (f : A \rightarrow B)\ (g : B \rightarrow B \rightarrow B) \rightarrow \textsf{tree}(A) \rightarrow B \\
    &\quad\mid\ \textsf{Leaf}\ x \Rightarrow f\ x \\
    &\quad\mid\ \textsf{Node}\ l\ r \Rightarrow g\ (\textsf{reduce}\ f\ g\ l)\ (\textsf{reduce}\ f\ g\ r)
  \end{alignedat}
\end{alignat*}
\endgroup
In this definition, the type \Un{} of $A$ is the universe of \emph{unbound}
(i.e. non-linear) types in \TLLC{}. So \textsf{tree} is parameterized by $A$
which represents the type of data stored at the leaf nodes. The \emph{sequential}
\textsf{map} and \textsf{reduce} functions for \textsf{tree} are all defined in a standard way.

To construct the concurrent map-reduce system, the protocol of map-reduce must
be able to branch depending on what operation the client requests to perform.
Unlike many prior session type systems~\cite{caires10,das20} which provide
built-in constructs (e.g. $\oplus$ and $\&$) for internal and external choice,
we implement branching protocols using just dependent protocols and type-level
pattern matching on sent or received messages. For our map-reduce system, we
define the kinds of operations that can be performed through the inductive type
\textsf{opr}:

\vspace{-1em}
\begingroup
\small
\addtolength{\jot}{-0.25em}
\begin{align*}
  \Inductive\ \textsf{opr}(A : \Un) :=\ &\textsf{Map}   : \forall \{B : \Un\}\ (f : A \rightarrow B) \rightarrow \textsf{opr}(A) \\
  \mid\ &\textsf{Reduce}: \forall \{B : \Un\}\ (f : A \rightarrow B)\ (g : B \rightarrow B \rightarrow B) \rightarrow \textsf{opr}(A) \\
  \mid\ &\textsf{Free}  : \textsf{opr}(A)
\end{align*}
\endgroup

The \textsf{opr} type has three constructors:
\begin{itemize}
  \item $\textsf{Map}\ f$ represents a map operation that applies the function
        $f : A \rightarrow B$ to each element of type $A$ and produces results of type $B$.
  \item $\textsf{Reduce}\ f\ g$ represents a reduce operation that first
        applies the function $f : A \rightarrow B$ to each element of type $A$ and then
        combines the results using the binary function $g : B \rightarrow B \rightarrow B$.
  \item $\textsf{Free}$ is the command that terminates the concurrent tree.
\end{itemize}

We are now ready to define the following \textsf{treeP} protocol 
to describe the interactions between nodes in the map-reduce tree.

\vspace{-1em}
\begingroup
\small
\addtolength{\jot}{-0.25em}
\begin{alignat*}{4}
  &\Def\ \textsf{treeP}\ (A : \Un)\ (t : \textsf{tree}\ A) :=\ ?(\ell : \textsf{opr}\ A).\\
  &\quad
    \begin{alignedat}{4}
      \Match\ \ell\ \With&\ \textsf{Map}\ \_\ f \Rightarrow\ \textsf{treeP}\ B\ (\textsf{map}\ f\ t) \\
                  \mid&\ \textsf{Reduce}\ \_\ f\ g \Rightarrow\ !(\textsf{sing}\ (\textsf{reduce}\ f\ g\ t)).\ \textsf{treeP}\ t \\
                  \mid&\ \textsf{Free} \Rightarrow \End
    \end{alignedat}
\end{alignat*}
\endgroup
For each node $n$ in the concurrent tree, it will be providing a channel of type
$\CH{\textsf{treeP}\ A\ t}$ to its parent. The parameter $t$ of type
$\textsf{tree}\ A$ represents the shape of the sub-tree rooted at $n$. The
\textsf{treeP} protocol states node $n$ will receive a message $\ell$ of type
$\textsf{opr}\ A$ from its parent.  The protocol then branches, via type-level
pattern matching on $\ell$, into three cases. If $\ell$ is of the form
$\textsf{Map}\ f$, then $n$ will continue the protocol as
$\textsf{treeP}\ B\ (\textsf{map}\ f\ t)$. Notice that the type parameter of
\textsf{treeP} is changed from $A$ to $B$ to reflect the fact that the data
stored at the leaves of the sub-tree is transformed from type $A$ to type
$B$. Furthermore, the shape of the sub-tree has also changed from $t$ to
$\textsf{map}\ f\ t$. In the second case where $\ell$ is of the form
$\textsf{Reduce}\ f\ g$, $n$ will first send the result of type
$\textsf{sing}\ (\textsf{reduce}\ f\ g\ t)$ to its parent. The type
$\textsf{sing}\ x$ is the \emph{singleton type} whose sole inhabitant is the
element $x$. After sending the result, $n$ will continue the protocol as
$\textsf{treeP}\ t$, i.e. remains unchanged. Finally, $n$ will terminate
the protocol when $\ell$ is $\textsf{Free}$.

Using the \textsf{treeP} protocol, we implement the processes that run at
each leaf of the concurrent tree. We have elided uninteresting
technical details regarding dependent pattern matching.

\vspace{-1em}
\begingroup
\small
\addtolength{\jot}{-0.25em}
\begin{align*}
  &\Def\ \textsf{leafWorker}\ \{A : \Un\}\ (x : A)\ (c : \CH{\textsf{treeP}\ A\ (\textsf{Leaf}\ x)}) : \CM{\textsf{unit}} := \\
  &\quad\Let\ \langle{\ell, c}\rangle := \Recv\ c\ \In \\
  &\quad
    \begin{alignedat}{4}
      &\Match\ \ell\ \With\\
      &\mid \textsf{Map} \Rightarrow \textsf{leafWorker}\ \{B\}\ (f\ x)\ c \\
      &\mid \textsf{Reduce} \Rightarrow \Let\ c \Leftarrow \Send\ c\ (\textsf{Just}\ (f\ x))\ \In\ \textsf{leafWorker}\ \{A\}\ x\ c  \\
      &\mid \textsf{Free} \Rightarrow \Close(c)
    \end{alignedat}
\end{align*}
\endgroup
The \textsf{leafWorker} function takes two non-ghost arguments: a data element
$x$ of type $A$ and a channel $c$ of type
$\CH{\textsf{treeP}\ A\ (\textsf{Leaf}\ x)}$. Through this channel $c$, the leaf
worker will receive requests from its parent and provide responses accordingly.
For instance, when the leaf worker receives a $\textsf{Map}\ f$ request, it will
apply $f: A \rightarrow B$ to its data element $x$ and continue as a leaf worker
with the new data element $f x$. In this case, the type parameter of
\textsf{leafWorker} has changed from $A$ to $B$ to reflect the transformation of
the data element.

To represent internal node workers we implement the following \textsf{nodeWorker}
function. This function takes (non-ghost) channels $c_l$ and $c_r$ of
types $\HC{\textsf{treeP}\ A\ l}$ and $\HC{\textsf{treeP}\ A\ r}$ for
communicating with its left and right children. Notice that the types of these
channels are indexed by ghost values $l$ and $r$ of type $\textsf{tree}\ A$
which represent the shapes of the concurrent sub-trees providing $c_l$ and
$c_r$. The \textsf{nodeWorker} communicates with its parent through the channel
$c$ whose type is indexed by the ghost value $\textsf{Node}\ l\ r$.

\vspace{-1em}
\begingroup
\small
\addtolength{\jot}{-0.25em}
\begin{align*}
  &\Def\ \textsf{nodeWorker}\ \{A : \Un\}\ \{l\ r : \textsf{tree}\ A\}\\
  &\qquad(c_l : \HC{\textsf{treeP}\ A\ l})\ (c_r : \HC{\textsf{treeP}\ A\ r})\ (c : \CH{\textsf{treeP}\ A\ (\textsf{Node}\ l\ r)}) : \CM{\textsf{unit}} := \\
  &\quad\Let\ \langle{\ell, c}\rangle := \Recv\ c\ \In \\
  &\quad
    \begin{alignedat}{4}
      &\Match\ \ell\ \With\\
      &\mid \textsf{Map}\ \_\ f \Rightarrow \\
      &\quad\Let\ c_l \Leftarrow \Send\ c_l\ (\textsf{Map}\ f)\ \In \\
      &\quad\Let\ c_r \Leftarrow \Send\ c_r\ (\textsf{Map}\ f)\ \In \\
      &\quad\Let\ c \Leftarrow \Send\ c\ (\textsf{Just}\ \textsf{unit})\ \In \\
      &\quad\textsf{nodeWorker}\ \{B\}\ \{(\textsf{map}\ f\ l)\ (\textsf{map}\ f\ r)\}\ c_l\ c_r\ c \\
      &\mid \textsf{Reduce}\ \_\ f\ g \Rightarrow \\
      &\quad\Let\ c_l \Leftarrow \Send\ c_l\ (\textsf{Reduce}\ f\ g)\ \In \\
      &\quad\Let\ c_r \Leftarrow \Send\ c_r\ (\textsf{Reduce}\ f\ g)\ \In \\
      &\quad\Let\ \langle{\textsf{Just}\ v_l, c_l}\rangle \Leftarrow \Recv\ c_l\ \In \\
      &\quad\Let\ \langle{\textsf{Just}\ v_r, c_r}\rangle \Leftarrow \Recv\ c_r\ \In \\
      &\quad\Let\ c \Leftarrow \Send\ c\ (\textsf{Just}\ (g\ v_l\ v_r))\ \In \\
      &\quad\textsf{nodeWorker}\ \{A\}\ \{l\ r\}\ c_l\ c_r\ c \\
      &\mid \textsf{Free} \Rightarrow \\
      &\quad\Let\ c_l \Leftarrow \Send\ c_l\ \textsf{Free}\ \In \\
      &\quad\Let\ c_r \Leftarrow \Send\ c_r\ \textsf{Free}\ \In \\
      &\quad\Wait(c_l); \Wait(c_r); \Close(c)
    \end{alignedat}
\end{align*}
\endgroup
Given the signature of \textsf{nodeWorker} and the definition of
the \textsf{treeP} protocol, the implementation of
\textsf{nodeWorker} is constrained to function exactly as intended. For instance,
in the case where \textsf{nodeWorker} receives a $\textsf{Map}\ f$ request from
its parent, the type of $c$ becomes $\CH{\textsf{treeP}\ B\ (\textsf{map}\ f\ (\textsf{Node}\ l\ r))}$
which simplifies to $\CH{\textsf{treeP}\ B\ (\textsf{Node}\ (\textsf{map}\ f\ l)\ (\textsf{map}\ f\ r))}$.
In other words, the type of $c$ forces the \textsf{nodeWorker} process to
recursively send the $\textsf{Map}\ f$ request to both of its children to
transform them into sub-trees of type $\HC{\textsf{treeP}\ B\ (\textsf{map}\ f\ l)}$ 
and $\HC{\textsf{treeP}\ B\ (\textsf{map}\ f\ r)}$.

\subsection{A Verified Interface for Map-Reduce}
Now that we have defined both leaf and internal node workers, we can wrap them
up into a more convenient interface as presented below.

\vspace{-1em}
\begingroup
\small
\addtolength{\jot}{-0.2em}
\begin{align*}
  &\Type\ \textsf{cTree}\ (A : \Un)\ (t : \textsf{tree}\ A) := \CM{\HC{\textsf{treeP}\ t}}
  \\[0.5em]
  &\Def\ \textsf{cLeaf}\ \{A : \Un\}\ (x : A) : \textsf{cTree}\ A\ (\textsf{Leaf}\ x) := \\
  &\quad\Fork(c : \CH{\textsf{treeP}\ A\ (\textsf{Leaf}\ x)})\ \With\ \textsf{leafWorker}\ x\ c
  \\[0.5em]
  &\Def\ \textsf{cNode}\ \{A : \Un\}\ \{l\ r : \textsf{tree}\ A\}\ (t_l : \textsf{cTree}\ A\ l)\ (t_r : \textsf{cTree}\ A\ r) : \textsf{cTree}\ (\textsf{Node}\ l\ r) := \\
  &\quad\Let\ c_l \Leftarrow t_l\ \In \\
  &\quad\Let\ c_r \Leftarrow t_r\ \In \\
  &\quad\Fork(c : \CH{\textsf{treeP}\ A\ (\textsf{Node}\ l\ r)})\ \With\ \textsf{nodeWorker}\ c_l\ c_r\ c
\end{align*}
\endgroup
The type alias \textsf{cTree} is defined to aid in the readability of the interface.
The wrapper functions \textsf{cLeaf} and \textsf{cNode} respectively create leaf
and internal node workers. This is accomplished by \emph{forking} a new process using
the \Fork{} construct of the concurrency monad. In particular, when given a channel type
$\CH{P}$, the \Fork{} construct will create a new channel and give one end of it to the caller
at type $\HC{P}$ and spawn a new process that runs the worker with the other end of the channel
at type $\CH{P}$. The duality of the channel types allows the caller and the worker to communicate.
Using these wrapper functions, one can construct a concurrent tree in virtually the same way
as one would construct a sequential tree. For example, the following code constructs a
concurrent tree with four leaf nodes containing integers $0, 1, 2$ and
$3$ respectively.
\begin{align*}
  \textsf{cNode}\ (\textsf{cNode}\ (\textsf{cLeaf}\ 0)\ (\textsf{cLeaf}\ 1))\ (\textsf{cNode}\ (\textsf{cLeaf}\ 2)\ (\textsf{cLeaf}\ 3))
\end{align*}
The type of this expression is rather verbose to write manually as it contains the full shape
of the concurrent tree. This is not a problem in practice as constant type arguments
(such as the tree shapes here) can almost always be inferred automatically by the type checker.

Finally, we implement the \textsf{cMap} and \textsf{cReduce} functions that provide the
map and reduce operations on concurrent trees. These functions are implemented by
simply sending the appropriate requests to the root worker of the concurrent tree.

\vspace{-1em}
\begingroup
\small
\addtolength{\jot}{-0.2em}
\begin{align*}
  &\Def\ \textsf{cMap}\ \{A\ B : \Un\}\ \{t : \textsf{tree}\ A\}\ (f : A \rightarrow B)\ (c : \textsf{cTree}\ A\ t) : \textsf{cTree}\ B\ (\textsf{map}\ f\ t) := \\
  &\quad\Let\ c \Leftarrow c\ \In \\
  &\quad\Let\ c \Leftarrow \Send\ c\ (\textsf{Map}\ f)\ \In \\
  &\quad\Return\ c
  \\[0.5em]
  &\Def\ \textsf{cReduce}\ \{A\ B : \Un\}\ \{t : \textsf{tree}\ A\}\ (f : A \rightarrow B)\ (g : B \rightarrow B \rightarrow B)\ (c : \textsf{cTree}\ A\ t) :\\
  &\qquad\CM{\textsf{sing}\ (\textsf{reduce}\ f\ g\ t) \otimes \textsf{cTree}\ A\ t} := \\
  &\quad\Let\ c \Leftarrow c\ \In \\
  &\quad\Let\ c \Leftarrow \Send\ c\ (\textsf{Reduce}\ f\ g)\ \In \\
  &\quad\Let\ \langle{v, c}\rangle \Leftarrow \Recv\ c \In \\
  &\quad\Return\ \langle{v, \Return\ c}\rangle
\end{align*}
\endgroup
From the type signature of \textsf{cMap}, we can see that it takes a function
$f$ and a concurrent tree of type $\textsf{cTree A\ t}$ and
returns a new concurrent tree of type $\textsf{cTree B\ (\textsf{map}\ f\ t)}$.
In other words, the type of \textsf{cMap} guarantees that the shape of the
concurrent tree is transformed in the same way as its sequential tree model under the
\textsf{map} function. Similarly, the \textsf{cReduce} takes a concurrent tree of type
$\textsf{cTree}\ A\ t$ and returns a
(linear) pair consisting of the result of type $\textsf{sing}\ (\textsf{reduce}\ f\ g\ t)$,
and the original concurrent tree. The correctness of \textsf{cReduce} is guaranteed
by the singleton type of its result: reducing a concurrent tree results in the same
value as reducing its sequential tree model.

\subsection{Concurrent Mergesort via Map-Reduce}
By properly instantiating the map-reduce interface defined previously, we can implement
more complex concurrent algorithms. Moreover, dependent session types allow us to easily 
verify the correctness of these derived concurrent algorithms relationally through their 
sequential models. As an extended example, we implement a concurrent version of the mergesort 
algorithm using the map-reduce interface and verify its correctness.

We define sequential \textsf{msort}, as a model of our concurrent implementation,
in the usual way using \textsf{split} and \textsf{merge} functions.
We will not go into further details regarding the well-founded recursion
of \textsf{msort} or the correctness of sorting as these are textbook results~\cite{cpdt,pierce10}.

\vspace{-1em}
\begingroup
\small
\addtolength{\jot}{-0.2em}
\begin{align*}
  &\Def\ \textsf{split}\ (xs : \textsf{list}\ \textsf{int}) : \textsf{list}\ \textsf{int} \times \textsf{list}\ \textsf{int} := ... \\
  &\Def\ \textsf{merge}\ (xs\ ys : \textsf{list}\ \textsf{int}) : \textsf{list}\ \textsf{int} := ...
  \\[0.5em]
  &\Def\ \textsf{msort}\ (xs : \textsf{list}\ \textsf{int}) : \textsf{list}\ \textsf{int} := \Match\ xs\ \With \\
  &\quad\mid\ \textsf{nil} \Rightarrow \textsf{nil} \\
  &\quad\mid\ x :: \textsf{nil} \Rightarrow x :: \textsf{nil} \\
  &\quad\mid\ zs \Rightarrow \Let\ \langle{xs, ys}\rangle := \textsf{split}\ zs\ \In\ \textsf{merge}\ (\textsf{msort}\ xs)\ (\textsf{msort}\ ys)
\end{align*}
\endgroup

Generally, to implement an algorithm using the map-reduce paradigm, one must first 
decompose the algorithm and data into a form that is amenable to parallelization. 
For mergesort, the input list can be recursively split into smaller sub-lists
which can be processed in parallel. To make this decomposition \emph{explicit},
we define the following \textsf{splittingTree} function that constructs a binary tree
representation of how the input list is split by the mergesort algorithm.

\vspace{-1em}
\begingroup
\small
\addtolength{\jot}{-0.2em}
\begin{align*}
  &\Def\ \textsf{splittingTree}\ (xs : \textsf{list}\ \textsf{int}) : \textsf{tree}\ (\textsf{list}\ \textsf{int}) := \Match\ xs\ \With \\
  &\quad\mid\ \textsf{nil} \Rightarrow \textsf{Leaf}\ \textsf{nil} \\
  &\quad\mid\ x :: \textsf{nil} \Rightarrow \textsf{Leaf}\ (x :: \textsf{nil}) \\
  &\quad\mid\ zs \Rightarrow \Let\ \langle{xs, ys}\rangle := \textsf{split}\ zs\ \In\ \textsf{Node}\ (\textsf{splittingTree}\ xs)\ (\textsf{splittingTree}\ ys)
\end{align*}
\endgroup

To apply map-reduce, we need to construct a concurrent representation of the
splitting tree with type $\textsf{cTree}\ (\textsf{list}\ \textsf{int})\ (\textsf{splittingTree}\ xs)$.
While it is tempting to directly convert the result of \textsf{splittingTree} into a concurrent tree
by recursively replacing \textsf{Leaf} with \textsf{cLeaf} and \textsf{Node} with \textsf{cNode},
such an approach would require traversing both the input list (to construct the splitting tree) 
and the resulting tree (to convert it into a concurrent tree). This would lead to a 
bottleneck in the performance of the overall algorithm as the traversals would be done 
sequentially without exploiting parallelism. Instead, we define the \textsf{splittingCTree} 
function that constructs the concurrent splitting tree in a concurrent manner.

\vspace{-1em}
\begingroup
\small
\addtolength{\jot}{-0.2em}
\begin{align*}
  &\Def\ \textsf{splittingCTree}\ (xs : \textsf{list}\ \textsf{int}) : 
      \CH{!(\textsf{cTree}\ (\textsf{list}\ \textsf{int})\ (\textsf{splittingTree}\ xs)).\ \End} \rightarrow \CM{\textsf{unit}} := \\
  &\quad\Match\ xs\ \With \\
  &\quad\mid\ \textsf{nil} \Rightarrow 
    \Let\ c \Leftarrow \Send\ c\ (\textsf{cLeaf}\ \textsf{nil})\ \In\ \Close(c);\ \Return\ () \\
  &\quad\mid\ x :: \textsf{nil} \Rightarrow 
    \Let\ c \Leftarrow \Send\ c\ (\textsf{cLeaf}\ (x :: \textsf{nil}))\ \In\ \Close(c);\ \Return\ () \\
  &\quad\mid\ zs \Rightarrow \\
  &\qquad\Let\ \langle{xs, ys}\rangle := \textsf{split}\ zs\ \In \\
  &\qquad\Let\ c_l \Leftarrow \Fork(c)\ \With\ \textsf{splittingCTree}\ xs\ c\ \In \\
  &\qquad\Let\ c_r \Leftarrow \Fork(c)\ \With\ \textsf{splittingCTree}\ ys\ c\ \In \\
  &\qquad ...
\end{align*}
\endgroup
The \textsf{splittingCTree} function takes an additional channel argument $c$ which is used to
send back the constructed concurrent tree to its caller. This small change allows the
recursive case to fork two new processes to construct the left and right sub-trees
in parallel. After both sub-trees have been constructed, the parent process can then
combine them into a single concurrent tree using \textsf{cNode} and send it back
to its caller. Notice that \textsf{splittingCTree} never calls the
sequential \textsf{splittingTree} function and only uses it at the type level to model
the concurrent tree being constructed. The complete implementation of
\textsf{splittingCTree} can be found in the supplementary materials but is shortened here for brevity.

Now that we have constructed a concurrent splitting tree of our input list, we can
apply the \textsf{cReduce} operation instantiated with $f := \lambda(x). x$
and $g := \textsf{merge}$ to perform merging in parallel.
This gives us an output of type
\begin{align*}
  \CM{\textsf{sing}\ (\textsf{reduce}\ (\lambda(x). x)\ \textsf{merge}\ (\textsf{splittingTree}\ xs)) \otimes \textsf{cTree}\ (\textsf{list}\ \textsf{int})\ (\textsf{splittingTree}\ xs)} 
\end{align*}
The singleton value 
$\textsf{sing}\ (\textsf{reduce}\ (\lambda(x). x)\ \textsf{merge}\ (\textsf{splittingTree}\ xs))$
returned by the monad relationally describes this series of concurrent computations using just 
sequential operations. This allows us to easily verify the correctness of our concurrent mergesort
implementation by proving the following theorem (in the internal logic of TLL) which states that 
reducing the splitting tree of a list is equivalent to performing mergesort on this list.

\vspace{-1em}
\begingroup
\small
\addtolength{\jot}{-0.2em}
\begin{align*}
  &\Theorem\ \textsf{reduceSplittingTree} : \\
  &\quad\forall (xs : \textsf{list}\ \textsf{int}) \rightarrow \textsf{reduce}\ (\lambda(x). x)\ \textsf{merge}\ (\textsf{splittingTree}\ xs) = \textsf{msort}\ xs
\end{align*}
\endgroup
Using this theorem, we can rewrite the singleton type returned by \textsf{cReduce} to
$\textsf{sing}\ (\textsf{msort}\ xs)$. In other words, the result of our concurrent
mergesort implementation is guaranteed to be exactly the same as that of the sequential
mergesort algorithm, thus completing our verification.

The full pipeline of concurrent mergesort is given in the following \textsf{cMSort} function.

\vspace{-1em}
\begingroup
\small
\addtolength{\jot}{-0.2em}
\begin{align*}
  &\Def\ \textsf{cMSort}\ (xs : \textsf{list}\ \textsf{int}) : \CM{\textsf{sing}\ (\textsf{msort}\ xs)} := \\
  &\quad\Let\ c \Leftarrow \Fork(c)\ \With\ \textsf{splittingCTree}\ xs\ c\ \In \\ 
  &\quad\Let\ \langle{\textit{ctree}, c}\rangle \Leftarrow \Recv\ c\ \In\ \Wait\ c; \\
  &\quad\Let\ \langle{v, \textit{ctree}}\rangle \Leftarrow \textsf{cReduce}\ (\lambda(x). x)\ \textsf{merge}\ \textit{ctree}\ \In \\
  &\quad\Let\ \textit{ctree} \Leftarrow \Send\ \textit{ctree}\ \textsf{Free}\ \In\ \Wait\ \textit{ctree}; \\
  &\quad\Return\ (\Rewrite[\textsf{reduceSplittingTree}\ xs]\ v)
\end{align*}
\endgroup


\section{Formal Theory of Dependent Session Types}\label{sec:theory}
\subsection{Core TLL}\label{sec:core-tll}
In this section, we give a brief summary of the Two-Level Linear dependent type theory (TLL)~\cite{fu25}. 
TLL is a dependent type theory that combines 
Martin-L\"{o}f-style dependent types~\cite{martinlof} 
with linear types~\cite{girard,wadler1990}. 
Notably, TLL supports \emph{essential linearity}~\cite{luo} through the use of
a stratified ``two-level'' typing system: the \emph{logical} level and the \emph{program} level. 
The typing judgments of the two levels are written and organized as follows:
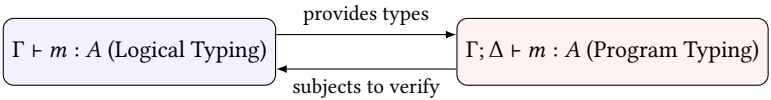
\begin{center}
\vspace{0.5em}
\begin{tikzpicture}[
    node distance=2.4cm,
    >=stealth, auto,
    every state/.style={rectangle, draw, rounded corners}
]
\node[state, fill=blue!5] (l)                {\small$\Gamma \vdash m : A\ \text{(Logical Typing)}$};
\node[state, fill=red!5]  (p) [right=of l]   {\small$\Gamma ; \Delta \vdash m : A\ \text{(Program Typing)}$};
\path[-latex,transform canvas={yshift=+1.5ex}] (l.east) edge node {\footnotesize{provides types}} (p.west);
\path[-latex,transform canvas={yshift=-1.5ex}] (p.west) edge node {\footnotesize{subjects to verify}} (l.east);
\end{tikzpicture}
\vspace{0.5em}
\end{center}

First, the \emph{logical} level is a standard dependent type system that supports unrestricted 
usage of types and terms. The primary purpose of the logical level is to provide typing rules
for types which will be used at the logical level. For example, the rules for dependent 
function type ($\Pi$-types) formation are defined at the logical level as follows:
\begin{mathpar}
  \inferrule[Explicit-Fun]
  { \Gamma \vdash A : s \\
    \Gamma, x : A \vdash B : r }
  { \Gamma \vdash \PiR{t}{x : A}{B} : t }

  \inferrule[Implicit-Fun]
  { \Gamma \vdash A : s \\
    \Gamma, x : A \vdash B : r }
  { \Gamma \vdash \PiI{t}{x : A}{B} : t }
\end{mathpar}
The symbols $s, r, t$ range over the \emph{sorts} of type universes, i.e. 
$\Un$ or $\Ln$. These sorts are used to classify types into two categories: 
unrestricted types ($A : \Un$) and linear types ($A : \Ln$).
Program level terms which inhabit unrestricted types can be freely duplicated or discarded,
while those which inhabit linear types must be used exactly once.
Note that this usage restriction is \emph{not} enforced at the logical level
as the logical level typing judgment is completely structural.
This is safe because the logical level will never be executed at runtime and 
is only used for type checking and verification. Thus, multiple uses of
a linear resource at the logical level will not lead to any runtime errors.

At the program level, the typing judgment $\Gamma ; \Delta \vdash m : A$ is used to
exclusively type \emph{terms}. In other words, no rules for forming types are defined
at the program level. All the types used in $\Gamma$, $\Delta$, $m$ and $A$ must be well-formed
according to the logical level typing judgment. This typing judgment possesses two contexts:
$\Gamma$ of all variables in scope, and $\Delta$ of all variables that are computationally relevant
in program $m$. Context $\Delta$ is crucial for enforcing linearity at the program level.
For example, consider the $\lambda$-abstraction rules:
\begin{mathpar}
  \inferrule[Explicit-Lam]
  { \Gamma, x : A ; \Delta, x :_s A \vdash m : B \\ 
    \Delta \triangleright t }
  { \Gamma ; \Delta \vdash \lamR{t}{x : A}{m} : \PiR{t}{x : A}{B} }

  \inferrule[Implicit-Lam]
  { \Gamma, x : A ; \Delta \vdash m : B \\
    \Delta \triangleright t }
  { \Gamma ; \Delta \vdash \lamI{t}{x : A}{m} : \PiI{t}{x : A}{B} }
\end{mathpar}
In \textsc{Explicit-Lam}, we can see that the bound variable $x$ is added to
both contexts $\Gamma$ and $\Delta$. This indicates that $x$ is a variable which
can be used both logically (in types and ghost values) through $\Gamma$, and
computationally (in real values) through $\Delta$. On the other hand, in the
\textsc{Implicit-Lam} rule, $x$ is only added to $\Gamma$ but not $\Delta$.
This indicates that $x$ is a ghost variable which can only be used logically.
A ubiquitous example of ghost variables is type parameters in polymorphic functions.
For instance, the polymorphic identity function can be implemented as
\begin{align*}
  \lamI{\Un}{A : \Un}{\lamR{\Un}{x : A}{x}}
\end{align*}
which has the type $\PiI{\Un}{A : \Un}{\PiR{\Un}{x : A}{A}}$.
Arguments to implicit functions are typed at the logical level, thus
allowing polymorphic functions to be instantiated with a type as an argument.
Additionally, as demonstrated in the examples of prior sections,
ghost variables also facilitate program verification by statically describing 
abstractions and invariants of program states.

In the two $\lambda$-abstraction rules above, 
the premise $\Delta \triangleright t$ is a simple side condition that states: if
$t = \Un$, then all variables in $\Delta$ must be unrestricted. In other words,
the $\lambda$-abstractions that can be applied unrestrictedly (with $t = \Un$)
are not allowed to capture linearly typed variables from $\Delta$. This is
similar to the restriction imposed on closures implementing the $\textsf{Fn}$
trait (i.e. those that can be called multiple times) in Rust~\cite{rust} where
capturing of mutable references is prohibited. If such a restriction is not
imposed, then evaluating a $\lambda$-abstraction (that captures a linear
variable) twice may lead to unsafe memory accesses such as double frees or
use-after-frees.

The application rules for both explicit and implicit functions are as follows:
\begin{mathpar}
  \inferrule[Explicit-App]
  { \Gamma ; \Delta_1 \vdash m : \PiR{t}{x : A}{B} \\ 
    \Gamma ; \Delta_2 \vdash n : A }
  { \Gamma ; \Delta_1 \dotcup \Delta_2 \vdash \appR{m}{n} : B[n/x] }

  \inferrule[Implicit-App]
  { \Gamma ; \Delta \vdash m : \PiI{t}{x : A}{B} \\ 
    \Gamma \vdash n : A }
  { \Gamma ; \Delta \vdash \appI{m}{n} : B[n/x] }
\end{mathpar}
In \textsc{Explicit-App}, the argument $n$ is a real value which must be typed
at the program level. The $\dotcup$ operator merges the two program context
$\Delta_1$ and $\Delta_2$ by contracting unrestricted variables and requiring
that linear variables be disjoint, thus preventing the sharing of linear
resources. In \textsc{Implicit-App}, the argument $n$ is a ghost value that is
typed at the logical level. Due to the fact that ghost values are erased prior
to runtime, the program context $\Delta$ in the conclusion only tracks the
computationally relevant variables used in $m$. Notice how in \textsc{Explicit-App}, 
the argument $n$ is substituted into the return type $B$. This allows types to depend 
on program level terms regardless of whether they are of linear or unrestricted types.

\paragraph{\textbf{Usage vs Uniqueness}}
Compared to other linear dependent type
theories~\cite{qtt,nothing,llf,vakar14,luo} which only enforce the linear
\emph{usage} of resources, the TLL type system prevents the \emph{sharing} of
linear resources as well. This is similar to the subtle distinction between
linear logic~\cite{girard} and bunched implications~\cite{ohearn99,ohearn03}
described by O'Hearn. 
Consider a linear function $f$, in the aforementioned dependent type theories,
of some type $A \multimap B$. When function $f$ is applied to some argument $v$
of type $A$, the argument $v$ is guaranteed to be used exactly once in the 
\emph{body} of $f$. Notice that this notion of linearity does not guarantee that
$f$ has unique access to $v$. If $v$ was obtain from some $!$-exponential or
$\omega$-quantity (the sharable quantity in graded systems \cite{qtt,nothing}),
then there may be other aliases of $v$ which can be used outside of $f$.

Wadler, in his seminal work~\cite{wadler1991}, made a similar distinction
between linearity and uniqueness in the context of functional programming,
noting that implicit uses of \emph{promotion} and \emph{dereliction} in linear
logic can lead to violations of uniqueness.  He coins the term \emph{steadfast
types} to refer to type systems that enforce both linearity and uniqueness. In
this sense, TLL is steadfast as its \emph{sort-uniqueness} property (i.e. types
uniquely inhabit either $\Un$ or $\Ln$) prohibits the implicit promotion and
dereliction of linear types, thus preventing the sharing of linear resources.
The heap semantics~\cite{turner99} of TLL shows that its programs enjoy the 
\emph{single-pointer} property which is a consequence of uniqueness at runtime.  In
the context of concurrency, the steadfast type system of TLL makes it especially
suitable for integration with session types: linear usage prevents replaying of
communication protocols and uniqueness ensures that a communication channel has
a single owner.

\subsection{Dependent Session Types of \TLLC{}}\label{sec:dependent-session-types}
In this section, we formally present the dependent session types of \TLLC{}.

\paragraph{\textbf{Basic Protocols and Channel Types}}
The intuitionistic session types of \TLLC{} are decoupled into \emph{protocols} and \emph{channel types}. 
The rules for forming protocols are as follows:
\begin{mathpar}
  \inferrule[Proto] 
  { \Gamma \vdash }
  { \Gamma \vdash \Proto : \Un }

  \inferrule[Explicit-Action]
  { \Gamma, x : A \vdash B : \Proto }
  { \Gamma \vdash \ActR{\rho}{x : A}{B} : \Proto }

  \inferrule[Implicit-Action]
  { \Gamma, x : A \vdash B : \Proto }
  { \Gamma \vdash \ActI{\rho}{x : A}{B} : \Proto }

  \inferrule[End]
  { \Gamma \vdash }
  { \Gamma \vdash \End : \Proto }

  \text{where } \rho \in \{!, ?\}
\end{mathpar}
Here, the \textsc{Proto} rule introduces the \Proto{} type which is the type of all protocols. 
Note that \Proto{} is an unrestricted type, thus protocols can be freely duplicated or discarded.
The \textsc{Explicit-Action} and \textsc{Implicit-Action} rules form dependent protocols which 
inhabit the \Proto{} type. The \textsc{End} rule marks the termination of a protocol.

Once a protocol is defined, we can form channel types using the following rules:
\begin{mathpar}
  \inferrule[ChType]
  { \Gamma \vdash A : \Proto }
  { \Gamma \vdash \CH{A} : \Ln }

  \inferrule[HcType]
  { \Gamma \vdash A : \Proto }
  { \Gamma \vdash \HC{A} : \Ln }
\end{mathpar}
Notice that the channel type constructors $\CH{\cdot}$ and $\HC{\cdot}$
lift protocols, which are unrestricted values, into linear types. This means that
channels must be used exactly once. Furthermore, as explained in the previous
section, the unique ownership of linear types in TLL ensures that only a single
entity has access to a channel at any point in time, thus preventing race conditions.

\paragraph{\textbf{Recursive Protocols}}
Recursive protocols can be formed using the $\fix{x : A}{m}$ construct:
\begin{mathpar}
  \inferrule[RecProto]
  { \Gamma, x : A \vdash m : A \\ 
    A\ \text{is an \emph{arity} ending on } \Proto \\
    x\ \text{is \emph{guarded} by protocol action in}\ m }
  { \Gamma \vdash \fix{x : A}{m} : A }
\end{mathpar}
For a $\fix{x : A}{m}$ term, we require that $A$ be an \emph{arity} ending on \Proto{}.
This prevents $\mu$ from introducing logical inconsistencies as it can only be used to
construct protocols but not proofs for arbitrary propositions. To ensure that protocols
defined through $\fix{x : A}{m}$ can be productively unfolded, recursive usages of $x$ must be
syntactically \emph{guarded} behind a protocol action in $m$. This enforces the 
\emph{contractiveness} condition for recursive session types~\cite{gay10}. Both the
arity and guardedness conditions are stable under substitution. Due to space limitations,
we present the rules of arities and guardedness in the appendix.

The difficulty of integrating recursive protocols in classical session type
systems is well documented~\cite{gay20}. The key challenge is to define a
suitable \emph{duality} operator that commutes with recursion. The following
example is due to Bernardi and Hennessy~\cite{bernardi16}. Suppose we define a 
reasonable, but naive, duality operator $(\cdot)^\bot$ which simply flips $!$
and $?$ in protocols.  For the dual of recursive protocol $\mu{X}.{?X}.X$, if we
first apply duality and then unfold the recursion, we get:
\begin{align*}
  (\mu{X}.{?X}.X)^\bot  = \mu{X}.{!X}.X =\; !(\mu{X}.{!X}.X).(\mu{X}.{!X}.X)
\end{align*} 
On the other hand, if we first unfold the recursion and then apply duality, we get:
\begin{align*}
  (\mu{X}.{?X}.X)^\bot = (?(\mu{X}.{?X}.X).(\mu{X}.{?X}.X))^\bot =\; !(\mu{X}.{?X}.X).(\mu{X}.{!X}.X)
\end{align*} 
Notice that the resulting protocols do not agree on the type of the sent message.
While solutions have been proposed to address this issue~\cite{bernardi16,bernardi14},
they do not generalize to dependent session types due to the presence of arbitrary
type-level computation. In \TLLC{}, the separation of protocols and channels types
allows us to sidestep the duality problem entirely. Suppose we define our previously
problematic recursive protocol in \TLLC{} as follows:
\begin{align*}
  T \triangleq \fix{X : \Proto}{?(\_ : X). X} =\; ?\big(\_ : \fix{X : \Proto}{?(\_ : X). X}\big).\ \fix{X : \Proto}{?(\_ : X). X}
\end{align*}
When viewed through the lens of channel type constructors $\CH{\cdot}$ and $\HC{\cdot}$,
the actions specified by the unfolded protocol are correctly dual to each other.
More specifically, a channel of type $\CH{T}$ receives a protocol of type $T$ whereas
a channel of type $\HC{T}$ sends a protocol of type $T$.

\paragraph{\textbf{Concurrency Monad}}
Concurrency is integrated into the pure functional core of TLL through a concurrency monad $\mcC$.
The basic components of the monad are given in the following rules.
\begin{mathpar}
  \inferrule[$\mcC$Type]
  { \Gamma \vdash A : s }
  { \Gamma \vdash \CM{A} : \Ln }

  \inferrule[Return]
  { \Theta ; \Gamma ; \Delta \vdash m : A }
  { \Theta ; \Gamma ; \Delta \vdash \return{m} : \CM{A} }

  \inferrule[Bind]
  { \Gamma \vdash B : s \\
    \Theta_1 ; \Gamma ; \Delta_1 \vdash m : \CM{A} \\\\
    \Theta_2 ; \Gamma, x : A ; \Delta_2, x :_r A \vdash n : \CM{B} }
  { \Theta_1 \dotcup \Theta_2 ; \Gamma ; \Delta_1 \dotcup \Delta_2 \vdash \letin{x}{m}{n} : \CM{B} }
\end{mathpar}

To reason about the communication channels that will appear at \emph{runtime}, the program level
typing judgment is extended to include a \emph{channel context} $\Theta$ which tracks the
channels used by the program. It is crucial to understand that the channel context is
largely a technical device for analyzing the type safety of \TLLC{}. Prior to runtime, the 
channel context is empty as no channels have been created. Programming
is carried out using normal variables in $\Delta$. At runtime, channels will be created and
substituted for appropriate variables in $\Delta$. It is these runtime channels that occupy the 
channel context $\Theta$ and are typed as follows:
\begin{mathpar}
  \inferrule[Channel-CH]
  { \Gamma ; \Delta \vdash \\ 
    \epsilon \vdash A : \Proto \\
    \Delta \triangleright \Un }
  { c :_\Ln \CH{A} ; \Gamma ; \Delta \vdash c : \CH{A} }

  \inferrule[Channel-HC]
  { \Gamma ; \Delta \vdash \\ 
    \epsilon \vdash A : \Proto \\
    \Delta \triangleright \Un }
  { c :_\Ln \HC{A} ; \Gamma ; \Delta \vdash c : \HC{A} }
\end{mathpar}
The protocol $A$ used in the channel types here must be \emph{closed}. 
This is because channels at runtime must follow fully concretized protocols.
The $\Gamma$ and $\Delta$ contexts are allowed to be non-empty for the
purely technical reason of facilitating proofs for renaming and substitution lemmas.

As explained in \Cref{sec:message-specification}, the protocol actions $!(x : A).B$ and $?(x : A).B$
are abstract constructs that need to be interpreted through channel types. Since $\CH{\cdot}$ and $\HC{\cdot}$
interpret protocol actions in opposite ways, we only present the typing rules for $\CH{\cdot}$ below.
\begin{mathpar}
  \inferrule[Explicit-Send-CH]
  { \Theta ; \Gamma ; \Delta \vdash m : \CH{\ActR{!}{x : A}{B}}}
  { \Theta ; \Gamma ; \Delta \vdash \Send\ m : \PiR{\Ln}{x : A}{\CM{\CH{B}}} }

  \inferrule[Explicit-Recv-CH]
  { \Theta ; \Gamma ; \Delta \vdash m : \CH{\ActR{?}{x : A}{B}}}
  { \Theta ; \Gamma ; \Delta \vdash \Recv\ m : \CM{\SigR{\Ln}{x : A}{\CH{B}}} }

  \inferrule[Implicit-Send-CH]
  { \Theta ; \Gamma ; \Delta \vdash m : \CH{\ActI{!}{x : A}{B}}}
  { \Theta ; \Gamma ; \Delta \vdash \SendI\ m : \PiI{\Ln}{x : A}{\CM{\CH{B}}} }

  \inferrule[Implicit-Recv-CH]
  { \Theta ; \Gamma ; \Delta \vdash m : \CH{\ActI{?}{x : A}{B}}}
  { \Theta ; \Gamma ; \Delta \vdash \RecvI\ m : \CM{\SigI{\Ln}{x : A}{\CH{B}}} }
\end{mathpar}

For the \textsc{Explicit-Send-CH} rule, a channel of type $\CH{!(x : A).B}$ is
applied to the \Send{} operator. This produces a function which takes a real
value $v$ of type $A$ and returns a concurrent computation of type
$\CM{\CH{B[v/x]}}$ which represents the continuation of the protocol after
sending a real value of type $A$. When this monadic value is bound by rule
\textsc{Bind} and executed at runtime, the value $v$ will be sent on channel
$m$. The dual \textsc{Explicit-Recv-HC} rule, as shown here,
\begin{mathpar}
  \inferrule[Explicit-Recv-HC]
  { \Theta ; \Gamma ; \Delta \vdash m : \HC{\ActR{!}{x : A}{B}}}
  { \Theta ; \Gamma ; \Delta \vdash \Recv\ m : \CM{\SigR{\Ln}{x : A}{\HC{B}}} }
\end{mathpar}
receives on a channel of type $\HC{!(x : A).B}$, which produces a (monadic)
dependent pair (similarly to \textsc{Explicit-Recv-CH}). The first component of
the pair is the value of type $A$ that is received, and the second component is
a channel of type $\HC{B[v/x]}$ representing the continuation of the protocol.
Notice that, due to the linearity of the $\mcC$ monad, all of the intermediate
monadic values are guaranteed to be bound by the \textsc{Bind} rule and executed.

The implicit send and receive rules are similar to their explicit counterparts,
except that they send and receive ghost values instead of real values. This 
distinction manifests by having the \SendI{} and \RecvI{} operators produce
implicit functions and implicit pairs respectively. When the implicit function
of \textsc{Implicit-Send-CH} is applied to a ghost argument using
\textsc{Implicit-App} (\Cref{sec:core-tll}), the ghost argument will be erased
prior to runtime. Similarly, the first component of the implicit pair produced
by \textsc{Implicit-Recv-CH} is also an erased ghost value. The underlying type
system of TLL ensures that these ghost values will only be used logically, thus
are safe to erase.

The last communication rules govern the creation and termination of channels:

\vspace{-1em}
\begin{small}
\begin{mathpar}
  \inferrule[Fork]
  { \Theta ; \Gamma, x : \CH{A} ; \Delta, x :_\Ln \CH{A} \vdash m : \CM{\unit} }
  { \Theta ; \Gamma ; \Delta \vdash \fork{x : \CH{A}}{m} : \CM{\HC{A}} }
  \and\hspace{-0.5em}
  \inferrule[Close]
  { \Theta; \Gamma ; \Delta \vdash c : \CH{\End} }
  { \Theta ; \Gamma ; \Delta \vdash \Close\ c : \CM{\unit} }
  \and\hspace{-1.5em}
  \inferrule[Wait]
  { \Theta; \Gamma ; \Delta \vdash c : \HC{\End} }
  { \Theta ; \Gamma ; \Delta \vdash \Wait\ c : \CM{\unit} }
\end{mathpar}
\end{small}

\noindent
\textsc{Close} and \textsc{Wait} are simple rules used to free channels whose
protocols have terminated. The \textsc{Fork} rule is used for creating a child
process which concurrently executes the monadic computation $m$. The child
process is provided with a fresh channel of type $\CH{A}$ which is bound to the
variable $x$ in $m$. Dually, the parent process obtains the channel endpoint of
type $\HC{A}$, which can be used to communicate with the spawned process. Note
that the newly spawned process $m$ is allowed to capture pre-existing channels
from $\Theta$ and program variables from $\Delta$. Compared to intuitionistic
session type systems based on the sequent
calculus~\cite{caires10,pfenning11,das20}, the $\CH{A}$ channel handed to the
child process behaves like the right-hand side of a sequent (i.e. the
\emph{provided} channel), while the $\HC{A}$ channel handed to the parent
process behaves like the left-hand side of a sequent (i.e. the \emph{consumed}
channels). Essentially, we have embedded intuitionistic session types into a
functional language without needing to reorganize the underlying type system
into a sequent calculus formulation.

\section{Semantics and Meta-Theory}\label{sec:semantics}
\subsection{Process Configurations}
In the previous section, we have presented the typing rules for \TLLC{} terms
which form individual processes. To compose multiple processes together, we
introduce the process level typing judgment $\Theta \Vdash P$ below. This judgment
formally states that a configuration of processes $P$ is well-typed under the
context $\Theta$, which tracks the channels used by the processes in $P$ at runtime.
\begin{mathpar}
  \inferrule[Expr] 
  { \Theta ; \epsilon ; \epsilon \vdash m : \CM{\unit} }
  { \Theta \Vdash \proc{m} }

  \inferrule[Par]
  { \Theta_1 \Vdash P_1 \\ 
    \Theta_2 \Vdash P_2 }
  { \Theta_1 \dotcup \Theta_2 \Vdash P_1 \mid P_2 }

  \inferrule[Scope]
  { \Theta, c :_\Ln \CH{A}, d :_\Ln \HC{A} \Vdash P }
  { \Theta \Vdash \scope{cd}{P} }
\end{mathpar}

The process configuration rules are standard. The \textsc{Expr} rule lifts
well-typed closed terms of type $\CM{\unit}$ to processes. It is important for
the term $m$ to be closed as processes in a configuration cannot rely on
external substitutions to resolve free variables. They can only communicate
through channels. In the \textsc{Par} rule, well-typed configurations $P$ and
$Q$ can be composed in parallel as long as their contexts $\Theta_1$ and
$\Theta_2$ can be combined. The \textsc{Scope} rule allows two dual channels
to be connected together, allowing processes holding channels $c$ and $d$ to communicate.

The structural congruence of process configurations is defined as the least
congruence relation generated by the following standard rules:
\begin{mathpar}
  P \mid Q \equiv Q \mid P 

  O \mid (P \mid Q) \equiv (O \mid P) \mid Q

  P \mid \proc{\return{\ii}} \equiv P
  \\
  \scope{cd}{P} \mid Q \equiv \scope{cd}{(P \mid Q)}

  \scope{cd}{P} \equiv \scope{dc}{P}

  \scope{cd}{\scope{c'd'}{P}} \equiv \scope{c'd'}{\scope{cd}{P}}
\end{mathpar}
Structural congruence states that parallel composition is commutative and
associative and compatible with channel scoping. Processes which terminate
with the unit value $\ii$ can be removed from a configuration.
Intuitively, two structurally congruent configurations should be
considered equivalent regarding their communication behavior.

\subsection{Semantics}
\paragraph{\textbf{Term Reduction}}
The operational semantics of \TLLC{} programs is mostly the same as that of
call-by-value TLL~\cite{fu25}.  The relation $m \Leadsto m'$ is used to denote a
single step of \emph{program} level reduction. Due to the monadic formulation of
concurrency in \TLLC{}, the only additional (non-trivial) program reduction rule
is the following \textsc{BindElim} rule which reduces a monadic
\bsf{let}-expression when its bound term is a \bsf{return} expression:
\begin{align*}
 (\textsc{BindElim})\qquad \letin{x}{\return{v}}{m} \Leadsto m[v/x] \tag*{(\text{where $v$ is a value})}
\end{align*}
Values now additionally include channels, partially applied communication operators
and thunked monadic expressions. We will use the metavariable $v$ to denote values
for the rest of this paper. The full definition of values is presented in the appendix.

\paragraph{\textbf{Process Reduction}}
The semantics of processes is defined through the relation $P \Rrightarrow Q$ which
states that process configuration $P$ reduces to process configuration $Q$
in one step. The process reduction rules are presented below.

\vspace{0.4em}
\begin{tabular}{r L C L}
  Evaluation Contexts & \mcM, \mcN & ::= & [\cdot] \mid \letin{x}{\mcM}{m}
\end{tabular}

\vspace{0.2em}
\begin{small}
\begin{tabular}{l L}
  (\textsc{Proc-Fork}) &
    \proc{
      \mcN[\fork{x : A}{m}]
    }
    \Rrightarrow
    \scope{cd}{(\proc{\mcN[\return{c}]} \mid \proc{m[d/x]})} 
  \\
  (\textsc{Proc-End}) 
    &\scope{cd}{(
        \proc{\mcM[\close{c}]} 
        \mid 
        \proc{\mcN[\wait{d}]}
      )}
    \Rrightarrow 
    \proc{\mcM[\return{\ii}]} \mid \proc{\mcN[\return{\ii}]} 
  \\
  (\textsc{Proc-Com}) 
    &\scope{cd}{(
        \proc{\mcM[\appR{\sendR{c}}{v}]} 
        \mid 
        \proc{\mcN[\recvR{d}]}
      )}
    \Rrightarrow 
    \scope{cd}{(
      \proc{\mcM[\return{c}]} 
      \mid 
      \proc{\mcN[\return{\pairR{v}{d}{\Ln}}]}
    )}
  \\
  (\textsc{Proc-\underline{Com}}) 
    &\scope{cd}{(
        \proc{\mcM[\appI{\sendI{c}}{o}]} 
        \mid 
        \proc{\mcN[\recvI{d}]}
      )}
     \Rrightarrow 
     \scope{cd}{(
        \proc{\mcM[\return{c}]} 
        \mid 
        \proc{\mcN[\return{\pairI{o}{d}{\Ln}}]}
      )}
\end{tabular}
\vspace{0.2em}
\begin{mathpar}
  \inferrule[(Proc-Expr)]
  { m \Leadsto m' }
  { \proc{m} \Rrightarrow \proc{m'} }

  \inferrule[(Proc-Par)]
  { P \Rrightarrow Q }
  { O \mid P \Rrightarrow O \mid Q }

  \inferrule[(Proc-Scope)]
  { P \Rrightarrow Q }
  { \scope{cd}{P} \Rrightarrow \scope{cd}{Q} }

  \inferrule[(Proc-Congr)]
  { P \equiv P' \\ 
    P' \Rrightarrow Q' \\ 
    Q' \equiv Q }
  { P \Rrightarrow Q }
\end{mathpar}
\end{small}
\vspace{0.2em}

\noindent
The first four rules define the synchronous communication semantics of \TLLC{}. 

The \textsc{Proc-Fork} rule creates a pair of dual channels $c$ and $d$ to connect
the continuation of the parent process with the newly forked child process $m$.
We can see here that $c$ is placed into the evaluation context $\mcN$ of the parent process. 
The child process $m$ receives the dual channel $d$ by substituting $d$ for
the bound variable $x$. The resulting configuration contains two processes
which can now communicate on channels $c$ and $d$.

The \textsc{Proc-End} rule synchronizes the termination of communicating on dual 
channels $c$ and $d$. The resulting process configuration contains two processes
which are no longer connected by any channels. Additionally, the close and wait 
operations are replaced by unit return values once the termination is synchronized.

The \textsc{Proc-Com} rule governs the communication of a real message $v$ from
a sender to a receiver. The sending process continues as $m$ with the channel
$c$ while the receiving process continues as $n$ with the received
message $v$ and the channel $d$ paired together as $\pairR{v}{d}{\Ln}$.

The \textsc{Proc-\underline{Com}} rule is similar to \textsc{Proc-Com} except
that it handles the communication of a ghost message $o$. While this rule seems
to indicate that ghost messages are communicated at runtime, we will later show
through the erasure safety theorem that ghost messages are always safe to be erased.
The exchange of ghost messages here is only for the purpose of establishing a 
reference point for reasoning about the correctness of erasure safety.

The remaining four rules are standard. The \textsc{Proc-Expr} rule allows a
singleton process to reduce its underlying term. The \textsc{Proc-Par} and
\textsc{Proc-Scope} rules allow a process to reduce in parallel composition and
under channel scope respectively. Finally, the \textsc{Proc-Congr} rule allows
processes to reduce up to structural congruence.

\subsection{Meta-Theory}
\paragraph{\textbf{Compatibility}}
We first show that the concurrency extensions of \TLLC{} are compatible with the
underlying TLL type system. To this end, we prove that \TLLC{} enjoys the same
meta-theoretical properties as TLL. Due to the fact that these properties do not involve 
concurrency, their proofs indicate that \TLLC{} is sound as a term calculus.
Here we present a few representative theorems. The full list of theorems and their proofs
can be found in the appendix.

The first theorem we present is the validity theorem which states that
well-typed terms have well-sorted types. This theorem is important as it
ensures that the types appearing in typing judgments are indeed valid (i.e. they inhabit a sort).
\begin{theorem}[Validity]
  Given $\Theta ; \Gamma ; \Delta \vdash m : A$, there exists
  sort $s$ such that $\Gamma \vdash A : s$.
\end{theorem}

In TLL and \TLLC{}, the sort of a type determines whether the type is unrestricted or linear.
This means that it is crucial for a type to have a unique sort, otherwise the same type
could be interpreted as both unrestricted and linear, leading to unsoundness. To address this
concern, we prove the sort uniqueness theorem below which states that a type can have at most one sort.
This ensures no ambiguity on whether a type is to be considered unrestricted or linear.
\begin{theorem}[Sort Uniqueness]
  Given $\Gamma \vdash A : s$ and $\Gamma \vdash A : t$,
  we have $s = t$.
\end{theorem}

The next theorem we present is the standard subject reduction theorem which states that
types are preserved under term reduction. This theorem is necessary for ensuring that
session fidelity holds during process reduction as singleton processes reduce by reducing their
underlying terms.
\begin{theorem}[Subject Reduction]
  Given $\Theta ; \epsilon ; \epsilon \vdash m : A$ and
  $m \Leadsto m'$, we have
  $\Theta ; \epsilon ; \epsilon \vdash m' : A$.
\end{theorem}

\paragraph{\textbf{Session Fidelity}}
The session fidelity theorem ensures that processes adhere to the communication
protocols specified by their types. This property guarantees that well-typed
processes will not encounter communication mismatches at runtime. Since we consider
processes up to structural congruence, we must first show that configuration
typing is preserved under structural congruence. This manifests as the following lemma.

\begin{lemma}[Congruence]
  Given $\Theta \Vdash P$ and $P \equiv Q$, we have $\Theta \Vdash Q$.
\end{lemma}

\noindent
The session fidelity theorem is then stated as follows.
\begin{theorem}[Session Fidelity]
  Given $\Theta \Vdash P$ and $P \Rrightarrow Q$, we have $\Theta \Vdash Q$.
\end{theorem}
\noindent

One of the primary challenges in proving session fidelity is to show that typing
is preserved during communication steps, specifically in the \textsc{Proc-Com}, and
\textsc{Proc-\underline{Com}} cases. In these cases, the message being
communicated is transported from the sender to the receiver without the use of a
substitution. We need to show that the message, after communication, is
consistently typed with regards to the receiver's context. Unlike simple type
systems where one could simply place a value into any context so long as
the value has the expected type, dependent type systems require more care. For
instance, the evaluation context $\pairR{[\cdot]}{\Refl}{} : \Sigma(x : \textsf{nat}).(x = 1)$ 
is well-typed if and only if the hole is filled with $1$. To address this challenge, 
we design the monadic \textsc{Bind} rule (\Cref{sec:dependent-session-types})
to disallow dependency on the bound value. More specifically, for
$\letin{x}{m}{n}$ expressions, the type of $n$ \emph{cannot} depend on $x$. This
restriction means that $m$ can be replaced by any other expression of the same
type without affecting the type of $n$. Consider the \textsc{Proc-Com} step below:
\begin{align*}
  &\scope{cd}{(\proc{\letin{x}{\appR{\sendR{c}}{v}}{m}} \mid \proc{\letin{y}{\recvR{d}}{n}})} \\
  &\quad\Rrightarrow \scope{cd}{(\proc{\letin{x}{\return{c}}{m}} \mid \proc{\letin{y}{\return{\pairR{v}{d}{\Ln}}}{n}})}
\end{align*}
This operation is carried out between two singleton processes that are
evaluating monadic \bsf{let}-expressions. Due to the dependency restriction of
the \textsc{Bind} rule, we can replace $\appR{\sendR{c}}{v}$ with $\return{c}$
and $\recvR{d}$ with $\return{\pairR{v}{d}{\Ln}}$ without affecting the types of
$m$ and $n$. Due to the fact that all communication operations in \TLLC{} are
carried out on \bsf{let}-expressions, the dependency restriction ensures that
session fidelity holds during communication steps.

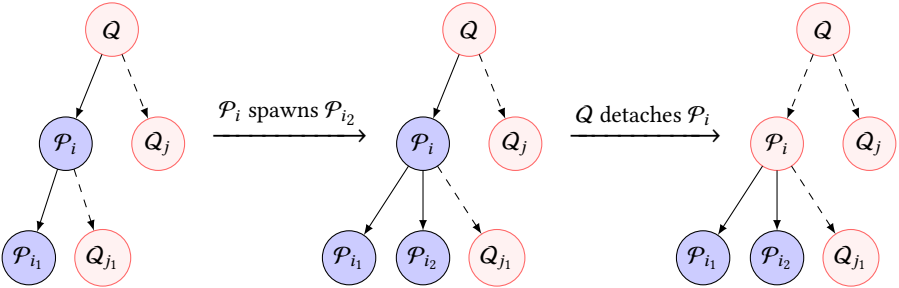
\begin{figure}[h]
\vspace{-0.5em}
\begin{tabular}[t]{c c c c c}
\begin{tikzpicture}
[
  node/.style = {shape=circle, solid,
                 minimum size=2em,
                 inner sep=2pt, outer sep=0pt,
                 draw, fill=blue!20},
  root/.style = {shape=circle, solid,
                 minimum size=2em,
                 inner sep=2pt, outer sep=0pt,
                 draw=red!60, fill=red!5},
  edge from parent/.style = {draw, -latex},
  level distance=1.5cm,
  level 1/.style={sibling distance=3.5em},
  level 2/.style={sibling distance=2.8em},
  baseline=(qj)
]
	\node[root] {\small$\mcQ$}
		child {
      node[node] (pi) {\small$\mcP_i$}
      child {node[node] {\small$\mcP_{i_1}$}}
      child[edge from parent/.style={draw,dashed,-latex}]{
        node[root] {\small$\mcQ_{j_1}$}
      }
    }
		child[edge from parent/.style={draw,dashed,-latex}] {
      node[root] (qj) {\small$\mcQ_j$}
    };
\end{tikzpicture}
&
{\Large$\xrightarrow{\mcP_i \text{ spawns } \mcP_{i_2}}$}
\hspace{-3em}
&
\begin{tikzpicture}
[
  node/.style = {shape=circle, solid,
                 minimum size=2em,
                 inner sep=2pt, outer sep=0pt,
                 draw, fill=blue!20},
  root/.style = {shape=circle, solid,
                 minimum size=2em,
                 inner sep=2pt, outer sep=0pt,
                 draw=red!60, fill=red!5},
  edge from parent/.style = {draw, -latex},
  level distance=1.5cm,
  level 1/.style={sibling distance=3.5em},
  level 2/.style={sibling distance=2.8em},
  baseline=(qj)
]
	\node[root] {\small$\mcQ$}
		child {
      node[node] (pi) {\small$\mcP_i$}
      child {node[node] {\small$\mcP_{i_1}$}}
      child {node[node] {\small$\mcP_{i_2}$}}
      child[edge from parent/.style={draw,dashed,-latex}]{
        node[root] {\small$\mcQ_{j_1}$}
      }
    }
		child[edge from parent/.style={draw,dashed,-latex}] {
      node[root] (qj) {\small$\mcQ_j$}
    };
\end{tikzpicture}
&
{\Large$\xrightarrow{\mcQ \text{ detaches } \mcP_i}$}
\hspace{-3em}
&
\begin{tikzpicture}
[
  node/.style = {shape=circle, solid,
                 minimum size=2em,
                 inner sep=2pt, outer sep=0pt,
                 draw, fill=blue!20},
  root/.style = {shape=circle, solid,
                 minimum size=2em,
                 inner sep=2pt, outer sep=0pt,
                 draw=red!60, fill=red!5},
  edge from parent/.style = {draw,-latex},
  level distance=1.5cm,
  level 1/.style={sibling distance=3.5em},
  level 2/.style={sibling distance=2.8em},
  baseline=(qj)
]
	\node[root] {\small$\mcQ$}
    child[edge from parent/.style={draw,dashed,-latex}]{
      node[root] (pi) {\small$\mcP_i$}
      child[edge from parent/.style = {draw,solid,-latex}] {
        node[node] {\small$\mcP_{i_1}$}
      }
      child[edge from parent/.style = {draw,solid,-latex}] {
        node[node] {\small$\mcP_{i_2}$}
      }
      child[edge from parent/.style={draw,dashed,-latex}]{
        node[root] {\small$\mcQ_{j_1}$}
      }
    }
		child[edge from parent/.style={draw,dashed,-latex}] {
      node[root] (qj) {\small$\mcQ_j$}
    };
\end{tikzpicture}
\end{tabular}
\caption{%
  Spawning Tree Transformations
  ({\color{blue!60}blue} for children, {\color{red!60}red} for detached sub-trees)
}
\label{fig:spawning-trees}
\vspace{-1em}
\end{figure}

\paragraph{\textbf{Global Progress}}
Global progress, i.e. deadlock-freedom, is a desirable property for concurrent programs.
Many session type systems~\cite{wadler12,caires10,das20} guarantee global progress by 
construction through a disciplined use of channels. However, there are also session 
type systems~\cite{honda93,honda16,ldst,balzer17} that eschew global progress in favor of
more expressive session types. \TLLC{} belongs to the latter category if we consider
arbitrary well-typed process configurations. This is because the process type system of
\TLLC{} does not prevent cyclic channel topologies that can lead to deadlocks.
However, we can prove a slightly weaker form of global progress for \TLLC{} by
considering only \emph{reachable} process configurations. 
Intuitively, reachable configurations are those that can be constructed
by \bsf{fork} operations starting from a singleton process. 
Formally, we define the structure of \emph{spawning trees} to capture the
spawning relationships between parent-to-children processes. 
This technique is inspired by the \emph{nested-multiverse semantics} for
reasoning about probabilistic session types~\cite{das23,fu_prob25} and
dominator trees from graph theory. The syntax of spawning trees is given below.

\begin{center}
  \vspace{0.5em}
  \begin{tabular}{r L C L}
    spawning tree & \mcP, \mcQ & ::= & \Root{m}{ \{ ( c_i, \mcP_i ) \}_{i \in \mcI}, \{ \mcQ_j \}_{j \in \mcJ} } \\
                  &            & \;| & \Node{d}{m}{ \{ ( c_i, \mcP_i ) \}_{i \in \mcI}, \{ \mcQ_j \}_{j \in \mcJ} }
  \end{tabular}
  \vspace{0.5em}
\end{center}

Each tree is associated with a term $m$ that performs computation and a set of
\emph{children} $\{ ( c_i, \mcP_i ) \}_{i \in \mcI}$ where $m$ communicates
with each child $\mcP_i$ through channel $c_i$. It also contains a set of
\emph{detached subtrees} $\{ \mcQ_j \}_{j \in \mcJ}$ that are dominated by $m$
but no longer in communication,
i.e. their connections to $m$ have been terminated by 
\Close{}/\Wait{}-operations.
In the case of internal nodes, $d$ is the channel
which $m$ uses to communicate with its parent. \Cref{fig:spawning-trees}
illustrates the spawning and detaching transformations of spawning trees
diagrammatically. The formal semantics of spawning trees is given in the
appendix. Basically, spawning trees are just process
configurations with disciplined channel topologies. 
This is evident in the flattening procedure $|\mcP| = P$, as defined below,
which transforms spawning tree $\mcP$ into a standard process configuration $P$. 

\vspace{-1em}
\begin{mathpar}\small
  \inferrule
  { \forall i \in \mcI,\ | \mcP_i | = (d_i, P_i) \\ 
    \forall j \in \mcJ,\ | \mcQ_j | = Q_j }
  { | \Root{m}{ \{ ( c_i, \mcP_i ) \}_{i \in \mcI}, \{ \mcQ_j \}_{j \in \mcJ} } | = 
    \scope{\overline{c_i d_i}}{(\proc{m} \mid \overline{P_i})}  
    \mid \overline{Q_j} }
  \textsc{(Flatten-Root)}

  \inferrule
  { \forall i \in \mcI,\ | \mcP_i | = (d_i, P_i) \\ 
    \forall j \in \mcJ,\ | \mcQ_j | = Q_j }
  { | \Node{d}{m}{ \{ ( c_i, \mcP_i ) \}_{i \in \mcI}, \{ \mcQ_j \}_{j \in \mcJ} } | = 
    (d, \scope{\overline{c_i d_i}}{(\proc{m} \mid \overline{P_i})} \mid \overline{Q_j}) }
  \textsc{(Flatten-Node)}
\end{mathpar}

To ensure that spawning trees are composed correctly,
we define the mutually inductive validity judgments
$\Vdash \mcQ$ and $\ch{\kappa}{A} \Vdash \mcP$ whose rules are shown below.

\vspace{-1em}
\begin{mathpar}\small
  \inferrule
  { \overline{c_i \tL \ch{\kappa_i}{A_i}} ; \epsilon ; \epsilon \vdash m :\CM{\unit} \\
    \forall i \in \mcI,\ \neg{\ch{\kappa_i}{A_i}} \Vdash \mcP_i  \\
    \forall j \in \mcJ,\ \Vdash \mcQ_j }
  { \Vdash \Root{m}{ \{ ( c_i, \mcP_i ) \}_{i \in \mcI}, \{ \mcQ_j \}_{j \in \mcJ} } }
  \textsc{(Valid-Root)}
  \\
  \inferrule
  {  \overline{c_i \tL \ch{\kappa_i}{A_i}}, d \tL \ch{\kappa}{A} ; \epsilon ; \epsilon \vdash m :\CM{\unit} \\
    \forall i \in \mcI,\ \neg{\ch{\kappa_i}{A_i}} \Vdash \mcP_i \\ 
    \forall j \in \mcJ,\ \Vdash \mcQ_j }
  { \ch{\kappa}{A} \Vdash \Node{d}{m}{ \{ ( c_i, \mcP_i ) \}_{i \in \mcI}, \{ \mcQ_j \}_{j \in \mcJ} } }
  \textsc{(Valid-Node)}
\end{mathpar}

\noindent
Here, we use $\ch{\kappa}{\cdot}$ to denote either $\CH{\cdot}$ or $\HC{\cdot}$
and $\neg{\ch{\kappa}{\cdot}}$ to denote its dual. 

For the root in rule \textsc{Valid-Root}, 
term $m$ must be well-typed under
context $\overline{c_i \tL \ch{\kappa_i}{A_i}}$ 
comprised of $\ch{\kappa_i}{A_i}$
channels connecting to its children $\mcP_i$.
The dual of each channel $\neg{\ch{\kappa_i}{A_i}}$ is propagated to type 
the corresponding child as $\neg{\ch{\kappa_i}{A_i}} \Vdash \mcP_i$.
All detached subtrees $\mcQ_j$ are also required to be valid.
When typing a node in rule \textsc{Valid-Node}, 
we require that $m$ be well-typed in a context
that includes $d \tL \ch{\kappa}{A}$, i.e. the channel that connects 
the node to its parent.

We now define the reachability of a configuration $P$ in terms of the
existence of a valid spawning tree that flattens to $P$.
The global progress theorem then uses reachability as an invariant.
\begin{definition}[Reachability]
  A configuration $P$ is \emph{reachable} if there exists a spawning tree $\mcP$
  such that $\Vdash \mcP$ and $|\mcP| = P$.
\end{definition}

\begin{theorem}[Global Progress]
  Given a reachable configuration $P$, either
  \begin{itemize}
    \item $P \equiv \proc{\return{\ii}}$, or
    \item there exists $Q$ such that $P \Rrightarrow Q$ and $Q$ is reachable.
  \end{itemize}
\end{theorem}

Since a well-typed singleton process $\Vdash \proc{m}$ is reachable
by virtue of $| \Root{m}{\emptyset, \emptyset} | = \proc{m}$,
the global progress theorem tells us that the configurations it transitions
to must be reachable as well.
Thus we have deadlock freedom for process configurations that originated
from a singleton process.

\paragraph{\textbf{Erasure Safety}}
To show that ghost messages are safe to erase, we define an erasure relation
$\Theta ; \Gamma ; \Delta \vdash m \sim m' : A$. This relation states that all
ghost arguments and type annotations in $m$ are replaced by a special opaque 
value $\square$ in $m'$. This relation is similar to the one defined for
the erasure of \emph{propositions} in standard dependent type 
theories~\cite{barras08,letouzey03,sozeau20}.
The most important erasure rule is shown below. The full set of erasure
rules can be found in the appendix.
\begin{mathpar}
  \inferrule[Erase-Implicit-App]
  { \Theta ; \Gamma ; \Delta \vdash m \sim m' : \PiI{t}{x : A}{B} \\ 
    \Gamma \vdash n : A }
  { \Theta ; \Gamma ; \Delta \vdash \appI{m}{n} \sim \appI{m'}{\square} : B[n/x] }
\end{mathpar}

The \textsc{Erase-Implicit-App} rule states that when erasing an implicit
application $\appI{m}{n}$, the ghost argument $n$ is replaced by $\square$ in the
erased term. Consider the $\SendI{}\ c$ operator for sending ghost messages on channel 
$c$. As defined in \Cref{sec:dependent-session-types}, this partially applied operator 
has a type of the form $\PiI{\Ln}{x : A}{\CM{B}}$. When fully applied as $\appI{\SendI{}\ c}{n}$, 
the ghost argument $n$ is erased to $\square$ by \textsc{Erase-Implicit-App}.
Since $\square$ is an opaque value, it cannot be inspected or pattern matched on.
Thus, if programs can be evaluated soundly after erasing all ghost arguments
and type annotations, we can conclude that ghost messages are safe to erase.

The erasure relation is then naturally lifted to the process level as $\Theta \Vdash P \sim P'$
where $P'$ is the erased version of $P$. The rules for this relation are as follows:

\vspace{-1em}
\begin{small}
\begin{mathpar}
  \inferrule[Erase-Expr]
  { \Theta ; \epsilon ; \epsilon \vdash m \sim m' : \CM{\unit} }
  { \Theta \Vdash \proc{m} \sim \proc{m'} }

  \inferrule[Erase-Par]
  { \Theta_1 \Vdash P \sim P' \\ 
    \Theta_2 \Vdash Q \sim Q' }
  { \Theta_1 \dotcup \Theta_2 \Vdash (P \mid Q) \sim (P' \mid Q') }

  \inferrule[Erase-Scope]
  { \Theta, c :_\Ln \CH{A}, d :_\Ln \HC{A} \Vdash P \sim P' }
  { \Theta \Vdash \scope{cd}{P} \sim \scope{cd}{P'} }
\end{mathpar}
\end{small}

We show that erasure is safe through the following two theorems. These theorems
tell us that any possible reduction on an original object 
(either a term or process) can be simulated on its erased counterpart. 
Moreover, the erased object obtained after reduction also satisfies the erasure
relation with respect to the reduced original object.  Basically, these theorems
state that any possible evaluation path of the original object remains valid
after erasure.

\begin{theorem}[Term Simulation]
  Given $\Theta ; \epsilon ; \epsilon \vdash m \sim m' : A$ and $m \Leadsto n$, 
  there exists $n'$ such that 
  $m' \Leadsto n'$ and $\Theta ; \epsilon ; \epsilon \vdash n \sim n' : A$.
\end{theorem}

\begin{theorem}[Process Simulation]
  Given $\Theta \Vdash P \sim P'$ and reduction $P \Rrightarrow Q$, there exists $Q'$ such that
  $P' \Rrightarrow Q'$ and $\Theta \Vdash Q \sim Q'$.
\end{theorem}

\section{Implementation and Evaluation}\label{sec:implementation}
\subsection{Implementation}
We implement a prototype compiler for \TLLC{}. The main components of the compiler
are written in OCaml while a minimalistic runtime library is implemented in C.
The compiler takes \TLLC{} source files as input and generates safe C code which
can be further compiled into executable binaries on POSIX compliant systems.
In this section, we give an overview of the inference, linearity checking
and optimization phases of the compiler.

\paragraph{\textbf{Inference}}
To reduce code duplication and type annotation burden, we implement two forms of inference: 
(1) automatic instantiation of \emph{sort-polymorphic schemes} similarly to the TLL compiler and 
(2) elaboration of inferred arguments.
Consider the identity function below:
\begin{align*}
  \Def\ \textsf{id}\flq{s}\frq\ \%\!\{ A : \textsf{Type}\flq{s}\frq\}\ (x : A) : A := x
\end{align*}
This function is a sort-polymorphic scheme as it is parameterized over sort variable $s$.
Depending on the universe of $A$, sort $s$ can be instantiated to either $\Ln$ for linear types
or $\Un$ for unrestricted types. This eliminates the need to define two separate identity functions
for linear and unrestricted types. The type $A$ here is marked by $\%$ to indicate that it is an
inferred argument. Suppose $\textsf{id}$ is applied to a natural number $42$. The compiler creates
two metavariables $\hat{s}$ and $\hat{\alpha}$ to represent the elided sort and type arguments respectively.
Type inference produces the following constraints:
\begin{align*}
  \textsf{id}\ 42 
  \ \ 
  \xRightarrow{\text{desugar}} 
  \ \ 
  \appI{\textsf{id}\flq{\hat{s}}\frq}{\hat{\alpha}}\ 42
  \ \ 
  \xRightarrow{\text{infer}} 
  \ \ 
  \begin{cases}
    \hat{s} = \Un \\
    \hat{\alpha} = \textsf{nat}
  \end{cases}
  \!\!
  \xRightarrow{\text{mono}} 
  \ \ 
  \appI{\textsf{id}\flq{\Un}\frq}{\textsf{nat}}\ 42
\end{align*}
Once the constraints are solved through unification~\cite{abel11}, the metavariables are
replaced by their solutions. The monomorphized code is then passed to the next phase for
linearity checking.

\paragraph{\textbf{Linearity Checking}}
During the inference phase, the usage of linear variables is not tracked. The
type checking algorithm essentially treats \TLLC{} as a fully structural type
system. It is only after all sort-polymorphic schemes and inferred arguments are
instantiated that the linearity checking begins. A substructural type checking
algorithm is applied to determine if the elaborated program compiles with the actual
typing rules of \TLLC{}. We adopt this two-phase approach to simplify the linearity
checking algorithm. Although sort-polymorphism greatly reduces code duplication from
the user's perspective, it also obfuscates the classification of types into linear
and unrestricted ones. Thus, it is much easier to check linearity after monomorphization.


\paragraph{\textbf{Optimization}}
Once linearity checking is complete, ghost terms are erased in a type directed
manner. The intermediate representation (IR) obtained from erasure carries
metadata that mark the linearity of certain critical expressions.

One of the optimizations performed is constructor unboxing. The layouts of inductive type
constructors are analyzed to determine if the inductive type is suitable for unboxing.
For example, consider the singleton type defined as follows:
\begin{align*}
  \Inductive\ \textsf{sing}\flq{s}\frq\ \%\!\{A : \textsf{Type}\flq{s}\frq\}\ (x : A) := 
  \textsf{Just} : \forall (x : A) \rightarrow \textsf{sing}\ x
\end{align*}
Here, \textsf{Just} is the only constructor of type \textsf{sing}. This means that
pattern matching on a value of type \textsf{sing} is redundant as there is only one possible case.
Expressions of the form $\textsf{Just}\ m$ are unboxed to $m$ to reduce the number of indirections
at runtime. In general, an inductive type can be unboxed if it has a single constructor and
the constructor has a single non-ghost field.

To reduce the time spent on allocating and deallocating heap objects, we utilize
in-place updates for linear values. This optimization is similar to recent works on
function in-place programming~\cite{lorenzen23,perceus} where allocated heap memory
is reused instead of being garbage collected. Unlike these works which utilize reference
counting to dynamically check the viability of an in-place update, the metadata in our 
IR is sufficient to statically determine if an in-place optimization is safe.

\subsection{Evaluation}
\paragraph{\textbf{Compilation Performance}}
\begin{table}[h]
  \caption{%
    Evaluation of \TLLC{} compiler.
    LOC = lines of code;
    Defs = number of definitions;
    SVars = number of sort meta-variables;
    IVars = number of inference meta-variables;
    Eqns = number of generated unification equations;
    Time = total compilation time;
    Mem = peak memory usage;
    LOC(C) = lines of generated C code.
  }
  \vspace{-0.8em}
  \begin{tabular}{l r r r r r r r r}
    \textbf{Program} & \textbf{LOC} & \textbf{Defs} & 
    \textbf{SVars} & \textbf{IVars} & \textbf{Eqns} & 
    \textbf{Time(s)} & \textbf{Mem(MB)} & \textbf{LOC(C)} \\
    \hline
    list            & 66  & 8  & 266 & 235 & 5368  & 0.03 & 17.71 & 1345 \\
    vector          & 38  & 9  & 148 & 162 & 2663  & 0.02 & 15.43 & 1313 \\
    additive pair   & 105 & 13 & 176 & 166 & 2994  & 0.02 & 15.59 & 1636 \\
    DH key-exchange & 62  & 7  & 111 & 126 & 2416  & 0.02 & 15.67 & 1305 \\
    RSA encryption  & 102 & 9  & 138 & 164 & 3906  & 0.04 & 20.56 & 1314 \\
    queue           & 107 & 17 & 232 & 237 & 10460 & 0.10 & 20.85 & 1772 \\
    nat induction   & 107 & 17 & 143 & 162 & 4279  & 0.03 & 17.05 & 1765 \\
    map-reduce      & 237 & 24 & 329 & 469 & 15456 & 0.10 & 22.72 & 2585 \\
    mergesort       & 132 & 13 & 223 & 177 & 4654  & 0.03 & 16.92 & 2252 \\
    key-value store & 149 & 20 & 185 & 405 & 27305 & 0.11 & 27.34 & 1803 \\
    \hline
  \end{tabular}
  \label{table:compiler-evaluation}
\end{table}

To evaluate the practicality of \TLLC{}, we implement a suite of example programs
that cover a wide range of applications of dependent session types. The programs
include implementations of common data structures such as lists and vectors,
cryptographic protocols such as Diffie-Hellman key exchange and RSA encryption, 
concurrent data structures such as queues and key-value stores, and
concurrent algorithms such as map-reduce and mergesort.

Our experiments are performed on a laptop with an Apple M4 Pro CPU and 24 GB
RAM. \Cref{table:compiler-evaluation} summarizes the results of our compiler
evaluation. We observe that even for complex programs such as the key-value
store, which involves heavy use of sort-polymorphic schemes and inferred
arguments, the total compilation time is only around 0.1 seconds with peak
memory usage under 30 MB. The generated C code is around 10 times the
size of the original \TLLC{} source code. Overall, these results demonstrate that
our compiler is efficient and scalable to non-trivial programs.

\paragraph{\textbf{Runtime Performance}}
We evaluate the runtime performance of \TLLC{} by implementing variants of the
mergesort algorithm. We implement three variants: (1) sequential on unrestricted lists, 
(2) sequential on linear lists, and (3) concurrent on linear lists.
We compare the performance of these three variants
with sequential implementations of mergesorts in OCaml and the
Rocq~\cite{coq} theorem prover (extracted to OCaml). 
\Cref{fig:runtime-evaluation} shows the results of our evaluation.

\vspace{-0.2em}
\begin{figure}[h]
\begin{subfigure}[t]{0.48\textwidth}
\centering
\begin{tikzpicture}[framed,scale=0.9]
  \begin{axis}[
    xbar,
    x = 2.2cm,
    y = 0.6cm,
    y axis line style = { opacity = 0 },
    axis x line       = none,
    tickwidth         = 0pt,
    ytick             = data,
    enlarge y limits  = 0.1,
    enlarge x limits  = 0.1,
    symbolic y coords = {Rocq, OCaml, ConL, SeqL, SeqU},
    nodes near coords,
  ]
  \addplot[color=black,fill=blue!20] coordinates {
    (0.84,SeqU) 
    (0.42,SeqL)
    (0.16,ConL) 
    (0.51,OCaml) 
    (1.87,Rocq)
  };
  \end{axis}
\end{tikzpicture}
\caption{Runtime performance (seconds)}
\end{subfigure}
\begin{subfigure}[t]{0.45\textwidth}
\centering
\begin{tikzpicture}[framed,scale=0.9]
  \begin{axis}[
    xbar,
    x = 0.012cm,
    y = 0.6cm,
    y axis line style = { opacity = 0 },
    axis x line       = none,
    tickwidth         = 0pt,
    ytick             = data,
    enlarge y limits  = 0.1,
    enlarge x limits  = 0.1,
    symbolic y coords = {Rocq, OCaml, ConL, SeqL, SeqU},
    nodes near coords,
  ]
  \addplot[color=black,fill=red!20] coordinates {
    (299.07,SeqU) 
    (82.59,SeqL)
    (205.32,ConL)
    (243.18,OCaml) 
    (382.27,Rocq)
  };
  \end{axis}
\end{tikzpicture}
\caption{Maximum memory usage (MB)}
\end{subfigure}
\vspace{-0.6em}
\caption{%
  Runtime performance comparison of variants of mergesort.
  SeqU  = sequential mergesort with unrestricted lists;
  SeqL  = sequential mergesort with linear lists;
  ConL  = concurrent mergesort with linear lists;
  OCaml = sequential mergesort in OCaml;
  Rocq  = sequential mergesort in Rocq.
}
\label{fig:runtime-evaluation}
\end{figure}
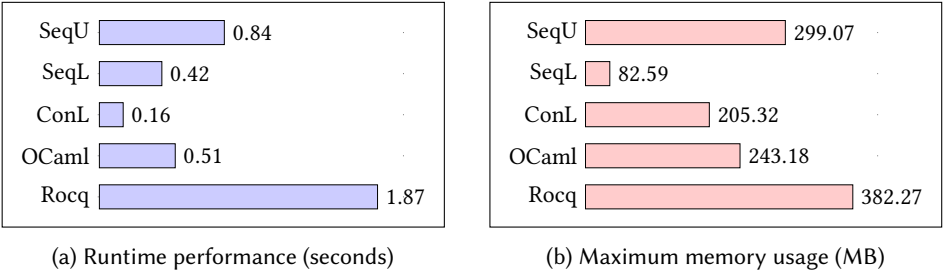
\vspace{-0.6em}

Compared to the OCaml implementation, our unrestricted mergesort is slower and
uses slightly more memory due to the lack of optimizations regarding garbage
collecting unrestricted types. However, our linear mergesorts outperform the
OCaml implementation in both runtime and memory usage thanks to the in-place update
optimizations for linear types. Rocq expectedly performs the worst as we
did not fine tune its extraction settings. Our concurrent mergesort on linear lists
achieves the best runtime performance among all variants. It exhibits a moderate
degree of memory overhead due to the use of concurrent data structures. Overall,
these results demonstrate that our compiler is capable of producing efficient
concurrent programs with competitive performance.

\section{Related Work}\label{sec:related}
Session types are a class of type systems pioneered by Honda~\cite{honda93} for
structuring dyadic communication in the $\pi$-calculus. Abramsky notices deep
connections between the Linear Logic~\cite{girard} of Girard and concurrency,
predicting that Linear Logic will play a foundation role in future theories of
concurrent computation~\cite{abramsky93,abramsky94}. Caires and Pfenning show an
elegant correspondence between session types and Linear Logic~\cite{caires10}.
Gay and Vasconcelos integrate session types with $\lambda$-calculus~\cite{gay10}
which allows one to express concurrent processes using standard functional
programming. Wadler further refines the calculus of Gay and Vasconcelos to be
deadlock free by construction~\cite{wadler12}.

Toninho together with Caires and Pfenning develops the first dependent session
type systems~\cite{toninho11,pfenning11}. These works extend the existing logic
of Caires and Pfenning~\cite{caires10} with universal and existential
quantifiers to precisely specify properties of communicated messages.

Toninho and Yoshida present an interesting language~\cite{toninho18} that
integrates both $\pi$-calculus style processes and $\lambda$-calculus style
terms using a contextual monad. Additionally, full $\lambda$-calculi are
embedded in both functional types and session types to enable large elimination.

Wu and Xi~\cite{wu17} implement session types in the ATS programming
language~\cite{ats} which supports DML style dependent types~\cite{dml}. This
allows them to specify the properties of concurrent programs and verify them
using proof automation. While DML style dependency is well suited for automatic
reasoning, certain properties can be difficult to encode due to restrictions on
the type level language.

Thiemann and Vasconcelos~\cite{ldst} introduce the LDST calculus which utilizes
label dependent session types to elegantly describe communication patterns.
Communication protocols written in non-dependent session type systems can
essentially be simulated through label dependency. On the other hand, LDST's
minimalist design limits its capabilities for general verification as label
dependency by itself is too weak to express many interesting program properties.

Das and Pfenning develop a refinement session type system~\cite{das20} where the
types of concurrent programs can be refined with logical predicates. Similarly
to DML style dependent types, the expressiveness of refinement session types is
intentionally limited to facilitate proof automation. The Martin-L\"{o}f style
dependent session types of \TLLC{} allow users to express and verify more
complex program properties at the cost of decidable proof automation.

Atkey proposes QTT~\cite{qtt} based on initial ideas of McBride~\cite{nothing}.
QTT is a dependent type theory which tracks resource usage through semi-ring
annotations on binders. By instantiating the semi-ring and its ordering relation
correctly, QTT can simulate linear types. The Idris 2 programming
language~\cite{idris2} (based on QTT) implements a session typed
DSL~\cite{brady21} around its raw communication primitives. The authors do not
formalize these session types or study its meta-theory. Unlike \TLLC{} where
a library provider could specify a type (such as channels) as linear and
automatically enforce its usage in client code through type checking, the
obligation of resource tracking is pushed to the client in QTT where binders
must be correctly annotated a priori. User mistakes in the annotations could
lead to resources being improperly tracked in a program despite passing type
checking. 

Hinrichsen et al. develop Actris~\cite{actris} which extends the
Iris~\cite{iris} separation logic framework with dependent separation protocols.
Compared to our work, Actris reasons about concurrent programs at a lower level
of abstraction. This gives it greater precision and flexibility when dealing
with imperative and unsafe programming features. However, the low level nature
of Actris reduces its effectiveness at providing guidance for writing programs.
In this regard, the interactivity of type systems is more beneficial to helping
users construct correct programs in the first place.

\section{Conclusion}\label{sec:conclusion}
\TLLC{} is a linear dependently typed programming language which extends the TLL
type theory with dependent session types. Through examples, we demonstrate how
dependent session types can be effectively applied to verify concurrent
programs. The expressive power of Martin-L\"{o}f style dependency allows \TLLC{}
session types to capture the expected semantics of concurrent programs. This
results in greater verification precision and flexibility when compared to other
type systems with more restricted forms of dependency. We study the meta-theory
of \TLLC{} and show that it is sound as both a term calculus and also as a
process calculus. A prototype compiler is implemented which compiles \TLLC{}
programs into safe concurrent C code.

A direction of research we intend to explore is the integration of dependency
with multi-party session types~\cite{honda16}. Protocols expressed through such
a session type system will be able to coordinate interactions between processes
from a global viewpoint. We predict dependency will again play a key role in
verifying the correctness of multi-party concurrent computation.

\bibliography{reference}
\clearpage

\begin{appendices}
\DoToC
\clearpage

\section{Syntax}
The full syntax of \TLLC{} is presented below.

\begin{figure}[H]
  \begin{tabular}{r l c l}
    variables & $x, y, z$   &     &                 \\
    channels  & $c, d$      &     &                 \\
    sorts     & $s, r, t$   & ::= & $\Un$ | $\Ln$   \\
    actions   & $\rho$      & ::= & $!$ | $?$       \\
    terms     & $m,n,A,B,C$ & ::= & $x$ | $c$ | $s$ \\
              &             & \;| & $\PiR{t}{x : A}{B}$ | $\PiI{t}{x : A}{B}$
                                    | $\SigR{t}{x : A}{B}$ | $\SigI{t}{x : A}{B}$ \\
              &             & \;| & $\lamR{t}{x : A}{m}$ | $\lamI{t}{x : A}{m}$
                                    | $\pairR{m}{n}{t}$ | $\pairI{m}{n}{t}$ \\
              &             & \;| & $\appR{m}{n}$ | $\appI{m}{n}$ | $\SigElim{[z]A}{m}{[x,y]n}$ | $\fix{x : A}{B}$ \\
              &             & \;| & $\unit$ | $\ii$ | $\Bool$ | $\bTrue$ | $\bFalse$
                                    | $\boolElim{[z]A}{m}{n_{1}}{n_{2}}$ \\
              &             & \;| & $\CM{A}$ | $\return{m}$ | $\letin{x}{m}{n}$ \\
              &             & \;| & $\Proto$ | $\End$
                                    | $\ActR{\rho}{x : A}{B}$ | $\ActI{\rho}{x : A}{B}$ | $\CH{A}$ | $\HC{A}$ \\
              &             & \;| & $\fork{x : A}{m}$ | $\recvR{m}$ | $\recvI{m}$
                                    | $\sendR{m}$ | $\sendI{m}$ \\
              &             & \;| & $\close{m}$ | $\wait{m}$ \\
    thunks    & $\tau$         & ::= & $\fork{x : A}{m}$ | $\recvR{c}$ | $\recvI{c}$  \\
              &             & \;| & $\appR{\sendR{c}}{v}$ | $\appI{\sendI{c}}{m}$ | $\close{c}$ | $\wait{c}$ \\
              &             & \;| & $\letin{x}{\tau}{m}$ \\ 
    values    & $u, v$      & ::= & $c$ | $\lamR{t}{x : A}{m}$ | $\lamI{t}{x : A}{m}$
                                    | $\pairR{u}{v}{t}$ | $\pairI{m}{v}{t}$ \\
              &             & \;| & $\ii$ | $\bTrue$ | $\bFalse$ | $\return{v}$ | $\tau$ | $\sendR{c}$ | $\sendI{c}$ \\
    process   & $O, P, Q$   & \;| & $\langle m \rangle$ | ($P \mid Q$) | $\scope{cd}{P}$
  \end{tabular}
\end{figure}

\section{Auxiliary Operators/Judgments}
In this section, we define several auxiliary operators/judgments used in the
formalization of \TLLC{}.

\paragraph{\textbf{Sort Ordering}}
The sort ordering relation $\sqsubseteq$ is defined as follows:
\begin{align*}
  (\textsc{Ord-$\Un$})\quad \Un \sqsubseteq s
  \qquad\qquad\qquad
  (\textsc{Ord-$\Ln$})\quad \Ln \sqsubseteq \Ln
\end{align*}
This relation is useful when defining the typing rules of dependent pairs by
ensuring that pairs only contain values of a lower or equal sort.

\paragraph{\textbf{Context Merge}}
The context merge operator $\dotcup$ is a partial function that combines two
contexts into one by selective applying the contraction rule on unrestricted
variables. The operator is undefined if both two contexts contain 
overlapping linear variables.
\begin{mathpar}
  \inferrule[Merge-Empty]
  { }
  { \epsilon \dotcup \epsilon = \epsilon }

  \inferrule[Merge-$\Un$]
  { \Delta_1 \dotcup \Delta_2 = \Delta \\
    x \notin \Delta }
  { (\Delta_1, x \tU A) \dotcup (\Delta_2, x \tU A) = (\Delta, x \tU A) }
  \\

  \inferrule[Merge-$\Ln_1$]
  { \Delta_1 \dotcup \Delta_2 = \Delta \\
    x \notin \Delta }
  { (\Delta_1, x \tL A) \dotcup \Delta_2 = (\Delta, x \tL A) }

  \inferrule[Merge-$\Ln_2$]
  { \Delta_1 \dotcup \Delta_2 = \Delta \\
    x \notin \Delta }
  { \Delta_1 \dotcup (\Delta_2, x \tL A) = (\Delta, x \tL A) }
\end{mathpar}

\paragraph{\textbf{Context Restriction}}
The context restriction operator $\triangleright$ is a predicate that is
useful for defining the typing rules of $\lambda$-expressions. In particular,
it prevents unrestricted functions from capturing linear variables in their closures.
\begin{mathpar}
  \inferrule
  { }
  { \epsilon \triangleright s }
  \textsc{(ReEmpty)}

  \inferrule
  { \Delta \triangleright \Un }
  { \Delta, x \tU A \triangleright \Un }
  \textsc{(Re-$\Un$)}

  \inferrule
  { \Delta \triangleright \Ln }
  { \Delta, x \ty{s} A \triangleright \Ln }
  \textsc{(Re-$\Ln$)}
\end{mathpar}

\paragraph{\textbf{Arity}}
For types $A$ and $X$, we say that $A$ is an \emph{arity} ending on $X$ if it
is either $X$ itself or a $\Pi$-type whose codomain is an arity ending on $X$.
Formally, we define the judgment $A~\arity{X}$ as follows:
\begin{mathpar}\small
  \inferrule[Arity-Base]
  { }
  { X~\arity{X} }

  \inferrule[Arity-Implicit]
  { B~\arity{X} }
  { \PiI{t}{x : A}{B}~\arity{X} }

  \inferrule[Arity-Explicit]
  { B~\arity{X} }
  { \PiR{t}{x : A}{B}~\arity{X} }
\end{mathpar}
This judgment is used for defining the typing rule of (parameterized) recursive protocols.

\paragraph{\textbf{Guarded}}
For variable $x$ and term $m$, we say that $x$ is \emph{guarded} in $m$ if
the judgment $m~\guard{x}$ is derivable. Intuitively, this means that every 
occurrence of $x$ in $m$ appears under a protocol action. This is important 
for ensuring that recursive protocols do not unfold indefinitely without
performing any actions. The judgment is defined as follows:

\begin{mathpar}\footnotesize
  \inferrule[Guard-Var]
  { x \neq y }
  { y~\guard{x} }

  \inferrule[Guard-Sort]
  { }
  { s~\guard{x} }

  \inferrule[Guard-Implicit-Fun]
  { A~\guard{x} \\ 
    B~\guard{x} }
  { \PiI{t}{x' : A}{B}~\guard{x} }

  \inferrule[Guard-Explicit-Fun]
  { A~\guard{x} \\ 
    B~\guard{x} }
  { \PiR{t}{x' : A}{B}~\guard{x} }

  \inferrule[Guard-Implicit-Lam]
  { A~\guard{x} \\ 
    m~\guard{x} }
  { \lamI{t}{x' : A}{m}~\guard{x} }

  \inferrule[Guard-Explicit-Lam]
  { A~\guard{x} \\ 
    m~\guard{x} }
  { \lamR{t}{x' : A}{m}~\guard{x} }

  \inferrule[Guard-Implicit-App]
  { m~\guard{x} \\ 
    n~\guard{x} }
  { \appI{m}{n}~\guard{x} }

  \inferrule[Guard-Explicit-App]
  { m~\guard{x} \\ 
    n~\guard{x} }
  { \appR{m}{n}~\guard{x} }

  \inferrule[Guard-Implicit-Sum]
  { A~\guard{x} \\ 
    B~\guard{x} }
  { \SigI{t}{x' : A}{B}~\guard{x} }

  \inferrule[Guard-Explicit-Sum]
  { A~\guard{x} \\ 
    B~\guard{x} }
  { \SigR{t}{x' : A}{B}~\guard{x} }

  \inferrule[Guard-Implicit-Pair]
  { m~\guard{x} \\ 
    n~\guard{x} }
  { \pairI{m}{n}{t}~\guard{x} }
  \and\hspace{-0.5em}
  \inferrule[Guard-Explicit-Pair]
  { m~\guard{x} \\ 
    n~\guard{x} }
  { \pairR{m}{n}~\guard{x} }
  \and\hspace{-0.5em}
  \inferrule[Guard-Explicit-SumElim]
  { A~\guard{x} \\ 
    m~\guard{x} \\
    n~\guard{x} }
  { \SigElim{[z]A}{m}{[x',y']n} }

  \inferrule[Guard-Unit]
  { }
  { \unit~\guard{x} }

  \inferrule[Guard-UnitVal]
  { }
  { \ii~\guard{x} }

  \inferrule[Guard-Bool]
  { }
  { \Bool~\guard{x} }

  \inferrule[Guard-True]
  { }
  { \bTrue~\guard{x} }

  \inferrule[Guard-False]
  { }
  { \bFalse~\guard{x} }

  \inferrule[Guard-BoolElim]
  { A~\guard{x} \\
    m~\guard{x} \\ 
    n_1~\guard{x} \\
    n_2~\guard{x} }
  { \boolElim{[z]A}{m}{n_{1}}{n_{2}}~\guard{x} }

  \inferrule[Guard-$\C$Type]
  { A~\guard{x} }
  { \CM{A}~\guard{x} }

  \inferrule[Guard-Return]
  { m~\guard{x} }
  { \return{m}~\guard{x} }

  \inferrule[Guard-Bind]
  { m~\guard{x} \\ 
    n~\guard{x} }
  { \letin{x'}{m}{n}~\guard{x} }

  \inferrule[Guard-Proto]
  { }
  { \Proto~\guard{x} }

  \inferrule[Guard-End]
  { }
  { \End~\guard{x} }

  \inferrule[Guard-RecProto]
  { A~\guard{x} \\ 
    m~\guard{x} }
  { \fix{x' : A}{m}~\guard{x} }

  \inferrule[Guard-Implicit-Action]
  { A~\guard{x} }
  { \ActI{\rho}{x' : A}{B}~\guard{x} }

  \inferrule[Guard-Explicit-Action]
  { A~\guard{x} }
  { \ActR{\rho}{x' : A}{B}~\guard{x} }

  \inferrule[Guard-CH]
  { A~\guard{x} }
  { \CH{A}~\guard{x} }

  \inferrule[Guard-HC]
  { A~\guard{x} }
  { \HC{A}~\guard{x} }

  \inferrule[Guard-Channel]
  { }
  { c~\guard{x} }

  \inferrule[Guard-Fork]
  { A~\guard{x} \\ 
    m~\guard{x} }
  { \fork{x': A}{m}~\guard{x} }

  \inferrule[Guard-Implicit-Recv]
  { m~\guard{x} }
  { \recvI{m}~\guard{x} }

  \inferrule[Guard-Explicit-Recv]
  { m~\guard{x} }
  { \recvR{m}~\guard{x} }

  \inferrule[Guard-Implicit-Send]
  { m~\guard{x} }
  { \sendI{m}~\guard{x} }

  \inferrule[Guard-Explicit-Send]
  { m~\guard{x} }
  { \sendR{m}~\guard{x} }

  \inferrule[Guard-Wait]
  { m~\guard{x} }
  { \wait{m}~\guard{x} }

  \inferrule[Guard-Close]
  { m~\guard{x} }
  { \close{m}~\guard{x} }
\end{mathpar}
\clearpage

\section{Formal Typing Rules}
In this section, we present the full typing rules of \TLLC{}. We organize the
typing rules into logical level, program level and process level.

\subsection{Logical Level}
The typing judgment for the logical level has the form $\Gamma \vdash m : A$.
This judgment states that under the \emph{logical context} $\Gamma$, term $m$
has type $A$. The logical level is completely \emph{structural}.

\paragraph{\textbf{Logical Context}}
The logical context $\Gamma$ is a sequence of variable bindings of the form 
$x_0 : A_0, x_1 : A_1, \dots, x_n : A_n$. Each variable $x_i$ is bound to a type $A_i$.
Variables in the logical context are unrestricted and can be used arbitrarily many
times. The empty context is denoted by $\epsilon$. To ensure the validity of types
in the logical context, we define the context validity judgment $\Gamma \vdash $.
\begin{mathpar}
  \inferrule[Ctx-Empty]
  { }
  { \epsilon \vdash }

  \inferrule[Ctx-Var]
  { \Gamma \vdash \quad \Gamma \vdash A : s \\
    x \notin \Gamma }
  { \Gamma, x : A \vdash }
\end{mathpar}
Note that the context validity judgment is \emph{mutually inductively} defined with the typing judgment. 

\paragraph{\textbf{Core Typing}}
The core typing rules is responsible for the functional fragment of \TLLC{}.
The convertibility relation $A \simeq B$ is used in the conversion rule to
allow type equivalence up to $\beta$-reduction. We will present the definition
of the convertibility relation in \Cref{appendix:logical-semantics}. 
\begin{mathpar}\footnotesize
  \inferrule[Sort]
  { \Gamma \vdash }
  { \Gamma \vdash s : \Un }

  \inferrule[Var]
  { \Gamma, x : A \vdash }
  { \Gamma, x : A \vdash x : A }

  \inferrule[Conversion]
  { \Gamma \vdash B : s \\
    \Gamma \vdash m : A \\
    A \simeq B }
  { \Gamma \vdash m : B }

  \inferrule[Explicit-Fun]
  { \Gamma \vdash A : s \\
    \Gamma, x : A \vdash B : r }
  { \Gamma \vdash \PiR{t}{x : A}{B} : t }

  \inferrule[Implicit-Fun]
  { \Gamma \vdash A : s \\
    \Gamma, x : A \vdash B : r }
  { \Gamma \vdash \PiI{t}{x : A}{B} : t }

  \inferrule[Explicit-Lam]
  { \Gamma, x : A \vdash m : B }
  { \Gamma \vdash \lamR{t}{x : A}{m} : \PiR{t}{x : A}{B} }

  \inferrule[Implicit-Lam]
  { \Gamma, x : A \vdash m : B }
  { \Gamma \vdash \lamI{t}{x : A}{m} : \PiI{t}{x : A}{B} }

  \inferrule[Explicit-App]
  { \Gamma \vdash m : \PiR{t}{x : A}{B} \\
    \Gamma \vdash n : A }
  { \Gamma \vdash \appR{m}{n} : B[n/x] }

  \inferrule[Implicit-App]
  { \Gamma \vdash m : \PiI{t}{x : A}{B} \\
    \Gamma \vdash n : A }
  { \Gamma \vdash \appI{m}{n} : B[n/x] }

  \inferrule[Explicit-Sum]
  { s \sqsubseteq t \\ r \sqsubseteq t \\
    \Gamma \vdash A : s \\
    \Gamma, x : A \vdash B : r }
  { \Gamma \vdash \SigR{t}{x : A}{B} : t }

  \inferrule[Implicit-Sum]
  { r \sqsubseteq t \\
    \Gamma \vdash A : s \\
    \Gamma, x : A \vdash B : r }
  { \Gamma \vdash \SigI{t}{x : A}{B} : t }

  \inferrule[Explicit-Pair]
  { \Gamma \vdash \SigR{t}{x : A}{B} : t \\
    \Gamma \vdash m : A \\
    \Gamma \vdash n : B[m/x] }
  { \Gamma \vdash \pairR{m}{n}{t} : \SigR{t}{x : A}{B} }

  \inferrule[Implicit-Pair]
  { \Gamma \vdash \SigI{t}{x : A}{B} : t \\
    \Gamma \vdash m : A \\
    \Gamma \vdash n : B[m/x] }
  { \Gamma \vdash \pairI{m}{n}{t} : \SigI{t}{x : A}{B} }

  \inferrule[Explicit-SumElim]
  { \Gamma, z : \SigR{t}{x : A}{B} \vdash C : s \\
    \Gamma \vdash m : \SigR{t}{x : A}{B} \\
    \Gamma, x : A, y : B \vdash n : C[\pairR{x}{y}{t}/z] }
  { \Gamma \vdash \SigElim{[z]C}{m}{[x,y]n} : C[m/z] }

  \inferrule[Implicit-SumElim]
  { \Gamma, z : \SigI{t}{x : A}{B} \vdash C : s \\
    \Gamma \vdash m : \SigI{t}{x : A}{B} \\
    \Gamma, x : A, y : B \vdash n : C[\pairI{x}{y}{t}/z] }
  { \Gamma \vdash \SigElim{[z]C}{m}{[x,y]n} : C[m/z] }
\end{mathpar}
\clearpage

\paragraph{\textbf{Data Typing}}
The data typing rules govern the typing of base types such as the unit type and
the boolean type. The rules are presented below.
\begin{mathpar}\small
  \inferrule[Unit]
  { \Gamma \vdash }
  { \Gamma \vdash \unit : \Un }

  \inferrule[UnitVal]
  { \Gamma \vdash }
  { \Gamma \vdash \ii : \unit }

  \inferrule[Bool]
  { \Gamma \vdash }
  { \Gamma \vdash \Bool : \Un }

  \inferrule[True]
  { \Gamma \vdash }
  { \Gamma \vdash \bTrue : \Bool }

  \inferrule[False]
  { \Gamma \vdash }
  { \Gamma \vdash \bFalse : \Bool }\\

  \inferrule[BoolElim]
  { \Gamma, z : \Bool \vdash A : s \\
    \Gamma \vdash m : \Bool \\
    \Gamma \vdash n_1 : A[\bTrue/z] \\
    \Gamma \vdash n_2 : A[\bFalse/z] }
  { \Gamma \vdash \boolElim{[z]A}{m}{n_1}{n_2} : A[m/z] }\\
\end{mathpar}

\paragraph{\textbf{Monadic Typing}}
The monadic typing rules govern the composition of monadic computations.
The standard rules for monadic return and bind are presented below.
\begin{mathpar}\small
  \inferrule[$\C$Type]
  { \Gamma \vdash A : s }
  { \Gamma \vdash \CM{A} : \Ln }

  \inferrule[Return]
  { \Gamma \vdash m : A }
  { \Gamma \vdash \return{m} : \CM{A} }

  \inferrule[Bind]
  { \Gamma \vdash B : s \\
    \Gamma \vdash m : \CM{A} \\
    \Gamma, x : A \vdash n : \CM{B} }
  { \Gamma \vdash \letin{x}{m}{n} : \CM{B} }\\
\end{mathpar}

\paragraph{\textbf{Session Typing}}
The session typing rules govern the typing of protocol, channels and concurrency
primitives. The rules are presented below.
\begin{mathpar}\footnotesize
  \inferrule[Proto]
  { \Gamma \vdash }
  { \Gamma \vdash \Proto : \Un }

  \inferrule[End]
  { \Gamma \vdash }
  { \Gamma \vdash \End : \Proto }

  \inferrule[Explicit-Action]
  { \Gamma, x : A \vdash B : \Proto }
  { \Gamma \vdash \ActR{\rho}{x : A}{B} : \Proto }

  \inferrule[Implicit-Action]
  { \Gamma, x : A \vdash B : \Proto }
  { \Gamma \vdash \ActI{\rho}{x : A}{B} : \Proto }

  \inferrule[RecProto]
  { \Gamma, x : A \vdash m : A \\
    A~\arity{\Proto} \\ 
    m~\guard{x} }
  { \Gamma \vdash \fix{x : A}{m} : A }

  \inferrule[ChType]
  { \Gamma \vdash A : \Proto }
  { \Gamma \vdash \CH{A} : \Ln }

  \inferrule[HcType]
  { \Gamma \vdash A : \Proto }
  { \Gamma \vdash \HC{A} : \Ln }

  \inferrule[Channel-CH]
  { \Gamma \vdash \\
    \epsilon \vdash A : \Proto }
  { \Gamma \vdash c : \CH{A} }

  \inferrule[Channel-HC]
  { \Gamma \vdash \\
    \epsilon \vdash A : \Proto }
  { \Gamma \vdash c : \HC{A} }

  \inferrule[Explicit-Send-CH]
  { \Gamma \vdash m : \CH{\ActR{!}{x : A}{B}} }
  { \Gamma \vdash \sendR{m} : \PiR{\Ln}{x : A}{\CM{\CH{B}}} }

  \inferrule[Explicit-Send-HC]
  { \Gamma \vdash m : \HC{\ActR{?}{x : A}{B}} }
  { \Gamma \vdash \sendR{m} : \PiR{\Ln}{x : A}{\CM{\HC{B}}} }

  \inferrule[Implicit-Send-CH]
  { \Gamma \vdash m : \CH{\ActI{!}{x : A}{B}} }
  { \Gamma \vdash \sendI{m} : \PiI{\Ln}{x : A}{\CM{\CH{B}}} }

  \inferrule[Implicit-Send-HC]
  { \Gamma \vdash m : \HC{\ActI{?}{x : A}{B}} }
  { \Gamma \vdash \sendI{m} : \PiI{\Ln}{x : A}{\CM{\HC{B}}} }

  \inferrule[Explicit-Recv-CH]
  { \Gamma \vdash m : \CH{\ActR{?}{x : A}{B}} }
  { \Gamma \vdash \recvR{m} : \CM{\SigR{\Ln}{x : A}{\CH{B}}} }

  \inferrule[Explicit-Recv-HC]
  { \Gamma \vdash m : \HC{\ActR{!}{x : A}{B}} }
  { \Gamma \vdash \recvR{m} : \CM{\SigR{\Ln}{x : A}{\HC{B}}} }

  \inferrule[Implicit-Recv-CH]
  { \Gamma \vdash m : \CH{\ActI{?}{x : A}{B}} }
  { \Gamma \vdash \recvI{m} : \CM{\SigI{\Ln}{x : A}{\CH{B}}} }

  \inferrule[Implicit-Recv-HC]
  { \Gamma \vdash m : \HC{\ActI{!}{x : A}{B}} }
  { \Gamma \vdash \recvI{m} : \CM{\SigI{\Ln}{x : A}{\HC{B}}} }

  \inferrule[Fork]
  { \Gamma, x : \CH{A} \vdash m : \CM{\unit} }
  { \Gamma \vdash \fork{x : \CH{A}}{m} : \CM{\HC{A}} }

  \inferrule[Close]
  { \Gamma \vdash m : \CH{\End} }
  { \Gamma \vdash \close{m} : \CM{\unit} }

  \inferrule[Wait]
  { \Gamma \vdash m : \HC{\End} }
  { \Gamma \vdash \wait{m} : \CM{\unit} }
\end{mathpar}
\clearpage

\subsection{Program Level}\label{appendix:program-typing}
The typing judgment for the program level has the form $\Theta ; \Gamma ; \Delta \vdash m : A$.
This judgment states that under the channel context $\Theta$, logical context $\Gamma$ and the
\emph{program context} $\Delta$, term $m$ has type $A$. The program level is \emph{substructural}
as the usage of variables in the program context is tracked.

\paragraph{\textbf{Program Context}}
The program context $\Delta$ is a sequence of variable bindings of the form
$x_0 :_{s_0} A_0, x_1 :_{s_1} A_1, \dots, x_n :_{s_n} A_n$. 
Each variable $x_i$ is bound to a type $A_i$ with a sort annotation $s_i$.
The variables in the program context are allowed to appear in computationally
relevant positions inside $m$. To ensure that all types appear in the program
context are well-formed, we define the program context validity judgment
$\Gamma ; \Delta \vdash $. The rules for this judgment are presented below.
\begin{mathpar}\small
  \inferrule[Ctx-Empty]
  { }
  { \epsilon ; \epsilon \vdash }

  \inferrule[Ctx-Implicit-Var]
  { \Gamma ; \Delta \vdash \\
    \Gamma \vdash A : s \\
    x \notin \Gamma }
  { \Gamma, x : A ; \Delta \vdash }

  \inferrule[Ctx-Explicit-Var]
  { \Gamma ; \Delta \vdash \\
    \Gamma \vdash A : s \\
    x \notin \Gamma }
  { \Gamma, x : A ; \Delta, x \ty{s} A \vdash }
\end{mathpar}
From these rules we can see that $\text{dom}(\Delta)$ is a subset of
$\text{dom}(\Gamma)$. Additionally, the sort annotation $s$ in each program
context binding $x \ty{s} A$ is the sort of the associated $A$ type.

\paragraph{\textbf{Core Typing}}
The core typing rules is responsible for the functional fragment of \TLLC{}.
\begin{mathpar}\footnotesize
  \inferrule[Var]
  { \epsilon ; \Gamma, x : A ; \Delta, x \ty{s} A \vdash \\
    \Delta \triangleright \Un }
  { \epsilon ; \Gamma, x : A ; \Delta, x \ty{s} A \vdash x : A }

  \inferrule[Conversion]
  { \Gamma \vdash B : s \\
    \Theta ; \Gamma ; \Delta \vdash m : A \\
    A \simeq B }
  { \Theta ; \Gamma ; \Delta \vdash m : B }

  \inferrule[Explicit-Lam]
  { \Theta ; \Gamma, x : A; \Delta, x \ty{s} A \vdash m : B \\
    \Theta \triangleright t \\
    \Delta \triangleright t }
  { \Theta ; \Gamma ; \Delta \vdash \lamR{t}{x : A}{m} : \PiR{t}{x : A}{B} }

  \inferrule[Implicit-Lam]
  { \Theta ; \Gamma, x : A; \Delta \vdash m : B \\
    \Theta \triangleright t \\
    \Delta \triangleright t }
  { \Theta ; \Gamma ; \Delta \vdash \lamI{t}{x : A}{m} : \PiI{t}{x : A}{B} }

  \inferrule[Explicit-App]
  { \Theta_1 ; \Gamma ; \Delta_1 \vdash m : \PiR{t}{x : A}{B} \\
    \Theta_2 ; \Gamma ; \Delta_2 \vdash n : A }
  { \Theta_1 \dotcup \Theta_2 ; \Gamma ; \Delta_1 \dotcup \Delta_2 \vdash \appR{m}{n} : B[n/x] }

  \inferrule[Implicit-App]
  { \Theta ; \Gamma ; \Delta \vdash m : \PiI{t}{x : A}{B} \\
    \Gamma \vdash n : A }
  { \Theta ; \Gamma ; \Delta \vdash \appI{m}{n} : B[n/x] }

  \inferrule[Explicit-Pair]
  { \Gamma \vdash \SigR{t}{x : A}{B} : t \\\\
    \Theta_1 ; \Gamma ; \Delta_1 \vdash m : A \\
    \Theta_2 ; \Gamma ; \Delta_2 \vdash n : B[m/x] }
  { \Theta_1 \dotcup \Theta_2 ; \Gamma ; \Delta_1 \dotcup \Delta_2 \vdash \pairR{m}{n}{t} : \SigR{t}{x : A}{B} }

  \inferrule[Implicit-Pair]
  { \Gamma \vdash \SigI{t}{x : A}{B} : t \\\\
    \Gamma \vdash m : A \\
    \Theta ; \Gamma ; \Delta \vdash n : B[m/x] }
  { \Theta ; \Gamma ; \Delta \vdash \pairI{m}{n}{t} : \SigI{t}{x : A}{B} }

  \inferrule[Explicit-SumElim]
  { \Gamma, z : \SigR{t}{x : A}{B} \vdash C : s \\
    \Theta_1 ; \Gamma ; \Delta_1 \vdash m : \SigR{t}{x : A}{B} \\\\
    \Theta_2 ; \Gamma, x : A, y : B; \Delta_2, x \ty{r1} A, y \ty{r2} B \vdash n : C[\pairR{x}{y}{t}/z] }
  { \Theta_1 \dotcup \Theta_2 ; \Gamma ; \Delta_1 \dotcup \Delta_2 \vdash \SigElim{[z]C}{m}{[x,y]n} : C[m/z] }

  \inferrule[Implicit-SumElim]
  { \Gamma, z : \SigI{t}{x : A}{B} \vdash C : s \\
    \Theta_1 ; \Gamma ; \Delta_1 \vdash m : \SigI{t}{x : A}{B} \\\\
    \Theta_2 ; \Gamma, x : A, y : B; \Delta_2, y \ty{r} B \vdash n : C[\pairI{x}{y}{t}/z] }
  { \Theta_1 \dotcup \Theta_2 ; \Gamma ; \Delta_1 \dotcup \Delta_2 \vdash \SigElim{[z]C}{m}{[x,y]n} : C[m/z] }\\
\end{mathpar}

\paragraph{\textbf{Data Typing}}
The data typing rules govern the typing of base types such as the unit type and
the boolean type. The rules are presented below.
\begin{mathpar}\footnotesize
  \inferrule[UnitVal]
  { \Gamma ; \Delta \vdash \\ \Delta \triangleright \Un }
  { \epsilon ; \Gamma ; \Delta \vdash \ii : \unit }

  \inferrule[True]
  { \Gamma ; \Delta \vdash \\ \Delta \triangleright \Un }
  { \epsilon ; \Gamma ; \Delta \vdash \bTrue : \Bool }

  \inferrule[False]
  { \Gamma ; \Delta \vdash \\ \Delta \triangleright \Un }
  { \epsilon ; \Gamma ; \Delta \vdash \bFalse : \Bool }

  \inferrule[BoolElim]
  { \Gamma, z : \Bool \vdash A : s \\
    \Theta_1 ; \Gamma ; \Delta_1 \vdash m : \Bool \\
    \Theta_2 ; \Gamma ; \Delta_2 \vdash n_1 : A[\bTrue/z] \\
    \Theta_2 ; \Gamma ; \Delta_2 \vdash n_2 : A[\bFalse/z] }
  { \Theta_1 \dotcup \Theta_2 ; \Gamma ; \Delta_1 \dotcup \Delta_2 \vdash 
    \boolElim{[z]A}{m}{n_1}{n_2} : A[m/z] }
\end{mathpar}

\paragraph{\textbf{Monadic Typing}}
The monadic typing rules govern the composition of monadic computations.
The standard rules for monadic return and bind are presented below.
\begin{mathpar}\small
  \inferrule[Return]
  { \Theta ; \Gamma ; \Delta \vdash m : A }
  { \Theta ; \Gamma ; \Delta \vdash \return{m} : \CM{A} }

  \inferrule[Bind]
  { \Gamma \vdash B : s \\
    \Theta_1 ; \Gamma ; \Delta_1 \vdash m : \CM{A} \\
    \Theta_2 ; \Gamma, x : A ; \Delta_2, x \ty{r} A \vdash n : \CM{B} }
  { \Theta_1 \dotcup \Theta_2 ; \Gamma ; \Delta_1 \dotcup \Delta_2 \vdash \letin{x}{m}{n} : \CM{B} }
\end{mathpar}

\paragraph{\textbf{Session Typing}}
The session typing rules govern the typing of channels and concurrency
primitives. The rules are presented below.
\begin{mathpar}\footnotesize
  \inferrule[Channel-CH]
  { \Gamma ; \Delta \vdash \\
    \epsilon \vdash A : \Proto \\
    \Delta \triangleright \Un }
  { c \tL \CH{A} ; \Gamma ; \Delta \vdash c : \CH{A} }

  \inferrule[Channel-HC]
  { \Gamma ; \Delta \vdash \\
    \epsilon \vdash A : \Proto \\
    \Delta \triangleright \Un }
  { c \tL \HC{A} ; \Gamma ; \Delta \vdash c : \HC{A} }

  \inferrule[Explicit-Send-CH]
  { \Theta ; \Gamma ; \Delta \vdash m : \CH{\ActR{!}{x : A}{B}} }
  { \Theta ; \Gamma ; \Delta \vdash \sendR{m} : \PiR{\Ln}{x : A}{\CM{\CH{B}}} }

  \inferrule[Explicit-Send-HC]
  { \Theta ; \Gamma ; \Delta \vdash m : \HC{\ActR{?}{x : A}{B}} }
  { \Theta ; \Gamma ; \Delta \vdash \sendR{m} : \PiR{\Ln}{x : A}{\CM{\HC{B}}} }

  \inferrule[Implicit-Send-CH]
  { \Theta ; \Gamma ; \Delta \vdash m : \CH{\ActI{!}{x : A}{B}} }
  { \Theta ; \Gamma ; \Delta \vdash \sendI{m} : \PiI{\Ln}{x : A}{\CM{\CH{B}}} }

  \inferrule[Implicit-Send-HC]
  { \Theta ; \Gamma ; \Delta \vdash m : \HC{\ActI{?}{x : A}{B}} }
  { \Theta ; \Gamma ; \Delta \vdash \sendI{m} : \PiI{\Ln}{x : A}{\CM{\HC{B}}} }

  \inferrule[Explicit-Recv-CH]
  { \Theta ; \Gamma ; \Delta \vdash m : \CH{\ActR{?}{x : A}{B}} }
  { \Theta ; \Gamma ; \Delta \vdash \recvR{m} : \CM{\SigR{\Ln}{x : A}{\CH{B}}} }

  \inferrule[Explicit-Recv-HC]
  { \Theta ; \Gamma ; \Delta \vdash m : \HC{\ActR{!}{x : A}{B}} }
  { \Theta ; \Gamma ; \Delta \vdash \recvR{m} : \CM{\SigR{\Ln}{x : A}{\HC{B}}} }

  \inferrule[Implicit-Recv-CH]
  { \Theta ; \Gamma ; \Delta \vdash m : \CH{\ActI{?}{x : A}{B}} }
  { \Theta ; \Gamma ; \Delta \vdash \recvI{m} : \CM{\SigI{\Ln}{x : A}{\CH{B}}} }

  \inferrule[Implicit-Recv-HC]
  { \Theta ; \Gamma ; \Delta \vdash m : \HC{\ActI{!}{x : A}{B}} }
  { \Theta ; \Gamma ; \Delta \vdash \recvI{m} : \CM{\SigI{\Ln}{x : A}{\HC{B}}} }

  \inferrule[Fork]
  { \Theta ; \Gamma, x : \CH{A} ; \Delta, x \tL \CH{A} \vdash m : \CM{\unit} }
  { \Theta ; \Gamma ; \Delta \vdash \fork{x : \CH{A}}{m} : \CM{\HC{A}} }

  \inferrule[Close]
  { \Theta ; \Gamma ; \Delta \vdash m : \CH{\End} }
  { \Theta ; \Gamma ; \Delta \vdash \close{m} : \CM{\unit} }

  \inferrule[Wait]
  { \Theta ; \Gamma ; \Delta \vdash m : \HC{\End} }
  { \Theta ; \Gamma ; \Delta \vdash \wait{m} : \CM{\unit} }
\end{mathpar}

\subsection{Process Level}
The typing judgment for the process level has the form $\Theta \Vdash P$. 
This judgment states that under the channel context $\Theta$, process $P$ is well-typed.
Unlike the logical and program levels which can type term that contain free variables, 
the process level only types processes whose terms are closed. Hence, there are no logical
or program contexts in the process typing judgment.
\begin{mathpar}\small
  \inferrule[Expr]
  { \Theta ; \epsilon ; \epsilon \vdash m : \CM{\unit} }
  { \Theta \Vdash \proc{m} }

  \inferrule[Par]
  { \Theta_1 \Vdash P \\ \Theta_2 \Vdash Q }
  { \Theta_1 \dotcup \Theta_2 \Vdash P \mid Q }

  \inferrule[Scope]
  { \Theta, c \tL \CH{A}, d \tL \HC{A} \Vdash P }
  { \Theta \Vdash \scope{cd}{P} }
\end{mathpar}
\clearpage

\section{Operational Semantics}\label{appendix:semantics}
In this section, we present the operational semantics of \TLLC{}. 
Similarly to the typing rules, we organize the presentation of the semantics
into the logical level, program level, and process level.

\subsection{Logical Level}\label{appendix:logical-semantics}
The semantics of the logical level is defined in terms of the 
\emph{parallel reduction} relation $m \Rightarrow m'$. This relation allows 
multiple redexes to be reduced simultaneously.

\paragraph{\textbf{Core Reduction}}
The parallel reduction for core functional terms is defined as follows:
\begin{mathpar}\footnotesize
  \inferrule[PStep-Var]
  { }
  { x \Rightarrow x }

  \inferrule[PStep-Sort]
  { }
  { s \Rightarrow s }

  \inferrule[PStep-Explicit-Fun]
  { A \Rightarrow A' \\ B \Rightarrow B' }
  { \PiR{s}{x : A}{B} \Rightarrow \PiR{s}{x : A'}{B'} }

  \inferrule[PStep-Implicit-Fun]
  { A \Rightarrow A' \\ B \Rightarrow B' }
  { \PiI{s}{x : A}{B} \Rightarrow \PiI{s}{x : A'}{B'} }

  \inferrule[PStep-Explicit-Lam]
  { A \Rightarrow A' \\ B \Rightarrow B' }
  { \lamR{s}{x : A}{B} \Rightarrow \lamR{s}{x : A'}{B'} }

  \inferrule[PStep-Implicit-Lam]
  { A \Rightarrow A' \\ B \Rightarrow B' }
  { \lamI{s}{x : A}{B} \Rightarrow \lamI{s}{x : A'}{B'} }

  \inferrule[PStep-Explicit-App]
  { m \Rightarrow m' \\ n \Rightarrow n' }
  { \appR{m}{n} \Rightarrow \appR{m'}{n'} }

  \inferrule[PStep-Implicit-App]
  { m \Rightarrow m' \\ n \Rightarrow n' }
  { \appI{m}{n} \Rightarrow \appI{m'}{n'} }

  \inferrule[PStep-Explicit-$\beta$]
  { m \Rightarrow m' \\ n \Rightarrow n' }
  { \appR{(\lamR{s}{x : A}{m})}{n} \Rightarrow m'[n'/x] }

  \inferrule[PStep-Implicit-$\beta$]
  { m \Rightarrow m' \\ n \Rightarrow n' }
  { \appI{(\lamI{s}{x : A}{m})}{n} \Rightarrow m'[n'/x] }

  \inferrule[PStep-Explicit-Sum]
  { A \Rightarrow A' \\ B \Rightarrow B' }
  { \SigR{s}{x : A}{B} \Rightarrow \SigR{s}{x : A'}{B'} }

  \inferrule[PStep-Implicit-Sum]
  { A \Rightarrow A' \\ B \Rightarrow B' }
  { \SigI{s}{x : A}{B} \Rightarrow \SigI{s}{x : A'}{B'} }

  \inferrule[PStep-Explicit-Pair]
  { m \Rightarrow m' \\ n \Rightarrow n' }
  { \pairR{m}{n}{s} \Rightarrow \pairR{m'}{n'}{s} }

  \inferrule[PStep-Implicit-Pair]
  { m \Rightarrow m' \\ n \Rightarrow n' }
  { \pairI{m}{n}{s} \Rightarrow \pairI{m'}{n'}{s} }

  \inferrule[PStep-SumElim]
  { A \Rightarrow A' \\ m \Rightarrow m' \\ n \Rightarrow n' }
  { \SigElim{[z]A}{m}{[x,y]n} \Rightarrow \SigElim{[z]A'}{m'}{[x,y]n'} }

  \inferrule[PStep-Explicit-PairElim]
  { m_1 \Rightarrow m_1' \\ m_2 \Rightarrow m_2' \\ n \Rightarrow n' }
  { \SigElim{[z]A}{\pairR{m_1}{m_2}{s}}{[x,y]n} \Rightarrow n'[m_1/x,m_2/y] }

  \inferrule[PStep-Implicit-PairElim]
  { m_1 \Rightarrow m_1' \\ m_2 \Rightarrow m_2' \\ n \Rightarrow n' }
  { \SigElim{[z]A}{\pairI{m_1}{m_2}{s}}{[x,y]n} \Rightarrow n'[m_1/x,m_2/y] }
\end{mathpar}

\paragraph{\textbf{Data Reduction}}
The parallel reduction for data terms is defined as follows:
\begin{mathpar}\footnotesize
  \inferrule[PStep-Unit]
  { }
  { \unit \Rightarrow \unit }

  \inferrule[PStep-UnitVal]
  { }
  { \ii \Rightarrow \ii }

  \inferrule[PStep-Bool]
  { }
  { \Bool \Rightarrow \Bool }

  \inferrule[PStep-True]
  { }
  { \bTrue \Rightarrow \bTrue }

  \inferrule[PStep-False]
  { }
  { \bFalse \Rightarrow \bFalse }

  \inferrule[PStep-BoolElim]
  { A \Rightarrow A' \\
    m \Rightarrow m' \\
    n_1 \Rightarrow n_1' \\
    n_2 \Rightarrow n_2' }
  { \boolElim{[z]A}{m}{n_1}{n_2} \Rightarrow \boolElim{[z]A'}{m'}{n_1'}{n_2'} }

  \inferrule[PStep-TrueElim]
  { n_1 \Rightarrow n_1' }
  { \boolElim{[z]A}{\bTrue}{n_1}{n_2} \Rightarrow n_1' }

  \inferrule[PStep-FalseElim]
  { n_2 \Rightarrow n_2' }
  { \boolElim{[z]A}{\bFalse}{n_1}{n_2} \Rightarrow n_2' }
\end{mathpar}

\paragraph{\textbf{Monadic Reduction}}
The parallel reduction for monadic terms is defined as follows:
\begin{mathpar}\footnotesize
  \inferrule[PStep-$\mcC$Type]
  { A \Rightarrow A' }
  { \CM{A} \Rightarrow \CM{A'} }

  \inferrule[PStep-Return]
  { m \Rightarrow m' }
  { \return{m} \Rightarrow \return{m'} }

  \inferrule[PStep-Bind]
  { m \Rightarrow m' \\ n \Rightarrow n' }
  { \letin{x}{m}{n} \Rightarrow \letin{x}{m'}{n'} }

  \inferrule[PStep-ReturnBind]
  { m \Rightarrow m' \\ n \Rightarrow n' }
  { \letin{x}{\return{m}}{n} \Rightarrow n'[m'/x] }
\end{mathpar}
\clearpage

\paragraph{\textbf{Session Reduction}}
The parallel reduction for protocols, channels and concurrency primitives are defined as follows:
\begin{mathpar}\footnotesize
  \inferrule[PStep-Proto]
  { }
  { \Proto \Rightarrow \Proto }

  \inferrule[PStep-End]
  { }
  { \End \Rightarrow \End }

  \inferrule[PStep-RecProto]
  { A \Rightarrow A' \\ m \Rightarrow m' }
  { \fix{x : A}{m} \Rightarrow \fix{x : A'}{m'} }

  \inferrule[PStep-RecUnfold]
  { A \Rightarrow A' \\ m \Rightarrow m' }
  { \fix{x : A}{m} \Rightarrow m'[(\fix{x : A'}{m'})/x] }

  \inferrule[PStep-Explicit-Action]
  { A \Rightarrow A' \\ B \Rightarrow B' }
  { \ActR{\rho}{x : A}{B} \Rightarrow \ActR{\rho}{x : A'}{B'} }

  \inferrule[PStep-Implicit-Action]
  { A \Rightarrow A' \\ B \Rightarrow B' }
  { \ActI{\rho}{x : A}{B} \Rightarrow \ActI{\rho}{x : A'}{B'} }

  \inferrule[PStep-CH]
  { A \Rightarrow A' }
  { \CH{A} \Rightarrow \CH{A'} }

  \inferrule[PStep-HC]
  { A \Rightarrow A' }
  { \HC{A} \Rightarrow \HC{A'} }

  \inferrule[PStep-Channel]
  { }
  { c \Rightarrow d }

  \inferrule[PStep-Fork]
  { A \Rightarrow A' \\ m \Rightarrow m' }
  { \fork{x : A}{m} \Rightarrow \fork{x : A'}{m'} }

  \inferrule[PStep-Explicit-Send]
  { m \Rightarrow m' }
  { \sendR{m} \Rightarrow \sendR{m'} }

  \inferrule[PStep-Implicit-Send]
  { m \Rightarrow m' }
  { \sendI{m} \Rightarrow \sendI{m'} }

  \inferrule[PStep-Explicit-Recv]
  { m \Rightarrow m' }
  { \recvR{m} \Rightarrow \recvR{m'} }

  \inferrule[PStep-Implicit-Recv]
  { m \Rightarrow m' }
  { \recvI{m} \Rightarrow \recvI{m'} }

  \inferrule[PStep-Close]
  { m \Rightarrow m' }
  { \close{m} \Rightarrow \close{m'} }

  \inferrule[PStep-Wait]
  { m \Rightarrow m' }
  { \wait{m} \Rightarrow \wait{m'} }
\end{mathpar}

\paragraph{\textbf{Convertibility Relation}}
The convertibility relation $A \simeq B$ is the reflexive, symmetric and
transitive closure of the parallel reduction relation. It can be inductively
defined as follows:
\begin{mathpar}
  \inferrule[Conv-Refl]
  { }
  { A \simeq A }

  \inferrule[Conv-PStep]
  { A \simeq B \\ 
    B \Rightarrow C }
  { A \simeq C }

  \inferrule[Conv-PStep-Rev]
  { A \simeq B \\ 
    C \Rightarrow B }
  { A \simeq C }
\end{mathpar}
Note that the program level \textsc{Conversion} rule (\Cref{appendix:program-typing}) 
also uses this convertibility relation.
\clearpage

\subsection{Program Level}\label{appendix:program-semantics}
The semantics of the program level is defined in terms of a small-step reduction
relation $m \Leadsto m'$. Unlike the logical level which has a non-deterministic
reduction strategy, the program level follows call-by-value evaluation.
Arguments are fully evaluated before substitution into functions.

\paragraph{\textbf{Core Reduction}}
The small-step reduction for core functional terms is defined as follows:
\begin{mathpar}\footnotesize
  \inferrule[Step-Explicit-App$_1$]
  { m \Leadsto m' }
  { \appR{m}{n} \Leadsto \appR{m'}{n} }

  \inferrule[Step-Explicit-App$_2$]
  { n \Leadsto n' }
  { \appR{m}{n} \Leadsto \appR{m}{n'} }

  \inferrule[Step-Implicit-App$_1$]
  { m \Leadsto m' }
  { \appI{m}{n} \Leadsto \appI{m'}{n} }
  
  \inferrule[Step-Explicit-$\beta$]
  { }
  { \appR{(\lamR{s}{x : A}{m})}{v} \Leadsto m[v/x] }

  \inferrule[Step-Implicit-$\beta$]
  { }
  { \appI{(\lamI{s}{x : A}{m})}{n} \Leadsto m[n/x] }

  \inferrule[Step-Explicit-Pair$_1$]
  { m \Leadsto m' }
  { \pairR{m}{n}{s} \Leadsto \pairR{m'}{n}{s} }

  \inferrule[Step-Explicit-Pair$_2$]
  { n \Leadsto n' }
  { \pairR{m}{n}{s} \Leadsto \pairR{m}{n'}{s} }

  \inferrule[Step-Implicit-Pair$_2$]
  { n \Leadsto n' }
  { \pairI{m}{n}{s} \Leadsto \pairI{m}{n'}{s} }

  \inferrule[Step-SumElim$_1$]
  { m \Leadsto m' }
  { \SigElim{[z]A}{m}{[x,y]n} \Leadsto \SigElim{[z]A}{m'}{[x,y]n} }

  \inferrule[Step-Explicit-PairElim]
  { }
  { \SigElim{[z]A}{\pairR{u}{v}{s}}{[x,y]n} \Leadsto n[u/x,v/y] }

  \inferrule[Step-Implicit-PairElim]
  { }
  { \SigElim{[z]A}{\pairI{m}{v}{s}}{[x,y]n} \Leadsto n[m/x,v/y] }
\end{mathpar}

\paragraph{\textbf{Data Reduction}}
The small-step reduction for data terms is defined as follows:
\begin{mathpar}\small
  \inferrule[Step-BoolElim]
  { m \Leadsto m' }
  { \boolElim{[z]A}{m}{n_1}{n_2} \Leadsto \boolElim{[z]A}{m'}{n_1}{n_2} }

  \inferrule[Step-TrueElim]
  { }
  { \boolElim{[z]A}{\bTrue}{n_1}{n_2} \Leadsto n_1 }

  \inferrule[Step-FalseElim]
  { }
  { \boolElim{[z]A}{\bFalse}{n_1}{n_2} \Leadsto n_2 }
\end{mathpar}

\paragraph{\textbf{Monadic Reduction}}
The small-step reduction for monadic terms is defined as follows:
\begin{mathpar}\small
  \inferrule[Step-Return]
  { m \Leadsto m' }
  { \return{m} \Leadsto \return{m'} }

  \inferrule[Step-Bind]
  { m \Leadsto m' }
  { \letin{x}{m}{n} \Leadsto \letin{x}{m'}{n} }

  \inferrule[Step-ReturnBind]
  { }
  { \letin{x}{\return{v}}{n} \Leadsto n[v/x] }
\end{mathpar}

\paragraph{\textbf{Session Reduction}}
The small-step reduction for session terms is defined as follows:
\begin{mathpar}\small
  \inferrule[Step-Explicit-Send]
  { m \Leadsto m' }
  { \sendR{m} \Leadsto \sendR{m'} }

  \inferrule[Step-Implicit-Send]
  { m \Leadsto m' }
  { \sendI{m} \Leadsto \sendI{m'} }

  \inferrule[Step-Explicit-Recv]
  { m \Leadsto m' }
  { \recvR{m} \Leadsto \recvR{m'} }

  \inferrule[Step-Implicit-Recv]
  { m \Leadsto m' }
  { \recvI{m} \Leadsto \recvI{m'} }

  \inferrule[Step-Close]
  { m \Leadsto m' }
  { \close{m} \Leadsto \close{m'} }

  \inferrule[Step-Wait]
  { m \Leadsto m' }
  { \wait{m} \Leadsto \wait{m'} }
\end{mathpar}
\clearpage

\subsection{Process Level}\label{appendix:process-semantics}
The semantics of the process level is defined in terms of a small-step reduction
relation $P \Rrightarrow Q$. This relation is what gives \TLLC{} its concurrent
behavior. Before we present the reduction rules, we first define the notion of
\emph{structural congruence} $\equiv$ which identifies processes that are
the same up to reordering of parallel components and renaming of bound channels.

\paragraph{\textbf{Structural Congruence}}
The structural congruence relation $\equiv$ is defined as follows:
\begin{mathpar}\small
  P \mid Q \equiv Q \mid P 

  O \mid (P \mid Q) \equiv (O \mid P) \mid Q

  P \mid \proc{\return{\ii}} \equiv P
  \\
  \scope{cd}{P} \mid Q \equiv \scope{cd}{(P \mid Q)}

  \scope{cd}{P} \equiv \scope{dc}{P}

  \scope{cd}{\scope{c'd'}{P}} \equiv \scope{c'd'}{\scope{cd}{P}}
\end{mathpar}

\paragraph{\textbf{Process Reduction}}
The small-step reduction for processes is defined as follows:

\begin{tabular}{r L C L}
  Evaluation Contexts & \mcM, \mcN & ::= & [\cdot] \mid \letin{x}{\mcM}{m}
\end{tabular}

\vspace{0.5em}
\begin{small}
\begin{tabular}{l L}
  (\textsc{Proc-Fork}) &
    \proc{
      \mcN[\fork{x : A}{m}]
    }
    \Rrightarrow
    \scope{cd}{(\proc{\mcN[\return{c}]} \mid \proc{m[d/x]})} 
  \\
  (\textsc{Proc-End}) 
    &\scope{cd}{(
        \proc{\mcM[\close{c}]} 
        \mid 
        \proc{\mcN[\wait{d}]}
      )}
    \Rrightarrow 
    \proc{\mcM[\return{\ii}]} \mid \proc{\mcN[\return{\ii}]} 
  \\
  (\textsc{Proc-Com}) 
    &\scope{cd}{(
        \proc{\mcM[\appR{\sendR{c}}{v}]} 
        \mid 
        \proc{\mcN[\recvR{d}]}
      )}
    \Rrightarrow 
    \scope{cd}{(
      \proc{\mcM[\return{c}]} 
      \mid 
      \proc{\mcN[\return{\pairR{v}{d}{\Ln}}]}
    )}
  \\
  (\textsc{Proc-\underline{Com}}) 
    &\scope{cd}{(
        \proc{\mcM[\appI{\sendI{c}}{o}]} 
        \mid 
        \proc{\mcN[\recvI{d}]}
      )}
     \Rrightarrow 
     \scope{cd}{(
        \proc{\mcM[\return{c}]} 
        \mid 
        \proc{\mcN[\return{\pairI{o}{d}{\Ln}}]}
      )}
\end{tabular}
\vspace{0.2em}
\begin{mathpar}
  \inferrule[(Proc-Expr)]
  { m \Leadsto m' }
  { \proc{m} \Rrightarrow \proc{m'} }

  \inferrule[(Proc-Par)]
  { P \Rrightarrow Q }
  { O \mid P \Rrightarrow O \mid Q }

  \inferrule[(Proc-Scope)]
  { P \Rrightarrow Q }
  { \scope{cd}{P} \Rrightarrow \scope{cd}{Q} }

  \inferrule[(Proc-Congr)]
  { P \equiv P' \\ 
    P' \Rrightarrow Q' \\ 
    Q' \equiv Q }
  { P \Rrightarrow Q }
\end{mathpar}
\end{small}
\clearpage

\section{Meta-Theory}\label{appendix:metatheory}
In this section, we study the meta-theory of \TLLC{}.
The results are classified into 4 categories.
\begin{enumerate}
  \item \textbf{Compatibility}: the extensions of \TLLC{} are compatible with the underlying TLL theory.
  \item \textbf{Session Fidelity}: processes follow the protocols specified by their session types.
  \item \textbf{Global Progress}: well-typed reachable processes do not deadlock.
  \item \textbf{Erasure Safety}: evaluation of erased terms and processes is safe.
\end{enumerate}

\subsection{Compatibility}\label{appendix:compatibility}
To show that the \TLLC{} extensions made are compatible with the theory of TLL,
we prove the \TLLC{} terms enjoy the same properties as TLL ones. 

\paragraph{\textbf{Confluence}}
We show that the logical reduction relation is confluent. 
This property is important as it ensures that type convertibility can be checked
regardless of the order in which reductions are applied. 
Confluence is easy to prove here as the parallel reduction satisfies the diamond property.

\begin{lemma}[Diamond Property]\label[lemma]{lemma:diamond}
  If $m \Rightarrow m_1$ and $m \Rightarrow m_2$, then there exists $m'$ such that
  $m_1 \Rightarrow m'$ and $m_2 \Rightarrow m'$.
\end{lemma}
\begin{proof}
  By induction on the structure of the parallel reduction.
\end{proof}

\begin{lemma}\label[lemma]{lemma:strip}
  If $m \Rightarrow m_1$ and $m \Rightarrow^* m_2$, then there exists $m'$ such that
  $m_1 \Rightarrow^* m'$ and $m_2 \Rightarrow m'$.
\end{lemma}
\begin{proof}
  By induction on the derivation of $m \Rightarrow^* m_2$ and \Cref{lemma:diamond}.
\end{proof}

\begin{theorem}[Confluence]
  If $m \Rightarrow^* m_1$ and $m \Rightarrow^* m_2$, then there exists $m'$ such that
  $m_1 \Rightarrow^* m'$ and $m_2 \Rightarrow^* m'$.
\end{theorem}
\begin{proof}
  By induction on the derivation of $m \Rightarrow^* m_2$ and \Cref{lemma:strip}.
\end{proof}

The validity of the confluence property allows us to prove the injectivity of
the convertibility relation for types.

\begin{corollary}\label[corollary]{corollary:inj-sort}
  $s_1 \simeq s_2$ implies $s_1 = s_2$.
\end{corollary}

\begin{corollary}\label[corollary]{corollary:inj-implicit-fun}
  $\PiI{s}{x : A}{B} \simeq \PiI{s'}{x : A'}{B'}$ implies $s = s'$, $A \simeq A'$, and $B \simeq B'$.
\end{corollary}

\begin{corollary}\label[corollary]{corollary:inj-explicit-fun}
  $\PiR{s}{x : A}{B} \simeq \PiR{s'}{x : A'}{B'}$ implies $s = s'$, $A \simeq A'$, and $B \simeq B'$.
\end{corollary}

\begin{corollary}\label[corollary]{corollary:inj-implicit-sig}
  $\SigI{s}{x : A}{B} \simeq \SigI{s'}{x : A'}{B'}$ implies $s = s'$, $A \simeq A'$, and $B \simeq B'$.
\end{corollary}

\begin{corollary}\label[corollary]{corollary:inj-explicit-sig}
  $\SigR{s}{x : A}{B} \simeq \SigR{s'}{x : A'}{B'}$ implies $s = s'$, $A \simeq A'$, and $B \simeq B'$.
\end{corollary}

\begin{corollary}\label[corollary]{corollary:inj-monad}
  $\CM{A} \simeq \CM{B}$ implies $A \simeq B$.
\end{corollary}

\begin{corollary}\label[corollary]{corollary:inj-implicit-action}
  $\ActI{\rho}{x : A}{B} \simeq \ActI{\rho'}{x : A'}{B'}$ implies $\rho = \rho'$, $A \simeq A'$, and $B \simeq B'$.
\end{corollary}

\begin{corollary}\label[corollary]{corollary:inj-explicit-action}
  $\ActR{\rho}{x : A}{B} \simeq \ActR{\rho'}{x : A'}{B'}$ implies $\rho = \rho'$, $A \simeq A'$, and $B \simeq B'$.
\end{corollary}

\begin{corollary}\label[corollary]{corollary:inj-ch}
  $\CH{A} \simeq \CH{B}$ implies $A \simeq B$. 
\end{corollary}

\begin{corollary}\label[corollary]{corollary:inj-hc}
  $\HC{A} \simeq \HC{B}$ implies $A \simeq B$.
\end{corollary}

\paragraph{\textbf{Weakening}}
Weakening allows for the addition of unused variables to a typing context. 
The logical level type system allows weakening as it is a fully structural type system.
On the other hand, the program level type system only allows weakening of unrestricted variables, 
i.e. variables whose types inhabit $\Un$.

\begin{lemma}[Renaming Arity]
  Given renaming $\xi$, if there is $A~\arity{\Proto}$, then there is $A[\xi]~\arity{\Proto}$.
\end{lemma}
\begin{proof}
  By induction on the structure of $A$.
\end{proof}

\begin{lemma}[Renaming Guarded]
  Given renaming $\xi$, 
  if there is $\forall x, y, \xi(x) = \xi(y) \implies x = y$, 
  then given variable $x$ and $A~\guard{x}$, there is $A[\xi]~\guard{\xi(x)}$.
\end{lemma}
\begin{proof}
  By induction on the structure of $A$.
\end{proof}

\begin{lemma}[Logical Weakening]
  If $\Gamma \vdash m : A$ and $\Gamma \vdash B : s$, 
  then $\Gamma, x : B \vdash m : A$ where $x \not\in \Gamma$.
\end{lemma}
\begin{proof}
  By induction on the derivation of $\Gamma \vdash m : A$.
  For more details, see file \textsf{sta\_weak.v} of our Rocq development
  which uses a De Bruijn indices representation for variables.
\end{proof}

\begin{lemma}[Program Weakening (Explicit)]
  If $\Theta ; \Gamma ; \Delta \vdash m : A$ and $\Gamma \vdash B : \Un$,
  then $\Theta ; \Gamma, x : B ; \Delta, x :_\Un B \vdash m : A$ where $x \not\in \Gamma$.
\end{lemma}
\begin{proof}
  By induction on the derivation of $\Theta ; \Gamma ; \Delta \vdash m : A$.
  For more details, see file \textsf{dyn\_weak.v} of our Rocq development
  which uses a De Bruijn indices representation for variables.
\end{proof}

\begin{lemma}[Program Weakening (Implicit)]
  If $\Theta ; \Gamma ; \Delta \vdash m : A$ and $\Delta \vdash B : \Ln$,
  then $\Theta ; \Gamma, x : B ; \Delta \vdash m : A$ where $x \not\in \Delta$.
\end{lemma}
\begin{proof}
  By induction on the derivation of $\Theta ; \Gamma ; \Delta \vdash m : A$.
  For more details, see file \textsf{dyn\_weak.v} of our Rocq development
  which uses a De Bruijn indices representation for variables.
\end{proof}

\paragraph{\textbf{Substitution}}
The substitution lemma at the logical level is standard as the logical type system is
completely structural. The substitution lemma at the program level is more involved
as it needs to track linear variables in the program context.

\begin{lemma}[Substitution Arity]\label[lemma]{lemma:subst-arity}
  Given substitution $\sigma$, if there is $A~\arity{\Proto}$, then there is $A[\sigma]~\arity{\Proto}$.
\end{lemma}
\begin{proof}
  By induction on the structure of $A$.
\end{proof}

\begin{lemma}[Substitution Guarded]\label[lemma]{lemma:subst-guard}
  Given substitution $\sigma$ and variables $x, y$ and term $A$,
  if there is $\forall z, x \neq z \implies (\sigma\ z)~\guard{y}$
  and $A~\guard{x}$, then $A[\sigma]~\guard{y}$.
\end{lemma}
\begin{proof}
  By induction on the structure of $A$.
\end{proof}

\begin{lemma}[Logical Substitution]\label[lemma]{lemma:logical-subst}
  If $\Gamma, x : B \vdash m : A$ and $\Gamma \vdash n : B$, then $\Gamma \vdash m[n/x] : A[n/x]$.
\end{lemma}
\begin{proof}
  This lemma is proved through a more general lemma involving simultaneous substitutions. 
  For more details, see file \textsf{sta\_subst.v} of our Rocq development.
\end{proof}

\begin{corollary}\label[corollary]{corollary:logical-context-conv}
  If $\Gamma, x : A \vdash m : C$ and $A \simeq B$, then $\Gamma, x : B \vdash m : C$.
\end{corollary}

\begin{lemma}[Program Substitution (Explicit)]\label[lemma]{lemma:program-subst-explicit}
  If ${\Theta_1 ; \Gamma, x : B ; \Delta_1, x :_s B \vdash m : A}$ and
  ${\Theta_2 ; \Gamma ; \Delta_2 \vdash n : B}$ and $\Theta_2 \triangleright s$ and $\Delta_2 \triangleright s$, then
  ${\Theta_1 \dotcup \Theta_2 ; \Gamma ; \Delta_1 \dotcup \Delta_2 \vdash m[n/x] : A[n/x]}$.
\end{lemma}
\begin{proof}
  This lemma is proved through a more general lemma involving simultaneous substitutions. 
  For more details, see file \textsf{dyn\_subst.v} of our Rocq development.
\end{proof}

\begin{corollary}\label[corollary]{corollary:program-context-conv-explicit}
  If $\Theta ; \Gamma, x : A ; \Delta \vdash m : C$ and $\Gamma \vdash B : s$ and 
  $A \simeq B$, then $\Theta ; \Gamma, x : B ; \Delta \vdash m : C$.
\end{corollary}

\begin{lemma}[Program Substitution (Implicit)]\label[lemma]{lemma:program-subst-implicit}
  If ${\Theta ; \Gamma, x : B ; \Delta \vdash m : A}$ and
  ${\Gamma \vdash n : B}$, then ${\Theta ; \Gamma ; \Delta \vdash m[n/x] : A[n/x]}$.
\end{lemma}
\begin{proof}
  This lemma is proved through a more general lemma involving simultaneous substitutions. 
  For more details, see file \textsf{dyn\_subst.v} of our Rocq development.
\end{proof}

\begin{corollary}\label[corollary]{corollary:program-context-conv-implicit}
  If ${\Theta ; \Gamma, x : A ; \Delta, x :_s A \vdash m : C}$ and
  $A \simeq B$, then ${\Theta ; \Gamma, x : B ; \Delta, x :_s B \vdash m : C}$.
\end{corollary}

\paragraph{\textbf{Sort Uniqueness}}
Due to the fact that \TLLC{} utilizes the sort of types to determine
(sub)structural properties of their inhabitants, it is important for types
to have unique sorts. If a type could have multiple sorts, then it would be
ambiguous whether its inhabitants are linear or unrestricted.

One of the main challenges in proving sort uniqueness is that there is no
uniqueness of types in general. In particular, dependent pairs like $\pairI{m}{n}{s}$
do not have unique typing. For this reason, we prove a weaker property of
\emph{type similarity} instead of type uniqueness. Then from type similarity
we can derive sort uniqueness. We begin by defining the \emph{head similarity}
relation as follows:
\begin{mathpar}\footnotesize
  \inferrule  
  { }
  { \HeadSim{x}{x} }

  \inferrule
  { }
  { \HeadSim{s}{s} }

  \inferrule
  { \HeadSim{B_1}{B_2} }
  { \HeadSim{\PiR{s}{x : A_1}{B_1}}{\PiR{s}{x : A_2}{B_2}} }

  \inferrule
  { \HeadSim{B_1}{B_2} }
  { \HeadSim{\PiI{s}{x : A_1}{B_1}}{\PiI{s}{x : A_2}{B_2}} }

  \inferrule
  { }
  { \HeadSim{\lamR{s}{x : A}{m}}{\lamR{s}{x : A}{m}} }

  \inferrule
  { }
  { \HeadSim{\lamI{s}{x : A}{m}}{\lamI{s}{x : A}{m}} }

  \inferrule
  { }
  { \HeadSim{\appR{m}{n}}{\appR{m}{n}} }

  \inferrule
  { }
  { \HeadSim{\appI{m}{n}}{\appI{m}{n}} }

  \inferrule
  { }
  { \HeadSim{\SigR{s}{x : A_1}{B_1}}{\SigR{s}{x : A_2}{B_2}} }

  \inferrule
  { }
  { \HeadSim{\SigI{s}{x : A_1}{B_1}}{\SigI{s}{x : A_2}{B_2}} }

  \inferrule
  { }
  { \HeadSim{\pairR{m}{n}{s}}{\pairR{m}{n}{s}} }

  \inferrule
  { }
  { \HeadSim{\pairI{m}{n}{s}}{\pairI{m}{n}{s}} }

  \inferrule
  { }
  { \HeadSim{\SigElim{[z]A}{m}{[x,y]n}}{\SigElim{[z]A}{m}{[x,y]n}} }

  \inferrule
  { }
  { \HeadSim{\fix{x : A}{m}}{\fix{x : A}{m}} }

  \inferrule
  { }
  { \HeadSim{\unit}{\unit} }

  \inferrule
  { }
  { \HeadSim{\ii}{\ii} }

  \inferrule
  { }
  { \HeadSim{\Bool}{\Bool} }

  \inferrule
  { }
  { \HeadSim{\bTrue}{\bTrue} }

  \inferrule
  { }
  { \HeadSim{\bFalse}{\bFalse} }

  \inferrule
  { }
  { \HeadSim{\boolElim{[z]A}{m}{n_1}{n_2}}{\boolElim{[z]A}{m}{n_1}{n_2}} }

  \inferrule
  { }
  { \HeadSim{\CM{A}}{\CM{B}} }

  \inferrule
  { }
  { \HeadSim{\return{m}}{\return{m}} }

  \inferrule
  { }
  { \HeadSim{\letin{m}{x}{n}}{\letin{m}{x}{n}} }

  \inferrule
  { }
  { \HeadSim{\Proto}{\Proto} }

  \inferrule
  { }
  { \HeadSim{\End}{\End} }

  \inferrule
  { }
  { \HeadSim{\ActR{\rho}{x : A}{B}}{\ActR{\rho}{x : A}{B}} }

  \inferrule
  { }
  { \HeadSim{\ActI{\rho}{x : A}{B}}{\ActI{\rho}{x : A}{B}} }

  \inferrule
  { }
  { \HeadSim{\CH{A}}{\CH{B}} }

  \inferrule
  { }
  { \HeadSim{\HC{A}}{\HC{B}} }

  \inferrule
  { }
  { \HeadSim{c}{c} }

  \inferrule
  { }
  { \HeadSim{\fork{x : A}{m}}{\fork{x : A}{m}} }

  \inferrule
  { }
  { \HeadSim{\recvR{m}}{\recvR{m}} }

  \inferrule
  { }
  { \HeadSim{\recvI{m}}{\recvI{m}} }

  \inferrule
  { }
  { \HeadSim{\sendR{m}}{\sendR{m}} }

  \inferrule
  { }
  { \HeadSim{\sendI{m}}{\sendI{m}} }

  \inferrule
  { }
  { \HeadSim{\close{m}}{\close{m}} }

  \inferrule
  { }
  { \HeadSim{\wait{m}}{\wait{m}} }
\end{mathpar}

We then define the \emph{type similarity} relation as follows:
\begin{align*}
  \Sim{A}{B} \triangleq \exists A', B', A \simeq A' \land B \simeq B' \land \HeadSim{A'}{B'}
\end{align*}

The similarity relation is naturally extended to typing contexts as follows:
\begin{mathpar}
  \inferrule 
  { }
  { \Sim{\epsilon}{\epsilon} }

  \inferrule 
  { \Sim{A}{B} \\ 
    \Sim{\Gamma_1}{\Gamma_2} }
  { \Sim{(\Gamma_1, x : A)}{(\Gamma_2, x : B)} }
\end{mathpar}

The (head) similarity relation enjoys the following properties.
\begin{lemma}[HeadSim Reflexive]\label[lemma]{lemma:headsim-reflexive}
  For any term $A$, there is $\HeadSim{A}{A}$.
\end{lemma}
\begin{proof}
  By induction on the structure of $A$.
\end{proof}

\begin{lemma}[HeadSim Symmetric]\label[lemma]{lemma:headsim-symmetric}
  For any $\HeadSim{A}{B}$, there is $\HeadSim{B}{A}$.
\end{lemma}
\begin{proof}
  By induction on the derivation of $\HeadSim{A}{B}$.
\end{proof}

\begin{lemma}[HeadSim Substitution]\label[lemma]{lemma:headsim-substitution}
  Given substitution $\sigma$, if there is $\HeadSim{A}{B}$, then there is $\HeadSim{A[\sigma]}{B[\sigma]}$.
\end{lemma}
\begin{proof}
  By induction on the derivation of $\HeadSim{A}{B}$.
\end{proof}

\begin{lemma}[Sim Reflexive]\label[lemma]{lemma:sim-reflexive}
  For any term $A$, there is $\Sim{A}{A}$.
\end{lemma}
\begin{proof}
  By the reflexivity of $\simeq$ and \Cref{lemma:headsim-reflexive}.
\end{proof}

\begin{lemma}[Sim Transitive Left]
  For any $\Sim{A}{B}$ and $B \simeq C$, there is $\Sim{A}{C}$.
\end{lemma}
\begin{proof}
  By the transitivity of $\simeq$.
\end{proof}

\begin{lemma}[Sim Transitive Right]
  For any $\Sim{A}{B}$ and $A \simeq C$, there is $\Sim{C}{B}$.
\end{lemma}
\begin{proof}
  By the transitivity of $\simeq$.
\end{proof}

\begin{lemma}[Sim Symmetric]
  For any $\Sim{A}{B}$, there is $\Sim{B}{A}$.
\end{lemma}
\begin{proof}
  By the symmetry of $\simeq$ and \Cref{lemma:headsim-symmetric}.
\end{proof}

\begin{lemma}[Sim Substitution]
  Given substitution $\sigma$, if there is $\Sim{A}{B}$, then there is $\Sim{A[\sigma]}{B[\sigma]}$.
\end{lemma}
\begin{proof}
  By the substitutivity of $\simeq$ and \Cref{lemma:headsim-substitution}.
\end{proof}

\begin{lemma}[Sim Sort Injective]\label[lemma]{lemma:sim-sort-injective}
  If $\Sim{s_1}{s_2}$, then $s_1 = s_2$.
\end{lemma}
\begin{proof}
  By the definition of similarity and \Cref{corollary:inj-sort}.
\end{proof}

\begin{lemma}[Type Similarity]\label[lemma]{lemma:type-similarity}
  Given $\Gamma_1 \vdash m : A$ and $\Gamma_2 \vdash m : B$ and $\Sim{\Gamma_1}{\Gamma_2}$,
  then $\Sim{A}{B}$.
\end{lemma}
\begin{proof}
  By induction on the derivation of $\Gamma_1 \vdash m : A$.
\end{proof}

\begin{theorem}[Sort Uniqueness]\label[theorem]{theorem:sort-uniqueness}
  Given $\Gamma \vdash m : s$ and $\Gamma \vdash m : t$, then $s = t$.
\end{theorem}
\begin{proof}
  From \Cref{lemma:headsim-reflexive} we have $\Sim{\Gamma}{\Gamma}$.
  Then from \Cref{lemma:type-similarity} we have $\Sim{s}{t}$.
  Finally from \Cref{lemma:sim-sort-injective} we have $s = t$.
\end{proof}

\paragraph{\textbf{Inversion}}
Due to the presence of type conversion, inversion lemmas are necessary to reason about
typing derivations.

\begin{lemma}\label[lemma]{lemma:logical-inversion-explicit-fun}
  If $\Gamma \vdash \PiR{s}{x : A}{B} : C$, then there exists $t$
  such that $\Gamma, x : A \vdash B : t$ and $C \simeq s$.
\end{lemma}

\begin{lemma}\label[lemma]{lemma:logical-inversion-implicit-fun}
  If $\Gamma \vdash \PiI{s}{x : A}{B} : C$, then there exists $t$
  such that $\Gamma, x : A \vdash B : t$ and $C \simeq s$.
\end{lemma}

\begin{lemma}\label[lemma]{lemma:logical-inversion-explicit-lam}
  If $\Gamma \vdash \lamR{s_1}{x : A_1}{m} : \PiR{s_2}{x : A_2}{B}$, then $\Gamma, x : A_1 \vdash m : B$.
\end{lemma}

\begin{lemma}\label[lemma]{lemma:logical-inversion-implicit-lam}
  If $\Gamma \vdash \lamI{s_1}{x : A_1}{m} : \PiI{s_2}{x : A_2}{B}$, then $\Gamma, x : A_1 \vdash m : B$.
\end{lemma}

\begin{lemma}\label[lemma]{lemma:logical-inversion-explicit-sig}
  If $\Gamma \vdash \SigR{t}{x : A}{B} : C$, then there exists $s, r$ such that
  $s \sqsubseteq t$ and $r \sqsubseteq t$ and
  $\Gamma \vdash A : s$ and $\Gamma, x : A \vdash B : r$ and $C \simeq t$.
\end{lemma}

\begin{lemma}\label[lemma]{lemma:logical-inversion-implicit-sig}
  If $\Gamma \vdash \SigI{t}{x : A}{B} : C$, then there exists $s, r$ such that
  $r \sqsubseteq t$ and
  $\Gamma \vdash A : s$ and $\Gamma, x : A \vdash B : r$ and $C \simeq t$.
\end{lemma}

\begin{lemma}\label[lemma]{lemma:logical-inversion-monad}
  If $\Gamma \vdash \CM{A} : B$, then there exists $s$ such that $\Gamma \vdash A : s$ and $B \simeq \Ln$.
\end{lemma}

\begin{lemma}\label[lemma]{lemma:logical-inversion-ch}
  If $\Gamma \vdash \CH{A} : B$, then $\Gamma \vdash A : \Proto$ and $B \simeq \Ln$. 
\end{lemma}

\begin{lemma}\label[lemma]{lemma:logical-inversion-hc}
  If $\Gamma \vdash \HC{A} : B$, then $\Gamma \vdash A : \Proto$ and $B \simeq \Ln$. 
\end{lemma}

\begin{lemma}\label[lemma]{lemma:logical-inversion-explicit-action}
  If $\Gamma \vdash \ActR{\rho}{x : A}{B} : C$, then $\Gamma, x : A \vdash B : \Proto$.
\end{lemma}

\begin{lemma}\label[lemma]{lemma:logical-inversion-implicit-action}
  If $\Gamma \vdash \ActI{\rho}{x : A}{B} : C$, then $\Gamma, x : A \vdash B : \Proto$.
\end{lemma}

\begin{lemma}\label[lemma]{lemma:program-inversion-explicit-lam}
  If ${\Theta ; \Gamma ; \Delta \vdash \lamR{s}{x : A_2}{m} : \PiR{t}{x : A_1}{B}}$, then
  ${\Theta ; \Gamma, x : A_1 ; \Delta, x :_r A_1 \vdash m : B}$.
\end{lemma}

\begin{lemma}\label[lemma]{lemma:program-inversion-explicit-app}
  If ${\Theta ; \Gamma ; \Delta \vdash \appR{m}{n} : C}$, then
  there exists $A, B, s, \Theta_1, \Theta_2, \Delta_1, \Delta_2$ such that
  ${\Theta_1 ; \Gamma ; \Delta_1 \vdash m : \PiR{s}{x : A}{B}}$ and
  ${\Theta_2 ; \Gamma ; \Delta_2 \vdash n : A}$ and
  $\Theta = \Theta_1 \dotcup \Theta_2$ and $\Delta = \Delta_1 \dotcup \Delta_2$ and
  $C \simeq B[n/x]$.
\end{lemma}

\begin{lemma}\label[lemma]{lemma:program-inversion-implicit-app}
  If ${\Theta ; \Gamma ; \Delta \vdash \appI{m}{n} : C}$, then
  there exists $A, B, s$ such that
  ${\Theta ; \Gamma ; \Delta \vdash m : \PiI{s}{x : A}{B}}$ and
  ${\Gamma \vdash n : A}$ and $C \simeq B[n/x]$.
\end{lemma}

\begin{lemma}\label[lemma]{lemma:program-inversion-explicit-pair}
  If ${\Theta ; \Gamma ; \Delta \vdash \pairR{m}{n}{s} : \SigR{t}{x : A}{B}}$, then
  there exists $\Theta_1, \Theta_2, \Delta_1, \Delta_2$ such that
  ${\Theta_1 ; \Gamma ; \Delta_1 \vdash m : A}$ and
  ${\Theta_2 ; \Gamma ; \Delta_2 \vdash n : B[m/x]}$ and
  $\Theta = \Theta_1 \dotcup \Theta_2$ and $\Delta = \Delta_1 \dotcup \Delta_2$ and
  $s = t$.
\end{lemma}

\begin{lemma}\label[lemma]{lemma:program-inversion-unit}
  If ${\Theta ; \Gamma ; \Delta \vdash \ii : \unit}$, then
  $\Theta = \epsilon$ and $\Delta \triangleright \Un$.
\end{lemma}

\begin{lemma}\label[lemma]{lemma:program-inversion-true}
  If ${\Theta ; \Gamma ; \Delta \vdash \bTrue : \Bool}$, then
  $\Theta = \epsilon$ and $\Delta \triangleright \Un$.
\end{lemma}

\begin{lemma}\label[lemma]{lemma:program-inversion-false}
  If ${\Theta ; \Gamma ; \Delta \vdash \bFalse : \Bool}$, then
  $\Theta = \epsilon$ and $\Delta \triangleright \Un$.
\end{lemma}

\begin{lemma}\label[lemma]{lemma:program-inversion-return}
  If ${\Theta ; \Gamma ; \Delta \vdash \return{m} : \CM{A}}$, then
  ${\Theta ; \Gamma ; \Delta \vdash m : A}$.
\end{lemma}

\begin{lemma}\label[lemma]{lemma:program-inversion-bind}
  If ${\Theta ; \Gamma ; \Delta \vdash \letin{m}{x}{n} : \CM{B}}$, then
  there exists $A, \Theta_1, \Theta_2, \Delta_1, \Delta_2$ such that
  $\Theta_1 ; \Gamma ; \Delta_1 \vdash m : \CM{A}$ and
  $\Theta_2 ; \Gamma, x : A ; \Delta_2, x :_s A \vdash n : \CM{B}$ and
  $\Theta = \Theta_1 \dotcup \Theta_2$ and $\Delta = \Delta_1 \dotcup \Delta_2$. 
  and $x \notin \FV{B}$.
\end{lemma}

\begin{lemma}\label[lemma]{lemma:program-inversion-fork}
  If ${\Theta ; \Gamma ; \Delta \vdash \fork{x : A}{m} : B}$, then
  there exists $A'$ such that ${\Gamma \vdash A' : \Proto}$ and $A = \CH{A'}$ and
  $B \simeq \CM{\HC{A'}}$ and
  ${\Theta ; \Gamma, x : \CH{A'} ; \Delta, x :_\Ln \CH{A'} \vdash m : \CM{\unit}}$.
\end{lemma}

\begin{lemma}\label[lemma]{lemma:program-inversion-explicit-send}
  If ${\Theta ; \Gamma ; \Delta \vdash \sendR{m} : C}$, then there exists $A, B$
  such that either 
  $C \simeq \PiR{\Ln}{x : A}{\CM{\CH{B}}}$ and 
  ${\Theta ; \Gamma ; \Delta \vdash m : \CH{\ActR{!}{x : A}{B}}}$
  or
  $C \simeq \PiR{\Ln}{x : A}{\CM{\HC{B}}}$ and 
  ${\Theta ; \Gamma ; \Delta \vdash m : \HC{\ActR{?}{x : A}{B}}}$.
\end{lemma}

\begin{lemma}\label[lemma]{lemma:program-inversion-implicit-send}
  If ${\Theta ; \Gamma ; \Delta \vdash \sendI{m} : C}$, then there exists $A, B$
  such that either 
  $C \simeq \PiI{\Ln}{x : A}{\CM{\CH{B}}}$ and 
  ${\Theta ; \Gamma ; \Delta \vdash m : \CH{\ActI{!}{x : A}{B}}}$
  or
  $C \simeq \PiI{\Ln}{x : A}{\CM{\HC{B}}}$ and 
  ${\Theta ; \Gamma ; \Delta \vdash m : \HC{\ActI{?}{x : A}{B}}}$.
\end{lemma}

\begin{lemma}\label[lemma]{lemma:program-inversion-explicit-recv}
  If ${\Theta ; \Gamma ; \Delta \vdash \recvR{m} : C}$, then there exists $A, B$
  such that either 
  $C \simeq \CM{\SigR{\Ln}{x : A}{\CH{B}}}$ and 
  ${\Theta ; \Gamma ; \Delta \vdash m : \CH{\ActR{?}{x : A}{B}}}$
  or
  $C \simeq \CM{\SigR{\Ln}{x : A}{\HC{B}}}$ and 
  ${\Theta ; \Gamma ; \Delta \vdash m : \HC{\ActR{!}{x : A}{B}}}$.
\end{lemma}

\begin{lemma}\label[lemma]{lemma:program-inversion-implicit-recv}
  If ${\Theta ; \Gamma ; \Delta \vdash \recvI{m} : C}$, then there exists $A, B$
  such that either 
  $C \simeq \CM{\SigI{\Ln}{x : A}{\CH{B}}}$ and 
  ${\Theta ; \Gamma ; \Delta \vdash m : \CH{\ActI{?}{x : A}{B}}}$
  or
  $C \simeq \CM{\SigI{\Ln}{x : A}{\HC{B}}}$ and 
  ${\Theta ; \Gamma ; \Delta \vdash m : \HC{\ActI{!}{x : A}{B}}}$.
\end{lemma}

\begin{lemma}\label[lemma]{lemma:program-inversion-close}
  If ${\Theta ; \Gamma ; \Delta \vdash \close{m} : A}$, then
  ${\Theta ; \Gamma ; \Delta \vdash m : \CH{\End}}$ and $A \simeq \CM{\unit}$.
\end{lemma}

\begin{lemma}\label[lemma]{lemma:program-inversion-wait}
  If ${\Theta ; \Gamma ; \Delta \vdash \wait{m} : A}$, then
  ${\Theta ; \Gamma ; \Delta \vdash m : \HC{\End}}$ and $A \simeq \CM{\unit}$.
\end{lemma}

\paragraph{\textbf{Type Validity}}
We show that all types appearing in typing judgments are valid, i.e. they are well-sorted
at the logical level.

\begin{theorem}[Logical Type Validity]\label[theorem]{theorem:logical-type-validity}
  If $\Gamma \vdash m : A$, then there exists $s$ such that $\Gamma \vdash A : s$.
\end{theorem}
\begin{proof}
  By induction on the derivation of $\Gamma \vdash m : A$. We will show some representative cases.

\textbf{Case} (\textsc{Var}): 
  From the premise we have $\Gamma, x : A \vdash$ which implies $\Gamma \vdash A : s$ for some $s$.

\textbf{Case} (\textsc{Explicit-Lam}):
  From the induction hypothesis we have $\Gamma, x : A \vdash B : r$ for some $r$.
  The validity of context $\Gamma, x : A \vdash$ implies $\Gamma \vdash A : s$ for some $s$.
  Then from \textsc{Explicit-Fun} we have $\Gamma \vdash \PiR{s}{x : A}{B} : s$ which concludes this case.

\textbf{Case} (\textsc{Explicit-App}):
  From the induction hypothesis we have $\Gamma \vdash \PiR{s}{x : A}{B} : r$ for some $r$.
  From \Cref{lemma:logical-inversion-explicit-fun} we have $\Gamma, x : A \vdash B : t$ for some $t$.
  By \Cref{lemma:logical-subst} we have $\Gamma \vdash B[n/x] : t$ which concludes this case.

\textbf{Case} (\textsc{Explicit-Recv-CH}):
  From the induction hypothesis we have $\Gamma \vdash \CH{\ActR{?}{x : A}{B}} : s$ for some $s$.
  By \Cref{lemma:logical-inversion-ch} we have $\Gamma \vdash \ActR{?}{x : A}{B} : \Proto$.
  By \Cref{lemma:logical-inversion-explicit-action} we have $\Gamma, x : A \vdash B : \Proto$.
  From the validity of context $\Gamma, x : A \vdash$ we have $\Gamma \vdash A : t$ for some $t$.
  Applying \textsc{ChType}, we have $\Gamma, x : A \vdash \CH{B} : \Ln$.
  Applying \textsc{Explicit-Sum}, we have $\Gamma \vdash \SigR{\Ln}{x : A}{\CH{B}} : \Ln$.
  Applying \textsc{$\mcC$Type}, we have $\Gamma \vdash \CM{\SigR{\Ln}{x : A}{\CH{B}}} : \Ln$ which concludes this case.
\end{proof}

To show that the types appearing in program level typing judgments are valid, 
we first prove the lifting theorem which allows us to lift programs to the logical level.

\begin{theorem}[Lifting]\label[theorem]{theorem:lifting}
  If $\Theta ; \Gamma ; \Delta \vdash m : A$, then $\Gamma \vdash m : A$.
\end{theorem}
\begin{proof}
  By induction on the derivation of $\Theta ; \Gamma ; \Delta \vdash m : A$.
\end{proof}

\begin{theorem}[Program Type Validity]\label[theorem]{theorem:program-type-validity}
  If $\Theta ; \Gamma ; \Delta \vdash m : A$, then there exists $s$ such that $\Gamma \vdash A : s$.
\end{theorem}
\begin{proof}
  Immediate from \Cref{theorem:lifting} and \Cref{theorem:logical-type-validity}.
\end{proof}

\paragraph{\textbf{Subject Reduction}}
We show that both the logical and program level type systems enjoy subject reduction
under logical and program reductions respectively.

\begin{lemma}[Arity Preservation]\label[lemma]{lemma:arity-preservation}
  If $A \Rightarrow A'$ and $A~\arity{\Proto}$, then $A'~\arity{\Proto}$.
\end{lemma}
\begin{proof}
  By induction on the derivation of $A \Rightarrow A'$.
\end{proof}

\begin{lemma}[Guard Preservation]\label[lemma]{lemma:guard-preservation}
  If $A \Rightarrow A'$ and $A~\guard{x}$, then $A'~\guard{x}$.
\end{lemma}
\begin{proof}
  By induction on the derivation of $A \Rightarrow A'$ and appealing to  
  \Cref{lemma:subst-guard}.
\end{proof}

\begin{theorem}[Logical Subject Reduction]\label[theorem]{theorem:logical-subject-reduction}
  If $\Gamma \vdash m : A$ and $m \Rightarrow m'$, then $\Gamma \vdash m' : A$.
\end{theorem}
\begin{proof}
  By induction on the derivation of $\Gamma \vdash m : A$ and case analysis on the reduction $m \Rightarrow m'$.
  We present the following representative cases.

\textbf{Case} (\textsc{Explicit-Lam}):
  From case analysis on the reduction, we have $A \Rightarrow A'$ and $m \Rightarrow m'$.
  From the induction hypothesis we have $\Gamma, x : A \vdash m' : B$ and $\Gamma \vdash A' : s$ for some $s$.
  By definition of convertibility, we have $A \simeq A'$.
  By \Cref{corollary:logical-context-conv} we have $\Gamma, x : A' \vdash m' : B$.
  By \textsc{Explicit-Lam} we have $\Gamma \vdash \lamR{s}{x : A'}{m'} : \PiR{s}{x : A'}{B}$.
  By \textsc{Conversion} we have $\Gamma \vdash \lamR{s}{x : A'}{m'} : \PiR{s}{x : A}{B}$ which concludes this case.

\textbf{Case} (\textsc{Explicit-App}):
  From case analysis on the reduction we have two sub-cases:
  (1) \textsc{PStep-Explicit-App} and (2) \textsc{PStep-Explicit-$\beta$}.

  In sub-case (1) \textsc{PStep-Explicit-App}, we have $m \Rightarrow m'$ and $n \Rightarrow n'$.
  From the induction hypothesis we have $\Gamma \vdash m' : \PiR{s}{x : A}{B}$ and $\Gamma \vdash n' : A$.
  By \textsc{Explicit-App} we have $\Gamma \vdash \appR{m'}{n'} : B[n'/x]$.
  By definition of convertibility, we have $B[n/x] \simeq B[n'/x]$.
  By validity we have $\Gamma \vdash \PiR{s}{x : A}{B} : t$ for some $t$.
  By \Cref{lemma:logical-inversion-explicit-fun} we have $\Gamma, x : A \vdash B : r$ for some $r$.
  By \Cref{lemma:logical-subst} we have $\Gamma \vdash B[n/x] : r$.
  By \textsc{Conversion} we have $\Gamma \vdash \appR{m'}{n'} : B[n/x]$ which concludes this sub-case.

  In sub-case (2) \textsc{PStep-Explicit-$\beta$}, we have $m = \lamR{s}{x : A}{m_0}$ for some $m_0$
  and $m_0 \Rightarrow m_0'$ and $n \Rightarrow n'$.
  By \Cref{lemma:logical-inversion-explicit-lam} we have $\Gamma, x : A \vdash m_0 : B$.
  By the induction hypothesis we have $\Gamma, x : A \vdash m_0' : B$ and $\Gamma \vdash n' : A$.
  By \Cref{lemma:logical-subst} we have $\Gamma \vdash m_0'[n'/x] : B[n'/x]$.
  By definition of convertibility, we have $B[n/x] \simeq B[n'/x]$
  By validity we have $\Gamma \vdash \PiR{s}{x : A}{B} : t$ for some $t$.
  By \Cref{lemma:logical-inversion-explicit-fun} we have $\Gamma, x : A \vdash B : r$ for some $r$.
  By \Cref{lemma:logical-subst} we have $\Gamma \vdash B[n/x] : r$.
  By \textsc{Conversion} we have $\Gamma \vdash m_0'[n'/x] : B[n/x]$ which concludes this sub-case.

\textbf{Case} (\textsc{BoolElim})
  From case analysis on the reduction we have three sub-cases:
  (1) \textsc{PStep-BoolElim}, (2) \textsc{PStep-TrueElim}, and (3) \textsc{PStep-FalseElim}.

  In sub-case (1) \textsc{PStep-BoolElim}, we have 
  $A \Rightarrow A'$, $m \Rightarrow m'$, $n_1 \Rightarrow n_1'$, and $n_2 \Rightarrow n_2'$.
  By the induction hypothesis we have
  $\Gamma, z : \Bool \vdash A' : s$, 
  $\Gamma \vdash m' : \Bool$, 
  $\Gamma \vdash n_1' : A[\bTrue/z]$, and 
  $\Gamma \vdash n_2' : A[\bFalse/z]$.
  By definition of convertibility, we have $A \simeq A'$.
  By \Cref{lemma:logical-subst} we have $\Gamma \vdash A'[\bTrue/z] : s$ and $\Gamma \vdash A'[\bFalse/z] : s$.
  By \textsc{Conversion} we have $\Gamma \vdash n_1' : A'[\bTrue/z]$ and $\Gamma \vdash n_2' : A'[\bFalse/z]$.
  By \textsc{BoolElim} we have $\Gamma \vdash \boolElim{[z]A'}{m'}{n_1'}{n_2'} : A'[m'/z]$.
  By definition of convertibility, we have $A[m/z] \simeq A'[m'/z]$.
  By \Cref{lemma:logical-subst} we have $\Gamma \vdash A[m/z] : s$.
  By \textsc{Conversion} we have $\Gamma \vdash \boolElim{[z]A'}{m'}{n_1'}{n_2'} : A[m/z]$ which concludes this sub-case.

\textbf{Case} (\textsc{RecProto}):
  From case analysis on the reduction we have two sub-cases:
  (1) \textsc{PStep-RecProto} and (2) \textsc{PStep-RecUnfold}.

  In sub-case (1) \textsc{PStep-RecProto}, we have $A \Rightarrow A'$ and $m \Rightarrow m'$.
  By validity of context $\Gamma, x : A \vdash$ we have $\Gamma \vdash A : s$ for some $s$.
  By the induction hypothesis we have ${\Gamma, x : A \vdash m' : A}$ and ${\Gamma \vdash A' : s}$.
  By definition of convertibility, we have $A \simeq A'$.
  By \Cref{corollary:logical-context-conv} we have $\Gamma, x : A' \vdash m' : A$.
  By \textsc{Conversion} we have $\Gamma, x : A' \vdash m' : A'$.
  By \Cref{lemma:arity-preservation} we have $A'~\arity{\Proto}$.
  By \Cref{lemma:guard-preservation} we have $m'~\guard{x}$.
  By \textsc{RecProto} we have $\Gamma \vdash \fix{x : A'}{m'} : A'$.
  By \textsc{Conversion} we have $\Gamma \vdash \fix{x : A'}{m'} : A$ which concludes this sub-case.

  In sub-case (2) \textsc{PStep-RecUnfold}, we have $A \Rightarrow A'$ and $m \Rightarrow m'$.
  By validity of context $\Gamma, x : A \vdash$ we have $\Gamma \vdash A : s$ for some $s$.
  By the induction hypothesis we have ${\Gamma, x : A \vdash m' : A}$ and ${\Gamma \vdash A' : s}$.
  By definition of convertibility, we have $A \simeq A'$.
  By \Cref{corollary:logical-context-conv} we have $\Gamma, x : A' \vdash m' : A$.
  By \textsc{Conversion} we have $\Gamma, x : A' \vdash m' : A'$.
  By \Cref{lemma:arity-preservation} we have $A'~\arity{\Proto}$.
  By \Cref{lemma:guard-preservation} we have $m'~\guard{x}$.
  By \textsc{RecProto} we have $\Gamma \vdash \fix{x : A'}{m'} : A'$.
  By \Cref{lemma:logical-subst} we have $\Gamma \vdash m'[\fix{x : A'}{m'}/x] : A'$.
  By \textsc{Conversion} we have $\Gamma \vdash m'[\fix{x : A'}{m'}/x] : A$ which concludes this sub-case.
\end{proof}

In order to show subject reduction at the program level, we need to show that reduction of
redexes in dependent positions preserves typing. To do so, we prove the following lemma which
lifts program reductions into convertibility at the logical level.
\begin{lemma}[Program Step Convertible]\label[lemma]{lemma:program-step-convertible}
  If ${\Theta ; \epsilon ; \epsilon \vdash m : A}$ and $m \Leadsto n$, then $m \simeq n$.
\end{lemma}

The program level substitution lemma (\Cref{lemma:program-subst-explicit})
requires context restrictions $\Theta_2 \triangleright s$ and $\Delta_2 \triangleright s$ for the
substituted term $n$. To ensure that these restrictions are satisfied for \emph{values}, we prove the
following context bound lemma.
\begin{lemma}[Program Context Bound]\label[lemma]{lemma:program-context-bound}
  Given ${\Theta ; \Gamma ; \Delta \vdash v : A}$ and $\Gamma \vdash A : s$, then 
  $\Theta \triangleright s$ and $\Delta \triangleright s$.
\end{lemma}
\begin{proof}
  By induction on the derivation of ${\Theta ; \Gamma ; \Delta \vdash v : A}$ where $v$ is a value.
  We present the following representative cases.

\textbf{Case} (\textsc{Explicit-Lam}):
  From the premise we have $\Theta \triangleright t$ and $\Delta \triangleright t$.
  By \Cref{lemma:logical-inversion-explicit-fun} we have $t \simeq s$.
  By injectivity of sorts (\Cref{corollary:inj-sort}) we have $t = s$ which concludes this case.

\textbf{Case} (\textsc{Explicit-Pair}):
  From the assumption that the pair is a value, we have $v = \pairR{v_1}{v_2}{t}$ for some $t$.
  Additionally, we have $\Theta_1 ; \Gamma ; \Delta_1 \vdash v_1 : A$ and 
  $\Theta_2 ; \Gamma ; \Delta_2 \vdash v_2 : B[v_1/x]$ and $\Gamma \vdash \SigR{t}{x : A}{B} : t$.
  By \Cref{theorem:sort-uniqueness} we have $s = t$.
  By \Cref{lemma:logical-inversion-explicit-sig} we have 
  $\Gamma \vdash A : s$ and $\Gamma, x : A \vdash B : r$ and $s \sqsubseteq t$ and $r \sqsubseteq t$.
  By the induction hypothesis we have $\Theta_1 \triangleright s$ and $\Delta_1 \triangleright s$
  and $\Theta_2 \triangleright r$ and $\Delta_2 \triangleright r$.
  These context restrictions can then be weakened to $\Theta_1 \triangleright t$ and $\Delta_1 \triangleright t$ 
  and $\Theta_2 \triangleright t$ and $\Delta_2 \triangleright t$.
  The merged contexts now satisfy $\Theta_1 \dotcup \Theta_2 \triangleright t$ and 
  $\Delta_1 \dotcup \Delta_2 \triangleright t$ which concludes this case.

\textbf{Case} (\textsc{Return}):
  From the premise we have $\Theta ; \Gamma ; \Delta \vdash m : A$ and $\Gamma \vdash \CM{A} : s$.
  By \Cref{theorem:sort-uniqueness} we have $s = \Ln$. 
  The context restrictions $\Theta \triangleright \Ln$ and $\Delta \triangleright \Ln$ hold trivially which concludes this case.

\textbf{Case} (\textsc{Channel-CH})
  From the premise we have $\Theta \triangleright \Un$. From \textsc{Ord-$\Un$} we have $\Un \sqsubseteq s$
  which allows us to weaken the restriction into $\Theta \triangleright s$ and concluding this case.

\textbf{Case} (\textsc{Explicit-Send-CH}):
  From the premise we have $\Gamma \vdash \PiR{\Ln}{x : A}{\CM{\CH{B}}} : s$. 
  By \Cref{theorem:sort-uniqueness} we have $s = \Ln$.
  The context restrictions $\Theta \triangleright \Ln$ and $\Delta \triangleright \Ln$ hold trivially.
\end{proof}

\begin{theorem}[Program Subject Reduction]\label[theorem]{theorem:program-subject-reduction}
  If ${\Theta ; \epsilon ; \epsilon \vdash m : A}$ and ${m \Leadsto m'}$, then ${\Theta ; \epsilon ; \epsilon \vdash m' : A}$.
\end{theorem}
\begin{proof}
  By induction on the derivation of ${\Theta ; \epsilon ; \epsilon \vdash m : A}$ and case 
  analysis on the reduction $m \Leadsto m'$. We present the following representative cases.

\textbf{Case} (\textsc{Explicit-App}):
  From case analysis on the reduction we have three sub-cases:
  (1) \textsc{Step-Explicit-App$_1$}, (2) \textsc{Step-Explicit-App$_2$}, and (3) \textsc{Step-Explicit-$\beta$}.

  In sub-case (1) \textsc{Step-Explicit-App$_1$}, we have $m \Leadsto m'$.
  By the induction hypothesis we have ${\Theta_1 ; \Gamma ; \Delta_1 \vdash m' : \PiR{t}{x : A}{B}}$.
  By \textsc{Explicit-App} we have ${\Theta_1 \dotcup \Theta_2 ; \Gamma ; \Delta_1 \dotcup \Delta_2 \vdash \appR{m'}{n} : B[n/x]}$
  which concludes this sub-case.

  In sub-case (2) \textsc{Step-Explicit-App$_2$}, we have $n \Leadsto n'$.
  By the induction hypothesis we have ${\Theta_2 ; \Gamma ; \Delta_2 \vdash n' : A}$.
  By \textsc{Explicit-App} we have ${\Theta_1 \dotcup \Theta_2 ; \Gamma ; \Delta_1 \dotcup \Delta_2 \vdash \appR{m}{n'} : B[n'/x]}$.
  By \Cref{lemma:program-step-convertible} we have $n \simeq n'$ and $B[n/x] \simeq B[n'/x]$.
  By \Cref{theorem:lifting} we have $\Gamma \vdash n : A$.
  Applying \Cref{theorem:program-type-validity} on $\Theta_1; \Gamma ; \Delta_1 \vdash m : \PiR{t}{x : A}{B}$ we have
  $\Gamma \vdash \PiR{t}{x : A}{B} : s$ for some $s$.
  By \Cref{lemma:logical-inversion-explicit-fun} we have $\Gamma, x : A \vdash B : r$ for some $r$.
  By \Cref{lemma:logical-subst} we have $\Gamma \vdash B[n/x] : r$.
  By \textsc{Conversion} we have ${\Theta_1 \dotcup \Theta_2 ; \Gamma ; \Delta_1 \dotcup \Delta_2 \vdash \appR{m}{n'} : B[n/x]}$
  which concludes this sub-case.

  In sub-case (3) \textsc{Step-Explicit-$\beta$}, we have $m = \lamR{t}{x : A_0}{m_0}$ for some $A_0$ and $m_0$.
  By \Cref{lemma:program-inversion-explicit-lam} we have ${\Theta_1 ; \Gamma, x : A ; \Delta_1, x :_r A \vdash m_0 : B}$.
  From the validity of context ${\Theta_1 ; \Gamma, x : A ; \Delta_1, x :_r A \vdash}$ we have $\Gamma \vdash A : r$.
  By \Cref{lemma:program-context-bound} and ${\Theta_2 ; \Gamma ; \Delta_2 \vdash v : A}$ and $\Gamma \vdash A : r$
  we have $\Theta_2 \triangleright r$.
  By \Cref{lemma:program-subst-explicit} we have 
  ${\Theta_1 \dotcup \Theta_2 ; \Gamma ; \Delta_1 \dotcup \Delta_2 \vdash m_0[v/x] : B[v/x]}$ 
  which concludes this sub-case.

\textbf{Case} (\textsc{Explicit-SumElim}):
  From case analysis on the reduction we have two sub-cases:
  (1) \textsc{Step-SumElim$_1$} and (2) \textsc{Step-Explicit-PairElim}.

  In sub-case (1) \textsc{Step-SumElim$_1$}, we have $m \Leadsto m'$.
  By the induction hypothesis we have ${\Theta_1 ; \Gamma ; \Delta_1 \vdash m' : \SigR{t}{x : A}{B}}$
  By \textsc{Explicit-SumElim} we have
  ${\Theta_1 \dotcup \Theta_2 ; \Gamma ; \Delta_1 \dotcup \Delta_2 \vdash \SigElim{[z]C}{m'}{[x,y]n} : C[m'/z]}$.
  By \Cref{lemma:program-step-convertible} we have $m \simeq m'$ and $C[m/z] \simeq C[m'/z]$.
  By \Cref{theorem:lifting} we have $\Gamma \vdash m : \SigR{t}{x : A}{B}$.
  Applying \Cref{lemma:logical-subst} on assumption $\Gamma, z : \SigR{t}{x : A}{B} \vdash C : s$ and
  $\Gamma \vdash m : \SigR{t}{x : A}{B}$ we have $\Gamma \vdash C[m/z] : s$. By \textsc{Conversion} we have 
  ${\Theta_1 \dotcup \Theta_2 ; \Gamma ; \Delta_1 \dotcup \Delta_2 \vdash \SigElim{[z]C}{m'}{[x,y]n} : C[m/z]}$
  which concludes this sub-case.

  In sub-case (2) \textsc{Step-Explicit-PairElim}, we have $m = \pairR{u}{v}{t}$ for some $u, v$ and $t$.
  By \Cref{lemma:program-inversion-explicit-pair} we have
  ${\Theta_{11} ; \Gamma ; \Delta_{11} \vdash u : A}$ and
  ${\Theta_{12} ; \Gamma ; \Delta_{12} \vdash v : B[u/x]}$ and
  $\Theta_1 = \Theta_{11} \dotcup \Theta_{12}$ and $\Delta_1 = \Delta_{11} \dotcup \Delta_{12}$ and $s = t$.
  From the validity of context $\Theta_2 ; \Gamma, x : A, y : B ; \Delta_2, x :_{r_1} A, y :_{r_2} B \vdash$
  we have $\Gamma \vdash A : r_1$ and $\Gamma, x : A \vdash B : r_2$.
  Applying \Cref{theorem:lifting} to ${\Theta_{11} ; \Gamma ; \Delta_{11} \vdash u : A}$ we have
  $\Gamma \vdash u : A$. Applying \Cref{lemma:logical-subst} we have $\Gamma \vdash B[u/x] : r_2$.
  By \Cref{lemma:program-context-bound} and ${\Theta_{11} ; \Gamma ; \Delta_{11} \vdash u : A}$ and $\Gamma \vdash A : r_1$
  we have $\Theta_{11} \triangleright r_1$.
  By \Cref{lemma:program-context-bound} and ${\Theta_{12} ; \Gamma ; \Delta_{12} \vdash v : B[u/x]}$ and $\Gamma \vdash B[u/x] : r_2$
  we have $\Theta_{12} \triangleright r_2$.
  By \Cref{lemma:program-subst-explicit} we have
  ${\Theta_{11} \dotcup \Theta_{12} \dotcup \Theta_2 ; \Gamma ; \Delta_{11} \dotcup \Delta_{12} \dotcup \Delta_2 \vdash n[u/x, v/y] : C[\pairR{u}{v}{t}/z]}$
  which concludes this sub-case.

\textbf{Case} (\textsc{BoolElim}):
  By case analysis on the reduction we have three sub-cases:
  (1) \textsc{Step-BoolElim$_1$}, (2) \textsc{Step-TrueElim}, and (3) \textsc{Step-FalseElim}.

  In sub-case (1) \textsc{Step-BoolElim$_1$}, we have $m \Leadsto m'$.
  By the induction hypothesis we have ${\Theta_1 ; \Gamma ; \Delta_1 \vdash m' : \Bool}$.
  By \textsc{BoolElim} we have
  ${\Theta_1 \dotcup \Theta_2 ; \Gamma ; \Delta_1 \dotcup \Delta_2 \vdash \boolElim{[z]C}{m'}{n_1}{n_2} : C[m'/z]}$.
  By \Cref{lemma:program-step-convertible} we have $m \simeq m'$ and $C[m/z] \simeq C[m'/z]$.
  By \Cref{theorem:lifting} we have $\Gamma \vdash m : \Bool$.
  Applying \Cref{lemma:logical-subst} on ${\Gamma, z : \Bool \vdash C : s}$ and
  ${\Gamma \vdash m : \Bool}$ we have ${\Gamma \vdash C[m/z] : s}$.
  By \textsc{Conversion} we have
  ${\Theta_1 \dotcup \Theta_2 ; \Gamma ; \Delta_1 \dotcup \Delta_2 \vdash \boolElim{[z]C}{m'}{n_1}{n_2} : C[m/z]}$
  which concludes this sub-case.

  In sub-case (2) \textsc{Step-TrueElim}, we have $m = \bTrue$.
  By \Cref{lemma:program-inversion-true} we have $\Theta_1 = \epsilon$ and $\Delta_1 \triangleright \Un$.
  Thus we have $\Theta_1 \dotcup \Theta_2 = \Theta_2$ and $\Delta_1 \dotcup \Delta_2 = \Delta_2$.
  The assumption $\Theta_2 ; \Gamma ; \Delta_2 \vdash n_1 : C[\bTrue/z]$ gives us the desired result
  which concludes this sub-case.

  In sub-case (3) \textsc{Step-FalseElim}, we have $m = \bFalse$.
  By \Cref{lemma:program-inversion-false} we have $\Theta_1 = \epsilon$ and $\Delta_1 \triangleright \Un$.
  Thus we have $\Theta_1 \dotcup \Theta_2 = \Theta_2$ and $\Delta_1 \dotcup \Delta_2 = \Delta_2$.
  The assumption $\Theta_2 ; \Gamma ; \Delta_2 \vdash n_2 : C[\bFalse/z]$ gives us the desired result
  which concludes this sub-case.

\textbf{Case} (\textsc{Bind}):
  By case analysis on the reduction we have two sub-cases: 
  (1) \textsc{Step-Bind} and (2) \textsc{Step-ReturnBind}.

  In sub-case (1) \textsc{Step-Bind}, we have $m \Leadsto m'$.
  By the induction hypothesis we have ${\Theta_1 ; \Gamma ; \Delta_1 \vdash m' : \CM{A}}$.
  By \textsc{Bind} we have ${\Theta_1 \dotcup \Theta_2 ; \Gamma ; \Delta_1 \dotcup \Delta_2 \vdash \letin{x}{m'}{n} : C}$ 
  which concludes this sub-case.

  In sub-case (2) \textsc{Step-ReturnBind}, we have $m = \return{v}$ for some value $v$.
  By \Cref{lemma:program-inversion-return} we have ${\Theta_1 ; \Gamma ; \Delta_1 \vdash v : A}$.
  From the validity of context ${\Theta_2 ; \Gamma, x : A ; \Delta_2, x :_r A \vdash}$ we have $\Gamma \vdash A : r$.
  By \Cref{lemma:program-context-bound} and ${\Theta_1 ; \Gamma ; \Delta_1 \vdash v : A}$ and $\Gamma \vdash A : r$
  we have $\Theta_1 \triangleright r$.
  By \Cref{lemma:program-subst-explicit} we have
  ${\Theta_1 \dotcup \Theta_2 ; \Gamma ; \Delta_1 \dotcup \Delta_2 \vdash n[v/x] : \CM{B[v/x]}}$.
  From assumption $\Gamma \vdash B : s$ we know that $x \notin \FV{B}$, thus $B[v/x] = B$
  and ${\Theta_1 \dotcup \Theta_2 ; \Gamma ; \Delta_1 \dotcup \Delta_2 \vdash n[v/x] : \CM{B}}$.
  which concludes this sub-case.
\end{proof}

\paragraph{\textbf{Progress}}
Due to the presence of concurrency primitives, the values that the program level
terms can reduce to are not necessarily canonical forms. They can also be thunked monadic computations.
These thunked computations will eventually be reduced by the semantics of the process level.

The following canonical forms lemmas are used to prove program progress (\Cref{theorem:program-progress}). 
They are proved by induction on the typing derivation of the value.
\begin{lemma}\label[lemma]{lemma:program-explicit-fun-canonical}
  If ${\Theta ; \epsilon ; \epsilon \vdash v : \PiR{s}{x : A}{B}}$ then
  $v = \lamR{s}{x : A}{m}$ or $v = \sendR{u}$.
\end{lemma}

\begin{lemma}\label[lemma]{lemma:program-implicit-fun-canonical}
  If ${\Theta ; \epsilon ; \epsilon \vdash v : \PiI{s}{x : A}{B}}$ then
  $v = \lamI{s}{x : A}{m}$ or $v = \sendI{u}$.
\end{lemma}

\begin{lemma}\label[lemma]{lemma:program-sig-canonical}
  If ${\Theta ; \epsilon ; \epsilon \vdash v : \SigR{s}{x : A}{B}}$ then
  $v = \pairR{v_1}{v_2}{s}$.
\end{lemma}

\begin{lemma}\label[lemma]{lemma:program-implicit-sig-canonical}
  If ${\Theta ; \epsilon ; \epsilon \vdash v : \SigI{s}{x : A}{B}}$ then
  $v = \pairI{m}{v_2}{s}$.
\end{lemma}

\begin{lemma}\label[lemma]{lemma:program-unit-canonical}
  If ${\Theta ; \epsilon ; \epsilon \vdash v : \unit}$ then $v = \ii$.
\end{lemma}

\begin{lemma}\label[lemma]{lemma:program-bool-canonical}
  If ${\Theta ; \epsilon ; \epsilon \vdash v : \Bool}$ then
  $v = \bTrue$ or $v = \bFalse$.
\end{lemma}

\begin{lemma}\label[lemma]{lemma:program-monad-canonical}
  If ${\Theta ; \epsilon ; \epsilon \vdash v : \CM{A}}$ then $v = \return{u}$ or $v$ is a thunk.
\end{lemma}

\begin{theorem}[Program Progress]\label[theorem]{theorem:program-progress}
  If ${\Theta ; \epsilon ; \epsilon \vdash m : A}$, then $m$ is a value or there exists $m'$ such that $m \Leadsto m'$.
\end{theorem}
\begin{proof}
  By induction on the derivation of ${\Theta ; \epsilon ; \epsilon \vdash m : A}$.
  We present the following cases.

\textbf{Case} (\textsc{Var}): Impossible since the context is empty.

\textbf{Case} (\textsc{Explicit-Lam}): Trivial since $\lamR{t}{x : A}{m_0}$ is a value.

\textbf{Case} (\textsc{Explicit-App}):
  By the induction hypothesis we have that either $m$ is a value or there exists $m'$ such that $m \Leadsto m'$.
  If $m \Leadsto m'$, then we are done by \textsc{Step-Explicit-App$_1$}.
  If $m$ is a value, by \Cref{lemma:program-explicit-fun-canonical} we have two sub-cases:
  (1) $m = \lamR{t}{x : A}{m_0}$ and (2) $m = \sendR{u}$.

  In sub-case (1) $m = \lamR{t}{x : A}{m_0}$, by the induction hypothesis on $n$ we have
  that either $n$ is a value or there exists $n'$ such that $n \Leadsto n'$.
  If $n \Leadsto n'$, then we are done by \textsc{Step-Explicit-App$_2$}.
  If $n$ is a value, then we are done by \textsc{Step-Explicit-$\beta$}.

  In sub-case (2) $m = \sendR{u}$, by the induction hypothesis on $n$ we have
  that either $n$ is a value or there exists $n'$ such that $n \Leadsto n'$.
  If $n \Leadsto n'$, then we are done by \textsc{Step-Explicit-App$_2$}.
  If $n$ is a value $v$, then we are done as $\sendR{u}\ v$ is a value.

\textbf{Case} (\textsc{Explicit-Pair}):
  By assumption we have 
  ${\Theta_1 ; \Gamma ; \Delta_1 \vdash m_1 : A}$ and 
  ${\Theta_2 ; \Gamma ; \Delta_2 \vdash m_2 : B[m_1/x]}$.
  From the induction hypothesis we have that either $m_1$ is a value or there exists $m_1'$ such that $m_1 \Leadsto m_1'$.
  If $m_1 \Leadsto m_1'$, then we are done by \textsc{Step-Explicit-Pair$_1$}.
  If $m_1$ is a value $u$, then by the induction hypothesis on $m_2$ we have
  that either $m_2$ is a value or there exists $m_2'$ such that $m_2 \Leadsto m_2'$.
  If $m_2 \Leadsto m_2'$, then we are done by \textsc{Step-Explicit-Pair$_2$}.
  If $m_2$ is a value $v$, then we are done since $\pairR{u}{v}{t}$ is a value.

\textbf{Case} (\textsc{Explicit-SumElim}):
  By the induction hypothesis we have that either $m$ is a value or there exists $m'$ such that $m \Leadsto m'$.
  If $m \Leadsto m'$, then we are done by \textsc{Step-SumElim$_1$}.
  If $m$ is a value, by \Cref{lemma:program-sig-canonical} we have $m = \pairR{u}{v}{t}$ for some $u, v$ and $t$.
  We are done by \textsc{Step-Explicit-PairElim}.

\textbf{Case} (\textsc{BoolElim}):
  By the induction hypothesis we have that either $m$ is a value or there exists $m'$ such that $m \Leadsto m'$.
  If $m \Leadsto m'$, then we are done by \textsc{Step-BoolElim$_1$}.
  If $m$ is a value, by \Cref{lemma:program-bool-canonical} we have two sub-cases:
  (1) $m = \bTrue$ and (2) $m = \bFalse$.
  In sub-case (1) $m = \bTrue$, we are done by \textsc{Step-TrueElim}.
  In sub-case (2) $m = \bFalse$, we are done by \textsc{Step-FalseElim}.

\textbf{Case} (\textsc{Return}):
  By the induction hypothesis we have that either $m$ is a value or there exists $m'$ such that $m \Leadsto m'$.
  If $m \Leadsto m'$, then we are done by \textsc{Step-Return}.
  If $m$ is a value, then $\return{m}$ is a value.

\textbf{Case} (\textsc{Bind}):
  By the induction hypothesis we have that either $m$ is a value or there exists $m'$ such that $m \Leadsto m'$.
  If $m \Leadsto m'$, then we are done by \textsc{Step-Bind}.
  If $m$ is a value, then $\letin{x}{m}{n}$ is a value.

For the session typing rules, the term $m$ is a thunked computation and thus a value.
\end{proof}

\subsection{Session Fidelity}\label{appendix:fidelity}
The session fidelity property ensures that well-typed processes 
will adhere to the protocols specified by their session types during execution.
To prove this property, we first must prove that structural congruence preserves typing.

\begin{lemma}[Congruence]\label[lemma]{lemma:congruence}
  Given $\Theta \Vdash P$ and $P \equiv Q$, then $\Theta \Vdash Q$.  
\end{lemma}
\begin{proof}
  By induction on the derivation of $\Theta \Vdash P$ and case analysis on the congruence relation.

\textbf{Case} (\textsc{Expr}): 
  It suffices to consider the case
  ${\proc{m} \equiv (\proc{m} \mid \proc{\return{\ii}})}$, 
  i.e. $Q = \proc{m} \mid \proc{\return{\ii}}$.

  Since $\Theta \Vdash \proc{m}$ and $\epsilon \Vdash \proc{\return{\ii}}$, we have
  $\Theta \dotcup \epsilon \Vdash \proc{m} \mid \proc{\return{\ii}}$ by \textsc{Par}.

\textbf{Case} (\textsc{Par}):
  \begin{mathpar}
    \inferrule[Par]
    { \Theta_1 \Vdash P \\ \Theta_2 \Vdash Q }
    { \Theta_1 \dotcup \Theta_2 \Vdash P \mid Q }
  \end{mathpar}
  By case analysis on the congruence relation we have the following sub-cases:
  \begin{enumerate}
    \item $P \mid Q \equiv Q \mid P$
    \item $P \mid (Q_1 \mid Q_2) \equiv (P \mid Q_1) \mid Q_2$
    \item $P \mid \proc{\return{\ii}} \equiv P$
    \item $\scope{cd}{P} \mid Q \equiv \scope{cd}{(P \mid Q)}$
  \end{enumerate}

  In sub-case (1), by \textsc{Par} we have $\Theta_2 \dotcup \Theta_1 \Vdash Q \mid P$.
  By the commutativity of $\dotcup$, we have $\Theta_1 \dotcup \Theta_2 \Vdash Q \mid P$ which
  concludes this sub-case.

  In sub-case (2), we have $\Theta_2 \Vdash Q_1 \mid Q_2$. By inversion on its typing derivation, we have
  $\Theta_{21} \Vdash Q_1$ and $\Theta_{22} \Vdash Q_2$ such that $\Theta_2 = \Theta_{21} \dotcup \Theta_{22}$.
  By \textsc{Par} we have $\Theta_1 \dotcup \Theta_{21}  \Vdash P \mid Q_1$.
  By \textsc{Par} again we have $(\Theta_1 \dotcup \Theta_{21}) \dotcup \Theta_{22} \Vdash (P \mid Q_1) \mid Q_2$.
  By the associativity of $\dotcup$, we have $\Theta_1 \dotcup (\Theta_{21} \dotcup \Theta_{22}) \Vdash (P \mid Q_1) \mid Q_2$.
  By substituting $\Theta_2$, we have $\Theta_1 \dotcup \Theta_2 \Vdash (P \mid Q_1) \mid Q_2$ which concludes this sub-case.

  In sub-case (3), we have $Q = \proc{\return{\ii}}$.
  By assumption we have $\Theta_2 \Vdash \proc{\return{\ii}}$.
  By inversion on its typing derivation, we have $\Theta_2 ; \epsilon ; \epsilon \vdash \return{\ii} : \CM{\unit}$.
  By \Cref{lemma:program-inversion-return} we have $\Theta_2 ; \epsilon ; \epsilon \vdash \ii : \unit$.
  By \Cref{lemma:program-inversion-unit} we have $\Theta_2 = \epsilon$ and $\epsilon \triangleright \Un$.
  Thus we have $\Theta_1 \dotcup \Theta_2 \Vdash P$.

  In sub-case (4), we have $\Theta_1 \Vdash \scope{cd}{P}$ and $\Theta_2 \Vdash Q$ by assumption.
  By inversion on the typing derivation of $\Theta_1 \Vdash \scope{cd}{P}$, we have
  $\Theta_1, c \tL \CH{A}, d \tL \HC{A} \Vdash P$ for some protocol $A$.
  By \textsc{Par} we have $\Theta_1, c \tL \CH{A}, d \tL \HC{A} \dotcup \Theta_2 \Vdash P \mid Q$.
  Since $c$ and $d$ are not in $\Theta_2$, we have $(\Theta_1 \dotcup \Theta_2), c \tL \CH{A}, d \tL \HC{A} \Vdash P \mid Q$.
  By \textsc{Scope} we have $\Theta_1 \dotcup \Theta_2 \Vdash \scope{cd}{(P \mid Q)}$ which concludes this sub-case.

\textbf{Case} (\textsc{Scope}):
  \begin{mathpar}
    \inferrule[Scope]
    { \Theta, c \tL \CH{A}, d \tL \HC{A} \Vdash P }
    { \Theta \Vdash \scope{cd}{P} }
  \end{mathpar}

  By case analysis on the congruence relation we have the following sub-cases:
  \begin{enumerate}
    \item $\scope{cd}{(P \mid Q)} \equiv \scope{cd}{P} \mid Q$
    \item $\scope{cd}{P} \equiv \scope{dc}{P}$
    \item $\scope{cd}{\scope{c'd'}{P}} \equiv \scope{c'd'}{\scope{cd}{P}}$
  \end{enumerate}
  
  In sub-case (1), we have $\Theta, c \tL \CH{A}, d \tL \HC{A} \Vdash P \mid Q$ by assumption.
  By inversion on its typing derivation, there exists $\Theta_1$ and $\Theta_2$ such that
  $\Theta = \Theta_1 \dotcup \Theta_2$ and $\Theta_1, c \tL \CH{A}, d \tL \HC{A} \Vdash P$ and $\Theta_2 \Vdash Q$.
  Channels $c$ and $d$ must be distributed the typing judgment of $P$ as
  they are linear and do not appear in $Q$. Applying \textsc{Scope} to $\Theta_1, c \tL \CH{A}, d \tL \HC{A} \Vdash P$ we have
  $\Theta_1 \Vdash \scope{cd}{P}$. By \textsc{Par} we have $\Theta_1 \dotcup \Theta_2 \Vdash \scope{cd}{P} \mid Q$
  which concludes this sub-case.

  In sub-case (2), we have $\Theta, c \tL \CH{A}, d \tL \HC{A} \Vdash P$ by assumption.
  By exchange of the context, we have $\Theta, d \tL \HC{A}, c \tL \CH{A} \Vdash P$.
  By \textsc{Scope} we have $\Theta \Vdash \scope{dc}{P}$ which concludes this sub-case.

  In sub-case (3), we have $\Theta, c \tL \CH{A}, d \tL \HC{A}, c' \tL \CH{A'}, d' \tL \HC{A'} \Vdash P$ by assumption.
  By exchange, we have $\Theta, c' \tL \CH{A'}, d' \tL \HC{A'}, c \tL \CH{A}, d \tL \HC{A} \Vdash P$.
  Applying \textsc{Scope} twice we have $\Theta \Vdash \scope{c'd'}{\scope{cd}{P}}$ which concludes this sub-case.
\end{proof}

To pull terms out of evaluation contexts, we need the following lemma.
\begin{lemma}\label[lemma]{lemma:eval-context}
  If ${\Theta ; \epsilon ; \epsilon \vdash \mcM[m] : \CM{B}}$, then there exists $\Theta_1, \Theta_2, A$ 
  such that ${\Theta_1 ; \epsilon ; \epsilon \vdash m : \CM{A}}$ and $\Theta = \Theta_1 \dotcup \Theta_2$.
  For any $\Theta_3, n$ such that ${\Theta_3 ; \epsilon ; \epsilon \vdash n : \CM{A}}$ and
  $\Theta_2 \dotcup \Theta_3$ is well-defined, we have
  $\Theta_2 \dotcup \Theta_3 ; \epsilon ; \epsilon \vdash \mcM[n] : \CM{B}$.
\end{lemma}
\begin{proof}
  By induction on the structure of $\mcM$.

\noindent
\textbf{Case} ($\mcM = [\cdot]$):
  Trivial by choosing $\Theta_1 = \Theta$, $\Theta_2 = \epsilon$ and $A = B$.

\noindent
\textbf{Case} ($\mcM = \letin{x}{\mcN}{n}$):
  We have the typing judgment $${\Theta ; \epsilon ; \epsilon \vdash \letin{x}{\mcN[m]}{n} : \CM{B}}$$

  \noindent
  By \Cref{lemma:program-inversion-bind} there exists $\Theta_1', \Theta_2', A'$ such that
  \begin{align*}
    & \Theta_1' ; \epsilon ; \epsilon \vdash \mcN[m] : \CM{A'} \\
    & \Theta_2' ; x : A' ; x :_r A' \vdash n : \CM{B} \\
    & \Theta = \Theta_1' \dotcup \Theta_2' \qquad x \notin \FV{B}
  \end{align*}

  \noindent
  By the induction hypothesis on $\Theta_1' ; \epsilon ; \epsilon \vdash \mcN[m] : \CM{A'}$,
  there exists $\Theta_{11}', \Theta_{12}', A''$ such that
  \begin{align*}
    & \Theta_{11}' ; \epsilon ; \epsilon \vdash m : A'' \\
    & \Theta_1' = \Theta_{11}' \dotcup \Theta_{12}'
  \end{align*}
  and for any $\Theta_3, n'$ such that ${\Theta_3 ; \epsilon ; \epsilon \vdash n' : A''}$ and
  $\Theta_{12}' \dotcup \Theta_3$ is well-defined, we have
  $$\Theta_{12}' \dotcup \Theta_3 ; \epsilon ; \epsilon \vdash \mcN[n'] : \CM{A'}$$

  \noindent
  By choosing $\Theta_1 = \Theta_{11}'$, $\Theta_2 = \Theta_{12}' \dotcup \Theta_2'$ and $A = A''$,
  we have $\Theta_1 ; \epsilon ; \epsilon \vdash m : A$ and $\Theta = \Theta_1 \dotcup \Theta_2$.

  \noindent
  For any $\Theta_3, n'$ such that ${\Theta_3 ; \epsilon ; \epsilon \vdash n' : A}$ and
  $\Theta_2 \dotcup \Theta_3$ is well-defined, we know that $\Theta_{12}' \dotcup \Theta_3$ is well-defined,
  which means that $$\Theta_{12}' \dotcup \Theta_3 ; \epsilon ; \epsilon \vdash \mcN[n'] : \CM{A'}$$

  \noindent
  Now applying \textsc{Bind} we have
  $$\Theta_2 \dotcup \Theta_3 ; \epsilon ; \epsilon \vdash \letin{x}{\mcN[n']}{n} : \CM{B}$$
  which concludes this case.
\end{proof}

\begin{theorem}[Session Fidelity]\label{theorem:session-fidelity}
  If $\Theta \Vdash P$ and $P \Rrightarrow Q$, then $\Theta \Vdash Q$.
\end{theorem}
\begin{proof}
  By induction on the derivation of $P \Rrightarrow Q$ and case analysis on the typing judgment.

\noindent
\textbf{Case} (\textsc{Proc-Fork}):
  \begin{align*}
    \proc{\mcN[\fork{x : A}{m}]}
    \Rrightarrow
    \scope{cd}{(\proc{\mcN[\return{c}]} \mid \proc{m[d/x]})} 
  \end{align*}
  By inversion on $\Theta \Vdash \proc{\mcN[\fork{x : A}{m}]}$ we have
  $$\Theta ; \epsilon ; \epsilon \vdash \mcN[\fork{x : A}{m}] : \CM{\unit}$$

  \noindent
  By \Cref{lemma:eval-context} we have $\Theta_1, \Theta_2, B$ such that
  $$
    \Theta_1 ; \epsilon ; \epsilon \vdash \fork{x : A}{m} : \CM{B}
  $$
  and $\Theta = \Theta_1 \dotcup \Theta_2$ and for any $\Theta_3, n$ such that
  $\Theta_3 ; \epsilon ; \epsilon \vdash n : \CM{B}$
  and $\Theta_2 \dotcup \Theta_3$ is well-defined, we have
  $\Theta_2 \dotcup \Theta_3 ; \epsilon ; \epsilon \vdash \mcN[n] : \CM{\unit}$.

  \noindent
  Applying \Cref{lemma:program-inversion-fork} to
  $\Theta_1 ; \epsilon ; \epsilon \vdash \fork{x : A}{m} : \CM{B}$
  we have
  $${\Theta_1 ; x : \CH{A'} ; x :_t \CH{A'} \vdash m : \CM{\unit}}$$
  and $A = \CH{A'}$ and $\CM{B} \simeq \CM{\HC{A'}}$.
  By the validity of context ${\Theta_2 ; x : \CH{A'} ; x :_t \CH{A'} \vdash}$ we have
  $\epsilon \vdash \CH{A'} : t$. By \textsc{ChType} and \Cref{theorem:sort-uniqueness} we have $t = \Ln$.
  Applying \Cref{corollary:inj-monad} to $\CM{B} \simeq \CM{\HC{A'}}$ we have $B \simeq \HC{A'}$.

  \noindent
  By \textsc{Channel-HC}, for some fresh channel $c$ we have
  $c \tL \HC{A'} ; \epsilon ; \epsilon \vdash c : \HC{A'}$.

  \noindent
  By \textsc{Channel-HC}, for some fresh channel $d$ we have
  $d \tL \CH{A'} ; \epsilon ; \epsilon \vdash d : \CH{A'}$.

  \noindent
  By \textsc{Return} we have $c \tL \HC{A'} ; \epsilon ; \epsilon \vdash \return{c} : \CM{\HC{A'}}$.

  \noindent
  Since $c$ is fresh, we have
  ${\Theta_2, c \tL \HC{A'} ; \epsilon ; \epsilon \vdash \mcN[\return{c}] : \CM{\unit}}$.

  \noindent
  By \Cref{lemma:program-subst-explicit} we have ${\Theta_1, d \tL \CH{A'} ; \epsilon ; \epsilon  \vdash m[d/x] : \CM{\unit}}$.

  \noindent
  By \textsc{Expr} we have ${\Theta_1, d \tL \CH{A'} \Vdash \proc{m[d/x]}}$ and
  ${\Theta_2, c \tL \HC{A'} \Vdash \proc{\mcN[\return{c}]}}$.

  \noindent
  By \textsc{Par} we have
  ${\Theta_1, d \tL \CH{A'} \dotcup \Theta_2, c \tL \HC{A'} \Vdash \proc{\mcN[\return{c}]} \mid \proc{m[d/x]}}$.

  \noindent
  By \textsc{Scope} we have
  ${\Theta_1 \dotcup \Theta_2 \Vdash \scope{cd}{(\proc{\mcN[\return{c}]} \mid \proc{m[d/x]})}}$
  which concludes this case.

\noindent
\textbf{Case} (\textsc{Proc-End}):
  \begin{align*}
    \scope{cd}{(
      \proc{\mcM[\close{c}]} 
      \mid 
      \proc{\mcN[\wait{d}]}
    )}
    \Rrightarrow 
    \proc{\mcM[\return{\ii}]} \mid \proc{\mcN[\return{\ii}]} 
  \end{align*}
  By inversion on 
  $\Theta \Vdash \scope{cd}{(\proc{\mcM[\appR{\sendR{c}}{v}]} \mid \proc{\mcN[\recvR{d}]})} $ 
  we have either
  \begin{enumerate}
    \item $\Theta, c \tL \CH{A}, d \tL \HC{A} \Vdash \proc{\mcM[\close{c}]} \mid \proc{\mcN[\wait{d}]}$
    \item $\Theta, c \tL \HC{A}, d \tL \CH{A} \Vdash \proc{\mcM[\close{c}]} \mid \proc{\mcN[\wait{d}]}$
  \end{enumerate}

  In sub-case (2), by inversion on its typing derivation we have
  ${\Theta_1, c \tL \HC{A} \Vdash \proc{\mcM[\close{c}]}}$ and
  ${\Theta_2, d \tL \CH{A} \Vdash \proc{\mcN[\wait{d}]}}$ such that ${\Theta = \Theta_1 \dotcup \Theta_2}$.

  \noindent
  By inversion on $\Theta_1, c \tL \HC{A} \Vdash \proc{\mcM[\close{c}]}$ we have
  $$\Theta_1, c \tL \HC{A} ; \epsilon ; \epsilon \vdash \mcM[\close{c}] : \CM{\unit}$$

  \noindent
  By \Cref{lemma:eval-context} we have
  $\Theta_{11}, c \tL \HC{A} ; \epsilon ; \epsilon \vdash \close{c} : \CM{B}$.

  \noindent
  By \Cref{lemma:program-inversion-close} we have $\Theta_{11}, c \tL \HC{A} ; \epsilon ; \epsilon \vdash c : \CH{\End}$
  which is a contradiction since $c$ cannot be $\CH{\End}$ in this context. Thus this sub-case is impossible.

  In sub-case (1), by inversion on its typing derivation we have
  ${\Theta_1, c \tL \CH{A} \Vdash \proc{\mcM[\close{c}]}}$ and
  ${\Theta_2, d \tL \HC{A} \Vdash \proc{\mcN[\wait{d}]}}$ such that ${\Theta = \Theta_1 \dotcup \Theta_2}$.

  \noindent
  By inversion on $\Theta_1, c \tL \CH{A} \Vdash \proc{\mcM[\close{c}]}$ we have
  $$\Theta_1, c \tL \CH{A} ; \epsilon ; \epsilon \vdash \mcM[\close{c}] : \CM{\unit}$$

  \noindent
  By inversion on $\Theta_2, d \tL \HC{A} \Vdash \proc{\mcN[\wait{d}]}$ we have
  $$\Theta_2, d \tL \HC{A} ; \epsilon ; \epsilon \vdash \mcN[\wait{d}] : \CM{\unit}$$

  \noindent
  By \Cref{lemma:program-inversion-bind} we have
  $$
    \Theta_{11}, c \tL \CH{A} ; \epsilon ; \epsilon \vdash \close{c} : \CM{B_1}
    \text{ and } {\Theta_1 = \Theta_{11} \dotcup \Theta_{12}}
  $$
  By \Cref{lemma:program-inversion-bind} we have
  $$
    \Theta_{21}, d \tL \HC{A} ; \epsilon ; \epsilon \vdash \wait{d} : \CM{B_2}
    \text{ and } {\Theta_2 = \Theta_{21} \dotcup \Theta_{22}}
  $$

  \noindent
  Applying \Cref{lemma:program-inversion-close} on 
  $\Theta_{11}, c \tL \CH{A} ; \epsilon ; \epsilon \vdash \close{c} : \CM{B_1}$ gives us
  $$
    \Theta_{11}, c \tL \CH{A} ; \epsilon ; \epsilon \vdash c : \CH{\End}
    \quad\text{and}\quad
    \CM{B_1} \simeq \CM{\unit}
  $$
  which implies $A \simeq \End$, $B_1 \simeq \unit$, and $\Theta_{11} = \epsilon$.

  \noindent
  Applying \Cref{lemma:program-inversion-wait} on
  $\Theta_{21}, d \tL \HC{A} ; \epsilon ; \epsilon \vdash \wait{d} : \CM{B_2}$ gives us
  $$
    \Theta_{21}, d \tL \HC{A} ; \epsilon ; \epsilon \vdash d : \HC{\End}
    \quad\text{and}\quad
    \CM{B_2} \simeq \CM{\unit}
  $$
  which implies $A \simeq \End$, $B_2 \simeq \unit$ and $\Theta_{21} = \epsilon$.

  \noindent
  Since $\Theta_{11} = \epsilon$ and $\Theta_{21} = \epsilon$, we have
  $$
    \Theta_{11} ; \epsilon ; \epsilon \vdash \return{\ii} : \CM{\unit}
    \quad\text{and}\quad
    \Theta_{21} ; \epsilon ; \epsilon \vdash \return{\ii} : \CM{\unit}
  $$

  Applying the evaluation contexts $\mcM$ and $\mcN$, we have
  $$\Theta_{11} \dotcup \Theta_{12} ; \epsilon ; \epsilon \vdash \mcM[\return{\ii}] : \CM{\unit}$$
  and
  $$\Theta_{21} \dotcup \Theta_{22} ; \epsilon ; \epsilon \vdash \mcN[\return{\ii}] : \CM{\unit}$$

  \noindent
  Applying \textsc{Expr} and \textsc{Par}, we have
  $$
    (\Theta_{11} \dotcup \Theta_{12}) \dotcup (\Theta_{21} \dotcup \Theta_{22}) \Vdash 
    \proc{\mcM[\return{\ii}]} \mid \proc{\mcN[\return{\ii}]}
  $$
  which concludes this case since $(\Theta_{11} \dotcup \Theta_{12}) \dotcup (\Theta_{21} \dotcup \Theta_{22}) = \Theta$.

\noindent
\textbf{Case} (\textsc{Proc-Com}):
  \begin{align*}
    \scope{cd}{(
      \proc{\mcM[\appR{\sendR{c}}{v}]} 
      \mid 
      \proc{\mcN[\recvR{d}]}
    )}
    \Rrightarrow 
    \scope{cd}{(
      \proc{\mcM[\return{c}]} 
      \mid 
      \proc{\mcN[\return{\pairR{v}{d}{\Ln}}]}
    )}
  \end{align*}
  By inversion on $\Theta \Vdash \scope{cd}{(\proc{\mcM[\appR{\sendR{c}}{v}]} \mid \proc{\mcN[\recvR{d}]})}$ we have either
  \begin{enumerate}
    \item $\Theta, c \tL \CH{A}, d \tL \HC{A} \Vdash \proc{\mcM[\appR{\sendR{c}}{v}]} \mid \proc{\mcN[\recvR{d}]}$
    \item $\Theta, c \tL \HC{A}, d \tL \CH{A} \Vdash \proc{\mcM[\appR{\sendR{c}}{v}]} \mid \proc{\mcN[\recvR{d}]}$
  \end{enumerate}

  In sub-case (1), by inversion on its typing derivation we have
  ${\Theta_1, c \tL \CH{A} \Vdash \proc{\mcM[\appR{\sendR{c}}{v}]}}$ and
  ${\Theta_2, d \tL \HC{A} \Vdash \proc{\mcN[\recvR{d}]}}$ such that ${\Theta = \Theta_1 \dotcup \Theta_2}$.

  \noindent
  By inversion on $\Theta_1, c \tL \CH{A} \Vdash \proc{\mcM[\appR{\sendR{c}}{v}]}$ we have
  $$\Theta_1, c \tL \CH{A} ; \epsilon ; \epsilon \vdash \mcM[\appR{\sendR{c}}{v}] : \CM{\unit}$$

  \noindent
  By inversion on $\Theta_2, d \tL \HC{A} \Vdash \proc{\mcN[\recvR{d}]}$ we have
  $$\Theta_2, d \tL \HC{A} ; \epsilon ; \epsilon \vdash \mcN[\recvR{d}] : \CM{\unit}$$

  \noindent
  Applying \Cref{lemma:eval-context} to 
  $\Theta_1, c \tL \CH{A} ; \epsilon ; \epsilon \vdash \mcM[\appR{\sendR{c}}{v}] : \CM{\unit}$
  we have
  $$
    \Theta_{11}, c \tL \CH{A} ; \epsilon ; \epsilon \vdash \appR{\sendR{c}}{v} : \CM{B_1}
    \text{ and } {\Theta_1 = \Theta_{11} \dotcup \Theta_{12}}
  $$

  \noindent
  Applying \Cref{lemma:eval-context} to
  $\Theta_2, d \tL \HC{A} ; \epsilon ; \epsilon \vdash \mcN[\recvR{d}] : \CM{\unit}$
  we have
  $$
    \Theta_{21}, d \tL \HC{A} ; \epsilon ; \epsilon \vdash \recvR{d} : \CM{B_2}
    \text{ and } {\Theta_2 = \Theta_{21} \dotcup \Theta_{22}}
  $$

  \noindent
  Applying \Cref{lemma:program-inversion-explicit-app} on
  $\Theta_{11}, c \tL \CH{A} ; \epsilon ; \epsilon \vdash \appR{\sendR{c}}{v} : \CM{B_1}$
  then there exists $A', B_1'$ such that
  \begin{align*}
    & \Theta_{111}, c \tL \CH{A} ; \epsilon ; \epsilon \vdash \sendR{c} : \PiR{t}{x : A'}{B_1'} \\
    & \Theta_{112} ; \epsilon ; \epsilon \vdash v : A' \\
    & \CM{B_1} \simeq B_1'[v/x] \quad\text{and}\quad \Theta_{11} = \Theta_{111} \dotcup \Theta_{112}
  \end{align*}

  \noindent
  Applying \Cref{lemma:program-inversion-explicit-send} on
  $\Theta_{111}, c \tL \CH{A} ; \epsilon ; \epsilon \vdash \sendR{c} : \PiR{t}{x : A'}{B_1'}$
  gives us either
  \begin{enumerate}[(a)]
    \item $\PiR{t}{x : A'}{B_1'} \simeq \PiR{\Ln}{A''}{\CM{\CH{B_1''}}}$ and
          $\Theta_{111}, c \tL \CH{A} ; \epsilon ; \epsilon \vdash c : \CH{\ActR{!}{x : A''}{B_1''}}$.
    \item $\PiR{t}{x : A'}{B_1'} \simeq \PiR{\Ln}{A''}{\CM{\HC{B_1''}}}$ and
          $\Theta_{111}, c \tL \CH{A} ; \epsilon ; \epsilon \vdash c : \HC{\ActR{?}{x : A''}{B_1''}}$.
  \end{enumerate}
  In (b), we have a contradiction since $c$ cannot be $\HC{\ActR{?}{x : A''}{B_1''}}$ in this context.
  Thus (b) is impossible. In (a), we have $A' \simeq A''$ and $B_1' \simeq \CM{\CH{B_1''}}$ by \Cref{corollary:inj-explicit-fun}.
  Additionally, we have $A \simeq \ActR{!}{x : A''}{B_1''}$ and $B_1 \simeq \CH{B_1''[v/x]}$ and $\Theta_{111} = \epsilon$.

  \noindent
  Applying \Cref{lemma:program-inversion-explicit-recv} on
  $\Theta_{21}, d \tL \HC{A} ; \epsilon ; \epsilon \vdash \recvR{d} : \CM{B_2}$
  gives us either
  \begin{enumerate}[(a)]
    \item $\CM{B_2} \simeq \CM{\SigR{\Ln}{x : A'''}{\CH{B_2'}}}$ and
          $\Theta_{21}, d \tL \HC{A} ; \epsilon ; \epsilon \vdash d : \CH{\ActR{?}{x : A'''}{B_2'}}$.
    \item $\CM{B_2} \simeq \CM{\SigR{\Ln}{x : A'''}{\HC{B_2'}}}$ and
          $\Theta_{21}, d \tL \HC{A} ; \epsilon ; \epsilon \vdash d : \HC{\ActR{!}{x : A'''}{B_2'}}$.
  \end{enumerate}
  In (a), we have a contradiction since $d$ cannot be $\CH{\ActR{?}{x : A'''}{\CH{B_2'}}}$ in this context.
  Thus (a) is impossible. In (b), we have $B_2 \simeq \SigR{\Ln}{x : A'''}{\HC{B_2'}}$ by \Cref{corollary:inj-monad}.
  Additionally, we have $A \simeq \ActR{!}{x : A'''}{B_2'}$ and $\Theta_{21} = \epsilon$.

  \noindent
  Since $A \simeq \ActR{!}{x : A''}{B_1''}$ and $A \simeq \ActR{!}{x : A'''}{B_2'}$, by
  transitivity of $\simeq$ we have $$\ActR{!}{x : A''}{B_1''} \simeq \ActR{!}{x : A'''}{B_2'}$$
  By \Cref{corollary:inj-implicit-action} we have $A' \simeq A'' \simeq A'''$ and $B_1'' \simeq B_2'$. 

  \noindent
  By \textsc{Channel-CH}, we have $c \tL \CH{B_1''[v/x]} ; \epsilon ; \epsilon \vdash c : \CH{B_1''[v/x]}$.

  \noindent
  By \textsc{Channel-HC}, we have $d \tL \HC{B_1''[v/x]} ; \epsilon ; \epsilon \vdash d : \HC{B_1''[v/x]}$.

  \noindent
  Applying \textsc{Return} to ${c \tL \CH{B_1''[v/x]} ; \epsilon ; \epsilon \vdash c : \CH{B_1''[v/x]}}$
  we have $${c \tL \CH{B_1''[v/x]} ; \epsilon ; \epsilon \vdash \return{c} : \CM{\CH{B_1''[v/x]}}}$$

  \noindent
  Since there is $B_1 \simeq \CH{B_1''[v/x]}$, we can apply evaluation context $\mcM$ to create
  $$\Theta_{12}, c \tL \CH{B_1''[v/x]} ; \epsilon ; \epsilon \vdash \mcM[\return{c}] : \CM{\unit}$$

  \noindent
  Pairing with $v$ with $d$ using \textsc{Explicit-Pair} we have
  $$\Theta_{112}, d \tL \HC{B_1''[v/x]} ; \epsilon ; \epsilon \vdash \pairR{v}{d}{\Ln} : \SigR{\Ln}{x : A'}{\HC{B_1''}}$$

  \noindent
  Apply \textsc{Return} to 
  $\Theta_{112}, d \tL \HC{B_1''[v/x]} ; \epsilon ; \epsilon \vdash \pairR{v}{d}{\Ln} : \SigR{\Ln}{x : A'}{\HC{B_1''}}$
  we have
  $$\Theta_{112}, d \tL \HC{B_1''[v/x]} ; \epsilon ; \epsilon \vdash \return{\pairR{v}{d}{\Ln}} : \CM{\SigR{\Ln}{x : A'}{\HC{B_1''}}}$$

  \noindent
  Since $B_2 \simeq \SigR{\Ln}{x : A'''}{\HC{B_2'}} \simeq \SigR{\Ln}{x : A'}{\HC{B_1''}}$, we have
  $$\Theta_{112} \dotcup \Theta_{22}, d \tL \HC{B_1''[v/x]} ; \epsilon ; \epsilon \vdash \mcN[\return{\pairR{v}{d}{\Ln}}] : \CM{\unit}$$

  \noindent
  By \textsc{Expr} we have
  $$\Theta_{12}, c \tL \CH{B_1''[v/x]} \Vdash \proc{\mcM[\return{c}]}$$
  and
  $$\Theta_{112} \dotcup \Theta_{22}, d \tL \HC{B_1''[v/x]} \Vdash \proc{\mcN[\return{\pairR{v}{d}{\Ln}}]}$$

  \noindent
  By \textsc{Par} we have
  \begin{align*}
    (\Theta_{112} \dotcup \Theta_{12} \dotcup \Theta_{22}), c \tL \CH{B_1''[v/x]}, d \tL \HC{B_1''[v/x]} 
    \Vdash \proc{\mcM[\return{c}]} \mid \proc{\mcN[\return{\pairR{v}{d}{\Ln}}]}
  \end{align*}

  \noindent
  By \textsc{Scope} we have
  $${\Theta_{112} \dotcup \Theta_{12} \dotcup \Theta_{22} \Vdash 
    \scope{cd}{(\proc{\mcM[\return{c}]} \mid \proc{\mcN[\return{\pairR{v}{d}{\Ln}}]})}}$$

  \noindent
  Since $\Theta_{111} = \epsilon$, $\Theta_{21} = \epsilon$,
  we have 
  $\Theta = ((\Theta_{111} \dotcup \Theta_{112}) \dotcup \Theta_{12}) \dotcup (\Theta_{21} \dotcup \Theta_{22})$ which means
  $$\Theta \Vdash \scope{cd}{(\proc{\mcM[\return{c}]} \mid \proc{\mcN[\return{\pairR{v}{d}{\Ln}}]})}$$
  thus concluding this sub-case.

  For sub-case (2), the proof is similar to sub-case (1). The only difference is that the $\CH{\cdot}$ and
  $\HC{\cdot}$ types are swapped.

\noindent
\textbf{Case} (\textsc{Proc-\underline{Com}}):
  \begin{align*}
    \scope{cd}{(
      \proc{\mcM[\appI{\sendI{c}}{o}]} 
      \mid 
      \proc{\mcN[\recvI{d}]}
     )}
     \Rrightarrow 
     \scope{cd}{(
      \proc{\mcM[\return{c}]} 
      \mid 
      \proc{\mcN[\return{\pairI{o}{d}{\Ln}}]}
     )}
  \end{align*}
  The proof is similar to the previous case (\textsc{Proc-Com}) with the only difference being
  that explicit applications are replaced with implicit applications.

\noindent
\textbf{Case} (\textsc{Proc-Expr}):
\begin{mathpar}
  \inferrule
  { m \Leadsto m' }
  { \proc{m} \Rrightarrow \proc{m'} }
\end{mathpar}

  By inversion on $\Theta \Vdash \proc{m}$ we have $\Theta ; \epsilon ; \epsilon \vdash m : \CM{\unit}$.
  By \Cref{theorem:program-subject-reduction} we have $\Theta ; \epsilon ; \epsilon \vdash m' : \CM{\unit}$.
  By \textsc{Expr} we have $\Theta \Vdash \proc{m'}$ which concludes this case.

\noindent
\textbf{Case} (\textsc{Proc-Par}):
  \begin{mathpar}
    \inferrule
    { P \Rrightarrow Q }
    { O \mid P \Rrightarrow O \mid Q }
  \end{mathpar}
  By inversion on $\Theta \Vdash O \mid P$ we have
  $\Theta_1 \Vdash O$ and $\Theta_2 \Vdash P$ such that $\Theta = \Theta_1 \dotcup \Theta_2$.
  By the induction hypothesis we have $\Theta_2 \Vdash Q$.
  By \textsc{Par} we have $\Theta_1 \dotcup \Theta_2 \Vdash O \mid Q$ which concludes this case.

\noindent
\textbf{Case} (\textsc{Proc-Scope}):
  \begin{mathpar}
    \inferrule
    { P \Rrightarrow Q }
    { \scope{cd}{P} \Rrightarrow \scope{cd}{Q} }
  \end{mathpar}
  By inversion on $\Theta \Vdash \scope{cd}{P}$ we have either
  \begin{enumerate}
    \item $\Theta, c \tL \CH{A}, d \tL \HC{A} \Vdash P$
    \item $\Theta, c \tL \HC{A}, d \tL \CH{A} \Vdash P$
  \end{enumerate}

  In case (1), by the induction hypothesis we have
  $\Theta, c \tL \CH{A}, d \tL \HC{A} \Vdash Q$.

  \noindent
  By \textsc{Scope} we have $\Theta \Vdash \scope{cd}{Q}$ which concludes this sub-case.

  In case (2), by the induction hypothesis we have
  $\Theta, c \tL \HC{A}, d \tL \CH{A} \Vdash Q$.

  \noindent
  By \textsc{Scope} we have $\Theta \Vdash \scope{cd}{Q}$ which concludes this sub-case.

\noindent
\textbf{Case} (\textsc{Proc-Congr}):
  \begin{mathpar}
    \inferrule
    { P \equiv P' \\ P' \Rrightarrow Q' \\ Q' \equiv Q }
    { P \Rrightarrow Q }
  \end{mathpar}
  By \Cref{lemma:congruence} we have $\Theta \Vdash P'$.
  By the induction hypothesis we have $\Theta \Vdash Q'$.
  By \Cref{lemma:congruence} we have $\Theta \Vdash Q$ which concludes this case.
\end{proof}

\subsection{Global Progress}\label{appendix:progress}
The process level type system of \TLLC{} is insufficient to ensure that
arbitrary process configurations enjoy global progress. 
This is because cyclic channel topologies are also considered to be well-typed.
However, we can still prove a weaker form of progress for a class of configurations
we call \emph{reachable configurations}. Intuitively, a reachable configuration is one
that can be reached from a well-typed singleton process through \Fork{}-operations.

Formally, we define the structure of \emph{spawning trees} to capture the
spawning relationships between parent-to-children processes. 
This formalism is inspired by the \emph{nested-multiverse semantics} for
reasoning about probabilistic session types~\cite{das23,fu_prob25} and dominator trees from graph theory. 
\begin{center}
  \vspace{0.5em}
  \begin{tabular}{r L C L}
    spawning tree & \mcP, \mcQ & ::= & \Root{m}{ \{ ( c_i, \mcP_i ) \}_{i \in \mcI}, \{ \mcQ_j \}_{j \in \mcJ} } \\
                  &            & \;| & \Node{d}{m}{ \{ ( c_i, \mcP_i ) \}_{i \in \mcI}, \{ \mcQ_j \}_{j \in \mcJ} }
  \end{tabular}
  \vspace{0.5em}
\end{center}
Each tree is associated with a term $m$ that performs computation and a set of children processes
$\{ ( c_i, \mcP_i ) \}_{i \in \mcI}$ where $m$ communicates with each child process $\mcP_i$ through
channel $c_i$. It also contains a set of subtrees $\{ \mcQ_j \}_{j \in \mcJ}$ that are no longer in
communication with $m$ (i.e. they have been detached through \Close{}/\Wait{}-operations).
In the case of internal nodes, $d$ is the channel which $m$ uses to 
communicate with its parent process. 

The crucial ideal behind the spawning tree structure is that we are going to
define an alternative process semantics that operates on spawning trees. 
We will show that this alternative semantics can be simulated by the original process semantics.
Moreover, we will show that the spawning tree semantics enjoys global progress.
By defining reachable configurations as those that can be derived from
well-typed spawning trees, we can then prove that reachable configurations
enjoy global progress (induced by simulation). 

To make the typing rules of spawning trees easier to define, we first introduce 
the following notations for channel types and $\kappa \in \{ +, - \}$:
\begin{align*}
  \ch{+}{A} = \CH{A} && \neg{\ch{+}{A}} = \ch{-}{A} \\
  \ch{-}{A} = \HC{A} && \neg{\ch{-}{A}} = \ch{+}{A}
\end{align*}

We define the typing rules for 
spawning trees as follows:
\begin{mathpar}\small
  \inferrule
  { \overline{c_i \tL \ch{\kappa_i}{A_i}} ; \epsilon ; \epsilon \vdash m :\CM{\unit} \\
    \forall i \in \mcI,\ \neg{\ch{\kappa_i}{A_i}} \Vdash \mcP_i  \\
    \forall j \in \mcJ,\ \Vdash \mcQ_j }
  { \Vdash \Root{m}{ \{ ( c_i, \mcP_i ) \}_{i \in \mcI}, \{ \mcQ_j \}_{j \in \mcJ} } }
  \textsc{(Valid-Root)}
  \\
  \inferrule
  {  \overline{c_i \tL \ch{\kappa_i}{A_i}}, d \tL \ch{\kappa}{A} ; \epsilon ; \epsilon \vdash m :\CM{\unit} \\
    \forall i \in \mcI,\ \neg{\ch{\kappa_i}{A_i}} \Vdash \mcP_i \\ 
    \forall j \in \mcJ,\ \Vdash \mcQ_j }
  { \ch{\kappa}{A} \Vdash \Node{d}{m}{ \{ ( c_i, \mcP_i ) \}_{i \in \mcI}, \{ \mcQ_j \}_{j \in \mcJ} } }
  \textsc{(Valid-Node)}
\end{mathpar}
For the root node, we require that the term $m$ be well-typed in
channel context $\overline{c_i \tL \ch{\kappa_i}{A_i}}$ comprised of $\ch{\kappa_i}{A_i}$
channels connecting to its children $\mcP_i$.
The dual of each channel $\neg{\ch{\kappa_i}{A_i}}$ is propagated to type the corresponding
child $\neg{\ch{\kappa_i}{A_i}} \Vdash \mcP_i$.
When typing an internal node, we require that $m$ be well-typed in a channel context
that also includes $d \tL \ch{\kappa}{A}$, i.e. the channel connecting to its parent process.

We define the \emph{flattening} operation $| \mcP |$ that converts a spawning tree into
a standard process configuration. The operation is defined as follows:
\begin{mathpar}\small
  \inferrule[Flatten-Root]
  { \forall i \in \mcI,\ | \mcP_i | = (d_i, P_i) \\ 
    \forall j \in \mcJ,\ | \mcQ_j | = Q_j }
  { | \Root{m}{ \{ ( c_i, \mcP_i ) \}_{i \in \mcI}, \{ \mcQ_j \}_{j \in \mcJ} } | = 
    \scope{\overline{c_i d_i}}{(\proc{m} \mid \overline{P_i})}  
    \mid \overline{Q_j} }

  \inferrule[Flatten-Node]
  { \forall i \in \mcI,\ | \mcP_i | = (d_i, P_i) \\ 
    \forall j \in \mcJ,\ | \mcQ_j | = Q_j }
  { | \Node{d}{m}{ \{ ( c_i, \mcP_i ) \}_{i \in \mcI}, \{ \mcQ_j \}_{j \in \mcJ} } | = 
    (d, \scope{\overline{c_i d_i}}{(\proc{m} \mid \overline{P_i})} \mid \overline{Q_j}) }
\end{mathpar}
The flattening operation recursively flattens each child process $\mcP_i$ into
a channel-process pair $(d_i, P_i)$ and each subtree $\mcQ_j$ into a process configuration $Q_j$.
It then composes $m$ with all the sub-processes in parallel. 
The channel pairs ${c_i, d_i}$ are restricted to ensure proper channel scoping.

We now connect the validity of spawning trees to the well-typedness of flattened process configurations
through \Cref{lemma:flatten-valid}.
\begin{lemma}[Flatten Valid]\label[lemma]{lemma:flatten-valid}
  If $\Vdash \mcP$ and $| \mcP | = P$, then $\Vdash P$
  and
  if ${\ch{\kappa}{A} \Vdash \mcP}$ and $| \mcP | = (d, P)$, then $d \tL \ch{\kappa}{A} \Vdash P$.
\end{lemma}
\begin{proof}
  By mutual induction on the derivation of $\Vdash \mcP$ and $A \Vdash \mcP$.

\noindent
\textbf{Case} (\textsc{Valid-Root}):
  From \textsc{Flatten-Root}, we have
  $| \mcP_i | = (d_i, P_i)$ for each child process $\mcP_i$ and
  $| \mcQ_j | = Q_j$ for each subtree $\mcQ_j$.
  By the induction hypothesis, we have $d_i \tL \neg{\ch{\kappa_i}{A_i}} \Vdash P_i$ for each $i \in \mcI$
  and $\Vdash Q_j$ for each $j \in \mcJ$.

  \noindent
  From the premise of \textsc{Valid-Root}, we have
  $\overline{c_i \tL \ch{\kappa_i}{A_i}} ; \epsilon ; \epsilon \vdash m :\CM{\unit}$.

  \noindent
  By \textsc{Expr}, we have
  $\overline{c_i \tL \ch{\kappa_i}{A_i}} \Vdash \proc{m}$.

  \noindent
  By applying \textsc{Par} repeated, we have
  $\overline{c_i \tL \ch{\kappa_i}{A_i}}, \overline{d_i \tL \neg{\ch{\kappa_i}{A_i}}} \Vdash (\proc{m} \mid \overline{P_i})$.

  \noindent
  By applying \textsc{Scope} repeatedly, we have
  $\Vdash \scope{\overline{c_i d_i}}{(\proc{m} \mid \overline{P_i})}$.

  \noindent
  By applying \textsc{Par} repeatedly, we have
  $\Vdash \scope{\overline{c_i d_i}}{(\proc{m} \mid \overline{P_i})} \mid \overline{Q_j}$
  which concludes this case.

\textbf{Case} (\textsc{Valid-Node}):
  From \textsc{Flatten-Node}, we have for each child process $\mcP_i$, $| \mcP_i | = (d_i, P_i)$ and
  for each subtree $\mcQ_j$, $| \mcQ_j | = Q_j$.
  By the induction hypothesis, we have $d_i \tL \neg{\ch{\kappa_i}{A_i}} \Vdash P_i$ for each $i \in \mcI$
  and $\Vdash Q_j$ for each $j \in \mcJ$.

  \noindent
  From the premise of \textsc{Valid-Node}, we have
  $\overline{c_i \tL \ch{\kappa_i}{A_i}}, d \tL \ch{\kappa}{A} ; \epsilon ; \epsilon \vdash m :\CM{\unit}$.

  \noindent
  By \textsc{Expr}, we have
  $\overline{c_i \tL \ch{\kappa_i}{A_i}}, d \tL \ch{\kappa}{A} \Vdash \proc{m}$.

  \noindent
  By applying \textsc{Par} repeated, we have
  $\overline{c_i \tL \ch{\kappa_i}{A_i}}, \overline{d_i \tL \neg{\ch{\kappa_i}{A_i}}}, d \tL \ch{\kappa}{A} \Vdash 
    (\proc{m} \mid \overline{P_i})$.

  \noindent
  By applying \textsc{Scope} repeatedly, we have
  $d \tL \ch{\kappa}{A} \Vdash \scope{\overline{c_i d_i}}{(\proc{m} \mid \overline{P_i})}$.

  \noindent
  By applying \textsc{Par} repeatedly, we have
  ${d \tL \ch{\kappa}{A} \Vdash \scope{\overline{c_i d_i}}{(\proc{m} \mid \overline{P_i})} \mid \overline{Q_j}}$
  which concludes this case.
\end{proof}

We now define the spawning tree semantics through the following reduction rules:

\begin{mathpar}\small
  \inferrule[Root-Fork]
  { 
    \mcI' = \{ i \in \mcI | c_i \in \FC{m} \}
  }
  { 
    \Root{ 
      \mcN[\fork{x : A}{m}]
    }{ 
      \{ ( c_i, \mcP_i ) \}_{i \in \mcI}, 
      \{ \mcQ_j \}_{j \in \mcJ}
    } \\
    \quad\Rrightarrow
    \Root{
      \mcN[\return{c}]
    }{
      \{ ( c_i, \mcP_i ) \}_{i \in \mcI \setminus \mcI'} 
      \cup
      \{ (c, 
          \Node{d}{m[d/x]}{ 
            \{ ( c_{i'}, \mcP_{i'} ) \}_{i' \in \mcI'},
            \emptyset
          }) 
      \},
      \{ \mcQ_j \}_{j \in \mcJ}
    }
  }

  \inferrule[Node-Fork]
  { 
    \mcI' = \{ i \in \mcI | c_i \in \FC{m} \}
  }
  { 
    \Node{d}{ 
      \mcN[\fork{x : A}{m}]
    }{ 
      \{ ( c_i, \mcP_i ) \}_{i \in \mcI}, 
      \{ \mcQ_j \}_{j \in \mcJ}
    } \\
    \quad\Rrightarrow
    \Node{d}{
      \mcN[\return{c}]   
    }{
      \{ ( c_i, \mcP_i ) \}_{i \in \mcI \setminus \mcI'} 
      \cup
      \{ (c, 
          \Node{d}{m[d/x]}{ 
            \{ ( c_{i'}, \mcP_{i'} ) \}_{i' \in \mcI'},
            \emptyset
          }) 
      \},
      \{ \mcQ_j \}_{j \in \mcJ}
    }
  }
\end{mathpar}

\begin{mathpar}\small
  \inferrule[Root-Wait]
  { k \in \mcI \\ 
    \mcP_k = \Node{d_k}{
      \mcN[\close{d_k}]
    }{ 
      \{ ( c_{i_k}, \mcP_{i_k} ) \}_{i_k \in \mcI_k},
      \{ \mcQ_{j_k} \}_{j_k \in \mcJ_k}
    } \\
    \mcQ_{k} = 
    \Root{
      \mcN[\return{\ii}]
    }{ 
      \{ ( c_{i_k}, \mcP_{i_k} ) \}_{i_k \in \mcI_k},
      \{ \mcQ_{j_k} \}_{j_k \in \mcJ_k}
    }
  }
  { 
    \Root{
      \mcM[\wait{c_k}]
    }{ 
      \{ ( c_i, \mcP_i ) \}_{i \in \mcI},
      \{ \mcQ_j \}_{j \in \mcJ}
    } \\\\
    \quad\Rrightarrow
    \Root{
      \mcM[\return{\ii}]
    }{ 
      \{ ( c_i, \mcP_i ) \}_{i \in \mcI \setminus \{k\}},
      \{ \mcQ_j \}_{j \in \mcJ \cup \{k\}}
    }
  }

  \inferrule[Node-Wait]
  { k \in \mcI \\ 
    \mcP_k = \Node{d_k}{
      \mcN[\close{d_k}]
    }{ 
      \{ ( c_{i_k}, \mcP_{i_k} ) \}_{i_k \in \mcI_k},
      \{ \mcQ_{j_k} \}_{j_k \in \mcJ_k}
    } \\
    \mcQ_{k} = 
    \Root{
      \mcN[\return{\ii}]
    }{ 
      \{ ( c_{i_k}, \mcP_{i_k} ) \}_{i_k \in \mcI_k},
      \{ \mcQ_{j_k} \}_{j_k \in \mcJ_k}
    }
  }
  { 
    \Node{d}{
      \mcM[\wait{c_k}]
    }{ 
      \{ ( c_i, \mcP_i ) \}_{i \in \mcI},
      \{ \mcQ_j \}_{j \in \mcJ}
    } \\\\
    \quad\Rrightarrow
    \Node{d}{
      \mcM[\return{\ii}]
    }{ 
      \{ ( c_i, \mcP_i ) \}_{i \in \mcI \setminus \{k\}},
      \{ \mcQ_j \}_{j \in \mcJ \cup \{k\}}
    }
  }
  
  \inferrule[Root-Close]
  { k \in \mcI \\ 
    \mcP_k = \Node{d_k}{
      \mcN[\wait{d_k}]
    }{ 
      \{ ( c_{i_k}, \mcP_{i_k} ) \}_{i_k \in \mcI_k},
      \{ \mcQ_{j_k} \}_{j_k \in \mcJ_k}
    } \\
    \mcQ_{k} = 
    \Root{
      \mcN[\return{\ii}]
    }{ 
      \{ ( c_{i_k}, \mcP_{i_k} ) \}_{i_k \in \mcI_k},
      \{ \mcQ_{j_k} \}_{j_k \in \mcJ_k}
    }
  }
  { 
    \Root{
      \mcM[\close{c_k}]
    }{ 
      \{ ( c_i, \mcP_i ) \}_{i \in \mcI},
      \{ \mcQ_j \}_{j \in \mcJ}
    } \\\\
    \quad\Rrightarrow
    \Root{
      \mcM[\return{\ii}]
    }{ 
      \{ ( c_i, \mcP_i ) \}_{i \in \mcI \setminus \{k\}},
      \{ \mcQ_j \}_{j \in \mcJ \cup \{k\}}
    }
  }

  \inferrule[Node-Close]
  { k \in \mcI \\ 
    \mcP_k = \Node{d_k}{
      \mcN[\wait{d_k}]
    }{ 
      \{ ( c_{i_k}, \mcP_{i_k} ) \}_{i_k \in \mcI_k},
      \{ \mcQ_{j_k} \}_{j_k \in \mcJ_k}
    } \\
    \mcQ_{k} = 
    \Root{
      \mcN[\return{\ii}]
    }{ 
      \{ ( c_{i_k}, \mcP_{i_k} ) \}_{i_k \in \mcI_k},
      \{ \mcQ_{j_k} \}_{j_k \in \mcJ_k}
    }
  }
  { 
    \Node{d}{
      \mcM[\close{c_k}]
    }{ 
      \{ ( c_i, \mcP_i ) \}_{i \in \mcI},
      \{ \mcQ_j \}_{j \in \mcJ}
    } \\\\
    \quad\Rrightarrow
    \Node{d}{
      \mcM[\return{\ii}]
    }{ 
      \{ ( c_i, \mcP_i ) \}_{i \in \mcI \setminus \{k\}},
      \{ \mcQ_j \}_{j \in \mcJ \cup \{k\}}
    }
  }
\end{mathpar}

\begin{mathpar}\small
  \inferrule[Root-Send]
  { k \in \mcI \\ 
    \mcP_k = \Node{d_k}{
      \mcN[\recvR{d_k}]
    }{
      \{ ( c_{i_k}, \mcP_{i_k} ) \}_{i_k \in \mcI_k},
      \{ \mcQ_{j_k} \}_{j_k \in \mcJ_k}
    } \\
    \mcI' = \{ i \in \mcI | c_i \in \FC{v} \} \\
    \mcP_k' = \Node{d_k}{
      \mcN[\return{\pairR{v}{d_k}{\Ln}}]
    }{
      \{ ( c_{i_k}, \mcP_{i_k} ) \}_{i_k \in \mcI_k \cup \mcI'},
      \{ \mcQ_{j_k} \}_{j_k \in \mcJ_k}
    }
  }
  { \Root{
      \mcM[\appR{\sendR{c_k}}{v}]
    }{
      \{ ( c_i, \mcP_i ) \}_{i \in \mcI},
      \{ \mcQ_j \}_{j \in \mcJ}
    } \\\\
    \Rrightarrow
    \Root{
      \mcM[\return{c_k}]
    }{
      \{ (c_i, \mcP_i) \}_{i \in \mcI \setminus (\{ k \} \cup \mcI')} \cup \{ (c_k, \mcP_k') \},
      \{ \mcQ_j \}_{j \in \mcJ}
    }
  }

  \inferrule[Node-Send]
  { k \in \mcI \\ 
    \mcP_k = \Node{d_k}{
      \mcN[\recvR{d_k}]
    }{
      \{ ( c_{i_k}, \mcP_{i_k} ) \}_{i_k \in \mcI_k},
      \{ \mcQ_{j_k} \}_{j_k \in \mcJ_k}
    } \\
    \mcI' = \{ i \in \mcI | c_i \in \FC{v} \} \\
    d \notin \FC{v} \\
    \mcP_k' = \Node{d_k}{
      \mcN[\return{\pairR{v}{d_k}{\Ln}}]
    }{
      \{ ( c_{i_k}, \mcP_{i_k} ) \}_{i_k \in \mcI_k \cup \mcI'},
      \{ \mcQ_{j_k} \}_{j_k \in \mcJ_k}
    }
  }
  { \Node{d}{
      \mcM[\appR{\sendR{c_k}}{v}]
    }{
      \{ ( c_i, \mcP_i ) \}_{i \in \mcI},
      \{ \mcQ_j \}_{j \in \mcJ}
    } \\\\
    \Rrightarrow
    \Node{d}{
      \mcM[\return{c_k}]
    }{
      \{ (c_i, \mcP_i) \}_{i \in \mcI \setminus (\{ k \} \cup \mcI')} \cup \{ (c_k, \mcP_k') \},
      \{ \mcQ_j \}_{j \in \mcJ}
    }
  }

  \inferrule[Root-Recv]
  { k \in \mcI \\
    \mcP_k = \Node{d_k}{
      \mcN[\appR{\sendR{d_k}}{v}]
    }{
      \{ ( c_{i_k}, \mcP_{i_k} ) \}_{i_k \in \mcI_k},
      \{ \mcQ_{j_k} \}_{j_k \in \mcJ_k}
    } \\
    \mcI' = \{ i \in \mcI_k | c_i \in \FC{v} \} \\
    \mcP_k' = \Node{d_k}{
      \mcN[\return{d_k}]
    }{
      \{ ( c_{i_k}, \mcP_{i_k} ) \}_{i_k \in \mcI_k \setminus \mcI'},
      \{ \mcQ_{j_k} \}_{j_k \in \mcJ_k}
    }
  }
  {
    \Root{
      \mcM[\recvR{c_k}]
    } 
    {
      \{ ( c_i, \mcP_i ) \}_{i \in \mcI},
      \{ \mcQ_j \}_{j \in \mcJ}
    } \\\\
    \Rrightarrow
    \Root{
      \mcM[\return{\pairR{v}{c_k}{\Ln}}]
    }{
      \{ (c_i, \mcP_i) \}_{i \in (\mcI \setminus \{ k \}) \cup \mcI'} \cup \{ (c_k, \mcP_k') \},
      \{ \mcQ_j \}_{j \in \mcJ}
    }
  }

  \inferrule[Node-Recv]
  { k \in \mcI \\
    \mcP_k = \Node{d_k}{
      \mcN[\appR{\sendR{d_k}}{v}]
    }{
      \{ ( c_{i_k}, \mcP_{i_k} ) \}_{i_k \in \mcI_k},
      \{ \mcQ_{j_k} \}_{j_k \in \mcJ_k}
    } \\
    \mcI' = \{ i \in \mcI_k | c_i \in \FC{v} \} \\
    \mcP_k' = \Node{d_k}{
      \mcN[\return{d_k}]
    }{
      \{ ( c_{i_k}, \mcP_{i_k} ) \}_{i_k \in \mcI_k \setminus \mcI'},
      \{ \mcQ_{j_k} \}_{j_k \in \mcJ_k}
    }
  }
  {
    \Node{d}{
      \mcM[\recvR{c_k}]
    } 
    {
      \{ ( c_i, \mcP_i ) \}_{i \in \mcI},
      \{ \mcQ_j \}_{j \in \mcJ}
    } \\\\
    \Rrightarrow
    \Node{d}{
      \mcM[\return{\pairR{v}{c_k}{\Ln}}]
    }{
      \{ (c_i, \mcP_i) \}_{i \in (\mcI \setminus \{ k \}) \cup \mcI'} \cup \{ (c_k, \mcP_k') \},
      \{ \mcQ_j \}_{j \in \mcJ}
    }
  }
\end{mathpar}
\clearpage

\begin{mathpar}\small
  \inferrule[Root-\underline{Send}]
  { k \in \mcI \\ 
    \mcP_k = \Node{d_k}{
      \mcN[\recvI{d_k}]
    }{
      \{ ( c_{i_k}, \mcP_{i_k} ) \}_{i_k \in \mcI_k},
      \{ \mcQ_{j_k} \}_{j_k \in \mcJ_k}
    } \\
    \mcP_k' = \Node{d_k}{
      \mcN[\return{\pairI{o}{d_k}{\Ln}}]
    }{
      \{ ( c_{i_k}, \mcP_{i_k} ) \}_{i_k \in \mcI_k},
      \{ \mcQ_{j_k} \}_{j_k \in \mcJ_k}
    }
  }
  { \Root{
      \mcM[\appI{\sendI{c_k}}{o}]
    }{
      \{ ( c_i, \mcP_i ) \}_{i \in \mcI},
      \{ \mcQ_j \}_{j \in \mcJ}
    } \\\\
    \Rrightarrow
    \Root{
      \mcM[\return{c_k}]
    }{
      \{ (c_i, \mcP_i) \}_{i \in \mcI \setminus \{ k \}} \cup \{ (c_k, \mcP_k') \},
      \{ \mcQ_j \}_{j \in \mcJ}
    }
  }

  \inferrule[Node-\underline{Send}]
  { k \in \mcI \\ 
    \mcP_k = \Node{d_k}{
      \mcN[\recvI{d_k}]
    }{
      \{ ( c_{i_k}, \mcP_{i_k} ) \}_{i_k \in \mcI_k},
      \{ \mcQ_{j_k} \}_{j_k \in \mcJ_k}
    } \\
    \mcP_k' = \Node{d_k}{
      \mcN[\return{\pairI{o}{d_k}{\Ln}}]
    }{
      \{ ( c_{i_k}, \mcP_{i_k} ) \}_{i_k \in \mcI_k},
      \{ \mcQ_{j_k} \}_{j_k \in \mcJ_k}
    }
  }
  { \Node{d}{
      \mcM[\appI{\sendI{c_k}}{o}]
    }{
      \{ ( c_i, \mcP_i ) \}_{i \in \mcI},
      \{ \mcQ_j \}_{j \in \mcJ}
    } \\\\
    \Rrightarrow
    \Node{d}{
      \mcM[\return{c_k}]
    }{
      \{ (c_i, \mcP_i) \}_{i \in \mcI \setminus \{ k \}} \cup \{ (c_k, \mcP_k') \},
      \{ \mcQ_j \}_{j \in \mcJ}
    }
  }

  \inferrule[Root-\underline{Recv}]
  { k \in \mcI \\
    \mcP_k = \Node{d_k}{
      \mcN[\appI{\sendI{d_k}}{o}]
    }{
      \{ ( c_{i_k}, \mcP_{i_k} ) \}_{i_k \in \mcI_k},
      \{ \mcQ_{j_k} \}_{j_k \in \mcJ_k}
    } \\
    \mcP_k' = \Node{d_k}{
      \mcN[\return{d_k}]
    }{
      \{ ( c_{i_k}, \mcP_{i_k} ) \}_{i_k \in \mcI_k},
      \{ \mcQ_{j_k} \}_{j_k \in \mcJ_k}
    }
  }
  {
    \Root{
      \mcM[\recvI{c_k}]
    } 
    {
      \{ ( c_i, \mcP_i ) \}_{i \in \mcI},
      \{ \mcQ_j \}_{j \in \mcJ}
    } \\\\
    \Rrightarrow
    \Root{
      \mcM[\return{\pairI{o}{c_k}{\Ln}}]
    }{
      \{ (c_i, \mcP_i) \}_{i \in \mcI \setminus \{ k \}} \cup \{ (c_k, \mcP_k') \},
      \{ \mcQ_j \}_{j \in \mcJ}
    }
  }

  \inferrule[Node-\underline{Recv}]
  { k \in \mcI \\
    \mcP_k = \Node{d_k}{
      \mcN[\appI{\sendI{d_k}}{o}]
    }{
      \{ ( c_{i_k}, \mcP_{i_k} ) \}_{i_k \in \mcI_k},
      \{ \mcQ_{j_k} \}_{j_k \in \mcJ_k}
    } \\
    \mcP_k' = \Node{d_k}{
      \mcN[\return{d_k}]
    }{
      \{ ( c_{i_k}, \mcP_{i_k} ) \}_{i_k \in \mcI_k},
      \{ \mcQ_{j_k} \}_{j_k \in \mcJ_k}
    }
  }
  {
    \Node{d}{
      \mcM[\recvI{c_k}]
    } 
    {
      \{ ( c_i, \mcP_i ) \}_{i \in \mcI},
      \{ \mcQ_j \}_{j \in \mcJ}
    } \\\\
    \Rrightarrow
    \Node{d}{
      \mcM[\return{\pairI{o}{c_k}{\Ln}}]
    }{
      \{ (c_i, \mcP_i) \}_{i \in \mcI \setminus \{ k \}} \cup \{ (c_k, \mcP_k') \},
      \{ \mcQ_j \}_{j \in \mcJ}
    }
  }

  \inferrule[Node-Forward]
  { k \in \mcI \\ 
    \mcP_k = \Node{d_k}{
      \mcN[\recvR{d_k}]
    }{
      \{ ( c_{i_k}, \mcP_{i_k} ) \}_{i_k \in \mcI_k},
      \{ \mcQ_{j_k} \}_{j_k \in \mcJ_k}
    } \\
    \mcI' = \{ i \in \mcI | c_i \in \FC{v} \} \\
    d \in \FC{v} \\
    \mcP_k' = \Node{c_k}{
      \mcM[\return{c_k}]
    }{
      \{ (c_i, \mcP_i) \}_{i \in \mcI \setminus (\{ k \} \cup \mcI')},
      \{ \mcQ_j \}_{j \in \mcJ}
    }
  }
  { \Node{d}{
      \mcM[\appR{\sendR{c_k}}{v}]
    }{
      \{ ( c_i, \mcP_i ) \}_{i \in \mcI},
      \{ \mcQ_j \}_{j \in \mcJ}
    } \\\\
    \Rrightarrow
    \Node{d}{
      \mcN[\return{\pairR{v}{d_k}{\Ln}}]
    }{
      \{ ( c_{i_k}, \mcP_{i_k} ) \}_{i_k \in \mcI_k \cup \mcI'} \cup \{ (d_k, \mcP_k') \},
      \{ \mcQ_{j_k} \}_{j_k \in \mcJ_k}
    }
  }

  \inferrule[Root-Child]
  { k \in \mcI \\
    \mcP_k \Rrightarrow \mcP_k' }
  { 
    \Root{m}{ 
      \{ ( c_i, \mcP_i ) \}_{i \in \mcI},
      \{ \mcQ_j \}_{j \in \mcJ} 
    } 
    \Rrightarrow
    \Root{m}{ 
      \{ ( c_i, \mcP_i ) \}_{i \in \mcI \setminus \{k\}} \cup \{ ( c_k, \mcP_k' ) \}, 
      \{ \mcQ_j \}_{j \in \mcJ}
    }
  }

  \inferrule[Node-Child]
  { k \in \mcI \\
    \mcP_k \Rrightarrow \mcP_k' }
  { 
    \Node{d}{m}{ 
      \{ ( c_i, \mcP_i ) \}_{i \in \mcI},
      \{ \mcQ_j \}_{j \in \mcJ} 
    } 
    \Rrightarrow
    \Node{d}{m}{ 
      \{ ( c_i, \mcP_i ) \}_{i \in \mcI \setminus \{k\}} \cup \{ ( c_k, \mcP_k' ) \}, 
      \{ \mcQ_j \}_{j \in \mcJ}
    }
  }

  \inferrule[Root-SubTree]
  { k \in \mcJ \\
    \mcQ_k \Rrightarrow \mcQ_k' }
  { 
    \Root{m}{ 
      \{ ( c_i, \mcP_i ) \}_{i \in \mcI},
      \{ \mcQ_j \}_{j \in \mcJ} 
    } 
    \Rrightarrow
    \Root{m}{ 
      \{ ( c_i, \mcP_i ) \}_{i \in \mcI},
      \{ Q_j \}_{j \in \mcJ \setminus \{k\}} \cup \{ \mcQ_k' \}
    }
  }

  \inferrule[Node-SubTree]
  { k \in \mcJ \\
    \mcQ_k \Rrightarrow \mcQ_k' }
  { 
    \Node{d}{m}{ 
      \{ ( c_i, \mcP_i ) \}_{i \in \mcI},
      \{ \mcQ_j \}_{j \in \mcJ} 
    } 
    \Rrightarrow
    \Node{d}{m}{ 
      \{ ( c_i, \mcP_i ) \}_{i \in \mcI},
      \{ Q_j \}_{j \in \mcJ \setminus \{k\}} \cup \{ \mcQ_k' \}
    }
  }
\end{mathpar}
\clearpage

\begin{mathpar}\small
  \inferrule[Root-Expr]
  { 
    m \Leadsto m'
  }
  { 
    \Root{m}{
      \{ ( c_i, \mcP_i ) \}_{i \in \mcI},
      \{ \mcQ_j \}_{j \in \mcJ}
    }
    \Rrightarrow
    \Root{m'}{
      \{ ( c_i, \mcP_i ) \}_{i \in \mcI},
      \{ \mcQ_j \}_{j \in \mcJ} 
    }
  }

  \inferrule[Node-Expr]
  { 
    m \Leadsto m'
  }
  { 
    \Node{d}{m}{
      \{ ( c_i, \mcP_i ) \}_{i \in \mcI},
      \{ \mcQ_j \}_{j \in \mcJ}
    }
    \Rrightarrow
    \Node{d}{m'}{
      \{ ( c_i, \mcP_i ) \}_{i \in \mcI},
      \{ \mcQ_j \}_{j \in \mcJ} 
    }
  }
\end{mathpar}

In the rules above, the \textsc{Root-Fork} and \textsc{Node-Fork} rules
describe how a \Fork{}-operation spawns a new child process and adds it to the
set of children processes. The new child process is represented as an internal node.
The channel connecting the child to its parent is fresh. Child processes that are
connected to channels in $\FC{m}$ (the free channels in the spawning expression)
are moved to be children of the newly spawned process.

The \textsc{Root-Wait}, \textsc{Node-Wait}, \textsc{Root-Close} and 
\textsc{Node-Close} rules describe how \Close{}/\Wait{}-operations
detach a child from its parent. The detached process is moved to the set of subtrees.

The \textsc{Root-Send} and \textsc{Node-Send} rules describe how a \Send{}-operation
sends a value to a child process. The child process must be waiting to receive a value
through a \Recv{}-operation. The sent value may contain channels that are connected to
other child processes. It is important to note that \textsc{Node-Send} only applies when
the sent value \emph{does not} contain the channel $d$ connecting to the parent, 
i.e. the side condition $d \notin \FC{v}$. When a value is sent, any child processes 
connected to channels in $\FC{v}$ are moved to be children of the receiving process.

The \textsc{Root-Recv} and \textsc{Node-Recv} rules describe how a \Recv{}-operation
receives a value from a child process. The child process must be waiting to send a value
through a \Send{}-operation. The received value may contain channels that are connected to
other child processes. When a value is received, any child processes connected to channels
in $\FC{v}$ are moved to be children of the receiving process. Note that, due to linearity,
the channel $d_k$ connecting the sending child to its parent cannot be in $\FC{v}$.
Thus, there is no side condition in \textsc{Node-Recv}. Moreover, this means that
cyclic channel dependencies cannot arise here.

The \textsc{Root-\underline{Send}}, \textsc{Node-\underline{Send}}, \textsc{Root-\underline{Recv}},
and \textsc{Node-\underline{Recv}} rules describe the sending and receiving of 
ghost messages through \SendI{} and \RecvI{} operations. Due to the fact that ghost messages
do not contain channels, there are no side conditions or changes to the spawning tree structure.

The \textsc{Node-Forward} rule describes how a \Send{}-operation can forward a
parent channel $d$ to a child process. The child process must be waiting to
receive a value. The sent value must contain the parent channel $d$, i.e. the
side condition $d \in \FC{v}$. When this happens, the child process takes over
the parent channel $d$ and tree is restructured so that the child process
becomes the new parent and the sending process (the original parent) becomes one
of its child processes. Other child processes that are connected to channels in
$\FC{v}$ are also moved to be children of the receiving process.

The \textsc{Root-Child}, \textsc{Node-Child}, \textsc{Root-SubTree}, and \textsc{Node-SubTree}
rules describe how a child process or subtree can take a reduction step.

The \textsc{Root-Expr} and \textsc{Node-Expr} rules describe how the expression
in a node can reduce.

The \textsc{Root-Unit} and \textsc{Node-Unit} rules describe how a subtree
that has finished (i.e. its expression is $\return{\ii}$ and it has no children or subtrees)
can be removed from the spawning tree.

We now state the simulation theorem between the spawning tree semantics
and the standard semantics (\Cref{appendix:process-semantics}).
With slight abuse of notation, we write $| \mcP |$ to denote just the process $P$
obtained by flattening the spawning tree $\mcP$.

\begin{lemma}[Spawning Tree Simulation]\label[lemma]{lemma:spawning-tree-simulation}
  If $\Vdash \mcP$ or $A \Vdash \mcP$, then given $\mcP \Rrightarrow \mcP'$ 
  there is $| \mcP | \Rrightarrow | \mcP' |$.
\end{lemma}
\begin{proof}
  By induction on the derivation of $\mcP \Rrightarrow \mcP'$.

\noindent
\textbf{Case} (\textsc{Root-Fork}):
  \begin{mathpar}\small
  \inferrule[Root-Fork]
  { 
    \mcI' = \{ i \in \mcI | c_i \in \FC{m} \}
  }
  { 
    \Root{ 
      \mcN[\fork{x : A}{m}]
    }{ 
      \{ ( c_i, \mcP_i ) \}_{i \in \mcI}, 
      \{ \mcQ_j \}_{j \in \mcJ}
    } \\
    \quad\Rrightarrow
    \Root{
      \mcN[\return{c}]
    }{
      \{ ( c_i, \mcP_i ) \}_{i \in \mcI \setminus \mcI'} 
      \cup
      \{ (c, 
          \Node{d}{m[d/x]}{ 
            \{ ( c_{i'}, \mcP_{i'} ) \}_{i' \in \mcI'},
            \emptyset
          }) 
      \},
      \{ \mcQ_j \}_{j \in \mcJ}
    }
  }
  \end{mathpar}
  Flattening the LHS, we have:
  \begin{align*}
    \scope{\overline{c_i d_i}}{(\proc{\mcN[\fork{x : A}{m}]} \mid \overline{P_i})} \mid \overline{Q_j}
  \end{align*}
  Repeated application of \textsc{Proc-Par} and \textsc{Proc-Scope} 
  and then \textsc{Proc-Fork} on the LHS gives us the reduced configuration
  \begin{align*}
    \scope{\overline{c_i d_i}_{(i \in \mcI)}}{(
        \scope{cd}{(
          \proc{\mcN[\return{c}]} \mid 
          \proc{m[d/x]}
        )} \mid \overline{P_i}
      )} \mid \overline{Q_j}
  \end{align*}
  By linearity of channels, we know that $\{ c_{i'}, d_{i'} \}_{(i' \in \mcI')}$
  do not appear in $\mcN[\return{c}]$. Thus, we can apply structural congruence to rearrange the
  scoping to obtain 
  \begin{align*}
    \scope{\overline{c_i d_i}_{(i \in \mcI \setminus \mcI')}}{(
      \scope{cd}{(
        \proc{\mcN[\return{c}]} \mid 
        \scope{\overline{c_{i'} d_{i'}}_{(i' \in \mcI')}}{(
          \proc{m[d/x]} \mid \overline{P_{i'}}_{(i' \in \mcI')}
        )}
      )} \mid \overline{P_i}
    )} \mid \overline{Q_j}
  \end{align*}
  which is the flattened RHS.

\noindent
\textbf{Case} (\textsc{Node-Fork}): Similar to the \textsc{Root-Fork} case.

\noindent
\textbf{Case} (\textsc{Root-Wait}):
  \begin{mathpar}\small
  \inferrule[Root-Wait]
  { k \in \mcI \\ 
    \mcP_k = \Node{d_k}{
      \mcN[\close{d_k}]
    }{ 
      \{ ( c_{i_k}, \mcP_{i_k} ) \}_{i_k \in \mcI_k},
      \{ \mcQ_{j_k} \}_{j_k \in \mcJ_k}
    } \\
    \mcQ_{j_k} = 
    \Root{
      \mcN[\return{\ii}]
    }{ 
      \{ ( c_{i_k}, \mcP_{i_k} ) \}_{i_k \in \mcI_k},
      \{ \mcQ_{j_k} \}_{j_k \in \mcJ_k}
    }
  }
  { 
    \Root{
      \mcM[\wait{c_k}]
    }{ 
      \{ ( c_i, \mcP_i ) \}_{i \in \mcI},
      \{ \mcQ_j \}_{j \in \mcJ}
    } \\\\
    \quad\Rrightarrow
    \Root{
      \mcM[\return{\ii}]
    }{ 
      \{ ( c_i, \mcP_i ) \}_{i \in \mcI \setminus \{k\}},
      \{ \mcQ_j \}_{j \in \mcJ \cup \{k\}}
    }
  }
  \end{mathpar}
  Flattening both sides, we have:
  \begin{small}
  \begin{align*}
    &\scope{\overline{c_i d_i}}{ (
        \proc{\mcM[\wait{c_k}]}
        \mid \overline{P_i}_{(i \in \mcI \setminus \{k\})}
        \mid (
          \scope{\overline{c_{i_k}d_{i_k}}}{(\proc{\mcN[\close{d_k}]} \mid \overline{P_{i_k}})}
          \mid \overline{Q_{j_k}}
        )
      )
    } \mid \overline{Q_j}
    \\
    &\scope{\overline{c_i d_i}_{(i \in \mcI \setminus \{k\})}}{(
      \proc{\mcM[\return{\ii}]} 
      \mid \overline{P_i}_{(i \in \mcI \setminus \{k\})}
    )} 
    \mid \overline{Q_j} 
    \mid (
      \scope{\overline{c_{i_k} d_{i_k}}}{(
        \proc{\mcN[\return{\ii}]} \mid \overline{P_{i_k}}
      )} \mid \overline{Q_{j_k}}
    )
  \end{align*}
  \end{small}
  Apply \textsc{Proc-Congr} to the LHS to rearrange the processes,
  then apply \textsc{Proc-Scope} and \textsc{Proc-Par} repeatedly
  to isolate the sub-configuration 
  $\scope{c_k d_k}{(\proc{\mcM[\wait{c_k}]} \mid \proc{\mcN[\close{d_k}]})}$.
  Finally, apply \textsc{Proc-Wait} to this sub-configuration
  to obtain the reduced configuration, which is structurally congruent
  to the RHS.

\noindent
\textbf{Case} (\textsc{Node-Wait}): Similar to the \textsc{Root-Wait} case.

\noindent
\textbf{Case} (\textsc{Root-Close}): Similar to the \textsc{Root-Wait} case.

\noindent
\textbf{Case} (\textsc{Node-Close}): Similar to the \textsc{Root-Wait} case.

\noindent
\textbf{Case} (\textsc{Root-Send}):
  \begin{mathpar}\small
  \inferrule[Root-Send]
  { k \in \mcI \\ 
    \mcP_k = \Node{d_k}{
      \mcN[\recvR{d_k}]
    }{
      \{ ( c_{i_k}, \mcP_{i_k} ) \}_{i_k \in \mcI_k},
      \{ \mcQ_{j_k} \}_{j_k \in \mcJ_k}
    } \\
    \mcI' = \{ i \in \mcI | c_i \in \FC{v} \} \\
    \mcP_k' = \Node{d_k}{
      \mcN[\return{\pairR{v}{d_k}{\Ln}}]
    }{
      \{ ( c_{i_k}, \mcP_{i_k} ) \}_{i_k \in \mcI_k \cup \mcI'},
      \{ \mcQ_{j_k} \}_{j_k \in \mcJ_k}
    }
  }
  { \Root{
      \mcM[\appR{\sendR{c_k}}{v}]
    }{
      \{ ( c_i, \mcP_i ) \}_{i \in \mcI},
      \{ \mcQ_j \}_{j \in \mcJ}
    } \\\\
    \Rrightarrow
    \Root{
      \mcM[\return{c_k}]
    }{
      \{ (c_i, \mcP_i) \}_{i \in \mcI \setminus (\{ k \} \cup \mcI')} \cup \{ (c_k, \mcP_k') \},
      \{ \mcQ_j \}_{j \in \mcJ}
    }
  }
  \end{mathpar}
  Flattening both sides, we have:
  \begin{small}
  \begin{align*}
    &\scope{\overline{c_i d_i}}{ (
        \proc{\mcM[\appR{\sendR{c_k}}{v}]}
        \mid \overline{P_i}_{(i \in \mcI \setminus \{k\})}
        \mid (
          \scope{\overline{c_{i_k}d_{i_k}}}{(\proc{\mcN[\recvR{d_k}]} \mid \overline{P_{i_k}})}
          \mid \overline{Q_{j_k}}
        )
      )
    } \mid \overline{Q_j}
    \\
    &\nu{\overline{c_i d_i}_{(i \in (\mcI \setminus (\{k\} \cup \mcI')) \cup \{ k \})}}.(
      \proc{\mcM[\return{c_k}]}
      \mid \overline{P_i}_{(i \in \mcI \setminus (\{k\} \cup \mcI'))} \\
      &\qquad\mid
      \scope{\overline{c_{i_k} d_{i_k}}}{
      \scope{\overline{c_i d_i}_{(i \in \mcI')}}{(
        \proc{\mcN[\return{\pairR{v}{d_k}{\Ln}}]} \mid \overline{P_{i_k}} \mid \overline{P_i}_{(i \in \mcI')}
      )}
      } \mid \overline{Q_{j_k}}
    )
    \mid \overline{Q_j}
  \end{align*}
  \end{small}
  Apply \textsc{Proc-Congr} to the LHS to rearrange the processes,
  then apply \textsc{Proc-Scope} and \textsc{Proc-Par} repeatedly
  to isolate the sub-configuration 
  $\scope{c_k d_k}{(\proc{\mcM[\appR{\sendR{c_k}}{v}]} \mid \proc{\mcN[\recvR{d_k}]})}$.
  Now, apply \textsc{Proc-Send} to this sub-configuration to obtain the reduced configuration
  $$
    \scope{c_k d_k}{(
      \proc{\mcM[\return{c_k}]} 
      \mid \proc{\mcN[\return{\pairR{v}{d_k}{\Ln}}]} 
    )}
  $$
  Note that, by linearity, the channels $\{ c_i, d_i \}_{(i \in \mcI')}$ do not occur in 
  $\mcM[\return{c_k}]$.
  Thus, structural congruence can be applied to move the scope of these channels to
  $$\scope{\overline{c_i d_i}_{(i \in \mcI')}}{
    (\proc{\mcN[\return{\pairR{v}{d_k}{\Ln}}]} \mid \overline{P_i}_{(i \in \mcI')})}$$
  which gives us the desired result.

\noindent
\textbf{Case} (
  \textsc{Node-Send}, 
  \textsc{Root-Recv}
  \textsc{Node-Recv}
): Similar to the \textsc{Root-Send} case.

\noindent
\textbf{Case} (
  \textsc{Root-\underline{Send}}
  \textsc{Node-\underline{Send}}
  \textsc{Root-\underline{Recv}}
  \textsc{Node-\underline{Recv}}
): Similar to the \textsc{Root-Send} case.
The only difference is that scope restriction does not need to be applied
to move any channels since the ghost message $o$ does not contain channels.

\noindent
\textbf{Case} (\textsc{Node-Forward}):
  \begin{mathpar}\small
  \inferrule[Node-Forward]
  { k \in \mcI \\ 
    \mcP_k = \Node{d_k}{
      \mcN[\recvR{d_k}]
    }{
      \{ ( c_{i_k}, \mcP_{i_k} ) \}_{i_k \in \mcI_k},
      \{ \mcQ_{j_k} \}_{j_k \in \mcJ_k}
    } \\
    \mcI' = \{ i \in \mcI | c_i \in \FC{v} \} \\
    d \in \FC{v} \\
    \mcP_k' = \Node{c_k}{
      \mcM[\return{c_k}]
    }{
      \{ (c_i, \mcP_i) \}_{i \in \mcI \setminus (\{ k \} \cup \mcI')},
      \{ \mcQ_j \}_{j \in \mcJ}
    }
  }
  { \Node{d}{
      \mcM[\appR{\sendR{c_k}}{v}]
    }{
      \{ ( c_i, \mcP_i ) \}_{i \in \mcI},
      \{ \mcQ_j \}_{j \in \mcJ}
    } \\\\
    \Rrightarrow
    \Node{d}{
      \mcN[\return{\pairR{v}{d_k}{\Ln}}]
    }{
      \{ ( c_{i_k}, \mcP_{i_k} ) \}_{i_k \in \mcI_k \cup \mcI'} \cup \{ (d_k, \mcP_k') \},
      \{ \mcQ_{j_k} \}_{j_k \in \mcJ_k}
    }
  }
  \end{mathpar}
  Flattening both sides, we have:
  \begin{small}
  \begin{align*}
    &\scope{\overline{c_i d_i}}{ (
        \proc{\mcM[\appR{\sendR{c_k}}{v}]} 
        \mid \overline{P_i}_{(i \in \mcI \setminus \{k\})}
        \mid (
          \scope{\overline{c_{i_k}d_{i_k}}}{(\proc{\mcN[\recvR{d_k}]} \mid \overline{P_{i_k}})}
          \mid \overline{Q_{j_k}}
        )
      )
    } \mid \overline{Q_j}
    \\
    &\nu{\overline{c_{i_k} d_{i_k}}_{(i_k \in \mcI_k \cup \mcI' \cup \{ k \})}}.(
      \proc{\mcN[\return{\pairR{v}{d_k}{\Ln}}]} 
      \mid \overline{P_{i_k}}_{(i_k \in \mcI_k \cup \mcI')} \\
      &\qquad\mid
      \scope{\overline{c_i d_i}_{(i \in \mcI \setminus (\{k\} \cup \mcI'))}}{(
        \proc{\mcM[\return{c_k}]} 
        \mid \overline{P_i}_{(i \in \mcI \setminus (\{k\} \cup \mcI'))}
      )}
      \mid \overline{Q_j}
    ) \mid \overline{Q_{j_k}}
  \end{align*}
  \end{small}
  Apply \textsc{Proc-Congr} to the LHS to rearrange the processes,
  then apply \textsc{Proc-Scope} and \textsc{Proc-Par} repeatedly
  to isolate the sub-configuration 
  $\scope{c_k d_k}{(\proc{\mcM[\appR{\sendR{c_k}}{v}]} \mid \proc{\mcN[\recvR{d_k}]})}$.
  Now, apply \textsc{Proc-Send} to this sub-configuration to obtain the reduced configuration
  $$
    \scope{c_k d_k}{(
      \proc{\mcM[\return{c_k}]} 
      \mid \proc{\mcN[\return{\pairR{v}{d_k}{\Ln}}]}
    )}
  $$
  By symmetry of structural congruence, we have
  $$
    \scope{c_k d_k}{(
      \proc{\mcN[\return{\pairR{v}{d_k}{\Ln}}]}
      \mid \proc{\mcM[\return{c_k}]} 
    )}
  $$
  Since $\{ c_i, d_i \}_{(i \in \mcI \setminus (\{k\} \cup \mcI'))}$
  do not occur in $\mcN[\return{\pairR{v}{d_k}{\Ln}}]$, we
  can apply structural congruence to move the scope of these channels to
  $$\scope{\overline{c_i d_i}_{(i \in \mcI \setminus (\{k\} \cup \mcI'))}}{
    (\proc{\mcM[\return{c_k}]} \mid \overline{P_i}_{(i \in \mcI \setminus (\{k\} \cup \mcI'))})}$$
  which gives us the desired result.

\noindent
\textbf{Case} (
  \textsc{Root-Child},
  \textsc{Node-Child},
  \textsc{Root-SubTree},
  \textsc{Node-SubTree}
): By the induction hypothesis, we have $| \mcP_k | \Rrightarrow | \mcP_k' |$ or
  $| \mcQ_k | \Rrightarrow | \mcQ_k' |$.
  Repeated application of \textsc{Proc-Par} and \textsc{Proc-Scope} gives us the desired result.

\noindent
\textbf{Case} (\textsc{Root-Expr}, \textsc{Node-Expr}):
  By the assumption, we have $m \Leadsto m'$.
  Repeated application of \textsc{Proc-Par}, \textsc{Proc-Scope} 
  and then \textsc{Proc-Expr} gives us the desired result.
\end{proof}

In order to show that spawning trees are an adequate characterization of reachability,
we prove the following fidelity theorem. This theorem states that if a spawning tree is well-typed
and it takes a reduction step, then the resulting spawning tree is also well-typed.
Thus, starting from a well-typed singleton $\Root{m}{\emptyset, \emptyset}$,
the spawning trees reachable from it are all well-typed.

\begin{lemma}[Spawning Tree Fidelity]\label[lemma]{lemma:spawning-tree-fidelity}
  If $\Vdash \mcP$ or $\ch{\kappa}{A} \Vdash \mcP$, then given $\mcP \Rrightarrow \mcQ$ there
  is $\Vdash \mcQ$ or $\ch{\kappa}{A} \Vdash \mcQ$ respectively.
\end{lemma}
\begin{proof}
  By induction on the derivation of $\mcP \Rrightarrow \mcQ$ and by case analysis
  on the derivation of the typing judgment $\Vdash \mcP$ or $\ch{\kappa}{A} \Vdash \mcP$.

\noindent
\textbf{Case} (\textsc{Root-Fork}):
  \begin{mathpar}\small
  \inferrule[Root-Fork]
  { 
    \mcI' = \{ i \in \mcI | c_i \in \FC{m} \}
  }
  { 
    \Root{ 
      \mcN[\fork{x : A}{m}]
    }{ 
      \{ ( c_i, \mcP_i ) \}_{i \in \mcI}, 
      \{ \mcQ_j \}_{j \in \mcJ}
    } \\
    \quad\Rrightarrow
    \Root{
      \mcN[\return{c}]   
    }{
      \{ ( c_i, \mcP_i ) \}_{i \in \mcI \setminus \mcI'} 
      \cup
      \{ (c, 
          \Node{d}{m[d/x]}{ 
            \{ ( c_{i'}, \mcP_{i'} ) \}_{i' \in \mcI'},
            \emptyset
          }) 
      \},
      \{ \mcQ_j \}_{j \in \mcJ}
    }
  }
  \end{mathpar}
  By inversion on the typing judgment 
  $$
    \Vdash \Root{ 
      \mcN[\fork{x : A}{m}]
    }{ 
      \{ ( c_i, \mcP_i ) \}_{i \in \mcI}, 
      \{ \mcQ_j \}_{j \in \mcJ}
    }
  $$
  we have
  \begin{align*}
    & \overline{c_i \tL \ch{\kappa_i}{A_i}} ; \epsilon ; \epsilon \vdash
      \mcN[\fork{x : A}{m}] : 
      \CM{\unit} \\
    & \forall i \in \mcI, \neg{\ch{\kappa_i}{A_i}} \Vdash \mcP_i \\
    & \forall j \in \mcJ, \Vdash \mcQ_j
  \end{align*}
  Similarly to the reasoning in \Cref{theorem:session-fidelity} for the \textsc{Proc-Fork} case, 
  by inversion on the term typing judgment we know that $A = \CH{A'}$ and the following hold:
  \begin{align*}
    &\overline{c_i \tL \ch{\kappa_i}{A_i}}_{(i \in \mcI')}, d \tL \HC{A'} ; 
      \epsilon ; \epsilon \vdash m[d/x] : \CM{\unit}  \\
    &\overline{c_i \tL \ch{\kappa_i}{A_i}}_{(i \in \mcI \setminus \mcI')} ; c \tL \CH{A'} ; 
      \epsilon ; \epsilon \vdash \mcN[\return{c}] : \CM{\unit}
  \end{align*}

  \noindent
  We can partition $\{ ( c_i, \mcP_i ) \}_{i \in \mcI}$ into
  $\{ ( c_i, \mcP_i ) \}_{i \in \mcI \setminus \mcI'}$ and
  $\{ ( c_{i'}, \mcP_{i'} ) \}_{i' \in \mcI'}$.

  \noindent
  We now have well-typed node 
  $$
    \ch{-}{A'} \Vdash \Node{d}{m[d/x]}{ 
      \{ ( c_{i'}, \mcP_{i'} ) \}_{i' \in \mcI'},
      \emptyset
    }
  $$
  which in turn gives us the well-typed root
  $$
    \Vdash \Root{
      \mcN[\return{c}]   
    }{
      \{ ( c_i, \mcP_i ) \}_{i \in \mcI \setminus \mcI'} 
      \cup
      \{ (c, 
          \Node{d}{m[d/x]}{ 
            \{ ( c_{i'}, \mcP_{i'} ) \}_{i' \in \mcI'},
            \emptyset
          }) 
      \},
      \{ \mcQ_j \}_{j \in \mcJ}
    }
  $$
  and concludes this case.

\noindent
\textbf{Case} (\textsc{Node-Fork}): Similar to the \textsc{Root-Fork} case.

\noindent
\textbf{Case} (\textsc{Root-Wait}):
  \begin{mathpar}\small
  \inferrule[Root-Wait]
  { k \in \mcI \\ 
    \mcP_k = \Node{d_k}{
      \mcN[\close{d_k}]
    }{ 
      \{ ( c_{i_k}, \mcP_{i_k} ) \}_{i_k \in \mcI_k},
      \{ \mcQ_{j_k} \}_{j_k \in \mcJ_k}
    } \\
    \mcQ_{k} = 
    \Root{
      \mcN[\return{\ii}]
    }{ 
      \{ ( c_{i_k}, \mcP_{i_k} ) \}_{i_k \in \mcI_k},
      \{ \mcQ_{j_k} \}_{j_k \in \mcJ_k}
    }
  }
  { 
    \Root{
      \mcM[\wait{c_k}]
    }{ 
      \{ ( c_i, \mcP_i ) \}_{i \in \mcI},
      \{ \mcQ_j \}_{j \in \mcJ}
    } \\\\
    \quad\Rrightarrow
    \Root{
      \mcM[\return{\ii}]
    }{ 
      \{ ( c_i, \mcP_i ) \}_{i \in \mcI \setminus \{k\}},
      \{ \mcQ_j \}_{j \in \mcJ \cup \{k\}}
    }
  }
  \end{mathpar}
  By inversion on the typing judgment 
  $$
    \Vdash \Root{
      \mcM[\wait{c_k}]
    }{ 
      \{ ( c_i, \mcP_i ) \}_{i \in \mcI},
      \{ \mcQ_j \}_{j \in \mcJ}
    }
  $$
  we have
  \begin{align*}
    & \overline{c_i \tL \ch{\kappa_i}{A_i}} ; \epsilon ; \epsilon \vdash
      \mcM[\wait{c_k}] : 
      \CM{\unit} \\
    & \forall i \in \mcI, \neg{\ch{\kappa_i}{A_i}} \Vdash \mcP_i \\
    & \forall j \in \mcJ, \Vdash \mcQ_j
  \end{align*}

  \noindent
  Since $k \in \mcI$, we know 
  $$
    \neg{\ch{\kappa_k}{A_k}} \Vdash 
    \Node{d_k}{
      \mcN[\close{d_k}]
    }{ 
      \{ ( c_{i_k}, \mcP_{i_k} ) \}_{i_k \in \mcI_k},
      \{ \mcQ_{j_k} \}_{j_k \in \mcJ_k}
    }
  $$
  and by inversion on this typing judgment, we have
  \begin{align*}
    & \overline{c_{i_k} \tL \ch{\kappa_{i_k}}{A_{i_k}}}, d_k \tL \neg{\ch{\kappa_k}{A_k}} ; 
      \epsilon ; \epsilon \vdash
      \mcN[\close{d_k}] : 
      \CM{\unit} \\
    & \forall i_k \in \mcI_k, \neg{\ch{\kappa_{i_k}}{A_{i_k}}} \Vdash \mcP_{i_k} \\
    & \forall j_k \in \mcJ_k, \Vdash \mcQ_{j_k}
  \end{align*}
  
  \noindent
  Similarly to the reasoning in \Cref{theorem:session-fidelity} for the \textsc{Proc-End} case,
  we know $A_k \simeq \End$ and
  \begin{align*}
    & \overline{c_i \tL \ch{\kappa_i}{A_i}}_{(i \in \mcI \setminus \{k\})} ; 
      \epsilon ; \epsilon \vdash \mcM[\return{\ii}] : \CM{\unit} \\
    & \overline{c_{i_k} \tL \ch{\kappa_{i_k}}{A_{i_k}}} ; 
      \epsilon ; \epsilon \vdash \mcN[\return{\ii}] : \CM{\unit} 
  \end{align*}
  This gives us the well-typed root
  $$
    \Vdash \Root{
      \mcN[\return{\ii}]
    }{ 
      \{ ( c_{i_k}, \mcP_{i_k} ) \}_{i_k \in \mcI_k},
      \{ \mcQ_{j_k} \}_{j_k \in \mcJ_k}
    }
  $$
  which in turn gives us the well-typed root
  $$
    \Vdash \Root{
      \mcM[\return{\ii}]
    }{ 
      \{ ( c_i, \mcP_i ) \}_{i \in \mcI \setminus \{k\}},
      \{ \mcQ_j \}_{j \in \mcJ \cup \{k\}}
    }
  $$
  and concludes this case.

\noindent
\textbf{Case} (
  \textsc{Node-Wait}, 
  \textsc{Root-Close}, 
  \textsc{Node-Close}
): Similar to the \textsc{Root-Wait} case.

\noindent
\textbf{Case} (\textsc{Root-Send}):
  \begin{mathpar}\small
  \inferrule[Root-Send]
  { k \in \mcI \\ 
    \mcP_k = \Node{d_k}{
      \mcN[\recvR{d_k}]
    }{
      \{ ( c_{i_k}, \mcP_{i_k} ) \}_{i_k \in \mcI_k},
      \{ \mcQ_{j_k} \}_{j_k \in \mcJ_k}
    } \\
    \mcI' = \{ i \in \mcI | c_i \in \FC{v} \} \\
    \mcP_k' = \Node{d_k}{
      \mcN[\return{\pairR{v}{d_k}{\Ln}}]
    }{
      \{ ( c_{i_k}, \mcP_{i_k} ) \}_{i_k \in \mcI_k \cup \mcI'},
      \{ \mcQ_{j_k} \}_{j_k \in \mcJ_k}
    }
  }
  { \Root{
      \mcM[\appR{\sendR{c_k}}{v}]
    }{
      \{ ( c_i, \mcP_i ) \}_{i \in \mcI},
      \{ \mcQ_j \}_{j \in \mcJ}
    } \\\\
    \Rrightarrow
    \Root{
      \mcM[\return{c_k}]
    }{
      \{ (c_i, \mcP_i) \}_{i \in \mcI \setminus (\{ k \} \cup \mcI')} \cup \{ (c_k, \mcP_k') \},
      \{ \mcQ_j \}_{j \in \mcJ}
    }
  }
  \end{mathpar}
  By inversion on the typing judgment
  $$
    \Vdash \Root{
      \mcM[\appR{\sendR{c_k}}{v}]
    }{
      \{ ( c_i, \mcP_i ) \}_{i \in \mcI},
      \{ \mcQ_j \}_{j \in \mcJ}
    }
  $$
  we have
  \begin{align*}
    & \overline{c_i \tL \ch{\kappa_i}{A_i}} ; \epsilon ; \epsilon \vdash
      \mcM[\appR{\sendR{c_k}}{v}] : 
      \CM{\unit} \\
    & \forall i \in \mcI, \neg{\ch{\kappa_i}{A_i}} \Vdash \mcP_i \\
    & \forall j \in \mcJ, \Vdash \mcQ_j
  \end{align*}

  \noindent
  From $k \in \mcI$, we know
  \begin{align*}
    & \neg{\ch{\kappa_k}{A_k}} \Vdash 
      \Node{d_k}{
        \mcN[\recvR{d_k}]
      }{
        \{ ( c_{i_k}, \mcP_{i_k} ) \}_{i_k \in \mcI_k},
        \{ \mcQ_{j_k} \}_{j_k \in \mcJ_k}
      }
  \end{align*}
  and by inversion on this typing judgment, we have
  \begin{align*}
    & \overline{c_{i_k} \tL \ch{\kappa_{i_k}}{A_{i_k}}}, d_k \tL \neg{\ch{\kappa_k}{A_k}} ; 
      \epsilon ; \epsilon \vdash
      \mcN[\recvR{d_k}] : 
      \CM{\unit} \\
    & \forall i_k \in \mcI_k, \neg{\ch{\kappa_{i_k}}{A_{i_k}}} \Vdash \mcP_{i_k} \\
    & \forall j_k \in \mcJ_k, \Vdash \mcQ_{j_k}
  \end{align*}

  \noindent
  Similarly to the reasoning in \Cref{theorem:session-fidelity} for the 
  \textsc{Proc-Comm} case, there exists $B''$ such that
  \begin{small}
  \begin{align*}
    &\overline{c_i \tL \ch{\kappa_i}{A_i}}_{(i \in \mcI \setminus (\{k\} \cup \mcI'))}, 
     c_k \tL \ch{\kappa_k}{B''[v/x]} ; \epsilon ; \epsilon \vdash 
     \mcM[\return{c_k}] : \CM{\unit}
    \\
    &\overline{c_{i_k} \tL \ch{\kappa_{i_k}}{A_{i_k}}},
     \overline{c_i \tL \ch{\kappa_i}{A_i}}_{(i \in \mcI')},
     d_k \tL \neg{\ch{\kappa_k}{B''[v/x]}} ; \epsilon ; \epsilon \vdash
     \mcN[\return{\pairR{v}{d_k}{\Ln}}] : \CM{\unit}
  \end{align*}
  \end{small}

  \noindent
  We can partition $\{ ( c_i, \mcP_i ) \}_{i \in \mcI \setminus \{k\}}$ into
  $\{ ( c_i, \mcP_i ) \}_{i \in \mcI \setminus (\{ k \} \cup \mcI')}$ and
  $\{ ( c_{i}, \mcP_{i} ) \}_{i \in \mcI'}$.

  \noindent
  This gives us the well-typed node
  \begin{align*}
    \neg{\ch{\kappa_k}{B''[v/x]}} \Vdash
    \Node{d_k}{
      \mcN[\return{\pairR{v}{d_k}{\Ln}}]
    }{
      \{ ( c_{i_k}, \mcP_{i_k} ) \}_{i_k \in \mcI_k \cup \mcI'},
      \{ \mcQ_{j_k} \}_{j_k \in \mcJ_k}
    }
  \end{align*}
  which in turn gives us the well-typed root
  \begin{align*}
    &\Vdash \Root{
      \mcM[\return{c_k}]
    }{
      \{ (c_i, \mcP_i) \}_{i \in \mcI \setminus (\{ k \} \cup \mcI')} \cup \{ (c_k, \mcP_k') \},
      \{ \mcQ_j \}_{j \in \mcJ}
    }
  \end{align*}
  and concludes this case.

\noindent
\textbf{Case} (
  \textsc{Node-Send}, 
  \textsc{Root-Recv}, 
  \textsc{Node-Recv}
): Similar to the \textsc{Root-Send} case.

\noindent
\textbf{Case} (\textsc{Root-\underline{Send}}):
  \begin{mathpar}\small
  \inferrule[Root-\underline{Send}]
  { k \in \mcI \\ 
    \mcP_k = \Node{d_k}{
      \mcN[\recvI{d_k}]
    }{
      \{ ( c_{i_k}, \mcP_{i_k} ) \}_{i_k \in \mcI_k},
      \{ \mcQ_{j_k} \}_{j_k \in \mcJ_k}
    } \\
    \mcP_k' = \Node{d_k}{
      \mcN[\return{\pairI{o}{d_k}{\Ln}}]
    }{
      \{ ( c_{i_k}, \mcP_{i_k} ) \}_{i_k \in \mcI_k},
      \{ \mcQ_{j_k} \}_{j_k \in \mcJ_k}
    }
  }
  { \Root{
      \mcM[\appI{\sendI{c_k}}{o}]
    }{
      \{ ( c_i, \mcP_i ) \}_{i \in \mcI},
      \{ \mcQ_j \}_{j \in \mcJ}
    } \\\\
    \Rrightarrow
    \Root{
      \mcM[\return{c_k}]
    }{
      \{ (c_i, \mcP_i) \}_{i \in \mcI \setminus \{ k \}} \cup \{ (c_k, \mcP_k') \},
      \{ \mcQ_j \}_{j \in \mcJ}
    }
  }
  \end{mathpar}
  By inversion on the typing judgment
  \begin{align*}
    &\Vdash \Root{
      \mcM[\appI{\sendI{c_k}}{o}]
    }{
      \{ ( c_i, \mcP_i ) \}_{i \in \mcI},
      \{ \mcQ_j \}_{j \in \mcJ}
    }
  \end{align*}
  we have
  \begin{align*}
    & \overline{c_i \tL \ch{\kappa_i}{A_i}} ; \epsilon ; \epsilon \vdash
      \mcM[\appI{\sendI{c_k}}{o}] : 
      \CM{\unit} \\
    & \forall i \in \mcI, \neg{\ch{\kappa_i}{A_i}} \Vdash \mcP_i \\
    & \forall j \in \mcJ, \Vdash \mcQ_j
  \end{align*}

  \noindent
  From $k \in \mcI$, we know
  \begin{align*}
    & \neg{\ch{\kappa_k}{A_k}} \Vdash 
      \Node{d_k}{
        \mcN[\recvI{d_k}]
      }{
        \{ ( c_{i_k}, \mcP_{i_k} ) \}_{i_k \in \mcI_k},
        \{ \mcQ_{j_k} \}_{j_k \in \mcJ_k}
      }
  \end{align*}
  and by inversion on this typing judgment, we have
  \begin{align*}
    & \overline{c_{i_k} \tL \ch{\kappa_{i_k}}{A_{i_k}}}, d_k \tL \neg{\ch{\kappa_k}{A_k}} ; 
      \epsilon ; \epsilon \vdash
      \mcN[\recvI{d_k}] : 
      \CM{\unit} \\
    & \forall i_k \in \mcI_k, \neg{\ch{\kappa_{i_k}}{A_{i_k}}} \Vdash \mcP_{i_k} \\
    & \forall j_k \in \mcJ_k, \Vdash \mcQ_{j_k}
  \end{align*}
  Similarly to the reasoning in \Cref{theorem:session-fidelity} for the 
  \textsc{Proc-Comm} case, there exists $B''$ such that
  \begin{align*}
    &\overline{c_i \tL \ch{\kappa_i}{A_i}}_{(i \in \mcI \setminus \{k\})}, 
     c_k \tL \ch{\kappa_k}{B''[o/x]} ; \epsilon ; \epsilon \vdash 
     \mcM[\return{c_k}] : \CM{\unit}
    \\
    &\overline{c_{i_k} \tL \ch{\kappa_{i_k}}{A_{i_k}}}, 
     d_k \tL \neg{\ch{\kappa_k}{B''[o/x]}} ; \epsilon ; \epsilon \vdash
     \mcN[\return{\pairI{o}{d_k}{\Ln}}] : 
     \CM{\unit}
  \end{align*}

  \noindent
  This gives us the well-typed node
  \begin{align*}
    \neg{\ch{\kappa_k}{B''[o/x]}} \Vdash
    \Node{d_k}{
      \mcN[\return{\pairI{o}{d_k}{\Ln}}]
    }{
      \{ ( c_{i_k}, \mcP_{i_k} ) \}_{i_k \in \mcI_k},
      \{ \mcQ_{j_k} \}_{j_k \in \mcJ_k}
    }
  \end{align*}
  which in turn gives us the well-typed root
  \begin{align*}
    \Vdash \Root{
      \mcM[\return{c_k}]
    }{
      \{ (c_i, \mcP_i) \}_{i \in \mcI \setminus \{ k \}} \cup \{ (c_k, \mcP_k') \},
      \{ \mcQ_j \}_{j \in \mcJ}
    }
  \end{align*}
  and concludes this case.

\noindent
\textbf{Case} (
  \textsc{Node-\underline{Send}}, 
  \textsc{Root-\underline{Recv}}, 
  \textsc{Node-\underline{Recv}}
): Similar to the \textsc{Root-\underline{Send}} case.

\noindent
\textbf{Case} (\textsc{Node-Forward}):
  \begin{mathpar}\small
  \inferrule[Node-Forward]
  { k \in \mcI \\ 
    \mcP_k = \Node{d_k}{
      \mcN[\recvR{d_k}]
    }{
      \{ ( c_{i_k}, \mcP_{i_k} ) \}_{i_k \in \mcI_k},
      \{ \mcQ_{j_k} \}_{j_k \in \mcJ_k}
    } \\
    \mcI' = \{ i \in \mcI | c_i \in \FC{v} \} \\
    d \in \FC{v} \\
    \mcP_k' = \Node{c_k}{
      \mcM[\return{c_k}]
    }{
      \{ (c_i, \mcP_i) \}_{i \in \mcI \setminus (\{ k \} \cup \mcI')},
      \{ \mcQ_j \}_{j \in \mcJ}
    }
  }
  { \Node{d}{
      \mcM[\appR{\sendR{c_k}}{v}]
    }{
      \{ ( c_i, \mcP_i ) \}_{i \in \mcI},
      \{ \mcQ_j \}_{j \in \mcJ}
    } \\\\
    \Rrightarrow
    \Node{d}{
      \mcN[\return{\pairR{v}{d_k}{\Ln}}]
    }{
      \{ ( c_{i_k}, \mcP_{i_k} ) \}_{i_k \in \mcI_k \cup \mcI'} \cup \{ (d_k, \mcP_k') \},
      \{ \mcQ_{j_k} \}_{j_k \in \mcJ_k}
    }
  }
  \end{mathpar}
  By inversion on the typing judgment
  \begin{align*}
    &\ch{\kappa}{A} \Vdash \Node{d}{
      \mcM[\appR{\sendR{c_k}}{v}]
    }{
      \{ ( c_i, \mcP_i ) \}_{i \in \mcI},
      \{ \mcQ_j \}_{j \in \mcJ}
    }
  \end{align*}
  we have
  \begin{align*}
    & \overline{c_i \tL \ch{\kappa_i}{A_i}}, d \tL \ch{\kappa}{A} ; \epsilon ; \epsilon \vdash
      \mcM[\appR{\sendR{c_k}}{v}] : 
      \CM{\unit} \\
    & \forall i \in \mcI, \neg{\ch{\kappa_i}{A_i}} \Vdash \mcP_i \\
    & \forall j \in \mcJ, \Vdash \mcQ_j
  \end{align*}

  \noindent
  From $k \in \mcI$, we know
  \begin{align*}
    & \neg{\ch{\kappa_k}{A_k}} \Vdash 
      \Node{d_k}{
        \mcN[\recvR{d_k}]
      }{
        \{ ( c_{i_k}, \mcP_{i_k} ) \}_{i_k \in \mcI_k},
        \{ \mcQ_{j_k} \}_{j_k \in \mcJ_k}
      }
  \end{align*}
  and by inversion on this typing judgment, we have
  \begin{align*}
    & \overline{c_{i_k} \tL \ch{\kappa_{i_k}}{A_{i_k}}}, d_k \tL \neg{\ch{\kappa_k}{A_k}} ; 
      \epsilon ; \epsilon \vdash
      \mcN[\recvR{d_k}] : 
      \CM{\unit} \\
    & \forall i_k \in \mcI_k, \neg{\ch{\kappa_{i_k}}{A_{i_k}}} \Vdash \mcP_{i_k} \\
    & \forall j_k \in \mcJ_k, \Vdash \mcQ_{j_k}
  \end{align*}
  Similarly to the reasoning in \Cref{theorem:session-fidelity} for the 
  \textsc{Proc-Comm} case, there exists $B''$ such that
  \begin{small}
  \begin{align*}
    &\overline{c_i \tL \ch{\kappa_i}{A_i}}_{(i \in \mcI \setminus (\{k\} \cup \mcI'))}, 
     c_k \tL \ch{\kappa_k}{B''[v/x]} ; \epsilon ; \epsilon \vdash 
     \mcM[\return{c_k}] : \CM{\unit}
    \\
    &\overline{c_{i_k} \tL \ch{\kappa_{i_k}}{A_{i_k}}},
     \overline{c_i \tL \ch{\kappa_i}{A_i}}_{(i \in \mcI')},
     d_k \tL \neg{\ch{\kappa_k}{B''[v/x]}}, d \tL \ch{\kappa}{A} ; \epsilon ; \epsilon \vdash
     \mcN[\return{\pairR{v}{d_k}{\Ln}}] : \CM{\unit}
  \end{align*}
  \end{small}

  \noindent
  We can partition $\{ ( c_i, \mcP_i ) \}_{i \in \mcI \setminus \{k\}}$ into
  $\{ ( c_i, \mcP_i ) \}_{i \in \mcI \setminus (\{ k \} \cup \mcI')}$ and
  $\{ ( c_{i}, \mcP_{i} ) \}_{i \in \mcI'}$.

  \noindent
  This gives us the well-typed node
  \begin{align*}
    \ch{\kappa_k}{B''[v/x]} \Vdash
    \Node{c_k}{
      \mcM[\return{c_k}]
    }{
      \{ (c_i, \mcP_i) \}_{i \in \mcI \setminus (\{ k \} \cup \mcI')},
      \{ \mcQ_j \}_{j \in \mcJ}
    }
  \end{align*}
  which in turn gives us the well-typed node
  \begin{align*}
    &\ch{\kappa}{A} \Vdash \Node{d}{
      \mcN[\return{\pairR{v}{d_k}{\Ln}}]
    }{
      \{ ( c_{i_k}, \mcP_{i_k} ) \}_{i_k \in \mcI_k \cup \mcI'} \cup \{ (d_k, \mcP_k') \},
      \{ \mcQ_{j_k} \}_{j_k \in \mcJ_k}
    }
  \end{align*}
  and concludes this case.

\noindent
\textbf{Case} (\textsc{Root-Child}):
  \begin{mathpar}\small
  \inferrule[Root-Child]
  { k \in \mcI \\
    \mcP_k \Rrightarrow \mcP_k' }
  { 
    \Root{m}{ 
      \{ ( c_i, \mcP_i ) \}_{i \in \mcI},
      \{ \mcQ_j \}_{j \in \mcJ} 
    } 
    \Rrightarrow
    \Root{m}{ 
      \{ ( c_i, \mcP_i ) \}_{i \in \mcI \setminus \{k\}} \cup \{ ( c_k, \mcP_k' ) \}, 
      \{ \mcQ_j \}_{j \in \mcJ}
    }
  }
  \end{mathpar}
  By inversion on the typing judgment
  \begin{align*}
    \Vdash \Root{m}{ 
      \{ ( c_i, \mcP_i ) \}_{i \in \mcI},
      \{ \mcQ_j \}_{j \in \mcJ} 
    }
  \end{align*}
  we have
  \begin{align*}
    & \overline{c_i \tL \ch{\kappa_i}{A_i}} ; \epsilon ; \epsilon \vdash m : \CM{\unit} \\
    & \forall i \in \mcI, \neg{\ch{\kappa_i}{A_i}} \Vdash \mcP_i \\
    & \forall j \in \mcJ, \Vdash \mcQ_j
  \end{align*}

  \noindent
  From $k \in \mcI$ and the typing judgment $\neg{\ch{\kappa_k}{A_k}} \Vdash \mcP_k$,

  \noindent
  By the induction hypothesis, we have the well-typed process 
  $\neg{\ch{\kappa_k}{A_k}} \Vdash \mcP_k'$.

  \noindent
  This gives us the well-typed root
  \begin{align*}
    \Vdash \Root{m}{ 
      \{ ( c_i, \mcP_i ) \}_{i \in \mcI \setminus \{k\}} \cup \{ ( c_k, \mcP_k' ) \}, 
      \{ \mcQ_j \}_{j \in \mcJ}
    }
  \end{align*}
  and concludes this case.

\noindent
\textbf{Case} (
  \textsc{Node-Child},
  \textsc{Root-SubTree},
  \textsc{Node-SubTree}
): Similar to the \textsc{Root-Child} case.

\noindent
\textbf{Case} (\textsc{Root-Expr}):
  \begin{mathpar}\small
  \inferrule[Root-Expr]
  { 
    m \Leadsto m'
  }
  { 
    \Root{m}{
      \{ ( c_i, \mcP_i ) \}_{i \in \mcI},
      \{ \mcQ_j \}_{j \in \mcJ}
    }
    \Rrightarrow
    \Root{m'}{
      \{ ( c_i, \mcP_i ) \}_{i \in \mcI},
      \{ \mcQ_j \}_{j \in \mcJ} 
    }
  }
  \end{mathpar}
  By inversion on the typing judgment
  \begin{align*}
    \Vdash \Root{m}{
      \{ ( c_i, \mcP_i ) \}_{i \in \mcI},
      \{ \mcQ_j \}_{j \in \mcJ}
    }
  \end{align*}
  we have
  \begin{align*}
    & \overline{c_i \tL \ch{\kappa_i}{A_i}} ; \epsilon ; \epsilon \vdash m : \CM{\unit} \\ 
    & \forall i \in \mcI, \neg{\ch{\kappa_i}{A_i}} \Vdash \mcP_i \\
    & \forall j \in \mcJ, \Vdash \mcQ_j
  \end{align*}

  \noindent
  By \Cref{theorem:program-subject-reduction} we have 
  $$
    \overline{c_i \tL \ch{\kappa_i}{A_i}} ; \epsilon ; \epsilon \vdash m' : \CM{\unit}
  $$
  which gives us the well-typed root
  \begin{align*}
    \Vdash \Root{m'}{
      \{ ( c_i, \mcP_i ) \}_{i \in \mcI},
      \{ \mcQ_j \}_{j \in \mcJ} 
    }
  \end{align*}
  and concludes this case.

\noindent
\textbf{Case} (\textsc{Node-Expr}): Similar to the \textsc{Root-Expr} case.
\end{proof}

Now we can define the reachability of a configuration through spawning trees.
\begin{definition}[Reachability]
  A configuration $P$ is \emph{reachable} if there exists a spawning tree $\mcP$
  such that $\Vdash \mcP$ and $|\mcP| = P$.
\end{definition}

Now to prove the main progress theorem, we first need to define a few auxiliary judgments
that will help us characterize the state of a spawning tree.

The first judgment characterizes when a spawning tree is \emph{terminal}, i.e., it has no
further reductions. The expression in the root of the spawning tree must be a $\return{\ii}$ 
expression. The set of child processes should be empty and all sub-trees must be terminal.
\begin{mathpar}
  \inferrule
  { \forall j \in \mcJ,\quad  \mcQ_j~\Terminal }
  { \Root{\return{\ii}}{\emptyset, \{ Q_j \}_{j \in \mcJ}}~\Terminal }
  \textsc{(Terminal-Root)}
\end{mathpar}

\begin{lemma}[Terminal]\label[lemma]{lemma:terminal}
  If $\mcP~\Terminal$, then $|\mcP| \equiv \proc{\return{\ii}}$.
\end{lemma}
\begin{proof}
  By induction on the derivation of $\mcP~\Terminal$ and definition of structural congruence.
\end{proof}

The second judgment characterizes when a process is \emph{poised} to communicate on a
channel $d$. These poised expressions correspond to thunked concurrency primitives that
are waiting to be executed. The \Fork{}-operation is not include here because it does not
rely on any channel.
\begin{mathpar}
  \inferrule[Poised-Explicit-Recv]
  { }
  { \recvR{d}~\Poised{d} }

  \inferrule[Poised-Implicit-Recv]
  { }
  { \recvI{d}~\Poised{d} }

  \inferrule[Poised-Explicit-Send]
  { }
  { \appR{\sendR{d}}{v}~\Poised{d} }

  \inferrule[Poised-Implicit-Send]
  { }
  { \appI{\sendI{d}}{o}~\Poised{d} }

  \inferrule[Poised-Close]
  { }
  { \close{d}~\Poised{d} }

  \inferrule[Poised-Wait]
  { }
  { \wait{d}~\Poised{d} }

  \inferrule[Poised-Bind]
  { m~\Poised{d} }
  { \letin{x}{m}{n}~\Poised{d} }
\end{mathpar}

\begin{lemma}[Thunk Evaluation Context]\label[lemma]{lemma:thunk-eval-context}
  Given thunk $\tau$, there exists an evaluation context $\mcM$ such that
  $\tau = \mcM[m]$ where $m$ is one of the following forms:
  \[
    \fork{x : A}{n} \mid
    \recvR{c} \mid
    \recvI{c} \mid
    \appR{\sendR{c}}{v} \mid
    \appI{\sendI{c}}{o} \mid
    \close{c} \mid
    \wait{c}
  \]
\end{lemma}
\begin{proof}
  By induction on the structure of $\tau$.
\end{proof}

\begin{lemma}[Poised Evaluation Context]\label[lemma]{lemma:poised-eval-context}
  If $m~\Poised{d}$ then there exists an evaluation context $\mcM$ such that
  $m = \mcM[m']$ where $m'$ is one of the following forms:
  \[
    \recvR{d} \mid
    \recvI{d} \mid
    \appR{\sendR{d}}{v} \mid
    \appI{\sendI{d}}{o} \mid
    \close{d} \mid
    \wait{d}
  \]
\end{lemma}
\begin{proof}
  By induction on the derivation of $m~\Poised{d}$.
\end{proof}

The final judgment characterizes a spawning tree that is poised to communicate on
channel $d$ connected to its parent. 
\begin{mathpar}
  \inferrule[Poised-Node]
  { m~\Poised{d} }
  { 
    \Node{d}{m}{ 
      \{ ( c_i, \mcP_i ) \}_{i \in \mcI}, 
      \{ \mcQ_j \}_{j \in \mcJ} }
    ~\Poised{d} 
  }
\end{mathpar}

With these definitions, we can now prove the spawning tree progress lemma.
\begin{lemma}[Spawning Tree Progress]\label[lemma]{lemma:spawning-tree-progress}
  If $\Vdash \mcP$, then either $\mcP~\Terminal$ or there exists $\mcP'$ such that $\mcP \Rrightarrow \mcP'$.
  If $A \Vdash \mcP$, then either $\mcP~\Poised{d}$ or there exists $\mcP'$ such that $\mcP \Rrightarrow \mcP'$.
\end{lemma}
\begin{proof}
  By mutual induction on the derivation of $\Vdash \mcP$ and $A \Vdash \mcP$.

\noindent
\textbf{Case} (\textsc{Valid-Root})
  \begin{mathpar}\small
  \inferrule
  { \overline{c_i \tL \ch{\kappa_i}{A_i}} ; \epsilon ; \epsilon \vdash m :\CM{\unit} \\
    \forall i \in \mcI,\ \neg{\ch{\kappa_i}{A_i}} \Vdash \mcP_i  \\
    \forall j \in \mcJ,\ \Vdash \mcQ_j }
  { \Vdash \Root{m}{ \{ ( c_i, \mcP_i ) \}_{i \in \mcI}, \{ \mcQ_j \}_{j \in \mcJ} } }
  \textsc{(Valid-Root)}
  \end{mathpar}
  Apply the induction hypothesis on $\mcQ_j$. If any $\mcQ_j$ is not terminal, then
  there exists $\mcQ_k$ such that $\mcQ_k \Rrightarrow \mcQ_k'$. By the \textsc{Root-SubTree} rule,
  we are done. Thus we can assume that $\mcQ_j$ are all terminal.

  Next, we analyze the expression $m$. By \Cref{theorem:program-progress}, either
  $m \Leadsto m'$ or $m$ is a value. If $m \Leadsto m'$, then by the \textsc{Root-Expr} rule, we are done. Thus we can assume that $m$ is a value.

  By \Cref{lemma:program-monad-canonical} we know that either 
  \begin{enumerate}
    \item exists value $v$ such that $m = \return{v}$
    \item $m$ is a thunk
  \end{enumerate}

  In sub-case (1), if $m = \return{v}$, then by 
  \Cref{lemma:program-inversion-return} we have 
  $\overline{c_i \tL \ch{\kappa_i}{A_i}} ; \epsilon ; \epsilon \vdash v :\unit$.
  By \Cref{lemma:program-unit-canonical}, we have $v = \ii$.
  By \Cref{lemma:program-inversion-unit}, we have $\overline{c_i \tL \ch{\kappa_i}{A_i}} = \epsilon$.
  Thus $\mcI = \emptyset$ and we have the terminal spawning tree 
  $\Root{\return{\ii}}{\emptyset, \{ Q_j \}_{j \in \mcJ}}$.

  In sub-case (2), if $m$ is a thunk, then we can use \Cref{lemma:thunk-eval-context} to 
  decompose $m$ into an evaluation context $\mcM$ and a sub-expression $m'$ such that 
  $m = \mcM[m']$ where $m'$ is one of the following forms:
  \[
    \fork{x : A}{n} \mid
    \recvR{c} \mid
    \recvI{c} \mid
    \appR{\sendR{c}}{v} \mid
    \appI{\sendI{c}}{o} \mid
    \close{c} \mid
    \wait{c}
  \]
  We now perform case analysis on $m'$. Most cases are similar and straightforward, 
  so we only show the following interesting cases in detail:
  \begin{itemize}
    \item If $m' = \fork{x : A}{n}$, then applying \textsc{Root-Fork} yields a reduction.
    \item
      If $m' = \recvR{c}$, then by \Cref{lemma:eval-context} and 
      \Cref{lemma:program-inversion-explicit-recv} we have $c \in \{ c_i \}_{i \in \mcI}$.
      Let $k$ be such that $c_k = c$. Apply the induction hypothesis on all $\mcP_i$.
      If any $\mcP_i$ is not poised on $d_i$, then there exists $\mcP_i'$ such that
      $\mcP_i \Rrightarrow \mcP_i'$. By the \textsc{Root-Child} rule, we are done.
      Thus we can assume that all $\mcP_i$ are poised on $d_i$.
      In particular, $\mcP_k$ is poised on $d_k$.
      By inversion of \textsc{Poised-Node} and \Cref{lemma:poised-eval-context}, 
      we can decompose $\mcP_k$ into an evaluation context $\mcN$ and a sub-expression $m_k'$ such that
      $m_k'$ is one of the following forms:
      \[
        \recvR{d_k} \mid
        \recvI{d_k} \mid
        \appR{\sendR{d_k}}{v} \mid
        \appI{\sendI{d_k}}{o} \mid
        \close{d_k} \mid
        \wait{d_k}
      \]
      Since $\mcP_k$ is well-typed, the type of dual channel $d_k$ is $\neg{\ch{\kappa_k}{A_k}}$.
      So $m_k'$ must be $\appR{\sendR{d_k}}{v}$. Applying \textsc{Root-Recv} yields a reduction.
  \end{itemize}
  This concludes the \textsc{Valid-Root} case.

\noindent
\textbf{Case} (\textsc{Valid-Node}):
  \begin{mathpar}\small
  \inferrule
  {  \overline{c_i \tL \ch{\kappa_i}{A_i}}, d \tL \ch{\kappa}{A} ; \epsilon ; \epsilon \vdash m :\CM{\unit} \\
    \forall i \in \mcI,\ \neg{\ch{\kappa_i}{A_i}} \Vdash \mcP_i \\ 
    \forall j \in \mcJ,\ \Vdash \mcQ_j }
  { \ch{\kappa}{A} \Vdash \Node{d}{m}{ \{ ( c_i, \mcP_i ) \}_{i \in \mcI}, \{ \mcQ_j \}_{j \in \mcJ} } }
  \textsc{(Valid-Node)}
  \end{mathpar}
  Apply the induction hypothesis on $\mcQ_j$. If any $\mcQ_j$ is not terminal, then
  there exists $\mcQ_k$ such that $\mcQ_k \Rrightarrow \mcQ_k'$. By the \textsc{Node-SubTree} rule,
  we are done. Thus we can assume that $\mcQ_j$ are all terminal.

  Next, we analyze the expression $m$. By \Cref{theorem:program-progress}, either
  $m \Leadsto m'$ or $m$ is a value. If $m \Leadsto m'$, then by the \textsc{Node-Expr} rule, we are done. Thus we can assume that $m$ is a value.

  By \Cref{lemma:program-monad-canonical} we know that either 
  \begin{enumerate}
    \item exists value $v$ such that $m = \return{v}$
    \item $m$ is a thunk
  \end{enumerate}

  In sub-case (1), if $m = \return{v}$, then applying \Cref{lemma:program-inversion-return}
  yields 
  $\overline{c_i \tL \ch{\kappa_i}{A_i}}, d \tL \ch{\kappa}{A} ; \epsilon ; \epsilon \vdash v :\unit$.
  By \Cref{lemma:program-unit-canonical}, we have $v = \ii$.
  By \Cref{lemma:program-inversion-unit}, we have 
  $\overline{c_i \tL \ch{\kappa_i}{A_i}}, d \tL \ch{\kappa}{A} = \epsilon$.
  which is a contradiction since $d \tL \ch{\kappa}{A}$ is not empty.
  Thus this sub-case is impossible.

  In sub-case (2), if $m$ is a thunk, then we can use \Cref{lemma:thunk-eval-context} to 
  decompose $m$ into an evaluation context $\mcM$ and a sub-expression $m'$ such that 
  $m = \mcM[m']$ where $m'$ is one of the following forms:
  \[
    \fork{x : A}{n} \mid
    \recvR{c} \mid
    \recvI{c} \mid
    \appR{\sendR{c}}{v} \mid
    \appI{\sendI{c}}{o} \mid
    \close{c} \mid
    \wait{c}
  \]
  We now perform case analysis on $m'$. Most cases are similar and straightforward, 
  so we only show the following interesting cases in detail:
  \begin{itemize}
    \item If $m' = \fork{x : A}{n}$, then applying \textsc{Node-Fork} yields a reduction.
    \item 
      If $m' = \recvR{c}$, then by \Cref{lemma:eval-context} and
      \Cref{lemma:program-inversion-explicit-recv} we have either 
      $c = d$ or $c \in \{ c_i \}_{i \in \mcI}$. In the case that $c = d$, we are done
      by the \textsc{Poised-Node} rule. In the case that $c \in \{ c_i \}_{i \in \mcI}$,
      let $k$ be such that $c_k = c$. Apply the induction hypothesis on all $\mcP_i$.
      If any $\mcP_i$ is not poised on $d_i$, then there exists $\mcP_i'$ such that
      $\mcP_i \Rrightarrow \mcP_i'$. By the \textsc{Node-Child} rule, we are done.
      Thus we can assume that all $\mcP_i$ are poised on $d_i$.
      In particular, $\mcP_k$ is poised on $d_k$.
      By inversion of \textsc{Poised-Node} and \Cref{lemma:poised-eval-context}, 
      we can decompose $\mcP_k$ into an evaluation context $\mcN$ and a sub-expression $m_k'$ such that
      $m_k'$ is one of the following forms:
      \[
        \recvR{d_k} \mid
        \recvI{d_k} \mid
        \appR{\sendR{d_k}}{v} \mid
        \appI{\sendI{d_k}}{o} \mid
        \close{d_k} \mid
        \wait{d_k}
      \]
      Since $\mcP_k$ is well-typed, the type of dual channel $d_k$ is $\neg{\ch{\kappa_k}{A_k}}$.
      So $m_k'$ must be $\appR{\sendR{d_k}}{v}$. Applying \textsc{Node-Recv} yields a reduction.
    \item
      If $m' = \appR{\sendR{c}}{v}$, then by \Cref{lemma:eval-context} and
      \Cref{lemma:program-inversion-explicit-send} we have either 
      $c = d$ or $c \in \{ c_i \}_{i \in \mcI}$. In the case that $c = d$, we are done
      by the \textsc{Poised-Node} rule. In the case that $c \in \{ c_i \}_{i \in \mcI}$,
      let $k$ be such that $c_k = c$. Apply the induction hypothesis on all $\mcP_i$.
      If any $\mcP_i$ is not poised on $d_i$, then there exists $\mcP_i'$ such that
      $\mcP_i \Rrightarrow \mcP_i'$. By the \textsc{Node-Child} rule, we are done.
      Thus we can assume that all $\mcP_i$ are poised on $d_i$.
      In particular, $\mcP_k$ is poised on $d_k$.
      By inversion of \textsc{Poised-Node} and \Cref{lemma:poised-eval-context}, 
      we can decompose $\mcP_k$ into an evaluation context $\mcN$ and a sub-expression $m_k'$ such that
      $m_k'$ is one of the following forms:
      \[
        \recvR{d_k} \mid
        \recvI{d_k} \mid
        \appR{\sendR{d_k}}{v} \mid
        \appI{\sendI{d_k}}{o} \mid
        \close{d_k} \mid
        \wait{d_k}
      \]
      Since $\mcP_k$ is well-typed, the type of dual channel $d_k$ is $\neg{\ch{\kappa_k}{A_k}}$.
      So $m_k'$ must be $\recvR{d_k}$. 
      Now if $d \notin \FC{v}$, then applying \textsc{Node-Send} yields a reduction. 
      Otherwise if $d \in \FC{v}$, then applying \textsc{Node-Forward} yields a reduction.
  \end{itemize}
This concludes the \textsc{Valid-Node} case.
\end{proof}
\clearpage

Finally, we can prove the main global progress theorem.
\begin{theorem}[Global Progress]\label{theorem:global-progress}
  Given a reachable configuration $P$, either
  \begin{itemize}
    \item $P \equiv \proc{\return{\ii}}$, or
    \item there exists $Q$ such that $P \Rrightarrow Q$ and $Q$ is reachable.
  \end{itemize}
\end{theorem}
\begin{proof}
  Since $P$ is reachable, there exists a spawning tree $\mcP$ such that
  $\Vdash \mcP$ and $|\mcP| = P$. By \Cref{lemma:spawning-tree-progress}, either
  $\mcP~\Terminal$ or there exists $\mcP'$ such that $\mcP \Rrightarrow \mcP'$.

  In the case that $\mcP~\Terminal$ we have $P = |\mcP| \equiv \proc{\return{\ii}}$
  by \Cref{lemma:terminal}.

  In the case that there exists $\mcP'$ such that $\mcP \Rrightarrow \mcP'$, by
  \Cref{lemma:spawning-tree-fidelity} we have $\Vdash \mcP'$. Let $Q = |\mcP'|$.
  By \Cref{lemma:spawning-tree-simulation}, we have $P \Rrightarrow Q$ 
  and $Q$ is also reachable by definition.
\end{proof}

\subsection{Erasure Safety}\label{appendix:erasure}
To show that ghost messages can indeed be safely erased without affecting the
behavior of processes, we define an erasure relation which replaces all ghost
messages with the opaque value $\square$. Since $\square$ cannot be inspected,
it cannot affect the behavior of processes. Thus if each step of an original
process configuration can be simulated by its erased counterpart, then we can
safely erase ghost messages. We define two kinds of erasure relations for programs
and processes, respectively.
\begin{itemize}
  \item 
    $\Theta ; \Gamma ; \Delta \vdash m \sim m' : A$: 
    the program $m$ (of type $A$ under context $\Theta, \Gamma, \Delta$) is erased to $m'$.
  \item $\Theta \Vdash P \sim P'$:
    the process $P$ is erased to $P'$.
\end{itemize}
The formal rules for these relations are as follows.

\paragraph{\textbf{Core Erasure}}
The erasure rules for the functional core of \TLLC{}.
\begin{mathpar}\footnotesize
  \inferrule[Var]
  { \epsilon ; \Gamma, x : A ; \Delta, x \ty{s} A \vdash \\
    \Delta \triangleright \Un }
  { \epsilon ; \Gamma, x : A ; \Delta, x \ty{s} A \vdash x \sim x : A }

  \inferrule[Conversion]
  { \Gamma \vdash B : s \\
    \Theta ; \Gamma ; \Delta \vdash m \sim m' : A \\
    A \simeq B }
  { \Theta ; \Gamma ; \Delta \vdash m \sim m' : B }

  \inferrule[Explicit-Lam]
  { \Theta ; \Gamma, x : A; \Delta, x \ty{s} A \vdash m \sim m' : B \\
    \Theta \triangleright t \\
    \Delta \triangleright t }
  { \Theta ; \Gamma ; \Delta \vdash 
    \lamR{t}{x : A}{m} \sim \lamR{t}{x : \square}{m'} : \PiR{t}{x : A}{B} }

  \inferrule[Implicit-Lam]
  { \Theta ; \Gamma, x : A; \Delta \vdash m \sim m' : B \\
    \Theta \triangleright t \\
    \Delta \triangleright t }
  { \Theta ; \Gamma ; \Delta \vdash 
    \lamI{t}{x : A}{m} \sim \lamI{t}{x : \square}{m'} : \PiI{t}{x : A}{B} }

  \inferrule[Explicit-App]
  { \Theta_1 ; \Gamma ; \Delta_1 \vdash m \sim m' : \PiR{t}{x : A}{B} \\
    \Theta_2 ; \Gamma ; \Delta_2 \vdash n \sim n' : A }
  { \Theta_1 \dotcup \Theta_2 ; \Gamma ; \Delta_1 \dotcup \Delta_2 \vdash 
    \appR{m}{n} \sim \appR{m'}{n'} : B[n/x] }

  \inferrule[Implicit-App]
  { \Theta ; \Gamma ; \Delta \vdash m \sim m' : \PiI{t}{x : A}{B} \\
    \Gamma \vdash n : A }
  { \Theta ; \Gamma ; \Delta \vdash \appI{m}{n} \sim \appI{m'}{\square} : B[n/x] }

  \inferrule[Explicit-Pair]
  { \Gamma \vdash \SigR{t}{x : A}{B} : t \\\\
    \Theta_1 ; \Gamma ; \Delta_1 \vdash m \sim m' : A \\
    \Theta_2 ; \Gamma ; \Delta_2 \vdash n \sim n' : B[m/x] }
  { \Theta_1 \dotcup \Theta_2 ; \Gamma ; \Delta_1 \dotcup \Delta_2 \vdash 
    \pairR{m}{n}{t} \sim \pairR{m'}{n'}{t} : \SigR{t}{x : A}{B} }

  \inferrule[Implicit-Pair]
  { \Gamma \vdash \SigI{t}{x : A}{B} : t \\\\
    \Gamma \vdash m : A \\
    \Theta ; \Gamma ; \Delta \vdash n \sim n' : B[m/x] }
  { \Theta ; \Gamma ; \Delta \vdash 
    \pairI{m}{n}{t} \sim \pairI{\square}{n'}{t} : \SigI{t}{x : A}{B} }

  \inferrule[Explicit-SumElim]
  { \Gamma, z : \SigR{t}{x : A}{B} \vdash C : s \\
    \Theta_1 ; \Gamma ; \Delta_1 \vdash m \sim m' : \SigR{t}{x : A}{B} \\\\
    \Theta_2 ; \Gamma, x : A, y : B; \Delta_2, x \ty{r1} A, y \ty{r2} B \vdash 
    n \sim n' : C[\pairR{x}{y}{t}/z] }
  { \Theta_1 \dotcup \Theta_2 ; \Gamma ; \Delta_1 \dotcup \Delta_2 \vdash 
    \SigElim{[z]C}{m}{[x,y]n} \sim \SigElim{\square}{m'}{[x,y]n'} : C[m/z] }

  \inferrule[Implicit-SumElim]
  { \Gamma, z : \SigI{t}{x : A}{B} \vdash C : s \\
    \Theta_1 ; \Gamma ; \Delta_1 \vdash m \sim m' : \SigI{t}{x : A}{B} \\\\
    \Theta_2 ; \Gamma, x : A, y : B; \Delta_2, y \ty{r} B \vdash n \sim n' : C[\pairI{x}{y}{t}/z] }
  { \Theta_1 \dotcup \Theta_2 ; \Gamma ; \Delta_1 \dotcup \Delta_2 \vdash 
    \SigElim{[z]C}{m}{[x,y]n} \sim \SigElim{\square}{m'}{[x,y]n'} : C[m/z] }\\
\end{mathpar}

\paragraph{\textbf{Data Erasure}}
The erasure rules for basic data types.
\begin{mathpar}\footnotesize
  \inferrule[UnitVal]
  { \Gamma ; \Delta \vdash \\ \Delta \triangleright \Un }
  { \epsilon ; \Gamma ; \Delta \vdash \ii \sim \ii : \unit }

  \inferrule[True]
  { \Gamma ; \Delta \vdash \\ \Delta \triangleright \Un }
  { \epsilon ; \Gamma ; \Delta \vdash \bTrue \sim \bTrue : \Bool }

  \inferrule[False]
  { \Gamma ; \Delta \vdash \\ \Delta \triangleright \Un }
  { \epsilon ; \Gamma ; \Delta \vdash \bFalse \sim \bFalse : \Bool }

  \inferrule[BoolElim]
  { \Gamma, z : \Bool \vdash A : s \\
    \Theta_1 ; \Gamma ; \Delta_1 \vdash m \sim m' : \Bool \\\\
    \Theta_2 ; \Gamma ; \Delta_2 \vdash n_1 \sim n_1' : A[\bTrue/z] \\
    \Theta_2 ; \Gamma ; \Delta_2 \vdash n_2 \sim n_2' : A[\bFalse/z] }
  { \Theta_1 \dotcup \Theta_2 ; \Gamma ; \Delta_1 \dotcup \Delta_2 \vdash 
    \boolElim{[z]A}{m}{n_1}{n_2} \sim \boolElim{\square}{m'}{n_1'}{n_2'} : A[m/z] }
\end{mathpar}

\paragraph{\textbf{Monadic Erasure}}
The erasure rules for the monadic constructs.
\begin{mathpar}\footnotesize
  \inferrule[Return]
  { \Theta ; \Gamma ; \Delta \vdash m \sim m' : A }
  { \Theta ; \Gamma ; \Delta \vdash \return{m} \sim \return{m'} : \CM{A} }

  \inferrule[Bind]
  { \Gamma \vdash B : s \\
    \Theta_1 ; \Gamma ; \Delta_1 \vdash m \sim m' : \CM{A} \\\\
    \Theta_2 ; \Gamma, x : A ; \Delta_2, x \ty{r} A \vdash n \sim n' : \CM{B} }
  { \Theta_1 \dotcup \Theta_2 ; \Gamma ; \Delta_1 \dotcup \Delta_2 \vdash 
    \letin{x}{m}{n} \sim \letin{x}{m'}{n'} : \CM{B} }
\end{mathpar}

\paragraph{\textbf{Session Erasure}}
The erasure rules for concurrency primitives.
\begin{mathpar}\footnotesize
  \inferrule[Channel-CH]
  { \Gamma ; \Delta \vdash \\
    \epsilon \vdash A : \Proto \\
    \Delta \triangleright \Un }
  { c \tL \CH{A} ; \Gamma ; \Delta \vdash c \sim c : \CH{A} }

  \inferrule[Channel-HC]
  { \Gamma ; \Delta \vdash \\
    \epsilon \vdash A : \Proto \\
    \Delta \triangleright \Un }
  { c \tL \HC{A} ; \Gamma ; \Delta \vdash c \sim c : \HC{A} }

  \inferrule[Explicit-Send-CH]
  { \Theta ; \Gamma ; \Delta \vdash m \sim m' : \CH{\ActR{!}{x : A}{B}} }
  { \Theta ; \Gamma ; \Delta \vdash \sendR{m} \sim \sendR{m'} : \PiR{\Ln}{x : A}{\CM{\CH{B}}} }

  \inferrule[Explicit-Send-HC]
  { \Theta ; \Gamma ; \Delta \vdash m \sim m' : \HC{\ActR{?}{x : A}{B}} }
  { \Theta ; \Gamma ; \Delta \vdash \sendR{m} \sim \sendR{m'} : \PiR{\Ln}{x : A}{\CM{\HC{B}}} }

  \inferrule[Implicit-Send-CH]
  { \Theta ; \Gamma ; \Delta \vdash m \sim m' : \CH{\ActI{!}{x : A}{B}} }
  { \Theta ; \Gamma ; \Delta \vdash \sendI{m} \sim \sendI{m'} : \PiI{\Ln}{x : A}{\CM{\CH{B}}} }

  \inferrule[Implicit-Send-HC]
  { \Theta ; \Gamma ; \Delta \vdash m \sim m' : \HC{\ActI{?}{x : A}{B}} }
  { \Theta ; \Gamma ; \Delta \vdash \sendI{m} \sim \sendI{m'} : \PiI{\Ln}{x : A}{\CM{\HC{B}}} }

  \inferrule[Explicit-Recv-CH]
  { \Theta ; \Gamma ; \Delta \vdash m \sim m' : \CH{\ActR{?}{x : A}{B}} }
  { \Theta ; \Gamma ; \Delta \vdash \recvR{m} \sim \recvR{m'} : \CM{\SigR{\Ln}{x : A}{\CH{B}}} }

  \inferrule[Explicit-Recv-HC]
  { \Theta ; \Gamma ; \Delta \vdash m \sim m' : \HC{\ActR{!}{x : A}{B}} }
  { \Theta ; \Gamma ; \Delta \vdash \recvR{m} \sim \recvR{m'} : \CM{\SigR{\Ln}{x : A}{\HC{B}}} }

  \inferrule[Implicit-Recv-CH]
  { \Theta ; \Gamma ; \Delta \vdash m \sim m' : \CH{\ActI{?}{x : A}{B}} }
  { \Theta ; \Gamma ; \Delta \vdash \recvI{m} \sim \recvI{m'} : \CM{\SigI{\Ln}{x : A}{\CH{B}}} }

  \inferrule[Implicit-Recv-HC]
  { \Theta ; \Gamma ; \Delta \vdash m \sim m' : \HC{\ActI{!}{x : A}{B}} }
  { \Theta ; \Gamma ; \Delta \vdash \recvI{m} \sim \recvI{m'} : \CM{\SigI{\Ln}{x : A}{\HC{B}}} }

  \inferrule[Fork]
  { \Theta ; \Gamma, x : \CH{A} ; \Delta, x \tL \CH{A} \vdash m \sim m' : \CM{\unit} }
  { \Theta ; \Gamma ; \Delta \vdash \fork{x : \CH{A}}{m} \sim \fork{x : \square}{m'} : \CM{\HC{A}} }
  \\
  \inferrule[Close]
  { \Theta ; \Gamma ; \Delta \vdash m \sim m' : \CH{\End} }
  { \Theta ; \Gamma ; \Delta \vdash \close{m} \sim \close{m'} : \CM{\unit} }

  \inferrule[Wait]
  { \Theta ; \Gamma ; \Delta \vdash m \sim m' : \HC{\End} }
  { \Theta ; \Gamma ; \Delta \vdash \wait{m} \sim \wait{m'} : \CM{\unit} }
\end{mathpar}

\paragraph{\textbf{Process Erasure}}
The erasure rules for processes.
\begin{mathpar}\footnotesize
  \inferrule[Expr]
  { \Theta ; \epsilon ; \epsilon \vdash m \sim m' : \CM{\unit} }
  { \Theta \Vdash \proc{m} \sim \proc{m'} }

  \inferrule[Par]
  { \Theta_1 \vdash P \sim P' \\ \Theta_2 \vdash Q \sim Q' }
  { \Theta_1 \dotcup \Theta_2 \Vdash (P \mid Q) \sim (P' \mid Q') }

  \inferrule[Scope]
  { \Theta, c \tL \CH{A}, d \tL \HC{A} \vdash P \sim P' }
  { \Theta \Vdash \scope{cd}{P} \sim \scope{cd}{P'} }
\end{mathpar}
\clearpage

The steps to show erasure safety is similar to that of session fidelity~\Cref{appendix:fidelity}. 
We begin by stating the weakening and substitution lemmas for the erasure relation.
The proofs are routine inductions on the derivations of the erasure judgments.

\begin{lemma}[Implicit Weakening for Erasure]\label[lemma]{lemma:erasure-implicit-weakening}
  If $\Theta ; \Gamma ; \Delta \vdash m \sim m' : A$, then for any $\Gamma \vdash B : s$
  and $x \notin \Gamma$, we have $\Theta ; \Gamma, x : B ; \Delta \vdash m \sim m' : A$.
\end{lemma}

\begin{lemma}[Explicit Weakening for Erasure]\label[lemma]{lemma:erasure-explicit-weakening}
  If $\Theta ; \Gamma ; \Delta \vdash m \sim m' : A$, then for any $\Gamma \vdash B : \Un$
  and $x \notin \Gamma$, we have $\Theta ; \Gamma, x : B ; \Delta, x \tU B \vdash m \sim m' : A$.
\end{lemma}

\begin{lemma}[Implicit Substitution for Erasure]\label[lemma]{lemma:erasure-implicit-substitution}
  If $\Theta ; \Gamma, x : A ; \Delta \vdash m \sim m' : B$ and 
  $\Gamma \vdash n : A$, then
  $\Theta ; \Gamma ; \Delta \vdash m[n/x] \sim m'[\square/x] : B[n/x]$.
\end{lemma}

\begin{lemma}[Explicit Substitution for Erasure]\label[lemma]{lemma:erasure-explicit-substitution}
  If $\Theta_1 ; \Gamma, x : A ; \Delta_1, x :_s A \vdash m \sim m' : B$ and 
  $\Theta_2 ; \Gamma ; \Delta_2 \vdash n \sim n' : A$, and 
  $\Theta_2 \triangleright s$,
  $\Delta_2 \triangleright s$, then
  $\Theta_1 \dotcup \Theta_2 ; \Gamma ; \Delta_1 \dotcup \Delta_2 \vdash m[n/x] \sim m'[n'/x] : B[n/x]$
\end{lemma}

The erasure relation implies program typing.
\begin{lemma}[Erasure Typing]\label[lemma]{lemma:erasure-typing}
  If $\Theta ; \Gamma ; \Delta \vdash m \sim m' : A$, then $\Theta ; \Gamma ; \Delta \vdash m : A$
  and $\Gamma \vdash m : A$.
\end{lemma}
\begin{proof}
  To show $\Theta ; \Gamma ; \Delta \vdash m : A$, we proceed by induction on the
  derivation of ${\Theta ; \Gamma ; \Delta \vdash m \sim m' : A}$.
  All cases are straightforward as they are compatible with the typing rules.

  To show $\Gamma \vdash m : A$, we apply \Cref{theorem:lifting} to the 
  judgment $\Theta ; \Gamma ; \Delta \vdash m : A$.
\end{proof}

The types of the erasure relation are valid.
\begin{lemma}[Erasure Type Validity]\label[lemma]{lemma:erasure-type-validity}
  If $\Theta ; \Gamma ; \Delta \vdash m \sim m' : A$, then $\epsilon \vdash A : s$ for some $s$.
\end{lemma}
\begin{proof}
  By \Cref{lemma:erasure-typing}, we have $\Theta ; \Gamma ; \Delta \vdash m : A$.
  The result then follows from \Cref{theorem:program-type-validity}.
\end{proof}

\paragraph{\textbf{Shape Preservation}}
Next, we prove ``shape preservation'' lemmas for the erasure relation.
Intuitively, these lemmas state that the erasure relation preserve the overall shape of programs.
The proofs are routine inductions on the derivations of the erasure judgments.

\begin{lemma}\label[lemma]{lemma:erasure-shape-var}
  $\Theta ; \Gamma ; \Delta \vdash x \sim m' : B$, then $m' = x$.
\end{lemma}

\begin{lemma}\label[lemma]{lemma:erasure-shape-explicit-lam}
  If $\Theta ; \Gamma ; \Delta \vdash \lamR{t}{x : A}{m} \sim m' : B$, then 
  $m' = \lamR{t}{x : \square}{n'}$ for some $n'$.
\end{lemma}

\begin{lemma}\label[lemma]{lemma:erasure-shape-implicit-lam}
  If $\Theta ; \Gamma ; \Delta \vdash \lamI{t}{x : A}{m} \sim m' : B$, then 
  $m' = \lamI{t}{x : \square}{n'}$ for some $n'$.
\end{lemma}

\begin{lemma}\label[lemma]{lemma:erasure-shape-explicit-app}
  If $\Theta ; \Gamma ; \Delta \vdash \appR{m}{n} \sim m' : B$, then 
  $m' = \appR{n_1'}{n_2'}$ for some $n_1', n_2'$.
\end{lemma}

\begin{lemma}\label[lemma]{lemma:erasure-shape-implicit-app}
  If $\Theta ; \Gamma ; \Delta \vdash \appI{m}{n} \sim m' : B$, then 
  $m' = \appI{n_1'}{n_2'}$ for some $n_1', n_2'$.
\end{lemma}

\begin{lemma}\label[lemma]{lemma:erasure-shape-explicit-pair}
  If $\Theta ; \Gamma ; \Delta \vdash \pairR{m}{n}{t} \sim m' : B$, then 
  $m' = \pairR{n_1'}{n_2'}{t}$ for some $n_1', n_2'$.
\end{lemma}

\begin{lemma}\label[lemma]{lemma:erasure-shape-implicit-pair}
  If $\Theta ; \Gamma ; \Delta \vdash \pairI{m}{n}{t} \sim m' : B$, then 
  $m' = \pairI{n_1'}{n_2'}{t}$ for some $n_1', n_2'$.
\end{lemma}

\begin{lemma}\label[lemma]{lemma:erasure-shape-sumElim}
  If $\Theta ; \Gamma ; \Delta \vdash \SigElim{[z]C}{m}{[x,y]n} \sim m' : B$, then 
  $m' = \SigElim{\square}{n_1'}{[x,y]n_2'}$ for some $n_1', n_2'$.
\end{lemma}

\begin{lemma}\label[lemma]{lemma:erasure-shape-unit}
  If $\Theta ; \Gamma ; \Delta \vdash \ii \sim m' : B$, then $m' = \ii$. 
\end{lemma}

\begin{lemma}\label[lemma]{lemma:erasure-shape-true}
  If $\Theta ; \Gamma ; \Delta \vdash \bTrue \sim m' : B$, then $m' = \bTrue$. 
\end{lemma}

\begin{lemma}\label[lemma]{lemma:erasure-shape-false}
  If $\Theta ; \Gamma ; \Delta \vdash \bFalse \sim m' : B$, then $m' = \bFalse$. 
\end{lemma}

\begin{lemma}\label[lemma]{lemma:erasure-shape-boolElim}
  If $\Theta ; \Gamma ; \Delta \vdash \boolElim{[z]A}{m}{n_1}{n_2} \sim m' : B$, then 
  $m' = \boolElim{\square}{n'}{n_1'}{n_2'}$ for some $n', n_1', n_2'$.
\end{lemma}

\begin{lemma}\label[lemma]{lemma:erasure-shape-return}
  If $\Theta ; \Gamma ; \Delta \vdash \return{m} \sim m' : B$, then 
  $m' = \return{n'}$ for some $n'$.
\end{lemma}

\begin{lemma}\label[lemma]{lemma:erasure-shape-bind}
  If $\Theta ; \Gamma ; \Delta \vdash \letin{x}{m}{n} \sim m' : B$, then 
  $m' = \letin{x}{n_1'}{n_2'}$ for some $n_1', n_2'$.
\end{lemma}

\begin{lemma}\label[lemma]{lemma:erasure-shape-channel}
  If $c \tL \CH{A} ; \Gamma ; \Delta \vdash c \sim m' : B$, then $m' = c$. 
\end{lemma}

\begin{lemma}\label[lemma]{lemma:erasure-shape-explicit-send}
  If $\Theta ; \Gamma ; \Delta \vdash \sendR{m} \sim m' : B$, then 
  $m' = \sendR{n'}$ for some $n'$.
\end{lemma}

\begin{lemma}\label[lemma]{lemma:erasure-shape-explicit-recv}
  If $\Theta ; \Gamma ; \Delta \vdash \recvR{m} \sim m' : B$, then 
  $m' = \recvR{n'}$ for some $n'$.
\end{lemma}

\begin{lemma}\label[lemma]{lemma:erasure-shape-implicit-send}
  If $\Theta ; \Gamma ; \Delta \vdash \sendI{m} \sim m' : B$, then 
  $m' = \sendI{n'}$ for some $n'$.
\end{lemma}

\begin{lemma}\label[lemma]{lemma:erasure-shape-implicit-recv}
  If $\Theta ; \Gamma ; \Delta \vdash \recvI{m} \sim m' : B$, then 
  $m' = \recvI{n'}$ for some $n'$.
\end{lemma}

\begin{lemma}\label[lemma]{lemma:erasure-shape-fork}
  If $\Theta ; \Gamma ; \Delta \vdash \fork{x : \CH{A}}{m} \sim m' : B$, then 
  $m' = \fork{x : \square}{n'}$ for some $n'$.
\end{lemma}

\begin{lemma}\label[lemma]{lemma:erasure-shape-close}
  If $\Theta ; \Gamma ; \Delta \vdash \close{m} \sim m' : B$, then 
  $m' = \close{n'}$ for some $n'$.
\end{lemma}

\begin{lemma}\label[lemma]{lemma:erasure-shape-wait}
  If $\Theta ; \Gamma ; \Delta \vdash \wait{m} \sim m' : B$, then 
  $m' = \wait{n'}$ for some $n'$.
\end{lemma}

We also show that the erasure relation preserves the structure of thunks and values.
\begin{lemma}\label[lemma]{lemma:erasure-shape-value}
  Given $\Theta ; \Gamma ; \Delta \vdash m \sim m' : A$,
  if $m$ is a value, then $m'$ is also a value, and if $m$ is a thunk, then $m'$ is also a thunk.
\end{lemma}
\begin{proof}
  By straightforward induction on the derivation of $\Theta ; \Gamma ; \Delta \vdash m \sim m' : A$
  and applying the shape preservation lemmas.
\end{proof}

\paragraph{\textbf{Inversion Lemmas}}
We now define inversion lemmas for the erasure relation.
When used in conjunction with the shape preservation lemmas, these lemmas allow us to
invert an erasure judgment to obtain erasure judgments for its subterms.
The proofs are routine inductions on the derivations of the erasure judgments.

\begin{lemma}\label[lemma]{lemma:erasure-inversion-implicit-lam}
  If 
  $
    \Theta ; \Gamma ; \Delta \vdash 
    \lamI{t}{x : A_1}{m_1} \sim 
    \lamI{t}{x : A_2}{m_2} : \PiI{s}{x : A_3}{B}
  $, then 
  $\Theta ; \Gamma, x : A_3 ; \Delta \vdash m_1 \sim m_2 : B$.
\end{lemma}

\begin{lemma}\label[lemma]{lemma:erasure-inversion-explicit-lam}
  If 
  $
    \Theta ; \Gamma ; \Delta \vdash 
    \lamR{t}{x : A_1}{m_1} \sim 
    \lamR{t}{x : A_2}{m_2} : \PiR{s}{x : A_3}{B}
  $, then 
  $\Theta ; \Gamma, x : A_3 ; \Delta, x \ty{r} A_3 \vdash m_1 \sim m_2 : B$
  for some $r$.
\end{lemma}

\begin{lemma}\label[lemma]{lemma:erasure-inversion-implicit-pair}
  If 
  $
    \Theta ; \Gamma ; \Delta \vdash 
    \pairI{m_1}{n_1}{t} \sim 
    \pairI{m_2}{n_2}{t} : \SigI{s}{x : A}{B}
  $, then 
  $s = r$, $m_2 = \square$, $\Gamma \vdash m_1 : A$, and 
  $\Theta ; \Gamma ; \Delta \vdash n_1 \sim n_2 : B.[m_1/x]$.
\end{lemma}

\begin{lemma}\label[lemma]{lemma:erasure-inversion-explicit-pair}
  If 
  $
    \Theta ; \Gamma ; \Delta \vdash 
    \pairR{m_1}{n_1}{t} \sim 
    \pairR{m_2}{n_2}{t} : \SigR{s}{x : A}{B}
  $, then there exist $\Theta_1, \Theta_2, \Delta_1, \Delta_2$ such that
  $s = r$, $\Theta_1 \dotcup \Theta_2 = \Theta$, 
  $\Delta_1 \dotcup \Delta_2 = \Delta$, 
  $\Theta_1 ; \Gamma ; \Delta_1 \vdash m_1 \sim m_2 : A$, and 
  $\Theta_2 ; \Gamma ; \Delta_2 \vdash n_1 \sim n_2 : B[m_1/x]$. 
\end{lemma}

\begin{lemma}\label[lemma]{lemma:erasure-inversion-implicit-app}
  If 
  $
    \Theta ; \Gamma ; \Delta \vdash 
    \appI{m_1}{n_1} \sim 
    \appI{m_2}{n_2} : C
  $, then there exist $A, B, s$ such that
  $n_2 = \square$, $\Theta ; \Gamma ; \Delta \vdash m_1 \sim m_2 : \PiI{s}{x : A}{B}$,
  $\Gamma \vdash n_1 : A$, and $C \simeq B[n_1/x]$.
\end{lemma}

\begin{lemma}\label[lemma]{lemma:erasure-inversion-explicit-app}
  If 
  $
    \Theta ; \Gamma ; \Delta \vdash 
    \appR{m_1}{n_1} \sim 
    \appR{m_2}{n_2} : C
  $, then there exist $A, B, s, \Theta_1, \Theta_2, \Delta_1, \Delta_2$ such that
  $\Theta_1 \dotcup \Theta_2 = \Theta$, 
  $\Delta_1 \dotcup \Delta_2 = \Delta$, 
  $\Theta_1 ; \Gamma ; \Delta_1 \vdash m_1 \sim m_2 : \PiR{s}{x : A}{B}$,
  $\Theta_2 ; \Gamma ; \Delta_2 \vdash n_1 \sim n_2 : A$, and $C \simeq B[n_1/x]$.
\end{lemma}

\begin{lemma}\label[lemma]{lemma:erasure-inversion-implicit-return}
  If 
  $
    \Theta ; \Gamma ; \Delta \vdash 
    \return{m_1} \sim 
    \return{m_2} : \CM{A}
  $, then 
  $\Theta ; \Gamma ; \Delta \vdash m_1 \sim m_2 : A$.
\end{lemma}

\begin{lemma}\label[lemma]{lemma:erasure-inversion-bind}
  If 
  ${
    \Theta ; \Gamma ; \Delta \vdash 
    \letin{x}{m_1}{n_1} \sim 
    \letin{x}{m_2}{n_2} : \CM{B}
  }$, then there exists $\Theta_1, \Theta_2, \Delta_1, \Delta_2, A, s$ such that
  $\Theta_1 \dotcup \Theta_2 = \Theta$, 
  $\Delta_1 \dotcup \Delta_2 = \Delta$, 
  $\Theta_1 ; \Gamma ; \Delta_1 \vdash m_1 \sim m_2 : \CM{A}$, and 
  ${\Theta_2 ; \Gamma, x : A ; \Delta_2, x \ty{r} A \vdash n_1 \sim n_2 : \CM{B}}$ and
  $x \notin \FV{B}$.
\end{lemma}

\begin{lemma}\label[lemma]{lemma:erasure-inversion-explicit-send}
  If ${\Theta ; \Gamma ; \Delta \vdash \sendR{m} \sim \sendR{m'} : C}$, then there exists $A, B$
  such that either 
  $C \simeq \PiR{\Ln}{x : A}{\CM{\CH{B}}}$ and 
  $\Theta ; \Gamma ; \Delta \vdash m \sim m' : \CH{\ActR{!}{x : A}{B}}$
  or
  $C \simeq \PiR{\Ln}{x : A}{\CM{\HC{B}}}$ and 
  $\Theta ; \Gamma ; \Delta \vdash m \sim m' : \HC{\ActR{?}{x : A}{B}}$.
\end{lemma}

\begin{lemma}\label[lemma]{lemma:erasure-inversion-explicit-recv}
  If $\Theta ; \Gamma ; \Delta \vdash \recvR{m} \sim \recvR{m'} : C$, then there exists $A, B$
  such that either 
  $C \simeq \CM{\SigR{\Ln}{x : A}{\CH{B}}}$ and 
  $\Theta ; \Gamma ; \Delta \vdash m \sim m' : \CH{\ActR{?}{x : A}{B}}$
  or
  $C \simeq \CM{\SigR{\Ln}{x : A}{\HC{B}}}$ and 
  $\Theta ; \Gamma ; \Delta \vdash m \sim m' : \HC{\ActR{!}{x : A}{B}}$.
\end{lemma}

\begin{lemma}\label[lemma]{lemma:erasure-inversion-implicit-send}
  If $\Theta ; \Gamma ; \Delta \vdash \sendI{m} \sim \sendI{m'} : C$, then there exists $A, B$
  such that either 
  $C \simeq \PiI{\Ln}{x : A}{\CM{\CH{B}}}$ and 
  $\Theta ; \Gamma ; \Delta \vdash m \sim m' : \CH{\ActI{!}{x : A}{B}}$
  or
  $C \simeq \PiI{\Ln}{x : A}{\CM{\HC{B}}}$ and 
  $\Theta ; \Gamma ; \Delta \vdash m \sim m' : \HC{\ActI{?}{x : A}{B}}$.
\end{lemma}

\begin{lemma}\label[lemma]{lemma:erasure-inversion-implicit-recv}
  If $\Theta ; \Gamma ; \Delta \vdash \recvI{m} \sim \recvI{m'} : C$, then there exists $A, B$
  such that either 
  $C \simeq \CM{\SigI{\Ln}{x : A}{\CH{B}}}$ and 
  $\Theta ; \Gamma ; \Delta \vdash m \sim m' : \CH{\ActI{?}{x : A}{B}}$
  or
  $C \simeq \CM{\SigI{\Ln}{x : A}{\HC{B}}}$ and 
  $\Theta ; \Gamma ; \Delta \vdash m \sim m' : \HC{\ActI{!}{x : A}{B}}$.
\end{lemma}

\begin{lemma}\label[lemma]{lemma:erasure-inversion-fork}
  If 
  ${
    \Theta ; \Gamma ; \Delta \vdash 
    \fork{x: A_1}{m_1} \sim 
    \fork{x: A_2}{m_2} : B
  }$, then exists $A$ such that
  $A_1 = \CH{A}$,
  $A_2 = \square$,
  $B \simeq \CM{\HC{A}}$,
  $\Gamma \vdash A : \Proto$, and
  ${\Theta ; \Gamma, x : \CH{A} ; \Delta, x \tL \CH{A} \vdash m_1 \sim m_2 : \CM{\unit}}$.
\end{lemma}

\begin{lemma}\label[lemma]{lemma:erasure-inversion-close}
  If ${\Theta ; \Gamma ; \Delta \vdash \close{m} \sim \close{m'} : A}$, then
  ${\Theta ; \Gamma ; \Delta \vdash m \sim m' : \CH{\End}}$ and $A \simeq \CM{\unit}$.
\end{lemma}

\begin{lemma}\label[lemma]{lemma:erasure-inversion-wait}
  If ${\Theta ; \Gamma ; \Delta \vdash \wait{m} \sim \wait{m'} : A}$, then
  ${\Theta ; \Gamma ; \Delta \vdash m \sim m' : \HC{\End}}$ and $A \simeq \CM{\unit}$.
\end{lemma}

\paragraph{\textbf{Erasure Congruence}}
Similarly to the proof of session fidelity, we prove that the erasure relation 
respects structural congruence.

\begin{lemma}\label[lemma]{lemma:erasure-congruence}
  If $\Theta \Vdash P \sim P'$ and $P \equiv Q$, then there exists $Q'$ such that
  $P' \equiv Q'$ and $\Theta \Vdash Q \sim Q'$.
\end{lemma}
\begin{proof}
  By induction on the derivation of $\Theta \Vdash P \sim P'$. We have three cases:

\noindent
\textbf{Case} \textsc{Expr}:
  \begin{mathpar}\small
  \inferrule[Expr]
  { \Theta ; \epsilon ; \epsilon \vdash m \sim m' : \CM{\unit} }
  { \Theta \Vdash \proc{m} \sim \proc{m'} }
  \end{mathpar}
  If suffices to consider the case where ${\proc{m} \equiv (\proc{m} \mid \proc{\return{\ii}})}$,
  i.e. $Q = \proc{m} \mid \proc{\return{\ii}}$.
  From erasure rules \textsc{Return} and \textsc{UnitVal}, we have
  \begin{align*}
    \epsilon ; \epsilon ; \epsilon \vdash \return{\ii} \sim \return{\ii} : \CM{\unit}
  \end{align*}
  which gives us $\epsilon \Vdash \proc{\return{\ii}} \sim \proc{\return{\ii}}$ by rule \textsc{Expr}.
  By rule \textsc{Par}, we have
  \begin{align*}
    \Theta \Vdash (\proc{m} \mid \proc{\return{\ii}}) \sim (\proc{m'} \mid \proc{\return{\ii}})
  \end{align*}
  which concludes this case by taking $Q' = \proc{m'} \mid \proc{\return{\ii}}$.

\noindent
\textbf{Case} \textsc{Par}:
  \begin{mathpar}\small
  \inferrule[Par]
  { \Theta_1 \vdash P \sim P' \\ \Theta_2 \vdash Q_0 \sim Q_0' }
  { \Theta_1 \dotcup \Theta_2 \Vdash (P \mid Q_0) \sim (P' \mid Q_0') }
  \end{mathpar}
  By case analysis on the congruence relation we have the following sub-cases:
  \begin{enumerate}
    \item $P \mid Q_1 \equiv Q_1 \mid P$
    \item $P \mid (Q_1 \mid Q_2) \equiv (P \mid Q_1) \mid Q_2$
    \item $P \mid \proc{\return{\ii}} \equiv P$
    \item $\scope{cd}{P} \mid Q_1 \equiv \scope{cd}{(P \mid Q_1)}$
  \end{enumerate}

  In sub-case (1), by \textsc{Par} we have 
  $\Theta_2 \dotcup \Theta_1 \Vdash (Q_1 \mid P) \sim (Q_1' \mid P')$.
  By the commutativity of $\dotcup$, we conclude this case by taking $Q' = Q_1' \mid P'$.

  In sub-case (2), by \textsc{Par} we have
  ${(\Theta_1 \dotcup \Theta_2) \dotcup \Theta_3 \Vdash 
    ((P \mid Q_1) \mid Q_2) \sim ((P' \mid Q_1') \mid Q_2')}$.
  By the associativity of $\dotcup$, we conclude this case by taking 
  $Q' = (P' \mid Q_1') \mid Q_2'$.

  In sub-case (3), we have $Q_0 = \proc{\return{\ii}}$.
  By assumption we have $\Theta_2 \Vdash \proc{\return{\ii}} \sim Q_1'$ for some $Q_1'$.
  By inversion on the erasure judgment, we have $\Theta_2 = \epsilon$.
  Thus we have $\Theta \vdash P \sim P'$ and we can conclude this case by taking $Q' = P'$.

  In sub-case (4), we have $\Theta_1 \Vdash \scope{cd}{P} \sim P'$ and 
  $\Theta_2 \Vdash Q_1 \sim Q_1'$.
  By inversion on $\Theta_1 \Vdash \scope{cd}{P} \sim P'$ we know there exists $P''$ and $A$ 
  such that $P' = \scope{cd}{P''}$ and $\Theta_1, c \tL \CH{A}, d \tL \HC{A} \Vdash P \sim P''$.
  By \textsc{Par} we have
  $\Theta_1, c \tL \CH{A}, d \tL \HC{A} \dotcup \Theta_2 \Vdash P \mid Q_1 \sim P'' \mid Q_1'$.
  Since $c, d \notin \Theta_2$, we have
  $\Theta_1 \dotcup \Theta_2, c \tL \CH{A}, d \tL \HC{A} \Vdash P \mid Q_1 \sim P'' \mid Q_1'$.
  By \textsc{Scope} we have
  $\Theta_1 \dotcup \Theta_2 \Vdash \scope{cd}{(P \mid Q_1)} \sim \scope{cd}{(P'' \mid Q_1')}$.
  We conclude this case by taking $Q' = \scope{cd}{(P'' \mid Q_1')}$.

\noindent
\textbf{Case} \textsc{Scope}:
  \begin{mathpar}\small
  \inferrule[Scope]
  { \Theta, c \tL \CH{A}, d \tL \HC{A} \vdash P \sim P' }
  { \Theta \Vdash \scope{cd}{P} \sim \scope{cd}{P'} }
  \end{mathpar}
  By case analysis on the congruence relation we have the following sub-cases:
  \begin{enumerate}
    \item $\scope{cd}{(P \mid Q_1)} \equiv \scope{cd}{P} \mid Q_1$
    \item $\scope{cd}{P} \equiv \scope{dc}{P}$
    \item $\scope{cd}{\scope{c'd'}{P}} \equiv \scope{c'd'}{\scope{cd}{P}}$
  \end{enumerate}

  In sub-case (1), we have $\Theta, c \tL \CH{A}, d \tL \HC{A} \vdash P \mid Q_1 \sim P' \mid Q_1'$
  by assumption. By inversion on the erasure judgment, we know there exists $\Theta_1$ and $\Theta_2$
  such that $\Theta_1 \dotcup \Theta_2 = \Theta$ and
  $\Theta_1, c \tL \CH{A}, d \tL \HC{A} \vdash P \sim P'$ and
  $\Theta_2 \vdash Q_1 \sim Q_1'$.
  Channels $c$ and $d$ must be distributed to the erasure relation of $P$
  as they are linear and do not appear in $Q_1$.
  Applying \textsc{Scope} we have
  $\Theta_1 \Vdash \scope{cd}{P} \sim \scope{cd}{P'}$.
  By \textsc{Par} we have
  $\Theta_1 \dotcup \Theta_2 \Vdash \scope{cd}{P} \mid Q_1 \sim \scope{cd}{P'} \mid Q_1'$.
  We conclude this case by taking $Q' = \scope{cd}{P'} \mid Q_1'$.

  In sub-case (2), we have
  $\Theta, d \tL \CH{A}, c \tL \HC{A} \vdash P \sim P'$ by assumption.
  By exchange of the context, we have
  $\Theta, c \tL \CH{A}, d \tL \HC{A} \vdash P \sim P'$.
  By \textsc{Scope} we have
  $\Theta \Vdash \scope{cd}{P} \sim \scope{dc}{P'}$.
  We conclude this case by taking $Q' = \scope{dc}{P'}$.

  In sub-case (3), we have
  $\Theta, c \tL \CH{A}, d \tL \HC{A}, c' \tL \CH{A'}, d' \tL \HC{A'} \vdash P \sim P'$
  by assumption.
  By exchange, we have 
  $\Theta, c' \tL \CH{A'}, d' \tL \HC{A'}, c \tL \CH{A}, d \tL \HC{A} \Vdash P \sim P'$.
  Applying \textsc{Scope} twice we have
  $\Theta \Vdash \scope{c'd'}{\scope{cd}{P}} \sim \scope{c'd'}{\scope{cd}{P'}}$.
  We conclude this case by taking $Q' = \scope{c'd'}{\scope{cd}{P'}}$.
\end{proof}

\paragraph{\textbf{Erasure Simulation}}
We can now prove the simulation lemma for programs.
\begin{theorem}[Term Simulation]\label[theorem]{lemma:erasure-program-simulation}
  Given $\Theta ; \epsilon ; \epsilon \vdash m \sim m' : A$ and $m \Leadsto n$, 
  there exists $n_0'$ such that 
  $m' \Leadsto n_0'$ and $\Theta ; \epsilon ; \epsilon \vdash n \sim n_0' : A$.
\end{theorem}
\begin{proof}
  The proof is similar to the proof of program subject reduction 
  (\Cref{theorem:program-subject-reduction}).
  We process by induction on the derivation of 
  $\Theta ; \epsilon ; \epsilon \vdash m \sim m' : A$.
  We show the following representative cases:

\noindent
\textbf{Case} (\textsc{Explicit-App}):
  \begin{mathpar}\small
  \inferrule[Explicit-App]
  { \Theta_1 ; \epsilon ; \epsilon \vdash m \sim m' : \PiR{t}{x : A}{B} \\
    \Theta_2 ; \epsilon ; \epsilon \vdash n \sim n' : A }
  { \Theta_1 \dotcup \Theta_2 ; \epsilon ; \epsilon \vdash 
    \appR{m}{n} \sim \appR{m'}{n'} : B[n/x] }
  \end{mathpar}
  From case analysis on the reduction we have three sub-cases: 
  (1) \textsc{Step-Explicit-App$_1$}, 
  (2) \textsc{Step-Explicit-App$_2$}, and 
  (3) \textsc{Step-Explicit-$\beta$}.

  In sub-case (1), we have $m \Leadsto m_1$.
  By the induction hypothesis, there exists $m_1'$ such that
  $\Theta_1 ; \epsilon ; \epsilon \vdash m_1 \sim m_1' : \PiR{t}{x : A}{B}$ and
  $m' \Leadsto m_1'$.
  By \textsc{Explicit-App} we have
  $\Theta_1 \dotcup \Theta_2 ; \epsilon ; \epsilon \vdash 
    \appR{m_1}{n} \sim \appR{m_1'}{n'} : B[n/x]$.
  We conclude this case by taking $n_0' = \appR{m_1'}{n'}$
  since by \textsc{Step-Explicit-App$_1$} we have
  $\appR{m'}{n'} \Leadsto \appR{m_1'}{n'}$.

  In sub-case (2), we have $n \Leadsto n_1$.
  By the induction hypothesis, there exists $n_1'$ such that
  $\Theta_2 ; \epsilon ; \epsilon \vdash n_1 \sim n_1' : A$ and
  $n' \Leadsto n_1'$.
  By \textsc{Explicit-App} we have
  $\Theta_1 \dotcup \Theta_2 ; \epsilon ; \epsilon \vdash 
    \appR{m}{n_1} \sim \appR{m'}{n_1'} : B[n_1/x]$.
  By \Cref{lemma:program-step-convertible}, we have $B[n/x] \simeq B[n_1/x]$
  which gives us
  $\Theta_1 \dotcup \Theta_2 ; \epsilon ; \epsilon \vdash 
    \appR{m}{n_1} \sim \appR{m'}{n_1'} : B[n/x]$ by \textsc{Conversion}.
  We conclude this case by taking $n_0' = \appR{m'}{n_1'}$
  since by \textsc{Step-Explicit-App$_2$} we have
  $\appR{m'}{n'} \Leadsto \appR{m'}{n_1'}$.

  In sub-case (3), we have $m = \lamR{t}{x : A'}{m_0}$ and $n = v$ for some value $v$.

  \noindent
  Applying \Cref{lemma:erasure-shape-value} on the typing judgment
  \begin{align*}
    \Theta_2 ; \epsilon ; \epsilon \vdash v \sim n' : A
  \end{align*}
  we know that $n' = v'$ for some value $v'$.
  By \Cref{lemma:erasure-typing} we have $\Theta_2 ; \epsilon ; \epsilon \vdash v : A$.
  By \Cref{lemma:erasure-type-validity} we have $\epsilon \vdash A : s$ for some $s$.
  By \Cref{lemma:program-context-bound} we have $\Theta_2 \triangleright s$.

  \noindent
  Applying \Cref{lemma:erasure-shape-explicit-lam} on the typing judgment
  \begin{align*}
    \Theta_1 ; \epsilon ; \epsilon \vdash 
    \lamR{t}{x : A'}{m_0} \sim m' : \PiR{t}{x : A}{B}
  \end{align*}
  we know that $m' = \lamR{t}{x : \square}{m_0'}$ for some $m_0'$.
  By \Cref{lemma:erasure-inversion-explicit-lam} there exists $r$ such that 
  \begin{align*}
    \Theta_1 ; x : A ; x :_r A \vdash m_0 \sim m_0' : B 
  \end{align*}

  \noindent
  By validity of context $x : A ; x :_r A \vdash$ we have $\epsilon \vdash A : r$.
  By \Cref{theorem:sort-uniqueness} we have $s = r$.

  \noindent
  By \Cref{lemma:erasure-explicit-substitution} we have
  \begin{align*}
    \Theta_1 \dotcup \Theta_2 ; \epsilon ; \epsilon \vdash m_0[v/x] \sim m_0'[v'/x] : B[v/x]
  \end{align*}

  \noindent
  We conclude this case by taking $n_0' = m_0'[v'/x]$ since we have
  $\appR{(\lamR{t}{x : \square}{m_0'})}{v'} \Leadsto m_0'[v'/x]$
  by \textsc{Step-Explicit-$\beta$}.

\noindent
\textbf{Case} (\textsc{Implicit-App}):
  \begin{mathpar}\small
  \inferrule[Implicit-App]
  { \Theta ; \epsilon ; \epsilon \vdash m \sim m' : \PiI{t}{x : A}{B} \\
    \epsilon \vdash n : A }
  { \Theta ; \epsilon ; \epsilon \vdash \appI{m}{n} \sim \appI{m'}{\square} : B[n/x] }
  \end{mathpar}
  From case analysis on the reduction we have two sub-cases: 
  (1) \textsc{Step-Implicit-App} and 
  (2) \textsc{Step-Implicit-$\beta$}.

  In sub-case (1), we have $m \Leadsto m_1$.
  By the induction hypothesis, there exists $m_1'$ such that
  $\Theta ; \epsilon ; \epsilon \vdash m_1 \sim m_1' : \PiI{t}{x : A}{B}$ and
  $m' \Leadsto m_1'$.
  By \textsc{Implicit-App} we have
  $\Theta ; \epsilon ; \epsilon \vdash 
    \appI{m_1}{n} \sim \appI{m_1'}{\square} : B[n/x]$.
  We conclude this case by taking $n_0' = \appI{m_1'}{\square}$
  since by \textsc{Step-Implicit-App} we have
  $\appI{m'}{\square} \Leadsto \appI{m_1'}{\square}$.

  In sub-case (2), we have $m = \lamI{t}{x : A'}{m_0}$ for some $m_0$.

  \noindent 
  Applying \Cref{lemma:erasure-shape-implicit-lam} on the typing judgment
  \begin{align*}
    \Theta ; \epsilon ; \epsilon \vdash 
    \lamI{t}{x : A'}{m_0} \sim m' : \PiI{t}{x : A}{B}
  \end{align*}
  we know that $m' = \lamI{t}{x : \square}{m_0'}$ for some $m_0'$.
  By \Cref{lemma:erasure-inversion-implicit-lam} there is
  \begin{align*}
    \Theta ; x : A ; \epsilon \vdash m_0 \sim m_0' : B 
  \end{align*}

  \noindent
  By \Cref{lemma:erasure-implicit-substitution} we have
  \begin{align*}
    \Theta ; \epsilon ; \epsilon \vdash m_0[n/x] \sim m_0'[\square/x] : B[n/x]
  \end{align*}

  \noindent
  We conclude this case by taking $n_0' = m_0'[\square/x]$ since we have
  $\appI{(\lamI{t}{x : \square}{m_0'})}{\square} \Leadsto m_0'[\square/x]$
  by \textsc{Step-Implicit-$\beta$}.
\end{proof}
\clearpage

\Cref{lemma:erasure-eval-context} allows us to reason about evaluation contexts
in the erasure relation.

\begin{lemma}\label[lemma]{lemma:erasure-eval-context}
  If $\Theta ; \epsilon ; \epsilon \vdash \mcM[m] \sim m_0' : \CM{B}$ then 
  there exists $\Theta_1, \Theta_2, \mcM', m', A$ such that
  $\Theta = \Theta_1 \dotcup \Theta_2$ and
  $\Theta_1 ; \epsilon ; \epsilon \vdash m \sim m' : \CM{A}$ and
  $m_0' = \mcM'[m']$.
  For any $\Theta_3, n, n'$ such that 
  $\Theta_3 ; \epsilon ; \epsilon \vdash n \sim n' : \CM{A}$ and 
  $\Theta_2 \dotcup \Theta_3$ is well-defined, we have
  $\Theta_2 \dotcup \Theta_3 ; \epsilon ; \epsilon \vdash \mcM[n] \sim \mcM'[n'] : \CM{B}$.
\end{lemma}
\begin{proof}
  By induction on the structure of the evaluation context $\mcM$.

\noindent
\textbf{Case} ($\mcM = [\cdot]$):
  Trivial by taking 
  $\mcM' = [\cdot]$, $m' = m_0'$, $\Theta_1 = \Theta$, $\Theta_2 = \epsilon$, and $A = B$.

\noindent
\textbf{Case} ($\mcM = \letin{x}{\mcN}{n}$):
  We have the typing judgment
  \begin{align*}
    \Theta ; \epsilon ; \epsilon \vdash \letin{x}{\mcN[m]}{n_1} \sim m_0' : \CM{B}
  \end{align*}

  \noindent
  By \Cref{lemma:erasure-shape-bind} we have
  $m_0' = \letin{x}{m_1'}{n_1'}$ and
  \begin{align*}
    \Theta ; \epsilon ; \epsilon \vdash 
    \letin{x}{\mcN[m]}{n_1} \sim \letin{x}{m_1'}{n_1'} : \CM{B}
  \end{align*}

  \noindent
  By \Cref{lemma:erasure-inversion-bind} there exists 
  $\Theta_1', \Theta_2', A'$ such that
  \begin{align*}
    &\Theta_1' ; \epsilon ; \epsilon \vdash \mcN[m] \sim m_1' : \CM{A'} \\
    &\Theta_2' ; x : A' ; x \tL A' \vdash n_1 \sim n_1' : \CM{B} \\
    &\Theta = \Theta_1' \dotcup \Theta_2' \quad x \notin \FV{B}
  \end{align*}

  \noindent
  By the induction hypothesis, there exists $\Theta_{11}', \Theta_{12}', \mcN', m', A''$ such that
  \begin{align*}
    &\Theta_{11}' ; \epsilon ; \epsilon \vdash m \sim m' : \CM{A''} \\
    &\Theta_1' = \Theta_{11}' \dotcup \Theta_{12}' \\
    &m_1' = \mcN'[m']
  \end{align*}
  and for any $\Theta_3, n, n'$ such that 
  $\Theta_3 ; \epsilon ; \epsilon \vdash n \sim n' : \CM{A''}$ and 
  $\Theta_{12}' \dotcup \Theta_3$ is well-defined, we have
  \begin{align*}
    \Theta_{12}' \dotcup \Theta_3 ; \epsilon ; \epsilon \vdash \mcN[n] \sim \mcN'[n'] : \CM{A'}
  \end{align*}

  \noindent
  By choosing $\Theta_1 = \Theta_{11}'$, $\Theta_2 = \Theta_{12}' \dotcup \Theta_2'$,
  $\mcM' = \letin{x}{\mcN'}{n_1'}$, and $A = A''$, we have
  \begin{align*}
    &\Theta_1 ; \epsilon ; \epsilon \vdash m \sim m' : \CM{A''} \\
    &\Theta = \Theta_1 \dotcup \Theta_2 \\
    &m_0' = \letin{x}{\mcN'[m']}{n_1'} = \mcM'[m']
  \end{align*}

  \noindent
  For any $\Theta_3, n, n'$ such that 
  $\Theta_3 ; \epsilon ; \epsilon \vdash n \sim n' : \CM{A''}$ and 
  $\Theta_2 \dotcup \Theta_3$ is well-defined, we know that
  $\Theta_{12}' \dotcup \Theta_3$ is well-defined, which means that
  \begin{align*}
    \Theta_{12}' \dotcup \Theta_3 ; \epsilon ; \epsilon \vdash \mcN[n] \sim \mcN'[n'] : \CM{A'}
  \end{align*}
  Now applying \textsc{Bind} we have
  \begin{align*}
    \Theta_2 \dotcup \Theta_3 ; \epsilon ; \epsilon \vdash 
    \letin{x}{\mcN[n]}{n_1} \sim 
    \letin{x}{\mcN'[n']}{n_1'} : \CM{B}
  \end{align*}
  which concludes this case.
\end{proof}
\clearpage

We can now prove the main simulation theorem for processes.
\begin{theorem}[Process Simulation]\label[theorem]{theorem:erasure-simulation}
  Given $\Theta \Vdash P \sim P'$ and reduction $P \Rrightarrow Q_0$, there exists $Q_0'$ such that
  $P' \Rrightarrow Q_0'$ and $\Theta \Vdash Q_0 \sim Q_0'$.
\end{theorem}
\begin{proof}
  Similar to the proof of session fidelity (\Cref{theorem:session-fidelity}).
  We proceed by induction on the derivation of $P \Rrightarrow Q$.

\noindent
\textbf{Case} (\textsc{Proc-Fork}):
\begin{mathpar}\small
  \proc{
    \mcN[\fork{x : A}{m}]
  }
  \Rrightarrow
  \scope{cd}{(\proc{\mcN[\return{c}]} \mid \proc{m[d/x]})} 
\end{mathpar}

\noindent
By assumption, we have $\Theta \Vdash \proc{\mcN[\fork{x : A}{m}]} \sim P'$.
By inversion on the erasure judgment we know there exists $n'$ such that
$P' = \proc{n'}$ and
\begin{align*}
  &\Theta ; \epsilon ; \epsilon \vdash \mcN[\fork{x : A}{m}] \sim n' : \CM{\unit}
\end{align*}

\noindent
By \Cref{lemma:erasure-eval-context} there exists $\Theta_1, \Theta_2, \mcN', n'', B$ such that
\begin{align*}
  &\Theta_1 ; \epsilon ; \epsilon \vdash \fork{x : A}{m} \sim n'' : \CM{B} \\
  &\Theta = \Theta_1 \dotcup \Theta_2 \\
  &n' = \mcN'[n'']
\end{align*}

\noindent
By \Cref{lemma:erasure-shape-fork} we know that $n'' = \fork{x : A'}{m'}$ for some $A'$ and $m'$.

\noindent
By \Cref{lemma:erasure-inversion-fork} there exists $A''$ such that
\begin{align*}
  &\Theta_1 ; x : \CH{A''} ; x \tL \CH{A''} \vdash m \sim m' : \CM{\unit} \\
  &\epsilon \vdash A'' : \Proto \\
  &A' = \square \quad A = \CH{A''} \quad \CM{B} \simeq \CM{\HC{A''}}
\end{align*}

\noindent
By \textsc{Channel-HC} we have 
$d \tL \HC{A''} ; \epsilon ; \epsilon \vdash c \sim c : \HC{A''}$.

\noindent
By \textsc{Channel-CH} we have 
$c \tL \CH{A''} ; \epsilon ; \epsilon \vdash d \sim d : \CH{A''}$.

\noindent
By \textsc{Return} we have 
$d \tL \HC{A''} ; \epsilon ; \epsilon \vdash \return{c} \sim \return{c} : \CM{\HC{A''}}$.
Applying the evaluation contexts $\mcN$ and $\mcN'$ we have
\begin{align*}
  \Theta_2, c \tL \HC{A''} ; \epsilon ; \epsilon \vdash 
  \mcN[\return{c}] \sim \mcN'[\return{c}] : \CM{\unit}
\end{align*}

\noindent
By \Cref{lemma:erasure-explicit-substitution} we have
\begin{align*}
  \Theta_1, d \tL \CH{A''} ; \epsilon ; \epsilon \vdash m[d/x] \sim m'[d/x] : \CM{\unit}
\end{align*}

\noindent
By \textsc{Expr} and \textsc{Par} we have
\begin{align*}
  \Theta_1 \dotcup \Theta_2, d \tL \CH{A''}, c \tL \HC{A''} \Vdash
  (\proc{\mcN[\return{c}]} \mid \proc{m[d/x]}) \sim 
  (\proc{\mcN'[\return{c}]} \mid \proc{m'[d/x]})
\end{align*}

\noindent
By \textsc{Scope} we have
\begin{align*}
  \Theta \Vdash 
  \scope{cd}{(\proc{\mcN[\return{c}]} \mid \proc{m[d/x]})} \sim
  \scope{cd}{(\proc{\mcN'[\return{c}]} \mid \proc{m'[d/x]})}
\end{align*}
which concludes this case by taking
$Q_0' = \scope{cd}{(\proc{\mcN'[\return{c}]} \mid \proc{m'[d/x]})}$
and noting that
\begin{align*}
  \proc{\mcN'[\fork{x : \square}{m'}]} \Rrightarrow
  \scope{cd}{(\proc{\mcN'[\return{c}]} \mid \proc{m'[d/x]})}
\end{align*}

\noindent
\textbf{Case} (\textsc{Proc-End}):
  \begin{mathpar}\small
  \scope{cd}{(
    \proc{\mcM[\close{c}]} 
    \mid 
    \proc{\mcN[\wait{d}]}
  )}
  \Rrightarrow 
  \proc{\mcM[\return{\ii}]} \mid \proc{\mcN[\return{\ii}]} 
  \end{mathpar}
  By assumption, we have
  $\Theta \Vdash 
  \scope{cd}{(
    \proc{\mcM[\close{c}]} 
    \mid 
    \proc{\mcN[\wait{d}]}
  )} \sim P'$.
  By inversion on the erasure judgment we know there exists $A$ and $P''$ such that
  $P' = \scope{cd}{P''}$ and
  \begin{align*}
    &\Theta, c \tL \CH{A}, d \tL \HC{A} \Vdash 
    (\proc{\mcM[\close{c}]} \mid \proc{\mcN[\wait{d}]}) \sim P''
  \end{align*}

\noindent
By inversion on the erasure judgment we know there exists 
$\Theta_1, \Theta_2, m_1', n_1'$ such that
$\Theta = \Theta_1 \dotcup \Theta_2$ and
$P'' = \proc{m_1'} \mid \proc{n_1'}$ and
\begin{align*}
  &\Theta_1, c \tL \CH{A} ; \epsilon ; \epsilon \vdash \mcM[\close{c}] \sim m_1' : \CM{\unit} \\
  &\Theta_2, d \tL \HC{A} ; \epsilon ; \epsilon \vdash \mcN[\wait{d}] \sim n_1' : \CM{\unit}
\end{align*}

\noindent
By \Cref{lemma:erasure-eval-context} there exists
$\Theta_{11}, \Theta_{12}, \mcM', m_2', B$ such that
\begin{align*}
  &\Theta_{11}, c \tL \CH{A} ; \epsilon ; \epsilon \vdash \close{c} \sim m_2' : \CM{B} \\
  &\Theta_1 = \Theta_{11} \dotcup \Theta_{12} \\
  &m_1' = \mcM'[m_2']
\end{align*}

\noindent
By \Cref{lemma:erasure-eval-context} there exists
$\Theta_{21}, \Theta_{22}, \mcN', n_2', C$ such that
\begin{align*}
  &\Theta_{21}, d \tL \HC{A} ; \epsilon ; \epsilon \vdash \wait{d} \sim n_2' : \CM{C} \\
  &\Theta_2 = \Theta_{21} \dotcup \Theta_{22} \\
  &n_1' = \mcN'[n_2']
\end{align*}

\noindent
By \Cref{lemma:erasure-shape-close} we know that $m_2' = \close{c}$ and
by \Cref{lemma:erasure-shape-wait} we know that $n_2' = \wait{d}$.

\noindent
Now following the same reasoning as in the proof of session fidelity
(\Cref{theorem:session-fidelity}) for the \textsc{Proc-End} case, we have
\begin{align*}
  &\Theta_{12} \Vdash \proc{\mcM[\return{\ii}]} \sim \proc{\mcM'[\return{\ii}]} : \CM{\unit} \\
  &\Theta_{22} \Vdash \proc{\mcN[\return{\ii}]} \sim \proc{\mcN'[\return{\ii}]} : \CM{\unit} \\
  &\Theta_{11} = \Theta_{21} = \epsilon
\end{align*}

\noindent
By \textsc{Par} we have
\begin{align*}
  \Theta_{12} \dotcup \Theta_{22} \Vdash
  (\proc{\mcM[\return{\ii}]} \mid \proc{\mcN[\return{\ii}]}) \sim
  (\proc{\mcM'[\return{\ii}]} \mid \proc{\mcN'[\return{\ii}]})
\end{align*}
which concludes this case by taking
$Q_0' = \proc{\mcM'[\return{\ii}]} \mid \proc{\mcN'[\return{\ii}]}$
and noting that
\begin{align*}
  \scope{cd}{(
    \proc{\mcM'[\close{c}]} 
    \mid 
    \proc{\mcN'[\wait{d}]}
  )} \Rrightarrow 
  \proc{\mcM'[\return{\ii}]} \mid \proc{\mcN'[\return{\ii}]}
\end{align*}

\noindent
\textbf{Case} (\textsc{Proc-Com}):
  \begin{mathpar}\small
  \scope{cd}{(
    \proc{\mcM[\appR{\sendR{c}}{v}]} 
    \mid 
    \proc{\mcN[\recvR{d}]}
  )}
  \Rrightarrow 
  \scope{cd}{(
    \proc{\mcM[\return{c}]} 
    \mid 
    \proc{\mcN[\return{\pairR{v}{d}{\Ln}}]}
  )}
  \end{mathpar}
  By assumption, we have
  $\Theta \Vdash 
  \scope{cd}{(
    \proc{\mcM[\appR{\sendR{c}}{v}]} 
    \mid 
    \proc{\mcN[\recvR{d}]}
  )} \sim P'$.
  By inversion on the erasure judgment we know there exists $\kappa$, $A$ and $P''$ such that
  $P' = \scope{cd}{P''}$ and
  \begin{align*}
    &\Theta, c \tL \ch{\kappa}{A}, d \tL \neg{\ch{\kappa}{A}} \Vdash 
    (\proc{\mcM[\appR{\sendR{c}}{v}]} \mid \proc{\mcN[\recvR{d}]}) \sim P''
  \end{align*}

\noindent
By inversion on the erasure judgment we know there exists
$\Theta_1, \Theta_2, m_1', n_1'$ such that
$\Theta = \Theta_1 \dotcup \Theta_2$ and
$P'' = \proc{m_1'} \mid \proc{n_1'}$ and
\begin{align*}
  &\Theta_1, c \tL \ch{\kappa}{A} ; \epsilon ; \epsilon \vdash 
  \mcM[\appR{\sendR{c}}{v}] \sim m_1' : \CM{\unit} \\
  &\Theta_2, d \tL \neg{\ch{\kappa}{A}} ; \epsilon ; \epsilon \vdash 
  \mcN[\recvR{d}] \sim n_1' : \CM{A}
\end{align*}

\noindent
By \Cref{lemma:erasure-eval-context} there exists
$\Theta_{11}, \Theta_{12}, \mcM', m_2', B$ such that
\begin{align*}
  &\Theta_{11}, c \tL \ch{\kappa}{A} ; \epsilon ; \epsilon \vdash 
   \appR{\sendR{c}}{v} \sim m_2' : \CM{B} \\
  &\Theta_1 = \Theta_{11} \dotcup \Theta_{12} \\
  &m_1' = \mcM'[m_2']
\end{align*}

\noindent
By \Cref{lemma:erasure-eval-context} there exists
$\Theta_{21}, \Theta_{22}, \mcN', n_2', C$ such that
\begin{align*}
  &\Theta_{21}, d \tL \neg{\ch{\kappa}{A}} ; \epsilon ; \epsilon \vdash 
   \recvR{d} \sim n_2' : \CM{C} \\
  &\Theta_2 = \Theta_{21} \dotcup \Theta_{22} \\
  &n_1' = \mcN'[n_2']
\end{align*}

\noindent
By 
\Cref{lemma:erasure-shape-explicit-app},
\Cref{lemma:erasure-shape-explicit-send} and
\Cref{lemma:erasure-shape-channel} we know that 
$m_2' = \appR{\sendR{c}}{m_3'}$ for some $m_3'$.

\noindent
By \Cref{lemma:erasure-inversion-explicit-app} we know there exists 
$\Theta_{111}, \Theta_{112}, A'$ and $r$ such that
\begin{align*}
  &\Theta_{111}, c \tL \ch{\kappa}{A} ; \epsilon ; \epsilon \vdash 
  \sendR{c} \sim \sendR{c} : \PiR{t}{x : A'}{\CM{B}} \\
  &\Theta_{112} ; \epsilon ; \epsilon \vdash v \sim m_3' : A'
\end{align*}
By \Cref{lemma:erasure-shape-value} we know that $m_3' = v'$ for some value $v'$.

\noindent
By \Cref{lemma:erasure-shape-explicit-recv} we know that $n_2' = \recvR{d}$.

\noindent
Now following the same reasoning as in the proof of session fidelity
(\Cref{theorem:session-fidelity}) for the \textsc{Proc-Com} case, we know there exists $B'$ such that
\begin{align*}
  &\Theta_{12}, c \tL \ch{\kappa}{B'} \Vdash 
   \proc{\mcM[\return{c}]} \sim \proc{\mcM'[\return{c}]} \\
  &\Theta_{22} \dotcup \Theta_{112}, d \tL \neg{\ch{\kappa}{B'}} \Vdash 
    \proc{\mcN[\return{\pairR{v}{d}{\Ln}}]} \sim \proc{\mcN'[\return{\pairR{v'}{d}{\Ln}}]} \\
  &\Theta_{111} = \Theta_{21} = \epsilon
\end{align*}

\noindent
By \textsc{Par} and \textsc{Scope} we have
\begin{align*}
  &\Theta_{12} \dotcup \Theta_{22} \dotcup \Theta_{112} \Vdash \\
  &\qquad\scope{cd}{(
    \proc{\mcM[\return{c}]} 
    \mid 
    \proc{\mcN[\return{\pairR{v}{d}{\Ln}}]}
  )} \sim
  \scope{cd}{(
    \proc{\mcM'[\return{c}]} 
    \mid 
    \proc{\mcN'[\return{\pairR{v'}{d}{\Ln}}]}
  )}
\end{align*}
which concludes this case by taking
$Q_0' = \scope{cd}{(
    \proc{\mcM'[\return{c}]} 
    \mid 
    \proc{\mcN'[\return{\pairR{v'}{d}{\Ln}}]}
  )}$
and noting
\begin{align*}
  \scope{cd}{(
    \proc{\mcM'[\appR{\sendR{c}}{v'}]} 
    \mid 
    \proc{\mcN'[\recvR{d}]}
  )} \Rrightarrow 
  \scope{cd}{(
    \proc{\mcM'[\return{c}]} 
    \mid 
    \proc{\mcN'[\return{\pairR{v'}{d}{\Ln}}]}
  )}
\end{align*}

\noindent
\textbf{Case} (\textsc{Proc-\underline{Com}}):
  Similar to the \textsc{Proc-Com} case.

\noindent
\textbf{Case} (\textsc{Proc-Expr}):
  \begin{mathpar}
  \inferrule
  { m \Leadsto n }
  { \proc{m} \Rrightarrow \proc{n} }
  \textsc{(Proc-Expr)}
  \end{mathpar}
  By assumption, we have $\Theta \Vdash \proc{m} \sim P'$.
  By inversion on the erasure judgment we know there exists $m'$ such that
  $\Theta ; \epsilon ; \epsilon \vdash m \sim m' : \CM{A}$ and
  $P' = \proc{m'}$.

\noindent
By \Cref{lemma:erasure-program-simulation} there exists $n'$ such that
$\Theta ; \epsilon ; \epsilon \vdash n \sim n' : \CM{A}$ and
$m' \Leadsto n'$.

\noindent
By \textsc{Expr} we have $\Theta \Vdash \proc{n} \sim \proc{n'}$.
We conclude this case by taking $Q_0' = \proc{n'}$ and noting that
$\proc{m'} \Rrightarrow \proc{n'}$ by \textsc{Proc-Expr}.

\noindent
\textbf{Case} (\textsc{Proc-Par}):
  \begin{mathpar}
  \inferrule
  { P \Rrightarrow Q }
  { O \mid P \Rrightarrow O \mid Q }
  \textsc{(Proc-Par)}
  \end{mathpar}
  By assumption, we have $\Theta \Vdash (O \mid P) \sim P'$.
  By inversion on the erasure judgment we know there exists $\Theta_1, \Theta_2, O', P''$ such that
  $\Theta = \Theta_1 \dotcup \Theta_2$, $P' = O' \mid P''$ and
  \begin{align*}
    &\Theta_1 \Vdash O \sim O' \\
    &\Theta_2 \Vdash P \sim P''
  \end{align*}

  \noindent
  By the induction hypothesis, there exists $Q'$ such that
  $P'' \Rrightarrow Q'$ and $\Theta_2 \Vdash Q \sim Q'$.
  By \textsc{Par} we have
  $\Theta \Vdash (O \mid Q) \sim (O' \mid Q')$.
  We conclude this case by taking $Q_0' = O' \mid Q'$ and noting that
  $O' \mid P'' \Rrightarrow O' \mid Q'$ by \textsc{Proc-Par}.

\noindent
\textbf{Case} (\textsc{Proc-Scope}):
  \begin{mathpar}
  \inferrule
  { P \Rrightarrow Q }
  { \scope{cd}{P} \Rrightarrow \scope{cd}{Q} }
  \textsc{(Proc-Scope)}
  \end{mathpar}
  By assumption, we have $\Theta \Vdash \scope{cd}{P} \sim P'$.
  By inversion on the erasure judgment we know there exists $P''$ such that
  $P' = \scope{cd}{P''}$ and
  $\Theta, c \tL \CH{A}, d \tL \HC{A} \Vdash P \sim P''$ for some $A$.

  \noindent
  By the induction hypothesis, there exists $Q'$ such that
  $P'' \Rrightarrow Q'$ and $\Theta, c \tL \CH{A}, d \tL \HC{A} \Vdash Q \sim Q'$.
  By \textsc{Scope} we have
  $\Theta \Vdash \scope{cd}{Q} \sim \scope{cd}{Q'}$.
  We conclude this case by taking $Q_0' = \scope{cd}{Q'}$ and noting that
  $\scope{cd}{P''} \Rrightarrow \scope{cd}{Q'}$ by \textsc{Proc-Scope}.

\noindent
\textbf{Case} (\textsc{Proc-Congr}):
  \begin{mathpar}
  \inferrule
  { P \equiv P_1 \\ P_1 \Rrightarrow Q_1 \\ Q_1 \equiv Q }
  { P \Rrightarrow Q }
  \end{mathpar}
  By assumption, we have $\Theta \Vdash P \sim P'$.
  By \Cref{lemma:erasure-congruence} there exists $P_1'$ such that
  $\Theta \Vdash P_1 \sim P_1'$ and $P' \equiv P_1'$.
  By the induction hypothesis, there exists $Q_1'$ such that
  $P_1' \Rrightarrow Q_1'$ and $\Theta \Vdash Q_1 \sim Q_1'$.
  By \Cref{lemma:erasure-congruence} there exists $Q_0'$ such that
  $\Theta \Vdash Q \sim Q_0'$ and $Q_1' \equiv Q_0'$.
  We conclude this case by noting that $P' \Rrightarrow Q_0'$
  by \textsc{Proc-Congr}.

\end{proof}
\end{appendices}

\end{document}